\documentclass[11pt,a4paper]{book}

\usepackage[T1]{fontenc}		
\usepackage{amsmath,esint}				
\usepackage{amssymb}				
\usepackage{array}				
\usepackage{multicol}				
\usepackage{xspace}				
\usepackage[pdftex]{graphicx}			
\usepackage{indentfirst} 			
\usepackage{latexsym}				
\usepackage{setspace}				
\usepackage{makeidx}				
\usepackage{vmargin}				

\usepackage{amscd}
\usepackage{xspace}
\usepackage{verbatim}
\newtheorem{defn}{Definition}[section]
\newtheorem{prop}{Proposition}[section]
\newtheorem{lem}{Lemma}[section]
\newtheorem{thm}{Theorem}[section]
\newtheorem{cor}{Corollary}[section]

\newtheorem{example}{Example}
\newtheorem{conj}{Conjecture}
\newtheorem{rem}{Remark}

\numberwithin{equation}{section}

\newenvironment{proof}[1][Proof]{\textbf{#1.} }{\hfill\rule{0.5em}{0.5em}}
{\catcode`\@=11\global\let\AddToReset=\@addtoreset
\AddToReset{equation}{section}

\usepackage{fancyhdr}
\fancyhead{}
\fancyhead[L]{\rightmark}			
\fancyfoot[C]{\thepage}	
\setcounter{secnumdepth}{3}

\newcommand{\nc}{\newcommand}

\nc{\cW}{{\cal W}}
\nc{\W}{{\mathcal W}}
\nc{\cWN}{\stackrel{\footnotesize[N]}{\cW}}
\nc{\NWex}[1]{\stackrel{\footnotesize[#1]}{\cW}}

\makeindex
\begin{document}
\author{
\textbf{Thi-Thao VU}\\
~\\ ~\\
Laboratoire de Math\'ematiques et Physique Th\'eorique, Universit\'e Fra\c cois Rabelai, Tours, \textsc{France}\\
~\\
\textit{Thi-thao.Vu@lmpt.univ-tours.fr}\\
}

\title{
\textbf{\huge THE HIGHER ORDER $q$-DOLAN-GRADY RELATIONS AND QUANTUM INTEGRABLE SYSTEMS}
\\
Version 
}
\chapter*{\begin{center}\huge{THE HIGHER ORDER $q$-DOLAN-GRADY RELATIONS AND QUANTUM INTEGRABLE SYSTEMS} \end{center}}
\begin{center}
TH\`ESE\\
LMPT, Universit\'e Fran\c cois-Rabelais, Tours, \textsc{France}\\
\vspace{1 cm}
\large{\textbf{Thi Thao Vu}}\\

\end{center}
\vspace{0.7cm}
		 
\begin{tabular}{lll}
\textbf{\large TH\`ESE dirig\'ee par : }  \vspace{0.01cm}\\
\textbf{M. P. BASEILHAC} & C.R. CNRS, Universit\'e Fran\c cois-Rabelais de Tours\\
	
\vspace{0.20cm} \\

\textbf{\large RAPPORTEURS :}  \vspace{0.01cm}\\
	\textbf{M. H. KONNO}& Professeur, Universit\'e des Sciences et Technologies Marines de Tokyo\\
	\textbf{M. P. TERWILLIGER}& Professeur, Universit\'e du Wisconsin-Madison\\
\vspace{0.20cm} \\

\textbf{\large JURY : }  \vspace{0.01cm} \\
	\textbf{M. P. BASEILHAC} & C.R. CNRS, Universit\'e Fran\c cois-Rabelais de Tours\\
	\textbf{M. C. LECOUVEY} & Professeur, Universit\'e Fran\c cois-Rabelais de Tours \\
	\textbf{M. V. ROUBTSOV} & Professeur, Universit\'e d'Angers \\
	\textbf{M. H. SALEUR} & Chercheur CEA, Institut de Physique Th\'eorique, CEA Saclay\\
	\textbf{M. P. TERWILLIGER}& Professeur, Universit\'e du Wisconsin-Madison
\end{tabular}
\newpage
\markright{\MakeUppercase{Abstract}}

\begin{center}  {\huge \bf Abstract} \end{center} \medskip
\paragraph{}
In this thesis, the connection between recently introduced algebraic structures (tridiagonal algebra, $q$-Onsager algebra, generalized $q-$Onsager algebras), related  representation theory (tridiagonal pair, Leonard pair, orthogonal polynomials), some properties of these algebras and the analysis of related quantum integrable models on the lattice (the $XXZ$ open spin chain at roots of unity) is first reviewed. Then, the main results of the thesis are described: (i) for the class of $q-$Onsager algebras associated with $\widehat{sl_2}$ and ADE type simply-laced affine Lie algebras, higher order analogs of Lusztig's relations are conjectured and proven in various cases, (ii) for the open $XXZ$ spin chain at roots of unity, new elements (that are divided polynomials of $q-$Onsager generators) are introduced and some of their properties are studied. These two elements together with the two basic elements of the $q-$Onsager algebra generate a new algebra, which can be understood as an analog of Lusztig's quantum group for the $q-$Onsager algebra. Some perspectives are presented.\\

\paragraph{Keywords : }

Tridiagonal algebra; Tridiagonal pair; $q$-Onsager algebra; Generalized $q$-Onsager algebra; $XXZ$ open spin chain; root of unity.

\chapter*{Notations}
\markright{\MakeUppercase{Notations}}
Throughout this thesis, we use the following notations:
\begin{enumerate}
\item Let $A, B$ denote generators, then
\begin{eqnarray*}
[A,B]&=&AB-BA.\\
~[A,B]_q&=&qAB-q^{-1}BA.
\end{eqnarray*}
\item Let $n, m$ be integers, then
\begin{eqnarray*}
\left[ \begin{array}{c}
n \\

m 
\end{array}\right]_q
&=&\frac{[n]_q!}{[m]_q!\,[n-m]_q!}\ , \qquad
[n]_q!=\prod_{l=1}^n[l]_q\ ,\qquad
[n]_q=\frac{q^n-q^{-n}}{q-q^{-1}}, \quad [0]_q=1 \ .
\end{eqnarray*}
\item $\{ x\}$ denotes the integer part of $x$. Let $j,m,n$ be integers, write $j=\overline{m,n}$  for $j=m,m+1,...,n-1,n$.
\item The Pauli matrices $\sigma_1, \sigma_2, \sigma_z, \sigma_{\pm}$:
\begin{eqnarray*}
\sigma_1&=&\left(\begin{array}{cc}
0&1\\
1&0
\end{array}\right),~
\sigma_2=
\left(\begin{array}{cc}
0&-i\\
i&0
\end{array}\right),~
\sigma_z=
\left(\begin{array}{cc}
1&0\\
0&-1
\end{array}\right),~\\
\sigma_+&=&\left(\begin{array}{cc}
0&1\\
0&0
\end{array}\right),~
\sigma_-=\left(\begin{array}{cc}
0&0\\
1&0
\end{array}\right).\nonumber
\end{eqnarray*}.
\end{enumerate}
 
\tableofcontents
\newpage

\thispagestyle{empty}
\pagenumbering{arabic} \setcounter{page}{1} 
\chapter*{Introduction}
\addcontentsline{toc}{chapter}{Introduction}
\markright{\MakeUppercase{Introduction}}
In the literature, the Onsager algebra and the Dolan-Grady relations first appeared in study of integrable systems (the $XY$, the Ising models,\dots) \cite{O, DG, B.D, B.D1}. Later on, they appeared in the context of mathematics in relation with certain subalgebras of $\widehat{sl_2}$ \cite{DR}. From 2003, a $q-$deformed analog of the Dolan-Grady relations appeared in the context of mathematics as a special case of tridiagonal algebras  [Terwilliger et al.]. Almost simultaneously, the $q-$Dolan-Grady relations appeared in the context of quantum integrable systems on the lattice and continuum: the $q-$Onsager
 algebra was defined (which $q-$Dolan-Grady relations are the defining relations) in relation with the quantum reflection equation, as an algebra generating a large class of quantum integrable systems on the lattice or continuum [Baseilhac et al.].\vspace{1mm}

There is now a rather vast literature on the subject of tridiagonal algebras \cite{PTIII}, the representation theory of tridiagonal pairs \cite{TD00} and Leonard pairs \cite{Ter01}, the
$q-$Onsager algebra and its generalizations \cite{T05, B1}, the connections with coideal subalgebras of $U_q(\widehat{sl_2})$ \cite{BB3} and with a new infinite dimensional algebra called ${\cal A}_q$ \cite{BasS}. In the context of mathematical physics, these structures and the explicit analysis of some of their properties lead to several new exact non-perturbative results for the open XXZ spin chain \cite{BK3}, for the half-infinite XXZ spin chain \cite{BB3, BK14}, for the open affine Toda field theories \cite{BB1, BF}. From a general point of view, a new approach called `$q-$Onsager approach' has emerged as an alternative to existing ones in quantum integrable systems (Bethe ansatz \cite{Bet, FST}, separation of variable \cite{Sk1}, $q-$vertex operators \cite{JM, JKKKMW}). Since 2007, this approach has been currently developed in different directions. \vspace{1mm}

In this thesis, we explore one direction which overlaps between  mathematics and physics. Namely, we investigate in detail some properties of the $q-$Onsager algebras (in particular the existence and explicit construction of higher order relations between monomials of the fundamental generators) which will find application in the analysis of the open XXZ spin chain at roots of unity (characterization of the symmetry of the Hamiltonian at roots of unity). 
At the moment, the results of this thesis have been published in two articles \cite{BV1, BV2}. There is another article in preparation \cite{BGSV}.\vspace{1mm} 

The manuscript of the thesis is divided into three main Chapters. 
\\

\textbf{Chapter 1.} \textit{We summarize without proofs the relevant material on tridiagonal algebras, the $q$-Onsager algebra and some aspects of its representation theory: Leonard pairs, tridiagonal pairs and orthogonal polynomials.} \\

In the first part, tridiagonal algebras are defined by generators and relations. Several special cases corresponding to particular parameter sequences of the tridiagonal relations are mentioned such as the $q$-Serre relations or the Dolan-Grady relations.\vspace{1mm}

The $q$-Onsager algebra is introduced in a second part. Its defining relations are the $q-$Dolan-Grady relations: these are $\rho_i$-deformed analogues of the $q$-Serre relations, and correspond to a special parameter sequence of the tridiagonal algebra. In connection with the quantum affine algebra $U_q(\widehat{sl_2})$ and the $U_q(sl_2)$-loop algebra, we thus recall homomorphisms from the $q$-Onsager algebra to these algebras. Finally, we recall the reflection equation algebra and indicate its relation with the $q$-Onsager algebra: the reflection equation algebra is defined by generators which are entries of the solution of the ``RKRK'' equations for the $U_q(\widehat{sl_2})$ $R$-matrix. The isomorphisms between the reflection equation algebra, the current algebra $O_q(\widehat{sl_2})$, and the infinite dimensional algebra $\mathcal{A}_q$ generated by $\{\mathcal{W}_{-k},  \mathcal{W}_{k+1}, \mathcal{G}_{k+1}, \tilde{\mathcal{G}}_{k+1}|k \in \mathbb{Z}_+\}$ are recalled. We also recall the construction of a coaction map for the $q$-Onsager algebra and the defining relations of the $K$-matrix as an intertwiner of irreducible finite dimensional representations of the $q$-Onsager algebra. Thus, a quotient of the $q$-Onsager algebra is isomorphic to a quotient of the reflection equation algebra.\vspace{1mm}

The last part recalls some aspects of the representation theory of tridiagonal algebras (including the case of the $q-$Onsager algebra), in particular the results of Terwilliger et al. about irreducible finite dimensional representations and the concept of tridiagonal pairs. For convenience, Leonard pairs (a subclass of the tridiagonal pairs) are introduced first. We recall the notion of Leonard pair, Leonard system as well as modification of a given Leonard pair in several ways. The relation between the Leonard pair and the tridiagonal algebra is also clarified. Namely, there exists a scalar sequence such that the Leonard pair satisfies the corresponding Askey-Wilson relations. Inversely, a pair of linear transformations satisfying the Askey-Wilson relations allows to define a Leonard pair under certain conditions. One more important result is the classification of Leonard pairs, it is asserted that a sequence of scalars satisfying conditions $(i)-(v)$ in Theorem (\ref{clas}) is necessary and sufficient to obtain a Leonard pair. In addition, we show that Leonard pairs arise naturally in relation with the Lie algebra $sl_2$ and $U_q(sl_2)$. A more general object than the concept of Leonard pair, namely the concept of tridiagonal pair, is also introduced and described in some details. We first recall the concept of a tridiagonal pair, of a tridiagonal system as well as properties of its (dual) eigenvalue sequence, the corresponding (dual) eigenspace sequence. It is asserted that the tridiagonal pair of $q$-Racah type satisfies the tridiagonal relations, inversely a tridiagonal pair can be obtained from a tridiagonal algebra under several conditions. We also describe some special classes of tridiagonal pairs, such as Leonard pairs, tridiagonal pairs of the $q$-Serre type, mild tridiagonal pairs, sharp tridiagonal pairs. Especially, the classification of the sharp tridiagonal pairs is clarified \cite{INT}. Last but not least, the relation between tridiagonal algebras and the theory of orthogonal polynomials is briefly described. We recall hypergeometric orthogonal polynomials, and describe the connection between the theory of Leonard pairs and the Askey-scheme of orthogonal polynomials \cite{Ter06}. Note that the extension
to the theory of tridiagonal pairs leads to hypergeometric polynomials of several variables defined on a discrete support (Gasper-Rahman), as recently discovered in \cite{BM}.

\textbf{Chapter 2.} \textit{We recall the known presentations of the Onsager algebra and how the so-called $q-$Onsager algebra appeared in the context of mathematical physics. We briefly recall the `$q-$Onsager approach'. 
}
\\

First, we provide a historical background about the two different known presentations of the Onsager algebra: either the original presentation with generators $A_n,G_m$ \cite{O} or the presentation in terms of the Dolan-Grady relations \cite{DG}.
We also recall the relation with the loop algebra of  $sl_2$.
Secondly, we recall how the $q-$Onsager algebra surprisingly appeared in 2004 in the context of quantum integrable systems, through an analysis of solutions of the reflection equation. In particular, we recall how the new infinite dimensional algebra ${\cal A}_q$ arises and how it is related with the $q-$Onsager algebra.
It is explained how its connections with the quantum loop algebra of $sl_2$ and with $U_q(\widehat{sl_2})$ naturally appear from the Yang-Baxter and reflection equation algebra formulation.
Then, the so-called `$q-$Onsager approach' is briefly recalled.  
\\

\textbf{Chapter 3.} \textit{The three main results of the thesis are presented in some details. (i) and (ii): For the family of $q-$Onsager algebras ($\widehat{sl_2}$ and ADE type), analogs of the higher order relations of Lusztig  are conjectured and supporting evidence is presented in detail; (iii) The open XXZ spin chain is considered at roots of unity in the framework of the $q-$Onsager approach. A new algebra, an analog of Lusztig's quantum group, naturally arises. For a class of finite dimensional representations, explicit generators and relations are described. With respect to the new algebra, symmetries of the Hamiltonian are explored.}
\\

Recall the homomorphism from the $q-$Onsager algebra to the quantum affine Lie algebra $U_q(\widehat{sl_2})$ \cite{Bas1, BB3}. Recall the homomorphism from the generalized $q-$Onsager algebras (higher rank generalizations of the $q-$Onsager algebra) to the quantum affine Lie algebra $U_q(\widehat{g})$ \cite{BB1, Kolb}. By analogy with Lusztig's higher order relations \cite{Luszt} which arise for any quantum affine Lie algebra, it is thus expected that higher order relations are satisfied. Successively, we obtained:

\begin{enumerate}
\item \textit{The higher order relations for the $q$-Onsager algebra (the $\widehat{sl_2}$ case)} \cite{BV1}\vspace{1mm}

Let $A,A^*$ be the standard generators of the $q-$Onsager algebra. The $r-$th higher order relations for the $q$-Onsager algebra are conjectured. First, a generalization of the conjecture is proven for the case of tridiagonal pairs (i.e. certain irreducible finite dimensional representations on which $A,A^*$ act).  Two-variable polynomials which determine the relations are given. Then, the special case of the $q-$Onsager algebra is considered in details. The conjecture is proven for $r = 2, 3$. For $r$ generic, the conjecture is studied recursively.
A Maple software program is used to check the conjecture, which is confirmed for $r\leq 10$. Also, for a special case, the higher order relations of Lusztig are recovered.\\
 
\item \textit{The higher order relations for the generalized $q$-Onsager algebra (the ADE serie)}  \cite{BV2}\\

For each affine Lie algebra, a generalized $q$-Onsager algebra has been defined in \cite{BB1}. Let $A_i$, $i=0,1,...,rank(g)$ be the standard generators of this algebra. By analogy with the $\widehat{sl_2}$ case, for any simply-laced affine Lie algebra analogues of Lusztig's higher order relations are conjectured. The conjecture is proven for $r \le 5$. For $r$ generic, the conjecture is studied recursively. A Maple software program is used to check the conjecture, which is confirmed for $r\leq 10$.
According to the parity of $r$, two new families of two-variable polynomials are proposed, which determine the structure of the higher order relations. Several independent checks are done, which support the conjecture.\\

\item \textit{The $XXZ$ open spin chain at roots of unity} \cite{BGSV}\\ 

Inspired by the fact that the XXZ periodic spin chain at roots of unity enjoys a $sl_2$ loop algebra symmetry in certain sectors of the spectrum \cite{DFM}, the aim is to settle an algebraic framework for the analysis of the open XXZ spin chain at roots of unity within the $q-$Onsager approach. First, the two basic generators of the $q-$Onsager algebra are recalled, and their properties are studied for $q$ a root of unity (spectrum, structure of the eigenspaces and action). They form a new object that we call a `cyclic tridiagonal pair'. Secondly, two new operators that are  divided polynomials of the fundamental generators of the $q$-Onsager algebra are introduced. We study some of their properties (spectrum, structure of the eigenspaces and action). The relations satisfied by the four operators are described in details. They generate an explicit realization and first example of an analog of Lusztig quantum group for the $q-$Onsager algebra.  Finally, we briefly discuss the conditions on the boundary parameters such that the Hamiltonian of the open XXZ spin chain commutes with some of the generators.   

\end{enumerate}
In the end of this thesis, three families of open problems are presented in Chapter 4 and appendices are reported in Chapter 5. 
\chapter{Mathematics: background}
The aim of this Chapter is to present various aspects of the tridiagonal and $q$-Onsager algebras, their relations with other types of quantum algebras, their relations with the theory of Leonard and tridiagonal pairs, and finally the remarkable connection with the theory of orthogonal polynomials and Askey-scheme.\vspace{1mm}

In the first part, the definition of the tridiagonal algebra is given in terms of generators and relations. Several special cases of tridiagonal relations corresponding to the particular parameter sequences are recalled such as the $q$-Serre relations, the Dolan-Grady relations.\vspace{1mm}

In the second part, the $q$-Onsager algebra is defined as the $\rho_i$-deformed analogues of the $q$-Serre relations. Indeed, it is a special case of the tridiagonal algebra. Furthermore, the connection  between the $q$-Onsager algebra and other algebras such that the quantum enveloping algebra $U_q(\widehat{sl_2})$, the quantum loop algebra $U_q({\cal L}(sl_2)$, the reflection equation algebra and the current algebra denoted $O_q(\widehat{sl_2})$ recently introduced is briefly described.\vspace{1mm}

In the last part, several results due to Terwilliger {et al.} about the representation theory of tridiagonal algebras - in particular the $q-$Onsager algebra - are recalled. Finite dimensional irreducible representations of the tridiagonal algebras and of the $q-$Onsager algebras have been considered in details by Terwilliger {\it et al.}.  Provided two matrices $A,A^*$ satisfy the defining relations of 
the tridiagonal algebra, are diagonalizable on the vector space and the representation is irreducible,  it implies the following: in the basis in which the first matrix $A$ is diagonal with degeneracies, the other matrix $A^*$ takes a block tridiagonal structure. Furthermore, there exists another basis with respect to which the first matrix $A$ transforms into a block tridiagonal matrix, whereas the second one transforms into a diagonal matrix with degeneracies. In the simplest case (no degeneracy in the spectra), tridiagonal pairs are called Leonard pairs, which definition and properties will be first recalled. Then, we will describe definition and properties of tridiagonal pairs based on Terwilliger and collaborators' investigations. In addition, the connection between the theory of special functions and orthogonal polynomial and the $q-$Onsager algebra
is briefly described. Note that this subject was one of the motivations for studying tridiagonal algebras and the theory of tridiagonal pairs. For the simplest examples of tridiagonal pairs, namely the Leonard pairs, Askey-Wilson polynomials arise.\vspace{1mm}

Note that some of the material presented in this Chapter will play a crucial role in the analysis of the higher-order $q-$Dolan-Grady relations of the $q-$Onsager algebra, that will be considered in Chapter 3 entitled ``MAIN RESULTS''.
\section{Tridiagonal algebras}
Tridiagonal algebras come up in the theory of $Q$-polynomial distance-regular graphs \cite[Lemma 5.4]{PTIII} and tridiagonal pairs \cite[Theorem 10.1]{TD00}, \cite[Theorem 3.10]{Ter03}. A tridiagonal algebra has a presentation by two generators and two relations as follows
\begin{defn}\cite{Ter03}
Let $\mathbb{K}$ denote a field, and let $\beta, \gamma, \gamma^*, \delta, \delta^*$ denote a sequence of scalars taken from $\mathbb{K}$. The corresponding tridiagonal algebra $T$ is the associative $\mathbb{K}$-algebra with 1 generated by the generators $A, A^*$ subject to the relations 
\begin{eqnarray}
\label{tri1}
\lbrack A, A^2A^*-\beta AA^*A+A^*A^2-\gamma(AA^*+A^*A)-\delta A^*\rbrack&=& 0,\\
\label{tri2}
\lbrack A^*, {A^*}^2A-\beta A^*AA^*+A{A^*}^2-\gamma^*(A^*A+AA^*)-\delta^* A\rbrack&=& 0.
\end{eqnarray}
These relations are called the tridiagonal relations.
\end{defn}
The algebra generated by $A$ and $A^*$ satisfy (\ref{tri1}), (\ref{tri2}) is known as the subconstituent algebra or the Terwilliger algebra \cite{Terc1, Terc2, PTIII}. Furthermore, relations (\ref{tri1}) and (\ref{tri2}) are satisfied by the generators of both the classical and quantum `Quadratic Askey-Wilson algebra' introduced by Granovskii et al. \cite{GLZ}.\\

A tridiagonal algebra depends on the parameter sequence $\beta, \gamma, \gamma^*, \delta, \delta^*$. Now we consider some cases of the tridiagonal algebras that correspond to the particular parameter sequences.\\
 
A special case of (\ref{tri1}) and (\ref{tri2}) occurs in the context of quantum groups.
For $\beta = q^2+q^{-2}, \gamma = \gamma^* =0, \delta =\delta^* =0$, the tridiagonal relations are the $q$-Serre relations \cite{CP}.
\begin{eqnarray}
\label{1.3}
A^3A^*-[3]_qA^2A^*A+[3]_qAA^*A^2-A^*A^3&=&0,\\
\label{1.4}
{A^*}^3A-[3]_q{A^*}^2AA^*+[3]_qA^*A{A^*}^2-A{A^*}^3&=&0.
\end{eqnarray}
Note that equations (\ref{1.3}) and (\ref{1.4}) are among the defining relations for the quantum affine algebra $U_q(\widehat{sl_2})$.\\

Another special case of (\ref{tri1}) and (\ref{tri2}) has come up in the context of exactly solvable models in statistical mecanics. For $\beta = 2, \gamma=\gamma^* =0, \delta =b^2, \delta^*={b^*}^2$, the tridiagonal relations are the so-called Dolan-Grady relations \cite{DG}
\begin{eqnarray}
\lbrack A,\lbrack A,\lbrack A,A^*\rbrack\rbrack\rbrack&=&b^2[A,A^*],\\
\lbrack A^*,\lbrack A^*,\lbrack A^*,A\rbrack\rbrack\rbrack&=&{b^*}^2[A^*,A].
\end{eqnarray}
One more example of a tridiagonal algebra is the $q$-Onsager algebra. It will be considered in the following part.

\section{The $q$-Onsager algebra from different points of view}
In \cite{T05, B1, Bas1}, the $q$-Onsager algebra has been defined by standard generators and relations which are called the `$q$-deformed' Dolan-Grady relations. Applications of the $q$-Onsager algebra to tridiagonal pairs can be found in \cite{TD00, IT2, IT004, Ter03, T05}. The $q$-Onsager algebra has applications to quantum integrable models \cite{B1, Bas1, Bas2, Bas3, BK, BK1, BK3, BasS, BB3, BK} and quantum symmetric pairs \cite{Kolb}.
\begin{defn}\cite{T05, B1, Bas1}
Let $\mathbb{K}$ denote a field. The $q$-Onsager algebra is the associative algebra with unit and standard generators $A, A^*$ subject to the following relations
\begin{eqnarray}
\label{qO1}
\sum\limits_{k=0}^3{(-1)^k\left[\begin{array}{c}
3\\
k
\end{array}
\right]_qA^{3-k}A^*A^k}=\rho_0[A,A^*],\\
\label{qO2}
\sum\limits_{k=0}^3{(-1)^k\left[\begin{array}{c}
3\\
k
\end{array}
\right]_q{A^*}^{3-k}A{A^*}^k}=\rho_1[A^*,A],
\end{eqnarray}
where $q$ is a deformation parameter, and $\rho_0, \rho_1$ are fixed scalars in $\mathbb{K}$.
\end{defn}
The relations (\ref{qO1}), (\ref{qO2}) can be seen as $\rho_i$-deformed analogues of the $q$-Serre relations because for $\rho_0 = \rho_1 =0$ these relations are reduced to the $q$-Serre relations.\\
Clearly, the generators $A, A^*$ of the $q$-Onsager algebra satisfy the defining relations of the tridiagonal algebra corresponding to the scalars $\beta= q^2+q^{-2}, \gamma = \gamma^* = 0$, and $\delta= \rho_0, \delta^*= \rho_1$. This parameter sequence is said to be reduced.\\
For simplicity, the $q$-Onsager algebra can be defined by generators $A, A^*$ subject to the relations
 \begin{eqnarray}
 \label{126}
 \lbrack A, A^2A^*-\beta AA^*A+A^*A^2\rbrack &=&\rho[A,A^*],\\
 \label{127}
 \lbrack A^*, {A^*}^2A-\beta A^*AA^*+A{A^*}^2\rbrack &=&\rho[A^*,A],
 \end{eqnarray}
 where $\beta = q^2+q^{-2}$ and $\rho_0=\rho_1 = \rho $ \cite{IT004}.
 Relations (\ref{126}) and (\ref{127}) can be regarded as a $q$-analogue of the Dolan-Grady relations.\\
 
 The $q$-Onsager algebra (\ref{126}), (\ref{127}) has a basis as follows \cite{IT004}.  \\
 
 Let $r$ denote a positive integer. Set
 \begin{eqnarray*}
 \Lambda_r &=& \{\lambda=(\lambda_0, \lambda_1,\dots, \lambda_r) \in \mathbb{Z}^{r+1}|\lambda_0\ge 0, \lambda_i \ge 1 (1 \le i \le r)\},\\
 \Lambda &=&\bigcup\limits_{r \in  \mathbb{N}\cup \{0\}} {\Lambda_r}. 
 \end{eqnarray*}
 If there exists an integer $i~ (0 \le i\le r)$ such that 
 \[\lambda_0 < \lambda_1 < \dots <\lambda_i \ge \lambda_{i+1}\ge \dots \ge \lambda_{r},\]
 then $\lambda=(\lambda_0, \lambda_1, \dots, \lambda_r)$ is said to be irreducible.\\
 
 Denote $\Lambda^{irr}= \{\lambda\in \Lambda|\lambda~ \text{is irreducible}\}$.\\
 
 Let $X, Y$ denote noncommuting indeterminates. For $\lambda = (\lambda_0, \lambda_1, \dots, \lambda_r) \in \Lambda$, put
 \[\omega_\lambda(X,Y)=\left\{\begin{array}{c}
 X^{\lambda_0}Y^{\lambda_1}\dots X^{\lambda_r}~~~~ \text{ if}~ r~ \text{is even}\\
 X^{\lambda_0}Y^{\lambda_1}\dots Y^{\lambda_r}~~~~\text{if}~ r~ \text{is odd}
 \end{array},\right.\]
 where $X^{\lambda_0}=1$ if $\lambda_0 = 0$.
 \begin{thm}\cite{IT004}
 \label{basis}
 The following set is a basis of the $q$-Onsager algebra as a $\mathbb{C}$-vector space:
 \[\{\omega_\lambda(A,A^*)|\lambda \in \Lambda^{irr}\}.\]
 \end{thm}

\subsection{The $q$-Onsager algebra and $U_q(\widehat{sl_2})$}
There is an explicit relationship between the $q$-Onsager algebra and the quantum affine algebra $U_q(\widehat{sl_2})$. There are homomorphisms from the $q$-Onsager algebra into $U_q(\widehat{sl_2})$ \cite{Bas1, Bas3, BB3, Kolb}. Firstly, recall the definition of $U_q(\widehat{sl_2})$ in Drinfeld-Jimbo presentation.
\begin{defn}\cite{CP}
\label{defiUq}
Let $\mathbb{K}$ denote an algebraically closed field. The quantum affine algebra $U_q(\widehat{sl_2})$ is the associative $\mathbb{K}$-algebra with unit 1, defined by generators $e_i^{\pm}, K_i^{\pm 1}, i \in \{0,1\}$ and the following relations:
\begin{eqnarray*}
K_iK_i^{-1} &=& K_i^{-1}K_i = 1,\\
K_0K_1 &=&K_1K_0,\\
K_ie_i^{\pm}K_i^{-1}&=&q^{\pm2}e_i^{\pm},\\
K_ie_j^{\pm}K_i^{-1}&=&q^{\mp 2}e_j^{\pm},~ i \ne j,\\
\lbrack e_i^+,e_i^-\rbrack&=&\frac{K_i-K_i^{-1}}{q-q^{-1}},\\
\lbrack e_0^{\pm},e_1^{\mp}\rbrack&=&0,
\end{eqnarray*}
\begin{equation}
\label{123}
(e_i^{\pm})^3e_j^{\pm}-[3]_q(e^{\pm}_i)^2e_j^{\pm}e_i^{\pm}+[3]_qe_i^{\pm}e_j^{\pm}(e_i^{\pm})^2-e_j^{\pm}(e_i^{\pm})^3 = 0, ~ i \ne j,
\end{equation}
where the expression $[r,s]$ means $rs - sr$. We call $e_i^{\pm}, K_i^{\pm}, i \in \{0, 1\}$ the Chevalley generators for $U_q(\widehat{sl_2})$.
\end{defn}
\begin{thm}\cite{BB3}
There exist algebra homomorphisms $\varphi, \varphi^*$ from the $q$-Onsager algebra to $U_q(\widehat{sl_2})$
\begin{eqnarray}
\varphi(A)&=& k_-e_0^++k_+q^{-1}e_0^-K_0+\epsilon_- K_0,\\
\varphi(A^*)&=& k_+e_1^++k_-q^{-1}e_1^-K_1+\epsilon_+K_1,
\end{eqnarray}
where $k_+, k_-, \epsilon_+, \epsilon_- \in \mathbb{K}$, and $\rho= k_+k_-(q+q^{-1})^2$.\\
And
\begin{eqnarray}
\varphi^*(A)&=&\bar{k}_-q^{-1}e_0^+K_0^{-1}+\bar{k}_+e_0^-+\bar{\epsilon}_+K_0^{-1},\\
\varphi^*(A^*)&=&\bar{k}_+q^{-1}e_1^+K_0^{-1}+\bar{k}_-e_1^-+\bar{\epsilon}_-K_1^{-1},
\end{eqnarray}
where $\bar{k}_+, \bar{k}_-,  \bar{\epsilon}_+, \bar{\epsilon}_- \in \mathbb{K}$, and $\rho=\bar{k}_+\bar{k}_-(q+q^{-1})^2$.
\end{thm}

\subsection{The $q$-Onsager algebra and the $U_q(sl_2)$-loop algebra}
There are algebra homomorphisms from the $q$-Onsager algebra into the $U_q(sl_2)$-loop algebra \cite{IT2, IT004}.
\begin{defn}\cite{CP}
The $U_q(sl_2)$-loop algebra $\mathcal{L}$ is the associative $\mathbb{C}$-algebra with 1 generated by $e_i^+, e_i^-, k_i, k_i^{-1}$ $(i =0, 1)$ subject to the relations
\begin{eqnarray*}
k_ik_i^{-1} &=& k_i^{-1}k_i = 1,\\
k_0k_1 &=&k_1k_0 =1,\\
k_ie_i^{\pm}k_i^{-1}&=&q^{\pm2}e_i^{\pm},\\
k_ie_j^{\pm}k_i^{-1}&=&q^{\mp 2}e_j^{\pm},~ i \ne j,\\
\lbrack e_i^+,e_i^-\rbrack&=&\frac{k_i-k_i^{-1}}{q-q^{-1}},\\
\lbrack e_0^{\pm},e_1^{\mp}\rbrack&=&0,
\end{eqnarray*}
\[(e_i^{\pm})^3e_j^{\pm}-[3]_q(e^{\pm}_i)^2e_j^{\pm}e_i^{\pm}+[3]_qe_i^{\pm}e_j^{\pm}(e_i^{\pm})^2-e_j^{\pm}(e_i^{\pm})^3 = 0, ~ i \ne j.\]
\end{defn}
Note that if we replace $k_0k_1=k_1k_0 =1$ in the definition for $\mathcal{L}$ by $k_0k_1=k_1k_0$ then we have the quantum affine algebra $U_q(\widehat{sl_2})$. Namely, $\mathcal{L}$ is isomorphic to the quotient algebra of $U_q(\widehat{sl_2})$ by the two-sided ideal generated by $k_0k_1-1$.
\begin{thm}\cite{IT005}
For arbitrary nonzero $s, t \in \mathbb{C}$, there exists an algebra homomorphism $\varphi_{s,t}$ from the $q$-Onsager algebra to $\mathcal{L}$ that sends $A,A^*$ to 
\begin{eqnarray*}
A_t(s)&=&x(s)+tk(s)+t^{-1}k(s)^{-1},\\
A_t^*(s)&=&y(s)+t^{-1}k(s)+tk(s)^{-1},
\end{eqnarray*}
respectively, where 
\begin{eqnarray*}
x(s)&=&\alpha(se_0^++s^{-1}e_1^-k_1)~~\text{with}~ \alpha=-q^{-1}(q-q^{-1})^2,\\
y(s)&=&se_0^{-}k_0+s^{-1}e_1^+,\\
k(s)&=&sk_0,\\
\rho & =& -(q^2-q^{-2})^2.
\end{eqnarray*}
Moreover $\varphi_{s,t}$ is injective.
\end{thm}
\begin{thm}\cite{Bas3}
Let $k_+, k_-, \epsilon_+, \epsilon_-$ denote scalars in $\mathbb{K}$. There is an algebra homomorphism from the $q$-Onsager algebra to $\mathcal{L}$ such that
\begin{eqnarray}
A& \mapsto & k_-e_0^+k_0^{1/2}+k_+e_0^-k_0^{1/2}+\epsilon_-k_0,\\
A^*& \mapsto & k_+e_1^+k_1^{1/2}+k_-e_1^-k_1^{1/2}+\epsilon_+k_1,
\end{eqnarray}
with $\rho_0=\rho_1=k_+k_-(q+q^{-1})^2$.
\end{thm}

\subsection{The $q$-Onsager algebra, the reflection equation and the algebra $\mathcal{A}_q$}
\label{Onsref}
The aim of this part is to recall the relation between the reflection equation algebra \cite{Cher84, Skly88} and the $q$-Onsager algebra (\ref{qO1}), (\ref{qO2}). In the literature, the connection between the two algebras appeared as follows:
First, the structure of the solutions the $K$-operators satisfying ``RKRK" defining relations for the $U_q(\widehat{sl_2})$ with $R$-matrix  (\ref{soluybequation}) had been studied in details in the case where the entries of the $K-$matrix act on  an irreducible finite dimensional vector space. Recall that the entries of the $K-$matrix depends on a spectral parameter $u$. Expanding the entries in terms of a new spectral parameter $U=(qu^2+q^{-1}u^{-2})/(q+q^{-1})$, it was observed that the first modes of the expansion of the diagonal entries of the $K-$matrix generate a $q-$Onsager algebra \cite{BK2}. This will be recalled in details in Chapter \ref{ChapPhys}.  This observation suggested that the reflection equation algebra and the $q-$Onsager algebra are closely related. \vspace{1mm}

In \cite{BasS}, this connection was further studied. Let $K(u)$ be a solution of the reflection equation algebra with $U_q(\widehat{sl_2})$ with $R$-matrix  (\ref{soluybequation}). Assume $q$ is not a root of unity. Let $V$ be an irreducible finite dimensional vector space on which  the entries of the  $K-$matrix act. Then, it is known that the solution of the reflection equation is unique (up to an overall scalar factor). Independently, consider the $q-$Onsager algebra. Let $\delta$ denote a coaction map (see the definition below) that is introduced explicitly.  Then, one introduces the $K-$matrix as the intertwiner between tensor product representations of the $q-$Onsager algera. By construction, the $K-$matrix automatically satisfies the reflection equation algebra.
As a consequence, for $V$ an irreducible finite dimensional vector space, the uniqueness of the solution $K$ implies that a quotient of the $q-$Onsager algebra is isomorphic to a quotient of the reflection equation algebra.\vspace{1mm}

Below, in a first part we recall the reflection equation algebra. Then, we recall the structure of the general $K-$matrix solutions of the reflection equation algebra for the $R-$matrix in terms of currents. The modes of these currents are  denoted $\{\mathcal{W}_{-k}, \mathcal{W}_{k+1}, \mathcal{G}_{k+1}, \tilde{\mathcal{W}}_{k+1}|k\in \mathbb{Z}_+\}$ and generate the algebra ${\cal A}_q$. The isomorphism between the current algebra and the reflection equation algebra is given (see Theorem \ref{isocurA}).  In a second part, we independently recall the construction of a coaction map of the $q-$Onsager algebra (see Proposition \ref{coa}). The defining relations of the $K-$matrix as an intertwiner of tensor product representations of the $q-$Onsager algebra is mentioned. One of the main result of \cite{BasS} is finally recalled: if the elements of the $q-$Onsager algebra act on an irreducible finite dimensional vector space (in which case a quotient of the $q-$Onsager algebra is considered), then it is isomorphic to a quotient of the algebra ${\cal A}_q$ or, equivalently, to a quotient of the reflection equation algebra.
\vspace{1mm}

Let $R: \mathbb{C}^* \to End(\mathcal{V}\otimes \mathcal{V})$ denote the intertwining operator between the tensor product of two fundamental representations $\mathcal{V}=\mathbb{C}^2$ associated with the algebra $U_q(\widehat{sl_2})$. The element $R(u)$ depends on the deformation parameter $q$ and is defined by \cite{Baxter1}
\begin{equation}
\label{soluybequation}
R(u)=\left(\begin{array}{cccc}
uq-u^{-1}q^{-1} & 0 &0  & 0 \\ 
0& u-u^{-1} & q-q^{-1} &0  \\ 
0& q-q^{-1} & u-u^{-1} &0  \\ 
0& 0 & 0 & uq-u^{-1}q^{-1} \\ 
\end{array} \right),
\end{equation}
where $u$ is called the special parameter.\\

By construction, $R(u)$ satisfies the Yang-Baxter equation in the space $\mathcal{V}\otimes \mathcal{V}\otimes\mathcal{V}$. Namely,
\begin{equation}
\label{ybequation}
R_{12}(u/v)R_{13}(u)R_{23}(v)=R_{23}(v)R_{13}(u)R_{12}(u/v),~~~\forall u, v.
\end{equation}
\begin{defn}\cite{Skly88} Reflection equation algebra $B_q(\widehat{sl_2})$ is an associative algebra with unit 1 and generators $K_{11}(u)\equiv A(u)$, $K_{12}(u)\equiv B(u)$, $K_{21}(u)\equiv C(u)$, $K_{22}(u)\equiv D(u)$ considered as the entries of the $2 \times 2$ square matrix $K(u)$ which obeys the defining relations 
\begin{equation}
\label{reflectequ}
R_{12}(u/v)(K(u)\otimes \mathbb{I})R_{12}(uv)(\mathbb{I}\otimes K(v))=(\mathbb{I}\otimes K(v))R_{12}(uv)(K(u)\otimes \mathbb{I})R_{12}(u/v), ~~~\forall u, v.
\end{equation}
\end{defn}
\begin{defn}\cite{BasS}
\label{dfcurren}
The current algebra $O_q(\widehat{sl_2})$ is an associative algebra with unit 1, current generators $\mathcal{W}_{\pm}(u)$,  $\mathcal{G}_{\pm}(u)$ and parameter $\rho \in \mathbb{C}^*$. Define the formal variables $U = (qu^2+q^{-1}u^{-2})/(q+q^{-1})$ and $V=(qv^2+q^{-1}v^{-2})/(q+q^{-1}), ~\forall u, v.$ The defining relations are:
\begin{equation}
[\mathcal{W}_{\pm}(u), \mathcal{W}_{\pm}(v)] = 0,
\end{equation}
\begin{equation}
[\mathcal{W}_+(u),\mathcal{W}_-(v)]+[\mathcal{W}_-(u),\mathcal{W}_+(v)]=0,
\end{equation}
\begin{eqnarray}
(U-V)[\mathcal{W}_{\pm}(u), \mathcal{W}_{\mp}(v)]&=& \frac{(q-q^{-1})}{\rho(q+q^{-1})}\left(\mathcal{G}_{\pm}(u)\mathcal{G}_{\mp}(v)-\mathcal{G}_{\pm}(v)\mathcal{G}_{\mp}(u)\right)~~\nonumber\\
& & +\frac{1}{q+q^{-1}}(\mathcal{G}_{\pm}(u)-\mathcal{G}_{\mp}(u)+\mathcal{G}_{\mp}(v)-\mathcal{G}_{\pm}(v)),~~~
\end{eqnarray}
\begin{eqnarray}
\mathcal{W}_{\pm}(u)\mathcal{W}_{\pm}(v)&-&\mathcal{W}_{\mp}(u)\mathcal{W}_{\mp}(v)+\frac{1}{\rho(q^2-q^{-2})}[\mathcal{G}_{\pm}(u),\mathcal{G}_{\mp}(v)]\nonumber\\
&+&\frac{1-UV}{U-V}(\mathcal{W}_{\pm}(u)\mathcal{W}_{\mp}(v)-\mathcal{W}_{\pm}(v)\mathcal{W}_{\mp}(u))=0,
\end{eqnarray}
\begin{eqnarray}
U[\mathcal{G}_{\mp}(v), \mathcal{W}_{\pm}(u)]_q&-& V[\mathcal{G}_{\mp}(u), \mathcal{W}_{\pm}(v)]_q-(q-q^{-1})(\mathcal{W}_{\mp}(u)\mathcal{G}_{\mp}(v)-\mathcal{W}_{\mp}(v)\mathcal{G}_{\mp}(u))\nonumber\\
&+&\rho(U\mathcal{W}_{\pm}(u)-V\mathcal{W}_{\pm}(v)-\mathcal{W}_{\mp}(u)+\mathcal{W}_{\mp}(v))=0,
\end{eqnarray}
\begin{eqnarray}
U[\mathcal{W}_{\mp}(u), \mathcal{G}_{\mp}(v)]_q&-& V[\mathcal{W}_{\mp}(v), \mathcal{G}_{\mp}(u)]_q-(q-q^{-1})(\mathcal{W}_{\pm}(u)\mathcal{G}_{\mp}(v)-\mathcal{W}_{\pm}(v)\mathcal{G}_{\mp}(u))\nonumber\\
&+&\rho(U\mathcal{W}_{\mp}(u)-V\mathcal{W}_{\mp}(v)-\mathcal{W}_{\pm}(u)+\mathcal{W}_{\pm}(v))=0,
\end{eqnarray}
\begin{equation}
[\mathcal{G}_\epsilon(u),\mathcal{W}_{\pm}(v)]+[\mathcal{W}_{\pm}(u), \mathcal{G}_\epsilon(v)]=0, ~\forall \epsilon = \pm,
\end{equation}
\begin{equation}
[\mathcal{G}_{\pm}(u), \mathcal{G}_{\pm}(v)]=0,
\end{equation}
\begin{equation}
[\mathcal{G}_+(u), \mathcal{G}_-(v)]+[\mathcal{G}_-(u), \mathcal{G}_+(v)]=0.
\end{equation}
\end{defn}
\begin{thm}\cite{BasS}
\label{isorefcur}
The map $\Phi: B_q(\widehat{sl_2})\longrightarrow O_q(\widehat{sl_2})$  defined by
\begin{eqnarray}
\label{1224}
A(u)&=& uq\mathcal{W}_+(u)-u^{-1}q^{-1}\mathcal{W}_-(u),\\
D(u)&=& uq\mathcal{W}_-(u)-u^{-1}q^{-1}\mathcal{W_+}(u),\\
B(u)& =& \frac{1}{k_-(q+q^{-1})}\mathcal{G}_+(u)+\frac{k_+(q+q^{-1})}{q-q^{-1}},\\
\label{1227}
C(u)&=& \frac{1}{k_+(q+q^{-1})}\mathcal{G}_-(u)+\frac{k_-(q+q^{-1})}{q-q^{-1}}
\end{eqnarray}
is an algebra isomorphism, where $q \ne 1, u \ne q^{-1}$ and $k_{\pm}\in \mathbb{C}^*$.
\end{thm}

\begin{defn}\cite{BasS}\label{alA}
$\mathcal{A}_q$ is an associative algebra with parameter $\rho \in \mathbb{C}^*$, unit 1 and generators $\{\mathcal{W}_{-k}$, $ \mathcal{W}_{k+1}$, $ \mathcal{G}_{k+1}$, $ \mathcal{\tilde{G}}_{k+1}|k \in \mathbb{Z}_+\}$ satisfying the following relations:
\begin{eqnarray}
\label{aq28}
[\mathcal{W}_0, \mathcal{W}_{k+1}]=[\mathcal{W}_{-k}, \mathcal{W}_1]=\frac{1}{q+q^{-1}}(\mathcal{\tilde{G}}_{k+1}-\mathcal{G}_{k+1}),\\
~[\mathcal{W}_0,\mathcal{G}_{k+1}]_q=[\mathcal{\tilde{G}}_{k+1}, \mathcal{W}_0]_q=\rho\mathcal{W}_{-k-1}-\rho\mathcal{W}_{k+1},\\
~[\mathcal{G}_{k+1},\mathcal{W}_1]_q=[\mathcal{W}_1, \mathcal{\tilde{G}}_{k+1}]_q=\rho\mathcal{W}_{k+2}-\rho\mathcal{W}_{-k},\\
~[\mathcal{W}_{-k},\mathcal{W}_{-l}]=0,~~~[\mathcal{W}_{k+1},\mathcal{W}_{l+1}]=0,\\
~[\mathcal{W}_{-k}, \mathcal{W}_{l+1}]+[\mathcal{W}_{k+1},\mathcal{W}_{-l}]=0,\\
~[\mathcal{W}_{-k}, \mathcal{G}_{l+1}]+[\mathcal{G}_{k+1},\mathcal{W}_{-l}]=0,\\
~[\mathcal{W}_{-k}, \mathcal{\tilde{G}}_{l+1}]+[\mathcal{\tilde{G}}_{k+1},\mathcal{W}_{-l}]=0,\\
~[\mathcal{W}_{k+1}, \mathcal{G}_{l+1}]+[\mathcal{G}_{k+1},\mathcal{W}_{l+1}]=0,\\
~[\mathcal{W}_{k+1}, \mathcal{\tilde{G}}_{l+1}]+[\mathcal{\tilde{G}}_{k+1},\mathcal{W}_{l+1}]=0,\\
~[\mathcal{G}_{k+1}, \mathcal{G}_{l+1}]=0,~~[\mathcal{\tilde{G}}_{k+1},\mathcal{\tilde{G}}_{l+1}]=0,\\
\label{aq38}
~[\mathcal{\tilde{G}}_{k+1}, \mathcal{G}_{l+1}]+[\mathcal{G}_{k+1},\mathcal{\tilde{G}}_{l+1}]=0,
\end{eqnarray}
where $\mathbb{Z}_+$ is the set of all nonnegative integers.
\end{defn}

By analogy with Drinfeld's construction, we are now looking for an infinite dimensional set of elements (the so-called `modes') of an algebra in terms of which the currents $\mathcal{W}_{\pm}(u)$, $\mathcal{G}_{\pm}(u)$ are expanded.
\begin{thm}\cite{BasS}
\label{isocurA}
Define the formal variable $U= (qu^2+q^{-1}u^{-2})/(q+q^{-1})$. Let $\Psi: O_q(\widehat{sl_2})\longrightarrow \mathcal{A}_q$ be the map defined by
\begin{eqnarray}
\mathcal{W}_+(u)=\sum\limits_{k\in \mathbb{Z}_+}{\mathcal{W}_{-k}U^{-k-1}},~~\mathcal{W}_-(u)= \sum\limits_{k\in \mathbb{Z}_+}{\mathcal{W}_{k+1}U^{-k-1}},\\
\mathcal{G}_+(u)=\sum\limits_{k\in \mathbb{Z}_+}{\mathcal{G}_{k+1}U^{-k-1}},~~\mathcal{G}_-(u)= \sum\limits_{k\in \mathbb{Z}_+}{\mathcal{\tilde{G}}_{k+1}U^{-k-1}}.
\end{eqnarray}
Then, $\Psi$ is an algebra isomorphism between $O_q(\widehat{sl_2})$ and $\mathcal{A}_q$.
\end{thm}

In \cite{Jim, Jim2}, Jimbo pointed out that intertwiners $R$ of quantum loop algebra lead to trigonometric solutions of the quantum Yang-Baxter equation (\ref{ybequation}). Any tensor product of two evaluation representations with generic evaluation parameters $u$ and $v$ being indecomposable, by Schur's lemma the solution $R$ is unique up to an overall scalar factor. Consider the quantum affine algebra $U_q(\widehat{sl_2})$ the construction of the solution $R(u)$ given by (\ref{soluybequation}) goes below. Similar arguments apply to solutions of the reflection equation algebra. If the vector space on which the entries of the $K-$matrix act is irreducible, then the solution is unique (up to an overall factor).\\

We now turn to the construction of intertwiners between representations of the $q-$Onsager algebra. First we recall some basic ingredients that are necessary for the discussion. \vspace{1mm}

Recall the realization of the quantum affine algebra $U_q(\widehat{sl_2})$ in the Chevalley presentation $\{H_j, E_j, F_j\}, j \in \{0, 1\}$ \cite{CP}
\begin{defn}
Define the extended Cartan matrix $\{a_{ij}\} ~ (a_{ii} = 2, a_{ij} = -2 ~\text{for}~ i \ne j)$. The quantum affine algebra $U_q{(\widehat{sl_2})}$ is generated by the elements $\{H_j, E_j, F_j\}, j \in \{0, 1\}$ which satisfy the defining relations
\begin{eqnarray}
[H_i, H_j] &=&0,~~~ [H_i, E_j] = a_{ij}E_j,\\
~[H_i, F_j] &=& -a_{ij}F_j,~~~ [E_i, F_j] = \delta_{ij}\frac{q^{H_i}-q^{-H_i}}{q-q^{-1}},
\end{eqnarray}
together with the $q$-Serre relations
\begin{eqnarray}
[E_i, [E_i, [E_i, E_j]_q]_{q^{-1}}] &=&0,\\
~[F_i,[F_i, [F_i, F_j]_q]_{q^{-1}}]&=&0.
\end{eqnarray}
\end{defn}
Clearly, this realization of $U_q(\widehat{sl_2})$ is equivalent to the one in Definition \ref{defiUq}. We endow the quantum group with a Hopf algebra structure that is ensured by the existence of a coproduct $\Delta: U_q(\widehat{sl_2})\longrightarrow U_q({\widehat{sl_2}})\otimes U_q(\widehat{sl_2})$, an antipode $\mathcal{S}: U_q(\widehat{sl_2})\longrightarrow U_q({\widehat{sl_2}})$ and a counit $\mathcal{E}: U_q(\widehat{sl_2})\longrightarrow \mathbb{C}$ with
\begin{eqnarray}
\Delta(E_i)&=&E_i\otimes q^{-H_i/2}+q^{H_i/2}\otimes E_i,\\
\Delta(F_i) &=& F_i\otimes q^{-H_i/2}+q^{H_i/2}\otimes F_i,\\
\Delta(H_i)&=& H_i\otimes \mathbb{I}+\mathbb{I}\otimes H_i,
\end{eqnarray}
\begin{equation}
\mathcal{S}(E_i) = -E_iq^{H_i}, ~~~ \mathcal{S}(F_i) = -q^{H_i}F_i, ~~~ \mathcal{S}(H_i) = -H_i, ~~~ \mathcal{S}(\mathbb{I}) = 1,
\end{equation}
and 
\begin{equation}
\mathcal{E}(E_i) = \mathcal{E}(F_i) = \mathcal{E}(H_i) = 0, ~~~\mathcal{E}(\mathbb{I}) = 1.
\end{equation}
Note that the opposite coproduct $\Delta'$ can be similarly defined with $\Delta' \equiv \sigma \circ \Delta$ where the permutation map $\sigma(x\otimes y) = y\otimes x$ for all $x, y \in U_q(\widehat{sl_2})$ is used.\\

Recall that, by definition the intertwiner $R(u/v): \mathcal{V}_u \otimes \mathcal{V}_v \longrightarrow \mathcal{V}_v \otimes \mathcal{V}_u$ between two fundamental $U_q(\widehat{sl_2})$- evaluation representation obeys
\begin{equation}
\label{interR}
R(u/v)(\pi_u\times \pi_v)[\Delta(x)]=(\pi_u \times \pi_v)[\Delta'(x)]R(u/v)~~~~ \forall x \in U_q(\widehat{sl_2}),
\end{equation}
where the (evaluation) homomorphism $\pi_u: U_q(\widehat{sl_2})\longrightarrow \text{End}(\mathcal{V}_u)$ is chosen such that $(\mathcal{V} \equiv \mathbb{C}^2)$ 
\begin{eqnarray}
\pi_u(E_1) &=& uq^{1/2}\sigma_+,~~~ \pi_u(E_0) = uq^{1/2}\sigma_-,\\
\pi_u(F_1) &=& u^{-1}q^{-1/2}\sigma_-,~~~ \pi_u(F_0) = u^{-1}q^{-1/2}\sigma_+,\\
\pi_u(q^{H_1})&=& q^{\sigma_z},~~~ \pi_u(q^{H_0}) = q^{-\sigma_z}.
\end{eqnarray}
It is not difficult to check the matrix $R(u)$ given by (\ref{soluybequation}) indeed satisfies the required conditions (\ref{interR}). The intertwiner $R(u)$ is unique (up to an overall scalar factor) and satisfies the Yang-Baxter algebra (\ref{ybequation}).\\

By analogy with the construction described above for the $R$-matrix and along the lines described in \cite{DelG, DM}, an intertwiner for the $q$-Onsager algebra can be easily constructed. Before, we need to introduce the concept of the comodule algebra using the analogue of the Hopf's algebra coproduct action called the coaction map.
\begin{defn}\cite{CP}
Given a Hopf algebra $\mathcal{H}$ with comultiplication $\Delta$ and counit $\mathcal{E}$, $\mathcal{I}$ is called a left $\mathcal{H}$-comodule if there exists a left coaction map $\delta: \mathcal{I}\longrightarrow \mathcal{H}\otimes \mathcal{I}$ such that
\begin{equation}
(\Delta \times  id)\circ \delta = (id\times \delta)\circ \delta, ~~~ (\mathcal{E}\times id)\circ \delta \cong id.
\end{equation}
Right $\mathcal{H}$-comodules are defined similarly.
\end{defn}

\begin{prop}\cite{B1}\label{coa}
Let $k_{\pm} \in \mathbb{C}^*$ and set $\rho_1=\rho_2 = k_+k_-(q+q^{-1})^2$. The $q$-Onsager algebra $\mathbb{T}$ is a left $U_q(\widehat{sl_2})$-comodule algebra with coaction map $\delta: \mathbb{T} \longrightarrow U_q(\widehat{sl_2})\otimes \mathbb{T}$ such that
\begin{eqnarray}
\delta(A)&=&(k_+E_1q^{H_1/2}+k_-F_1q^{H_1/2})\otimes \mathbb{I} +q^{H_1}\otimes A,\\
\delta(A^*)&=&(k_-E_0q^{H_0/2}+k_+F_0q^{H_0/2})\otimes \mathbb{I} +q^{H_0}\otimes A^*.
\end{eqnarray}   
\end{prop}
\begin{prop}\cite{BasS}
\label{propreflec}
Let $\pi_u: U_q(\widehat{sl_2})\longrightarrow  \text{End} (\mathcal{V}_u)$ be the evaluation homomorphism for $\mathcal{V}\equiv \mathbb{C}^2$. Let $W$ denote a vector space over $\mathbb{C}$ on which the elements of the $q$-Onsager algebra $\mathbb{T}$ act. Assume the tensor product $\mathcal{V}_u\otimes W$ is irreducible. Then, there exists an intertwiner
\begin{equation}
K(u): \mathcal{V}_u\otimes W \longrightarrow \mathcal{V}_{u^{-1}}\otimes W,
\end{equation}
such that
\begin{equation}
\label{interK}
K(u)(\pi_u \times id)[\delta(a)] = (\pi_{u^{-1}} \times id)[\delta(a)]K(u), ~~~ \forall a \in \mathbb{T}.
\end{equation}
It is unique (up to an overall scalar factor), and it satisfies the reflection equation (\ref{reflectequ}).
\end{prop}
According to Proposition (\ref{propreflec}), $K(u)$ is the unique intertwiner of the $q$-Onsager algebra $\mathbb{T}$ satisfying (\ref{interK}). By construction, it satisfies the reflection equation algebra (\ref{reflectequ}). This implies that a quotient of the $q$-Onsager algebra
is isomorphic to a quotient of the reflection equation algebra or, alternatively, a quotient of infinite dimensional algebra ${\cal A}_q$. Combining these results, it suggests that the $q$-Onsager algebra with $\rho_0=\rho_1=\rho$ has two different realizations: one in terms of the reflection equation algebra for the $U_q(\widehat{sl_2})$ $R$-matrix and another one in terms of the current algebra $O_q(\widehat{sl_2})\cong \mathcal{A}_q$. 
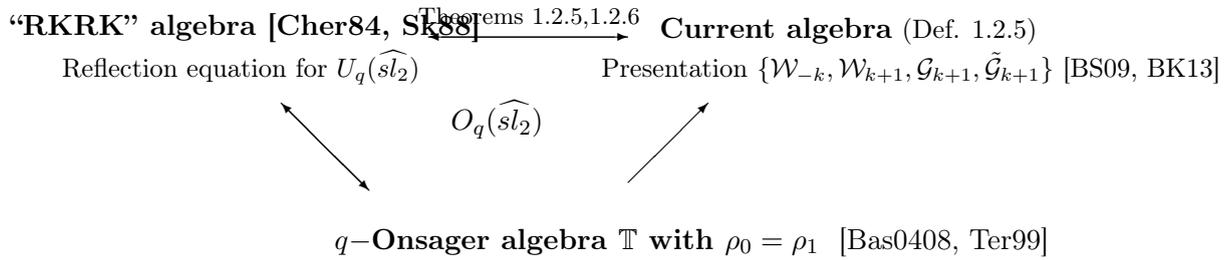
\begin{figure}[ht!]
\begin{picture}(390,90)
   \put(5,60){\shortstack[1]{\ \bf``RKRK'' algebra \cite{Cher84,Skly88} \\
        \small Reflection equation for $U_q(\widehat{sl_2})$}}
   \put(235,75){\vector(-2,0){70}}
   \put(170,75){\vector(2,0){70}}
   \put(152,80){\shortstack[l]{ \\
                                 \footnotesize \quad Theorems \ref{isorefcur},\ref{isocurA}}}
   \put(230,60){\shortstack[l]{\quad \quad  {\bf Current algebra} \small (Def. \ref{dfcurren})  \\
                                     \small Presentation \small $\{{\mathcal{W}}_{-k},{\mathcal{W}}_{k+1},{\mathcal{G}}_{k+1},{\tilde{\mathcal{G}}}_{k+1}\}$ \cite{BasS, BK}}} 
   \put(140,20){\vector(-1,1){30}}
   \put(113,47){\vector(1,-1){30}}
   \put(170,40){\shortstack[l]{\ {\bf $O_q({\widehat{sl_2}})$}}}                             
   \put(60,30){\shortstack[l]{}}
   
   \put(240,20){\vector(1,1){30}}
   \put(265,30){\shortstack[l]{}}

   \put(265,23){\shortstack[l]{}}
   \put(130,-5){\shortstack[l]{{\bf $q-$Onsager algebra  $\mathbb{T}$  with $\rho_0=\rho_1$} \ \cite{Bas1, Ter03}}}
\end{picture}
\vspace{5mm}
\caption{Presentations of $O_q({\widehat{sl_2}})$}
\end{figure}

\section{Representation theory}
The concept of a tridiagonal pair originated in algebraic graph theory, or more precisely, the theory of Q-polynomial distance-regular graphs. The standard generators for the subconstituent algebra (Terwilliger algebra) of a P- and Q- polynomial association scheme give rise to the concept of tridiagonal pair when they are restricted to an irreducible submodule of the standard module \cite[Example 1.4]{TD00}, \cite[Lemmas 3.9, 3.12]{Terc1}. This fact motivates the ongoing investigation of the tridiagonal pairs \cite{TD00, INT, IT, IT03, NT, Ter03}. \vspace{1mm} 

It is now well understood that the representation theory of tridiagonal algebras - in particular the $q-$Onsager algebra - is intimately connected with the theory of tridiagonal pairs. For instant, for an algebraically closed field and no  restrictions on $q$, note that a classification of tridiagonal pairs is given in \cite{INT} (see also \cite{IT2}). In this Section, we first recall what is a Leonard pair, and then turn to tridiagonal pairs. The material is taken from the work of Terwilliger et al.

\subsection{Leonard pairs}
We recall the notion of a Leonard pair, a special case of the tridiagonal pair \cite{TD00, Ter01, T05, Ter03} and illustrate how Leonard pairs arise in representation theory, combinatorics, and the theory of orthogonal polynomials \cite{Ter01}.
\begin{defn}
\label{dleo}\cite{T05}
Let $V$ denote a finite dimensional vector space over $\mathbb{K}$. By a Leonard pair we mean an ordered pair of linear transformations $A: V \to V,$ and $A^*: V \to V$ that satisfy the following conditions.\\

(i) There exists a basis for $V$ with respect to which the matrix representing $A$ is diagonal and the matrix representing $A^*$ is irreducible tridiagonal.\\

(ii) There exists a basis for $V$ with respect to which the matrix representing $A^*$ is diagonal and the matrix representing $A$ is irreducible tridiagonal.
\end{defn}
Note that a matrix is called tridiagonal whenever each nonzero entry lies on either the diagonal, the subdiagonal, or the superdiagonal. A tridiagonal matrix is said to be irreducible if all entries on the subdiagonal and superdiagonal are nonzero.\\

Suppose that $V$ has dimension $d+1$. Write $End(V)$ is the $\mathbb{K}$-algebra consisting of linear transformations from $V$ to $V$. Now let $A \in End(V)$, suppose that $A$ has $(n+1)$ distinct eigenvalues $\theta_0, \theta_1, \dots, \theta_n$. If $n = d$ and the eigenvalues $\theta_0, \theta_1, \dots, \theta_d \in \mathbb{K}$ then $A$ is said to be multiplicity-free. Assume $A$ is multiplicity-free, put
\[E_i=\prod\limits_{\begin{array}{c}
\scriptstyle 0 \le j \le d\\
\scriptstyle j \ne i \\
\end{array}}{\frac{A-\theta_j I}{\theta_i-\theta_j}}, ~~ 0 \le i \le d,\]
where $I$ denotes the identity of $End(V)$.

Apparently, by the elementary algebra,
\begin{eqnarray}
AE_i = E_iA = \theta_iE_i, ~~0 \le i \le d, \\
E_iE_j = \delta_{ij}E_i, ~~0\le i, j \le d,\\
\sum\limits_{i=0}^{d}{E_i} = I, ~~~~\sum\limits_{i = 0}^{d}{\theta_iE_i} = A,\\
V_i = E_iV ~\text{is eigenspace of}~ A ~ \text{corresponding with the eigenvalue}~\theta_i ,\\
E_0, E_1, \dots, E_d ~ \text{is a basic of subalgebra of}~ End(V)~ \text{generated by}~ A.
\end{eqnarray}
We call $E_i$ the primitive idempotent of $A$ associated with $\theta_i$. As a result,
\[V = E_0V+E_1V+\dots+E_dV~~~~~~\text{(direct sum)}.\]

\begin{lem}\cite{T05}
Let $V$ denote a finite dimensional vector space over a field $\mathbb{K}$. Let $(A, A^*)$ denote a Leonard pair on $V$. Then each of $A, A^*$ is multiplicity-free.
\end{lem}
When working with a Leonard pair, it is often convenient to consider a closely related object called a Leonard system.
\begin{defn}\cite{T05}
\label{dleosystem}
Let $d$ denote a nonnegative integer. Let $V$ denote a finite dimensional vector space over a field $\mathbb{K} $ with dimension $d+1$. By a Leonard system on $V$, we mean a sequence
\[\Phi=(A; E_0, E_1, \dots, E_d; A^*; E_0^*, E_1^*,\dots, E_d^*),\]
that satisfy the following conditions\\
(i) $A, A^*$ are both multiplicity-free elements.\\
(ii) $E_0, E_1, \dots, E_d$ is an ordering of the primitive idempotents of $A$.\\
(iii) $E_0^*, E_1^*,\dots, E_d^*$ is an ordering of the primitive idempotents of $A^*$.\\
(iv) $E_iA^*E_j=\left\{ \begin{array}{l}
0 ~~~~~~~~\text{if}~|i-j|>1\\
\ne 0 ~~~~~\text{if}~|i-j|=1
\end{array}\right. ~~(0\le i,j\le d)$. \\
(v) $E_i^*AE_j^*=\left\{ \begin{array}{l}
0 ~~~~~~~~\text{if}~|i-j|>1\\
\ne 0 ~~~~~\text{if}~|i-j|=1
\end{array}\right. ~~(0\le i,j\le d)$. 
\end{defn}
Ones refer to $d$ as the diameter of $\Phi$, and say $\Phi$ is over $\mathbb{K}$.\\

\begin{rem}
Let $\Phi=(A; E_0, E_1, \dots, E_d; A^*; E_0^*, E_1^*,\dots, E_d^*)$ denote a Leonard system on $V$, and let $\varphi$ denote an automorphism of $End(V)$. Then $\Psi=(\varphi(A); \varphi(E_0), \varphi(E_1), \dots, \varphi(E_d);$ $\varphi(A^*); \varphi(E_0^*), \varphi(E_1^*),\dots, \varphi(E_d^*))$ is also a Leonard system on $V$.
\end{rem}
The Leonard pairs $\Phi$ and $\Psi$ are said to be isomorphic whenever there exists an automorphism of $End(V)$ such that $\varphi(\Phi)=\Psi$.\\

A given Leonard system can be modified in several ways to get a new Leonard system. For instant, let $\Phi=(A; E_0, E_1, \dots, E_d; A^*; E_0^*, E_1^*,\dots, E_d^*)$ denote a Leonard system on $V$, let $\alpha, \alpha^*, \beta, \beta^*$ denote scalars in $\mathbb{K}$ such that $\alpha\ne 0, \alpha^*\ne 0$. Then,
\[(\alpha A+\beta I; E_0, E_1, \dots, E_d;\alpha^*A^*+\beta^*I; E_0^*, E_1^*,\dots, E_d^*) \]
is a Leonard system on $V$. Also,
\begin{eqnarray}
\Phi^*&:=&(A^*; E_0^*, E_1^*,\dots, E_d^*; A; E_0, E_1, \dots, E_d),\\
\Phi^{\downarrow}&:=&(A; E_0, E_1, \dots, E_d; A^*; E_d^*, E_{d-1}^*,\dots, E_0^*),\\
\Phi^{\Downarrow}&:=&(A; E_d, E_{d-1}, \dots, E_0; A^*; E_0^*, E_1^*,\dots, E_d^*)
\end{eqnarray}
are Leonard systems on $V$.\\

From the conditions $(ii)$ and $(iv)$ of Definition \ref{dleosystem}, with respect to an appropriate basis consisting of eigenvectors for $A$, the matrix representing $A^*$ is irreducible tridiagonal. Similarly, from the conditions $(iii)$ and $(v)$ of Definition \ref{dleosystem}, with respect to an appropriate basis consisting of eigenvectors for $A^*$, the matrix representing $A$ is irreducible tridiagonal. It means that in the Leonard system $\Phi$, $(A, A^*)$ is a Leonard pair. Inversely, if $(A, A^*)$ is a Leonard pair, and $E_i, i =0,\dots,d$ (resp.$E_i^*$) denotes an ordering of the primitive idempotents of $A$ (resp.$A^*$) corresponding with the basis of $V$ in the condition $(i)$ (resp.$(ii)$) of Definition \ref{dleo}, then $\Phi=(A; E_0, E_1, \dots, E_d; A^*; E_0^*, E_1^*,\dots, E_d^*)$ is a Leonard system.\\

Let $(A, A^*)$ be a Leonard pair on $V$. Let $\theta_0, \theta_1,\dots,  \theta_d$ (resp $\theta^*_0, \theta^*_1,\dots,\theta_d^*$) denote eigenvalues of $A$ (resp. $A^*$) such that the corresponding eigenvectors $v_0, v_1, \dots, v_d$ (resp. $v_0^*, v_1^*, \dots, v_d^*$) satisfy the condition \textit{(i)} (resp. \textit{(ii)}) in Definition \ref{dleo}. We refer to $\theta_0,$, $\theta_1$, $\dots$, $\theta_d$ as the eigenvalue sequence of $\Phi$, and $\theta^*_0, \theta^*_1,\dots,\theta_d^*$ as the dual eigenvalue sequence of $\Phi$. Let $E_i$ (resp. $E_i^*$) denote the primitive idempotent of $A$ (resp. $A^*$) associated with $\theta_i$ (resp. $\theta_i^*$), $0 \le i \le d$. Let $a_0, a_1,\dots, a_d$ be the diagonal of the matrix representing $A$ with respect to the basis $v_0^*, v_1^*,\dots, v_d^*$. Let $a_0^*, a_1^*,\dots, a_d^*$ be the diagonal of the matrix representing $A^*$ with respect to the basis $v_0, v_1,\dots, v_d$. As we know, $a_i = tr(E_i^*A),~ a_i^* = tr(E_iA^*)$ for $ 0\le i\le d$.
\begin{thm}\cite{Ter06}
Let $\Phi$ denote a Leonard system from Definition \ref{dleosystem}. Then the elements $A^rE_0^*A^s~ (0 \le r, s \le d)$ form a basis for the $\mathbb{K}$-vector space $End(V)$.
\end{thm}
\begin{cor}
Let $\Phi$ denote a Leonard system from Definition \ref{dleosystem}. Then the elements $A, E_0^*$ together generate $End(V)$. Moreover the elements $A, A^*$ together generate $End(V)$.
\end{cor}
\begin{cor}
Let $\Phi$ denote a Leonard system from Definition \ref{dleosystem}. Let $D$ denote the subalgebra of $End(V)$ generated by $A$. Let $X_0, X_1, \dots, X_d$ denote a basis for the $\mathbb{K}$-vector space $D$. Then the elements $X_rE_0^*X_s ~ (0 \le r, s \le d)$ form a basis for the $\mathbb{K}$-vector space $End(V)$.
\end{cor}
\begin{lem}
For $0 \le i, j \le d$, the following hold\\
(i)~~~ $E_i^*AE_j^*=\left\{ \begin{array}{l}
0 ~~~~~~~\text{if}~ |i-j|>1,\\
\ne 0 ~~~~\text{if}~ |i-j|=1,\\
a_i E_i^*~~~ \text{if}~ i =j.
\end{array}\right.$\\
(ii)~~~ $E_iA^*E_j=\left\{ \begin{array}{l}
0 ~~~~~~~\text{if}~ |i-j|>1,\\
\ne 0 ~~~~\text{if}~ |i-j|=1,\\
a_i^* E_i~~~ \text{if}~ i =j.
\end{array}\right.$\\
where $a_i = tr(E_i^*A), a_i^* = tr(E_iA^*)~ (0 \le i \le d)$
\end{lem}
\begin{lem}\cite{TV}
\label{lll}
For $0 \le r, s, i, j \le d$, the following hold\\
(i) ~~ $E_i^*A^rE_j^*= \left \{\begin{array}{l}
0 ~~~~~~\text{if}~~ 0 \le r < |i-j|,\\
\ne 0 ~~~\text{if}~~ r = |i -j|.
\end{array}
\right.$\\
(ii)~~ $E_i^*A^rA^*A^sE_j^* = \left\{\begin{array}{l}
\theta^*_{j+s}E_i^*A^{r+s}E_j^*, ~~~~\text{if}~ i-j = r + s,\\
\theta^*_{j-s}E_i^*A^{r+s}E_j^*, ~~~~\text{if}~ j-i = r + s,\\
0, ~~~~~~~~~~~~~~~~~~~~~\text{if}~|i-j| > r+s.
\end{array} \right.$
\end{lem}
Obviously, the role of $A$ and $A^*$ in the Leonard pair $(A, A^*)$ are the same, hence the properties similar to Lemma \ref{lll} are also given for $A^*$. \\

Based on the relations between $A, A^*$ and the primitive idempotents $E_i, E_i^*$, the eigenvalues $\theta_i, \theta_i^*$, P. Terwilliger found elements of the algebra generated by $A, A^*$ commuting with elements in the subalgebra of $End(V)$ generated by $A$ and elements in the subalgebra of $End(V)$ generated by $A^*$. As a result, $A, A^*$ satisfy the defining relations of a tridiagonal algebra.
\begin{thm}\cite{Ter03}
\label{tthm1}
Let $V$ denote a finite dimensional vector space over a field $\mathbb{K}$. Let $(A, A^*)$ denote a Leonard pair on $V$. There exists a sequence of scalars $ \beta, \gamma, \gamma^*, \delta, \delta^*$ in $\mathbb{K}$ such that $A, A^*$ satisfy the tridiagonal relations (\ref{tri1}), (\ref{tri2}). Furthermore, if the dimension of $V$ is at least 3, then the sequence is uniquely determined.
\end{thm}
The eigenvalues $\theta_i$ (resp. $\theta_i^*$) of $A$ (resp. $A^*$) satisfy recurrence relations below\\
\textit{(i)} ~The expressions $\frac{\theta_{i-2}-\theta_{i+1}}{\theta_{i-1}-\theta_i}$ and $\frac{\theta^*_{i-2}-\theta^*_{i+1}}{\theta^*_{i-1}-\theta^*_i}$ are independent of $i$ and equal to $\beta +1$, for all $2 \le i \le d-1$.\\
\textit{(ii)}~ $\theta_{i+1}-\beta\theta_i+\theta_{i-1}=\gamma, ~~\theta^*_{i+1}-\beta\theta^*_i+\theta^*_{i-1}=\gamma^*, ~~ 1 \le i \le d$.\\ 
\textit{(iii)}~ $\theta^2_i-\beta\theta_i\theta_{i-1}+\theta_{i-1}^2-\gamma(\theta_i+\theta_{i-1})= \delta,~~{\theta^*}^2_i-\beta{\theta}^*_i{\theta}^*_{i-1}+{\theta^*}_{i-1}^2-\gamma^*({\theta}^*_i+{\theta}^*_{i-1})= \delta^*, ~~ 1 \le i \le d$.\\

Theorem \ref{tthm1} is extended as follows.
\begin{thm}\cite{TV}
\label{lthm2}
Let $V$ denote a finite dimensional vector space over a field $\mathbb{K}$. Let $(A, A^*)$ denote a Leonard pair on $V$. There exists a sequence of scalars $ \beta, \gamma, \gamma^*, \delta, \delta^*, \omega, \eta , \eta^*$ in $\mathbb{K}$ such that both
\begin{eqnarray}
\label{AW1}
A^2A^*-\beta AA^*A+A^*A^2-\gamma(AA^*+A^*A)-\delta A^*&=&\gamma^*A^2+\omega A+\eta I,\\
\label{AW2}
{A^*}^2A-\beta A^*AA^*+A{A^*}^2-\gamma^*(A^*A+AA^*)-\delta^* A&=&\gamma {A^*}^2+\omega A^*+\eta^* I. ~~~~
\end{eqnarray}
If the dimension of $V$ is at least 4, then the sequence is uniquely determined.
\end{thm}
The relations (\ref{AW1}), (\ref{AW2}) are called the Askey-Wilson relations.\\

There is no doubt that if $A, A^*$ is a Leonard pair then $A, A^*$ satisfy the tridiagonal relations and the scalars $\beta, \gamma, \gamma^*, \delta, \delta^*$ in Theorem \ref{lthm2} coincide with the scalars of the tridiagonal relations. If the dimension of $V$ is at least 4, then they are uniquely defined by the recurrent relations. By \cite{TV} for the scalars $\omega, \eta, \eta^*$, they are obtained by the following\\
$(i)~~ \omega= a^*_i(\theta_i-\theta_{i+1})+a^*_{i-1}(\theta_{i-1}-\theta_{i-2})-\gamma^*(\theta_i+\theta_{i-1}), ~~2 \le i \le d-1$.\\
$(ii)~~ \omega= a_i(\theta^*_i-\theta^*_{i+1})+a_{i-1}(\theta^*_{i-1}-\theta^*_{i-2})-\gamma(\theta^*_i+\theta^*_{i-1}), ~~2 \le i \le d-1$.\\
$(iii)~~\eta = a^*_i(\theta_i-\theta_{i-1})(\theta_i-\theta_{i+1})-\gamma^*{\theta}_i^2-\omega\theta_i, ~~1 \le i \le d-1$.\\
$(iv)~~\eta^* = a_i(\theta^*_i-\theta^*_{i-1})(\theta^*_i-\theta^*_{i+1})-\gamma{\theta^*}_i^2-\omega\theta^*_i, ~~1 \le i \le d-1$.\\

We now show conditions in order to get a Leonard pair. Indeed, it is that a kind of converse of Theorem \ref{lthm2}.
\begin{thm}\cite{TV}
Let V denote a finite dimensional vector space on a field $\mathbb{K}$. Let $A, A^*$ denote linear transformations from $V$ to $V$. Suppose that:\\
(i) There exists a sequence of scalars $\beta, \gamma, \gamma^*, \delta, \delta^*, \omega, \eta, \eta^*$ in $\mathbb{K}$ which satisfy (\ref{AW1}), (\ref{AW2}).\\
(ii) There exists a scalar $q$ not a root of unity such that $q+q^{-1}=\beta$.\\
(iii) Each of $A, A^*$ is multiplicity-free.\\
(iv) There does not exist a subspace $W \subseteq V$ such that $AW \subseteq W, A^*W \subseteq W$ other than $0$ and $V$.\\
Then $(A, A^*)$ is a Leonard pair on $V$.
\end{thm}
\begin{lem}\cite{N}
Let $A, A^*$ denote linear transformations from $V$ to $V$. Suppose
\[A^2A^*-\beta AA^*A+A^*A^2-\gamma(AA^*+A^*A)-\delta A^*= P(A)\]
holds for some polynomial $P(x)$ in $\mathbb{K}(x)$, and for some scalars $\beta, \gamma, \delta$ in $\mathbb{K}$. The eigenspaces of $A$ and the eigenspaces of $A^*$ all have one common dimension.\\
Furthermore if $\mathbb{K}$ is algebraically closed or $q$ is not a root of unity, where $\beta = q+q^{-1}$, then the eigenspaces of $A$ and the eigenspaces of $A^*$ all have dimension 1.
\end{lem}
\begin{thm}\cite{N}
Suppose $\mathbb{K}$ is algebraically closed. Then the following are equivalent.\\
i) $(A, A^*)$ is a Leonard pair on $V$.\\
ii) $A, A^*$ satisfy the Askey-Wilson relations (\ref{AW1})-(\ref{AW2}) for some scalars $\beta, \gamma, \gamma^*, \delta, \delta^*$, $\omega, \eta, \eta^*$. 

\end{thm}
Now we move on to important properties about the classification of the Leonard pairs.
\begin{thm}\cite{T05}
Let $V$ denote a finite dimensional vector space over a field $\mathbb{K}$. Let $\Phi=(A; E_0, E_1, \dots, E_d; A^*; E_0^*, E_1^*,\dots, E_d^*)$ be a Leonard system. There exists an automorphism $\psi$ of $End(V)$, and there exist scalars $\varphi_1, \varphi_2, \dots, \varphi_d$ in $\mathbb{K}$ such that 
\[\psi \left( A \right) = \left[ {\begin{array}{*{20}{c}}
   {{\theta _0}} & {} & {} & {} & {} & 0  \\
   1 & {{\theta _1}} & {} & {} & {} & {}  \\
   {} & {1} & {{\theta _2}} & {} & {} & {}  \\
   {} & {} &  \cdot  &  \cdot  & {} & {}  \\
   {} & {} & {} &  \cdot  &  \cdot  & {}  \\
   0 & {} & {} & {} & 1 & {{\theta _d}}  \\
\end{array}}\right], ~~~ \psi(A^*)=\left[\begin{array}{*{20}{c}}
\theta^*_0 & \varphi_1 & {} & {} & {} & 0 \\
{} & \theta_1^* & \varphi_2 & {} &{}& {}\\
{}&{}&{\theta_2^*}&\cdot&{}&{}\\
{}&{}&{}&{\cdot}&\cdot{}&{}\\
{}&{}&{}&{}&{\cdot}&{\varphi_d}\\
0&{}&{}&{}&{}&\theta_d^*\\
\end{array} \right].\]
The sequence $\psi, \varphi_1, \dots, \varphi_d$ is uniquely determined by $\Phi$. Moreover, $\varphi_i \ne 0$ for $1 \le i \le d$.
\end{thm}
The Leonard system $\psi(\Phi)=(\psi(A);\psi(E_0), \psi(E_1),\dots,\psi(E_d);\psi(A^*);\psi(E_0^*), \psi(E_1^*),\dots, $ $\psi(E_d^*))$ is called the split canonical form of $\Phi$, and the sequence $\varphi_1, \varphi_2,\dots, \varphi_d$ is called the first split sequence of $\Phi$.
\begin{lem}
Let $\Phi$ and $\Phi'$ denote Leonard systems over $\mathbb{K}$. Then the following are equivalent:\\
(i) $\Phi$ and $\Phi'$ are isomorphic.\\
(ii) $\Phi$ and $\Phi'$ share the same eigenvalue sequence, dual eigenvalue sequence, and first split sequence.
\end{lem}
Let $\phi_1, \phi_2, \dots, \phi_d$ denote the first split sequence of $\Phi^{\Downarrow}$. We call $\phi_1, \phi_2, \dots, \phi_d$ the second split sequence of $\Phi$.
\begin{lem}
Let $\Phi$ and $\Phi'$ denote Leonard systems over $\mathbb{K}$. Then the following are equivalent:\\
(i) $\Phi$ and $\Phi'$ are isomorphic.\\
(ii) $\Phi$ and $\Phi'$ share the same eigenvalue sequence, dual eigenvalue sequence, and second split sequence.
\end{lem}
\begin{thm}\cite{T05}
\label{clas}
Let $d$ denote a nonnegative integer, let $\mathbb{K}$ denote a field, and let
\[\theta_0, \theta_1, \dots, \theta_d; ~~\theta_0^*, \theta_1^*, \dots, \theta_d^*;\]
\[\varphi_1, \varphi_2, \dots, \varphi_d;~~\phi_1, \phi_2, \dots, \phi_d\]
denote scalars in $\mathbb{K}$. Then there exists a Leonard system $\Phi$ over $\mathbb{K}$ with eigenvalues sequence $\theta_0, \theta_1, \dots, \theta_d$, dual eigenvalue sequence $\theta_0^*, \theta_1^*, \dots, \theta_d^*$, first split sequence $\varphi_1, \varphi_2, \dots, \varphi_d$, and second split sequence $\phi_1, \phi_2, \dots, \phi_d$ if and only if (i)-(v) hold below:\\
(i) $\varphi_i \ne 0,~ \phi_i \ne 0~~ (1 \le i \le d)$.\\
(ii) $\theta_i \ne \theta_j,~\theta^*_i \ne \theta^*_j,~~\text{if}~ i\ne j~~(0\le i, j \le d) $.\\
(iii) $\varphi_i = \phi_1\sum\limits_{h=0}^{i-1}{\frac{\theta_h-\theta_{d-h}}{\theta_0-\theta_d}}+(\theta_i^*-\theta_0^*)(\theta_{i-1}-\theta_d)~~ (1 \le i \le d).$\\
(iv) $\phi_i = \varphi_1\sum\limits_{h=0}^{i-1}{\frac{\theta_h-\theta_{d-h}}{\theta_0-\theta_d}}+(\theta_i^*-\theta_0^*)(\theta_{d-i+1}-\theta_0)~~ (1 \le i \le d).$\\
(v) The expressions
\[\frac{\theta_{i-2}-\theta_{i+1}}{\theta_{i-1}-\theta_{i}},~~\frac{\theta^*_{i-2}-\theta^*_{i+1}}{\theta^*_{i-1}-\theta^*_{i}}\]
are equal and independent of $i$ for $2\le i \le d-1$.\\
Moreover, if (i)-(v) hold above, then $\Phi$ is unique up to isomorphism of Leonard systems.
\end{thm}
We display the `most general' solution for the parameters in Theorem \ref{clas}.\\

For $0 \le i \le d,$
\begin{eqnarray}
\label{p1}
\theta_i&=&\theta_0+\frac{h(1-q^i)(1-sq^{i+1})}{q^i},\\
\label{p2}
\theta_i^*&=&\theta_0^*+\frac{h^*(1-q^i)(1-s^*q^{i+1})}{q^i}.
\end{eqnarray}

For $1 \le i \le d,$
\begin{eqnarray}
\label{p3}
\varphi_i=hh^*q^{1-2i}(1-q^i)(1-q^{i-d-1})(1-r_1q^i)(1-r_2q^i),\\
\label{p4}
\phi_i=\frac{hh^*q^{1-2i}(1-q^i)(1-q^{i-d-1})(r_1-s^*q^i)(r_2-s^*q^i)}{s^*},
\end{eqnarray}
where $q, h, h^*, r_1, r_2, s, s^*$ are nonzero scalars in the algebraic closure of $\mathbb{K}$ such that $r_1r_2=ss^*q^{d+1}$. For this solution the common value of $(v)$ in Theorem \ref{clas} equals $q+q^{-1}+1$.
\begin{cor}
Let $\Phi$ denote a Leonard system over $\mathbb{K}$ with diameter $d \ge 3$, eigenvalue sequence $\theta_0, \theta_1, \dots, \theta_d$, dual eigenvalue sequence $\theta_0^*, \theta_1^*, \dots, \theta_d^*$, first split sequence $\varphi_1, \varphi_2, \dots, \varphi_d$, and second split sequence $\phi_1, \phi_2, \dots, \phi_d$. Consider a sequence $\mathcal{L} $ of nine parameters consisting of $\theta_0, \theta_1, \theta_2, \theta_3, \theta_0^*, \theta_1^*, \theta_2^*, \theta_3^*$ and one of the parameters $\varphi_1, \varphi_d, \phi_1, \phi_d$. Then the isomorphism class of $\Phi$ as a Leonard system over $\mathbb{K}$ is determined by $\mathcal{L}$. 
\end{cor}
\begin{cor}
Let $d$ denote a nonnegative integer, and let $A$ and $A^*$ denote linear transformations in $End(V)$ of the form
\[A  = \left[ {\begin{array}{*{20}{c}}
   {{\theta _0}} & {} & {} & {} & {} & 0  \\
   1 & {{\theta _1}} & {} & {} & {} & {}  \\
   {} & {1} & {{\theta _2}} & {} & {} & {}  \\
   {} & {} &  \cdot  &  \cdot  & {} & {}  \\
   {} & {} & {} &  \cdot  &  \cdot  & {}  \\
   0 & {} & {} & {} & 1 & {{\theta _d}}  \\
\end{array}}\right], ~~~ A^*=\left[\begin{array}{*{20}{c}}
\theta^*_0 & \varphi_1 & {} & {} & {} & 0 \\
{} & \theta_1^* & \varphi_2 & {} &{}& {}\\
{}&{}&{\theta_2^*}&\cdot&{}&{}\\
{}&{}&{}&{\cdot}&\cdot{}&{}\\
{}&{}&{}&{}&{\cdot}&{\varphi_d}\\
0&{}&{}&{}&{}&\theta_d^*\\
\end{array} \right].\]
Then the following are equivalent:\\
(i) $(A, A^*)$ is a Leonard pair on $V$.\\
(ii) There exists a sequence of scalars $\phi_1, \phi_2, \dots, \phi_d$ taken from $\mathbb{K}$ such that (i)-(v) hold in Theorem \ref{clas}.
\end{cor}
Suppose (i) and (ii) hold above. Then
\[(A; E_0, E_1, \dots, E_d; A^*; E_0^*, E_1^*,\dots, E_d^*)\]
is a Leonard system on $V$, where $E_i$ (resp. $E_i^*$) denotes the primitive idempotent of $A$ (resp. $A^*$) associated with $\theta_i$ (resp. $\theta_i^*$), for $0 \le i \le d$. The Leonard system has eigenvalue sequence $\theta_0, \theta_1, \dots, \theta_d$, dual eigenvalue sequence $\theta_0^*, \theta_1^*, \dots, \theta_d^*$, first split sequence $\varphi_1, \varphi_2, \dots, \varphi_d$, and sequence split sequence $\phi_1, \phi_2, \dots, \phi_d$.

Leonard pairs arise naturally in Lie algebra. Consider the remark below.
\begin{defn}
Let $\mathbb{K}$ denote an algebraically closed field with characteristic zero. The algebra $sl_2(\mathbb{K})$ is a Lie algebra which has a basis $e, f, h$ satisfying 
\[[h,e]=2e, ~ [h,f]=-2f, ~ [e,f]=h, \]
\end{defn}
where $[,]$ denotes the Lie bracket.
\begin{lem}\cite{CK}
There exists a family $V_d, ~ d = 0, 1, 2\dots$ of irreducible finite dimensional $sl_2(\mathbb{K})$-modules such that the module $V_d$ has a basis $v_0, v_1, \dots, v_d$ satisfying 
\begin{eqnarray}
hv_i &=&(d-2i)v_i ~\text{for}~ 0 \le i \le d,\\
fv_i &=& (i+1)v_{i+1} ~\text{for}~ 0 \le i \le d-1,~~ fv_d = 0,\\
ev_i &=&(d-i+1)v_{i-1}~ \text{for}~ 1 \le i \le d, ~~ ev_0 =0.
\end{eqnarray}
Every irreducible finite dimensional $sl_2(\mathbb{K})$-module is isomorphic to exactly one of the modules $V_d, d = 0, 1, 2\dots$.
\end{lem}
\begin{rem}\cite{TD00}
Let $\mathbb{K}$ be an algebraically closed field with characteristic 0, let $A, A^*$ be semi-simple (diagonalizable) elements in the Lie algebra $sl_2(\mathbb{K})$, and let $V$ be an irreducible finite dimensional $sl_2(\mathbb{K})$-module. Assume $sl_2(\mathbb{K})$ is generated by the elements $A, A^*$, then the pair $A, A^*$ acts on $V$ as a Leonard pair.
\end{rem}
\begin{defn}
Let $\mathbb{K}$ denote an algebraically closed field. The algebra $U_q(sl_2)$ is an associative $\mathbb{K}$-algebra with unit 1 and is generated by the elements $e, f, k, k^{-1}$ that satisfy the following relations
\begin{eqnarray}
kk^{-1}=k^{-1}k =1,\\
ke=q^2ek, ~~kf=q^{-2}fk,\\
ef-fe=\frac{k-k^{-1}}{q-q^{-1}}.
\end{eqnarray} 
\end{defn}
\begin{lem}\cite[Lemma 6.2]{Ter01}
Let $d$ denote a nonnegative integer. There exists a family $V_{\epsilon, d},~ \epsilon \in \{1,-1\}$ of irreducible finite dimensional $U_q(sl_2)$-modules such that the module $V_{\epsilon, d}$ has a basis $u_0, u_1,\dots, u_d$ satisfying 
\begin{eqnarray}
ku_i=\epsilon q^{d-2i}u_i, ~0 \le i \le d, \\
fu_i=[i+1]_qu_{i+1}, ~0 \le i \le d-1, ~~fu_d =0,\\
eu_i=\epsilon[d-i+1]_qu_{i-1}, ~1 \le i \le d, ~~eu_0 =0.
\end{eqnarray}
Each irreducible finite dimensional $U_q(sl_2)$-module is isomorphic to exactly one of the modules $V_{\epsilon,d}$.
\end{lem}
\begin{thm}\cite[Example 6.3]{Ter01}
Let $d$ denote a nonnegative integer, and choose $\epsilon \in \{1, -1\}$. Let $e, f, k, k^{-1}$ denote generators of the algebra $U_q(sl_2)$. Let $\alpha, \beta$ denote scalars in $\mathbb{K}$. Define $A = \alpha f + \frac{k}{q-q^{-1}}, ~~A^* =\beta e+\frac{k^{-1}}{q-q^{-1}}$. Assume $\epsilon\alpha\beta$ is not among $q^{d-1}, q^{d-3},\dots, q^{1-d}$, then the pair $A, A^*$ acts on the irreducible finite dimensional $U_q(sl_2)$-module $V_{\epsilon, d}$ as a Leonard pair. 
\end{thm}

\subsection{Tridiagonal pairs}
A more general object called a tridiagonal pair is now considered. The concept of a tridiagonal pair is implicit in \cite[page 263]{BI}, \cite{Leo1} and more explicit in \cite[Theorem 2.1]{Terc1}. A systematic study began in \cite{TD00}. As research progressed, connections were found to representation theory \cite{AC, B1, IT, IT03, IT2, IT004, Ter03, TV}, partially ordered sets \cite{Ter01}, the bispectral problem \cite{GH1, GH2, GLZ, Zhe}. Parallel with these progress, tridiagonal pairs appeared in statistical mechanical models \cite{B1, Bas2, BK2, BK1, BK3, Bas3} and classical mechanics \cite{ZK}.\\

Let $V$ denote a vector space over a field $\mathbb{K}$ with finite
positive dimension. Let $End(V)$ denote the $\mathbb{K}-$algebra of all linear transformations from $V$ to $V$. Let $A$ denote a diagonalizable element of $End(V)$. Let $\{V_i\}_{i=0}^d$ denote an ordering of the eigenspaces of $A$ and let $\{\theta_i\}_{i=0}^d$ denote the corresponding ordering of the eigenvalues of $A$. For $0 \leq i\leq d$, define $E_i \in End(V)$ such that $(E_i-I)V_i=0$ and $E_iV_j=0$ for $j \not=i (0 \leq i \leq d)$. Here $I$ denotes the identity of $End(V)$. We call $E_i$ the \textit{primitive idempotent} of $A$ corresponding to $V_i$ (or $\theta_i$). 

\begin{lem}
The sequence $\{E_i\}_{i=0,\dots,d}$ satisfies the following properties,\\
 (i)~~ $I=\sum_{i=0}^d{E_i}$;\\
 (ii)~~$ E_iE_j=\delta_{i,j}E_i ~~(0 \leq i,j \leq d)$;\\ (iii)~~$ V_i=E_iV ~~(0
\leq i \leq d)$;\\
 (iv)~~$ A=\sum_{i=0}^d{\theta_iE_i}$;\\
 (v)~~ $E_0, E_1, \dots, E_d$ is a basic for subalgebra of $End(V)$ generated by $A$.
 \end{lem}
\begin{proof}
For all $v \in V$, there exist $a_j\in \mathbb{K}, v_j \in V_j, j=0,\dots,d$ such that $v=\sum_{j=0}^d{a_jv_j} $.\\
(i) $\sum_{i=1}^d{E_i}v = \sum_{i=1}^d{E_i\sum_{j=0}^d{a_jv_j}}=\sum_{i=1}^d{E_i(a_iv_i)}=\sum_{i=0}^d{a_iv_i}=v$. \\
(ii) $E_iE_jv=E_i(E_j(\sum_{k=0}^d{a_kv_k}))=E_i(a_jv_j)=\delta_{ij} E_i(\sum_{k=0}^d{a_kv_k})=\delta_{ij}E_iv$.\\
(iv) $Av=A\sum_{j=0}^d{a_jv_j}=\sum_{j=0}^d{a_j\theta_jv_j}$;~~ $\sum_{i=0}^d{\theta_iE_i}v=\sum_{i=0}^d{\theta_iE_i}\sum_{j=0}^d{a_jv_j}=\sum_{j=0}^d{a_j\theta_jv_j}$.\\
(v) Let $D$ denote the subalgebra of $End(V)$ generated by $A$. We have $D = \{\sum\limits_{i=0}^{n}{a_iA^i}|n\in \mathbb{N},a_i \in K, 0\le i \le n \}$. Since $A$ is diagonalizable, and $A$ has $d+1$ eigenvalues $\theta_0, \theta_1,\dots,\theta_d$, then the order of the minimal polynomial of $A$ is $d+1$. Hence, $D = \{\sum\limits_{i=0}^{d}{a_iA^i}|a_i \in K, 0\le i \le d \}$. Moreover, $E_0, E_1,\dots, E_d \in D$, and they are independent. The proof is straightforward. 
\end{proof}
\begin{defn}\label{defitri}\cite[Definition 1.1]{TD00}
Let $V$ denote a vector space over a field $\mathbb{K}$ with finite
positive dimension. 
By a {\it tridiagonal pair} (or {\it $TD$ pair})
on $V$
we mean an ordered pair of linear transformations
$A:V \to V$ and 
$A^*:V \to V$ 
that satisfy the following four conditions.
\begin{itemize}
\item[(i)] Each of $A,A^*$ is diagonalizable.
\item[(ii)] There exists an ordering $\lbrace V_i\rbrace_{i=0}^d$ of the  
eigenspaces of $A$ such that 
\begin{equation}
A^* V_i \subseteq V_{i-1} + V_i+ V_{i+1} \qquad \qquad 0 \leq i \leq d,
\label{eq:t1}
\end{equation}
where $V_{-1} = 0$ and $V_{d+1}= 0$.
\item[(iii)] There exists an ordering $\lbrace V^*_i\rbrace_{i=0}^{\delta}$ of
the  
eigenspaces of $A^*$ such that 
\begin{equation}
A V^*_i \subseteq V^*_{i-1} + V^*_i+ V^*_{i+1} 
\qquad \qquad 0 \leq i \leq \delta,
\label{eq:t2}
\end{equation}
where $V^*_{-1} = 0$ and $V^*_{\delta+1}= 0$.
\item[(iv)] There does not exist a subspace $W$ of $V$ such  that $AW\subseteq W$,
$A^*W\subseteq W$, $W\not=0$, $W\not=V$.
\end{itemize}
 \label{def1}
\end{defn}
Let $(A, A^*)$ denote a TD pair on $V$, the integers $d$ and $\delta$ from (ii), (iii) are equal and called the diameter of the pair. We will prove it later. An ordering of the eigenspaces of $A$ (resp. $A^*$) is said to be \textit{standard} whenever it satisfies (\ref{eq:t1}) (resp. (\ref{eq:t2})). Let $\{V_i\}_{i=0}^d$ (resp. $\{V_j^*\}_{j=0}^{\delta}$) denote a standard ordering of the eigenspaces of $A$ (resp. $A^*$). For $0\leq i \leq d, 0 \le j \le \delta$, let $\theta_i$ (resp. $\theta^*_j$) denote the eigenvalue of $A$ (resp. $A^*$) associated with $V_i$ (resp. $V_j^*$). By \cite{TD00}, for $0 \le i \le d$ the spaces $V_i, V_i^*$ have the same dimension; we denote this common dimension by $\rho_i$. By the construction $\rho_i \ne 0, i = 1, \dots, d$, and the sequence $\rho_0, \rho_1,\dots, \rho_d$ is symmetric and unimodal; that is $\rho_i = \rho_{d-i}$ for $0 \le i \le d$ and $\rho_{i-1} \le \rho_i$ for $1 \le i \le d/2$. We refer to the sequence $(\rho_0, \rho_1, \dots, \rho_d)$ as the shape vector of $A, A^*$. In particular, the shape vector of $A, A^*$ is independent of the choice of standard orderings of the eigenspaces of $A, A^*$.\\

Now let $(A, A^*)$ denote a TD pair on V. An ordering of the primitive idempotents of $A$ (resp. $A^*$) is said to be standard whenever the corresponding ordering of the eigenspaces of $A$ (resp. $A^*$) is standard.
\begin{defn}\cite[Definition 2.1]{TD00}
Let $V$ denote a vector space over $\mathbb{K}$ with finite positive dimention. By a tridiagonal system (or TD system) on V we mean a sequence 
\[\Phi = (A;\{E_i\}_{i=0}^d; A^*;\{E_i^*\}_{i=0}^\delta) \]
that satisfies (i)--(iii) below.
\begin{itemize}
\item[(i)]
$(A,A^*)$ is a TD pair on $V$.
\item[(ii)]
$\{E_i\}_{i=0}^d$ is a standard ordering
of the primitive idempotents of $A$.

\item[(iii)]
$\{E^*_i\}_{i=0}^\delta$ is a standard ordering
of the primitive idempotents of $A^*$.
\end{itemize}
We say that $\Phi$ is {\em over} $\mathbb{K}$.
\end{defn}

Actually, $E_i$ can be written as follows,
\[{E_i} = \prod\limits_{\scriptstyle 0 \le j \le d \hfill \atop 
  \scriptstyle ~~ j \ne i \hfill} {\frac{{A - {\theta _j}I}}{{{\theta _i} - {\theta _j}}}}.\]
We have that for all $v \in V_i$,
\[E_iv=(\prod\limits_{\scriptstyle 0 \le j \le d \hfill \atop 
  \scriptstyle ~~ j \ne i \hfill} {\frac{{A - {\theta _j}I}}{{{\theta _i} - {\theta _j}}}})v=v\prod\limits_{\scriptstyle 0 \le j \le d \hfill \atop 
    \scriptstyle ~~ j \ne i \hfill} {\frac{{\theta_i - {\theta _j}}}{{{\theta _i} - {\theta _j}}}}=v.\]
 And for all $v \in V_j, j\ne i$, $E_iv=0$. Therefore, $E_i$ is the projection of $V$ on $V_i$.\\
 Our proof is complete with the observation that there exists one and only one linear transformation such that it is the projection of V onto $V_i$. \\
\begin{lem}\cite{INT}
\label{lem:triplep}
Let $(A; \lbrace E_i\rbrace_{i=0}^d; A^*; \lbrace E^*_i\rbrace_{i=0}^\delta)$
denote a TD system. Then the following hold for $0 \leq i,j,r\leq d, 0 \le h,k,r \le \delta  $.
\begin{enumerate}
\item[\rm (i)] $E^*_hA^rE^*_k=0$ if $|h-k|>r$\ .
\item[\rm (ii)] $E_iA^{*r}E_j=0$ if $|i-j|>r$\ .
\end{enumerate}
\label{lem1}
\end{lem}
\begin{proof}
We give only the proof of (i), for (ii) it is also proved similarly.\\
For all $a \in V$, we have 
\begin{eqnarray*}
E_h^*A^rE_k^*(a) & \in & E_h^*A^rV_k^* \\
&\subset& E_h^*A^{r-1}(V^*_{k-1}+V^*_{k}+V^*_{k+1})\\
&\subset& E_h^*(V_{k-r}^*+V_{k-r+1}^*+\dots+V_{k+r-1}^*+V_{k+r}^*)\\ & = & 0.~~ (\text{Since}~ |h-k| >r). 
\end{eqnarray*}
It follows directly that $E^*_hA^rE^*_k=0$ if $|h-k|>r$. 
\end{proof}
\begin{lem}\cite{TD00}
\label{bode1}
Let $(A, A^*)$ denote a TD pair on $V$, the integers $d$ and $\delta$ from (ii), (iii) of Definition (\ref{def1}) are equal.
\end{lem}
\begin{proof}
Assume without loss of generality that $\delta \le d$. Set $V_{i,j}= (\sum_{h=0}^i{E_h^*V})\cap (\sum_{k=j}^d{E_kV})$ for all integers $i,j$, with convention that $\sum_{h=0}^i{E_h^*V}=\left\{\begin{array}{l}
0 ~~\text{if }~ i <0 \\
V ~\text{if }~ i > \delta \\
\end{array}\right.$  
and $\sum_{k=j}^d{E_kV}=\left\{\begin{array}{l}
V ~\text{if }~ j <0 \\
0 ~~\text{if }~ j > d \\
\end{array}\right.$.\\

Since $E_1^*+\dots+E_{\delta}^* = I$, then $AE_h^*V=(E_1^*+\dots+E_{\delta}^*)AE_h^*V$. By (i) of Lemma \ref{lem1}, $AE_h^*V \subset E_{h-1}^*V+E_{h}^*V+E_{h+1}^*V$. 
Consequently, $(A-\theta_j I)\sum_{h=0}^i{E_h^*V}\subset \sum_{h=0}^{i+1}{E_h^*V}$.\\
Moreover, 
\begin{eqnarray*}
(A-\theta_j I)\sum_{k=j}^d{E_kV}&=&\sum_{k=j}^d{(AE_kV-\theta_jE_kV)=\sum_{k=j}^d{(\theta_kE_kV-\theta_jE_kV)}}\\ &=& \sum_{k=j+1}^d{(\theta_k-\theta_j)E_kV}=\sum_{k=j+1}^d{E_kV}.
\end{eqnarray*}
It follows that $(A-\theta_j I)V_{i,j}\subset \left( \sum_{h=0}^{i+1}{E^*_hV}\right)\cap \left( \sum_{k=j+1}^d{E_kV}\right)=V_{i+1,j+1} $.
From that, we obtain $AV_{i,j} \subset V_{i,j} + V_{i+1,j+1}$.\\

Similarly, we get $A^*V_{i,j}\subset V_{i,j}+V_{i-1,j-1}$.\\

Put $W=V_{0,r}+V_{1,r+1}+\dots+V_{d-r,d}$ for $0 < r \le d$. By the above,
\begin{eqnarray*}
AW &=& A(V_{0,r}+\dots+V_{d-r,d})\subset V_{0,r}+V_{1,r+1}+\dots+V_{d-r,d}+V_{d-r+1,d+1}=W, \\
A^*W &=& A^*(V_{0,r}+\dots+V_{d-r,d})\subset V_{-1,r-1}+V_{0,r}+\dots+V_{d-r-1,d-1}+V_{d-r,d}=W.
\end{eqnarray*}
Use the condition (iv) of Definition \ref{def1}, $W=0$ or $W=V$.\\

Clearly, $V_{0,r}, \dots, V_{d-r,d}$ are subsets of $ \sum_{h=0}^{d-r}E^*_hV$, then $W \subset \sum_{h=0}^{d-r}E^*_hV$. 
Therefore $W \ne V$ ($\sum_{h=0}^{d-r}E^*_hV$ is the proper subset of $V$).\\
Suppose that $\delta <d$, put $ 0<r = d-\delta \le d$, we have
$W=V_{0,r}+V_{1,r+1}+\dots+V_{\delta,d} = 0$.
However, $V_{\delta,d}=E_dV \le 0$. This is a contradiction. Thus, we conclude that $\delta = d$.
\end{proof}
\begin{thm}\cite[Theorem 10.1]{TD00}
\label{dinhly2}
Let $V$ denote a vector space over $\mathbb{K}$ with finite positive dimension, and let $A, A^*$ denote a TD pair on $\mathbb{K}$. There exist scalars $\beta, \gamma, \gamma^*, \delta, \delta^*$ in $\mathbb{K}$ such that the tridiagonal relations are satisfied
\begin{eqnarray}
\lbrack A,A^2A^*-\beta AA^*A + 
A^*A^2 -\gamma (AA^*+A^*A)-\delta A^*\rbrack &=& 0\ , \label{eq:dolan1}
\\
\lbrack A^*,A^{*2}A-\beta A^*AA^* + AA^{*2} -\gamma^* (A^*A+AA^*)-
\delta^* A\rbrack &=&0\ . \label{eq:dolan2} \qquad  \quad 
\end{eqnarray}
Moreover, these scalars are unique if the diameter of the pair is at least 3. 
\end{thm}
\begin{proof}
Let $\Phi = (A;\{E_i\}_{i=0}^d; A^*;\{E_i^*\}_{i=0}^d) $ denote a TD system associated with $(A, A^*)$. Let $\theta_0,\theta_1,\dots, \theta_d$ (resp. $\theta_0^*,\theta_1^*,\dots, \theta_d^*$) denote the eigenvalue sequence (resp. the dual eigenvalue sequence) of $\Phi$.\\

Put $U_i=\sum\limits_{h=0}^i{E_h^*V}\cap\sum\limits_{k=i}^d{E_kV}$, for all $0\le i\le d$, and $W={U_0+U_1+\dots+U_d}$. We first prove that $V$ is the direct sum of $U_0,U_1,\dots,U_d$. Similarly, in the proof of Lemma \ref{bode1}, $AW \subset A, A^*W \subset W$, thus $W=0$, or $W = V$. Moreover, $0 \ne U_0=E_0^*V\subset W$. Clearly, $V = W$. On the other hand, for all $0 \le i \le d-1$, $\sum\limits_{l=0}^i{U_l} \subset \sum\limits_{h=0}^i{E_h^*V}, U_{i+1}\subset \sum\limits_{i+1}^d{E_kV}$. By the proof of Lemma \ref{bode1}, $\sum\limits_{l=0}^i{U_l}\cap U_{i+1} \subset \sum\limits_{h=0}^i{E_h^*V}\cap \sum\limits_{k=i+1}^d{E_kV} = \emptyset$. This gives that $V$ is the direct sum of $U_0, U_1, \dots, U_d$.\\

Let $F_i$ denote the projection map of $V$ onto $U_i$, for all $0 \le i \le d$. Put $R= A-\sum\limits_{i=0}^d{\theta_iF_i}, L= A^*-\sum\limits_{i=0}^d{\theta_i^*F_i}$.\\

Write $D = \{\sum\limits_{k=0}^n{a_kA^k}|n \in \mathbb{N}, a_k \in \mathbb K\}$. Since $A$ is diagonalizable, and $A$ has $d+1$ eigenvalues $\theta_0, \theta_1,\dots,\theta_d$, then the minimal polynomial of $A$ is $\prod\limits_{i=0}^d{(A-\theta_i{I})}$. Hence, we can rewrite $D= \{\sum\limits_{k=0}^d{a_kA^k}|a_k\in \mathbb K \}$.\\
For all $0\le i \le d$, by Lemma \ref{lem1}, $E_iA^*=E_iA^*(\sum\limits_{j=0}^d{E_j})=E_iA^*E_{i-1}+E_iA^*E_i+E_iA^*E_{i+1}, \text{with } E_{-1} = E_{d+1}=0 $. Similarly, $A^*E_i=E_{i-1}A^*E_i+E_iA^*E_i+E_{i+1}A^*E_i$. Put $L_l=\sum\limits_{i=0}^l{E_i}$, for all $0 \le l \le d,$  \[L_lA^*-A^*L_l=\sum\limits_{i=0}^l{E_i}A^*-A^*\sum\limits_{i=0}^l{E_i}=E_lA^*E_{l+1}-E_{l+1}A^*E_l.\]
In fact, $D$ is spanned by $E_0, E_1, \dots, E_d$ (the property (v) of $\{E_{i}, i=0,\dots, d\}$), then $L_0, L_1, \dots, L_d$ 
is also a basic of $D$.
\begin{eqnarray*}
\text{Span}\{XA^*Y-YA^*X|X,Y\in D \}&=&\text{Span}\{E_iA^*E_j-E_jA^*E_i|0 \le i,j \le d \}\\ &=& \text{Span}\{E_iA^*E_{i+1}-E_{i+1}A^*E_i|0 \le i \le d \}\\ &=& \text{Span}\{L_iA^*-A^*L_i|0 \le i \le d \}\\ &=& \{ XA^*-A^*X|X \in D\}.
\end{eqnarray*} 
It follows immediately that there exist uniquely scalars $a_i \in \mathbb K ~(i =0,1,\dots,d)$ such that 
\begin{equation}
\label{pt10}
A^2A^*A-AA^*A^2=\sum\limits_{i=0}^d{a_i(A^iA^*-A^*A^i)}
\end{equation}
Assume $d \ge 3$, put $t= \text{max}\{i \in \{0, 1, \dots, d \}|a_i \ne 0 \}$. Suppose that $t \ge 4$, multiply two sides of the equation (\ref{pt10}) on the left by $F_t$, and on the right by $F_0$. By Lemma 7.3 of \cite{TD00}, we have 
\[0=a_t(F_tA^tA^*F_0-F_tA^*A^tF_0)=a_t(\theta_0^*-\theta_t^*)R^tF_0\] 
By Corollary 6.7 of \cite{TD00}, $R^tF_0 \ne 0$, then $a_t(\theta_0^*-\theta_t^*)=0$, it is a contradiction. Hence $t \le 3$, the equation (\ref{pt10}) can be rewritten 
\begin{eqnarray}
\label{pt11}
A^2A^*A-AA^*A^2 = a_3(A^3A^*-A^*A^3)+a_2(A^2A^*-A^*A^2)+a_1(AA^*-A^*A).~~
\end{eqnarray}
Suppose $a_3 =0$, multiply the two sides of (\ref{pt11}) on the left by $F_3$ and on the right by $F_0$, we get $(\theta_1^*-\theta_2^*)R^3F_0=0$. This is again a contradiction. So $a_3 \ne 0$. We can put $\beta=\frac{1}{a_3}-1, \gamma=-\frac{a_2}{a_3}, \delta=-\frac{a_1}{a_3}$. The equation (\ref{pt11}) becomes
\[(A^3A^*-A^*A^3)-(\beta+1)(A^2A^*A-AA^*A^2)-\gamma(A^2A^*-A^*A^2)-\delta(AA^*-A^*A)=0\]
or $[A,A^2A^*-\beta AA^*A+A^*A^2-\gamma(AA^*+A^*A)-\delta A^*]=0$. Multiply each term in (\ref{eq:dolan1}) on the left by $F_{i-2}$ and on the right by $F_{i+1}$, for all $2 \le i \le d-1$, we find $(\theta_{i-2}^*-\theta_{i+1}^*-(\beta +1)(\theta_{i-1}^*-\theta_i^*))R^3F_{i-2}=0$. Consequently, $\theta_{i-2}^*-\theta_{i+1}^*-(\beta +1)(\theta_{i-1}^*-\theta_i^*)=0$. It follows $\theta_{i-2}^*-\beta \theta_{i-1}^*+\theta_i^*=\theta_{i-1}^*-\beta\theta_i^*+\theta_{i+1}$, for all $2 \le i \le d$. So the expression $\theta_{i-1}^*-\beta\theta_{i}^*+\theta_{i+1}^*$ does not depend on $i$. Let us put $\gamma* = \theta_{i-1}^*-\beta\theta_{i}^*+\theta_{i+1}^*$. Consider $(\theta_{i-1}-\theta_{i+1})(\theta_{i-1}-\beta\theta_i+\theta_{i=1}-\gamma^*)=0$, then ${\theta^*_{i-1}}^2-\beta\theta^*_{i-1}\theta_{i}^*+{\theta_i^*}^2-\gamma^*(\theta_{i-1}^*+\theta_i^*)={\theta^*_{i}}^2-\beta\theta^*_{i}\theta_{i+1}^*+{\theta_{i+1}^*}^2-\gamma^*(\theta_{i}^*+\theta_{i+1}^*)$. We can also put $\delta^*={\theta^*_{i}}^2-\beta\theta^*_{i}\theta_{i+1}^*+{\theta_{i+1}^*}^2-\gamma^*(\theta_{i}^*+\theta_{i+1}^*)$ because ${\theta^*_{i}}^2-\beta\theta^*_{i}\theta_{i+1}^*+{\theta_{i+1}^*}^2-\gamma^*(\theta_{i}^*+\theta_{i+1}^*)$ is independent of $i$. According to Lemma 9.3 of \cite{TD00}, we get (\ref{eq:dolan2}). Since $\beta, \gamma, \theta$ are determined uniquely, and $\gamma^*, \delta^*$ satisfy (\ref{eq:dolan2}) then $\gamma^*, \delta^*$ are also determined uniquely. We already finish proving for $d$ at least equals 3, then $\beta,\gamma,\delta $ are determined uniquely such that (\ref{eq:dolan1}), (\ref{eq:dolan2}). 

If $d =2$, let $\beta$ be any scalar in $\mathbb K$, put $\gamma= \theta_0-\beta \theta_1+\theta_2$, and $\delta = \theta_0^2-\beta\theta_0\theta_1+\theta_1^2-\gamma(\theta_{0}+\theta_1)$. If $d=1$, let $\beta, \gamma$ be any scalars in $\mathbb K$, put $\delta = \theta_0^2-\beta\theta_0\theta_1+\theta_1^2-\gamma(\theta_{0}+\theta_1)$. If $d=0$, let $\beta, \gamma, \delta$ be any scalars in $\mathbb K$. By Lemma 9.3 of \cite{TD00}, we get (\ref{eq:dolan1}). Similarly, there exist scalars $\gamma^*, \delta^*$ for (\ref{eq:dolan2}).
\end{proof}
\begin{cor}
\cite{Ter03}\label{ch}
Let $(A, A^*)$ be a TD pair on $V$, with eigenvalue sequence  $\theta_0, \theta_1, \dots, \theta_d$ and dual eigenvalue sequence $\theta_0^*, \theta_1^*,\dots ,\theta_d^*$. Then the expressions  
\begin{eqnarray}
\frac{\theta_{i-2}-\theta_{i+1}}{\theta_{i-1}-\theta_i}; ~~~~~\frac{\theta^*_{i-2}-\theta^*_{i+1}}{\theta^*_{i-1}-\theta^*_i}
\end{eqnarray}
are equal and independent of $i$ for $2\le i \le d-1$
\end{cor}
\begin{proof}
This lemma is straightforward from the equations $\theta_{i-2}-\theta_{i+1}-(\beta+1)(\theta_{i-1}-\theta_i)=0$, and $\theta^*_{i-2}-\theta^*_{i+1}-(\beta+1)(\theta^*_{i-1}-\theta^*_i)=0$ in the proof of the above theorem. 
\end{proof}
\begin{cor}
\label{hequa1}
Let $(A, A^*)$ be a TD pair on $V$, with eigenvalue sequence  $\theta_0, \theta_1, \dots, \theta_d$ and dual eigenvalue sequence $\theta_0^*, \theta_1^*,\dots ,\theta_d^*$. For all $~0 \le i \le d$, there exist scalars $\alpha_1, \alpha_2, \alpha_3, \alpha_1^*, \alpha^*_2, \alpha^*_3$ in $\mathbb K$ and $q$ in the algebraic closure of $\mathbb K$, $q \ne 0, q \ne \pm 1$ such that
\begin{eqnarray}
\theta_i&=&\alpha_1 +\alpha_2q^{i}+\alpha_3q^{-i}, \\
\theta^*_i&=&\alpha^*_1 +\alpha^*_2q^{i}+\alpha^*_3q^{-i}. 
\end{eqnarray}
\end{cor}
\begin{proof}
By Corollary \ref{ch} and Theorem \ref{dinhly2}, there exists scalar $\beta$ in $\mathbb K$, such that $\frac{\theta_{i-2}-\theta_{i+1}}{\theta_{i-1}-\theta_i}=\beta+1 $, for all $2 \le i \le d-1$. Let $q$ in the algebraic closure of $\mathbb K$, such that $q+q^{-1}=\beta$. We have the equation
\[\frac{\theta_{i-2}-\theta_{i+1}}{\theta_{i-1}-\theta_i}=q+q^{-1}+1.\]
It is equivalent to 
\[\frac{\theta_{i-2}-\theta_{i-1}+\theta_{i-1}-\theta_i+\theta_i-\theta_{i+1}}{\theta_{i-1}-\theta_i}=q+q^{-1}+1.\]
Put $u_i = \theta_{i-2}-\theta_{i-1}, $ it yields 
\[u_{i+2}-(q+q^{-1})u_{i+1}+u_{i}=0.\]
Since $q \ne \pm 1$, there exist scalars $a, b \in \mathbb K: u_i=aq^i+bq^{-i}$, or $\theta_{i-2}-\theta_{i-1}=a q^i + bq^{-i}$. Solving this equation, there exist scalars $\alpha_1, \alpha_2, \alpha_3 \in \mathbb K$ such that $\theta_i=\alpha_1+\alpha_2q^i+\alpha_3q^{-i}$. By the same arguments, there also exist scalars $\alpha^*_1, \alpha^*_2, \alpha^*_3 \in \mathbb K$ such that $\theta^*_i=\alpha^*_1+\alpha^*_2q^i+\alpha^*_3q^{-i}$.
\end{proof}

Note that the following equations are inferred directly from Corollary \ref{hequa1}. If $\beta = q^2+q^{-2}$, $q^4 \not=1$, there exist scalars $\alpha, \alpha^*$ in $\mathbb{K}$, and $b, b^*, c, c^*$ in the algebraic closure of ${\mathbb{K}}$ such that
\small
\begin{eqnarray}
\label{eq:const1}
&&\theta_i = \alpha + b q^{2i-d} + c q^{d-2i} 
 \qquad 0 \leq i \leq d \ ,
\\
\label{eq:const2}
&&\theta^*_j = \alpha^* + b^* q^{2j-\delta} + c^* q^{\delta-2j} 
 \qquad 0 \leq j \leq \delta \ ,
\label{eq:const3}
\end{eqnarray}
\normalsize
where $b, b^*,c, c^*$ are nonzero scalars. 

\begin{lem}
\label{lem:polypart}
For each positive integer $s$, there exist scalars $\beta_s, \gamma_s, \gamma^*_s, \delta_s, \delta^*_s$ in the algebraic cloture of $\mathbb{K}$ such that
\begin{eqnarray}
\label{pt1}
\theta_i^2 - \beta_s\theta_i\theta_j  + \theta_j^2 - \gamma_s (\theta_i+\theta_j) -\delta_s&=& 0 \quad \mbox{if} \quad |i-j|=s \ .~~(0\leq i,j\leq d),~~~~~~~~~~\\
\label{pt2}
{\theta_h^*}^2 - \beta_s{\theta_h^*}{\theta_k^*}  + {\theta_k^*}^2 - \gamma^*_s (\theta_h^*+\theta_k^*) -\delta^*_s &=& 0 \quad \mbox{if} \quad |h-k|=s \ .~~ (0\leq h,k\leq \delta).
\end{eqnarray}
\end{lem}
\begin{proof}
There is no loss of generality of assuming $j = i+ s$. Substituting (\ref{eq:const1}) into (\ref{pt1}) gives
\begin{eqnarray*}
0 = & & b^2q^{2s-2d}(q^{2s}+q^{-2s}-\beta_s)q^{4i}+c^2q^{2d-2s}(q^{2s}+q^{-2s}-\beta_s)q^{-4i}\\&+& b(q^{2s-d}+q^{-d})(2\alpha-\beta_s\alpha-\gamma_s)q^{2i}+c(q^d+q^{d-2s})(2\alpha-\beta_s\alpha-\gamma_s)q^{-2i}\\ &+& 2\alpha^2+4bc-\beta_s(\alpha^2+bc(q^{2s}+q^{-2s}))-2\alpha\gamma_s-\delta_s
\end{eqnarray*} 
If $b, c, q \ne 0$ and $q^{2s}+q^{-2s}\ne -2$, then (\ref{pt1}) is satisfied for all $0 \le i \le d$ if and only if
\[\left \{ \begin{array}{l}
q^{2s}+q^{-2s}-\beta_s=0\\
2\alpha-\beta_s\alpha-\gamma_s=0\\
2\alpha^2+4bc-\beta_s(\alpha^2+bc(q^{2s}+q^{-2s}))-2\alpha\gamma_s-\delta_s =0\\
\end{array} \right. \]
Therefore,
\[\left \{ \begin{array}{l}
\beta_s = q^{2s}+q^{-2s}\\
\gamma_s = \alpha(2-q^{2s}-q^{-2s})\\
\delta_s = \alpha^2(q^s-q^{-s})^2-bc(q^{2s}-q^{-2s})^2\\
\end{array}
\right.\ \]
Similarly, there exist scalars $\gamma^*_s, \delta^*_s$ s.t (\ref{pt2})
\end{proof}
\begin{rem}
Let $V$ denote a finite dimensional vector space over $\mathbb{K}$, and $(A, A^*)$ be a tridiagonal pair on $V$. Let $r, r^*, s, s^*$ denote scalars in $\mathbb{K}$ such that $r,r^*$ are nonzero, then the ordered pair $(rA+sI, r^*A^*+s^*I)$ is also a tridiagonal pair on $V$. Moreover, if $\beta, \gamma, \gamma^*, \delta, \delta^*$ is a sequence of parameters in the tridiagonal relations of $A, A^*$, then $\beta, r\gamma+s(2-\beta), r^*\gamma^*+s^*(2-\beta), r^2\delta-2rs\gamma+s^2(\beta-2), {r^*}^2\delta^*-2r^*s^*\gamma^*+{s^*}^2(\beta-2)$ is the sequence of parameters in the tridiagonal relations of $rA+sI, r^*A^*+s^*I$.
\end{rem}
\begin{thm}{\cite{Ter03}}\label{dl2}
Let $\beta, \gamma, \gamma^*, \delta, \delta^*$ be scalars in $\mathbb{K}$. Let $T$ denote a tridiagonal algebra over $\mathbb{K}$ with the parameters $\beta, \gamma, \gamma^*, \delta, \delta^*$, and generators $A, A^*$. Let $V$ denote an irreducible finite dimensional $T$-module. Assume that $q$ is not a root of unity such that $\beta = q+q^{-1}$, and $A, A^*$ are diagonalizable on $V$. Then $(A, A^*)$ acts on $V$ as a tridiagonal pair.
\end{thm}
By Theorem \ref{dinhly2} and Theorem \ref{dl2}, there exists a relationship between the tridiagonal pair and the generators of the $q$-Onsager algebra as follows.
\begin{rem}
If $A, A^*$ are generators of the $q$-Onsager algebra, $V$ is an irreducible finite dimensional vector space on which $A, A^*$ act, and $A, A^*$ are diagonalizable on $V$ then $(A, A^*)$ is a tridiagonal pair on $V$. Inversely, if $(A, A^*)$ is a tridiagonal pair on $V$, and the parameters $\beta=q^2+q^{-2},~~\gamma=\gamma^* = 0$ then $A, A^*$ satisfy the defining relations of the $q$-Onsager algebra.
\end{rem}
The conditions making a tridiagonal pair become a Leonard pair are established by the next theorem
\begin{thm}\cite{Ter03}
Let V denote a finite dimensional vector space over $\mathbb{K}$. Let $A, A^*$ be linear transformations from $V$ to $V$. Then the following are equivalent.\\
(i) $(A, A^*)$ is a Leonard pair on $V$.\\
(ii) $(A, A^*)$ is a tridiagonal pair, and for each of $A, A^*$ all eigenspaces have dimension 1. 
\end{thm}
\begin{lem}
Let $(A, A^*)$ denote a tridiagonal pair on $V$. If the shape vector of $A, A^*$ satisfies $\rho_0 = \rho_1 =1$, then $A, A^*$ is Leonard pair. 
\end{lem}
\begin{cor}
Let $(A, A^*)$ denote a tridiagonal pair on $V$. If $(A, A^*)$ is not a Leonard pair, the shape vector of $(A, A^*)$ satisfies $\rho_i \ge 2 ~(1 \le i \le d-1)$.
\end{cor}
\begin{thm}\cite{CP}
\label{cc1}
The quantum affine algebra $U_q(\widehat{sl_2})$ is isomorphic to the unital associative $\mathbb{K}$-algebra $\mathbb{U}$ with generators $y_i^{\pm}, k_i^{\pm}, i \in \{0, 1\}$ and the following relations:
\begin{eqnarray*}
k_ik_i^{-1}=k_i^{-1}k_i &=& 1,\\
k_0k_1 ~~\text{is central,}\\
\frac{qy_i^+k_i-q^{-1}k_iy_i^+}{q-q^{-1}}&=&1,\\
\frac{qk_iy_i^{-}-q^{-1}y_i^-k_i}{q-q^{-1}}&=&1,\\
\frac{qy_i^{-}y_i^+-q^{-1}y_i^+y_i^-}{q-q^{-1}}&=&1,\\
\frac{qy_i^{+}y_i^--q^{-1}y_i^-y_i^+}{q-q^{-1}}&=&k_0^{-1}k_1^{-1},~~i \ne j,
\end{eqnarray*}
\[(y_i^{\pm})^3y_j^{\pm}-[3]_q(y^{\pm}_i)^2y_j^{\pm}y_i^{\pm}+[3]_qy_i^{\pm}y_j^{\pm}(y_i^{\pm})^2-y_j^{\pm}(y_i^{\pm})^3 = 0, ~ i \ne j.\]
An isomorphism from $\mathbb{U}$ to $U_q(\widehat{sl_2})$ is given by:
\begin{eqnarray*}
k_i^{\pm}&\mapsto & K_i^{\pm},\\
y_i^{-} &\mapsto& K_i^{-1}+e_i^{-},\\
y_i^+ &\mapsto & K_i^{-1}-q(q-q^{-1})^2K_i^{-1}e_i^+.
\end{eqnarray*}
The inverse of this isomorphsim is given by:
\begin{eqnarray*}
K_i^{\pm}&\mapsto&k_i^{\pm},\\
e_i^{-}&\mapsto&y_i^--k_i^{-1},\\
e_i^+&\mapsto&\frac{1-k_iy_i^+}{q(q-q^{-1})^2}.
\end{eqnarray*}
\end{thm}
\begin{thm}\cite{IT}
Let $V$ denote a vector space over $\mathbb{K}$ with finite positive dimension and let $(A, A^*)$ denote a tridiagonal pair on $V$. Let $\theta_0, \theta_1, \dots, \theta_d$ (resp. $\theta_0^*, \theta_1^*, \dots, \theta_d^*$) denote a standard ordering of the eigenvalues of $A$ (resp. $A^*$). We assume there exist nonzero scalars $a, a^*$ in $\mathbb{K}$ such that $\theta_i = aq^{2i-d}$ and $\theta_i^* = a^*q^{d-2i}$ for $0 \le i \le d$. Then with reference to Theorem \ref{cc1} there exists a unique $U_q(\widehat{sl_2})$-module structure on $V$ such that $ay_1^{-}$ acts as $A$ and $a^*y_0^{-}$ acts as $A^*$. Moreover there exists a unique $U_q(\widehat{sl_2})$-module structure on $V$ such that $ay_0^+$ acts as $A$ and $a^*y_1^+$ acts as $A^*$. Both $U_q(\widehat{sl_2})$-module structures are irreducible.
\end{thm}
\begin{defn}
Let $V$ denote a vector space over $\mathbb{K}$ with finite positive dimension. Let $d$ denote a nonnegative integer. By a decomposition of $V$ of length $d$, we mean a sequence $U_0, U_1,\dots, U_d$ consisting of nonzero subspaces of $V$ such that
\[V = U_0 +U_1+\dots+U_d~~~ (\text{direct sum}).\]
\end{defn}
Now we are concerned about six decompositions of a finite dimensional vector space $V$. Recall that $(A, A^*)$ is the tridiagonal pair on $V$. Let $V_0, V_1, \dots, V_d$ (resp. $V^*_0, V_1^*, \dots, V_d^*$) denote a standard order of the eigenspaces of $A$ (resp. $A^*$). For $0 \le i \le d$, let $\theta_i$ (resp. $\theta_i^*$) denote the eigenvalue of $A$ (resp. $A^*$) associated with $V_i$ (resp. $V_i^*$)
\begin{lem}\cite{IT}
\label{decom}
The sequence of subspaces $U_0, U_1,\dots, U_d$ of $V$ is a decomposition of $V$ if one of the six following is satisfied\\
i) $U_i = V_i,~~~ i =0, 1, \dots, d,$\\
ii) $U_i = V_i^*, ~~~ i =0, 1, \dots, d$\\
iii) $U_i = (V_0^*+V_1^*+\dots+V_i^*)\cap (V_i+V_{i+1}+\dots+V_{d}), ~~~ i = 0, 1, \dots, d,$\\
iv) $U_i = (V_0^*+V_1^*+\dots+V_i^*)\cap (V_0+V_1+\dots+V_{d-i}),~~~ i = 0, 1, \dots, d,$\\
v) $U_i = (V^*_{d-i}+V^*_{d-i+1}+\dots+V^*_d)\cap(V_0+V_1+\dots+V_{d-i}), ~~~i = 0, 1, \dots, d,$\\
vi) $U_i = (V^*_{d-i}+V^*_{d-i+1}+\dots+V^*_d)\cap(V_i+V_{i+1}+\dots+V_{d}), ~~~i = 0, 1, \dots, d.$
\end{lem}
\begin{lem}\cite{IT}
Let $U_0, U_1, \dots, U_d$ denote one of the six decompositions of $V$ associated with the tridiagonal pair $(A, A^*)$ in Lemma \ref{decom}. For $0 \le i \le d,$ let $\rho_i$ denote the dimension of $U_i$. Then the sequence $\rho_0, \rho_1, \dots, \rho_d$ is independent of the decomposition. Moreover the sequence $\rho_0, \rho_1,\dots, \rho_d$ is unimodal and symmetric.
\end{lem}
\begin{lem}\cite{IT}
Let $U_0, U_1, \dots, U_d$ denote any one of the six decompositions of $V$ associated with a tridiagonal pair $(A, A^*)$ in Lemma \ref{decom}. Then for $0 \le i \le d$ the action of $A$ and $A^*$ on $U_i$ is described as follows.\\
i) If $U_i = V_i$ then $(A-\theta_iI)U_i = 0,~ A^*U_i \subseteq U_{i-1}+U_i+U_{i+1},$\\
ii) If $U_i = V_i^*$ then $AU_i \subseteq U_{i-1}+U_i+U_{i+1},~(A^*-\theta^*I)U_i =0,$\\
iii) If $U_i = (V_0^*+V_1^*+\dots+V_i^*)\cap (V_i+V_{i+1}+\dots+V_{d})$ then $(A-\theta_iI)U_i \subseteq U_{i+1},~(A^*-\theta_i^*I)U_i\subseteq U_{i-1},$\\
iv) If $U_i = (V_0^*+V_1^*+\dots+V_i^*)\cap (V_0+V_1+\dots+V_{d-i})$ then $(A-\theta_{d-i}I)U_i\subseteq U_{i+1}, ~(A^*-\theta_i^*I)U_i \subseteq U_{i-1},$\\
v) If $U_i = (V^*_{d-i}+V^*_{d-i+1}+\dots+V^*_d)\cap(V_0+V_1+\dots+V_{d-i})$ then $(A-\theta_{d-i}I)U_i \subseteq U_{i+1}, ~(A^*-\theta^*_iI)U_i \subseteq U_{i-1},$\\
vi) If $U_i = (V^*_{d-i}+V^*_{d-i+1}+\dots+V^*_d)\cap(V_i+V_{i+1}+\dots+V_{d})$ then $(A-\theta_iI)U_i\subseteq U_{i+1}, ~(A^*-\theta_{d-i}^*I)U_i \subseteq U_{i-1}.$
\end{lem}
\begin{thm}\cite{IT}
Let $B: V \to V$ denote the unique linear transformation such that for $0 \le i \le d,$
\[(V_0^*+V_1^*+\dots+V_i^*)\cap(V_0+V_1+\dots+V_{d-i})\]
is an eigenspace of $B$ with eigenvalues $bq^{2i-d}$. \\
Let $B^*: V \to V$ denote the unique linear transformation such that for $0 \le i \le d,$
\[(V_{d-i}^*+V_{d-i+1}^*+\dots+V_d^*)\cap(V_i+V_{i+1}+\dots+V_d)\]
is an eigenspace of $B^*$ with eigenvalue $b^*q^{d-2i}$.\\
Then $(B, B^*)$ is a tridiagonal pair on $V$. The sequence $bq^{2i-d} (0 \le i \le d)$ is a standard ordering of the eigenvalues of $B$ and the sequence $b^*q^{d-2i} (0 \le i \le d)$ is a standard ordering of the eigenvalues of $B^*$.
\end{thm}
Now we give two relations involving the tridiagonal pair $(A, A^*)$ which has a standard ordering of the eigevalues $\theta_i, i =0, 1, \dots, d$ of $A$ satisfying $\theta_i = aq^{2i-d}$ and a standard ordering of the eigenvalues $\theta_i^*, i = 0, 1, \dots, d$ of $A^*$ satisfying $\theta_i^*=a^*q^{d-2i}$ where $a, a^*$ are nonzero scalars in $\mathbb{K}$.
\begin{thm}\cite{IT}
Let $(A, A^*)$ denote a tridiagonal pair on $V$, let $\theta_0, \theta_1, \dots, \theta_d$ (resp. $\theta_0^*, \theta_1^*, \dots, \theta_d^*$) denote a standard ordering of the eigenvalues of $A$ (resp. $A^*$). Assume that $\theta_i = aq^{2i-d}, \theta_i^* = a^*q^{d-2i}, i =0, 1, \dots, d$ where $a, a^*$ are nonzero scalars in $\mathbb{K}$ then $A, A^*$ satisfy the $q$-Serre relations
\begin{eqnarray}
A^3A^*-[3]_qA^2A^*A+[3]_qA{A^*}A^2-A^*A^3&=&0,\\
{A^*}^3A-[3]_q{A^*}^2A{A^*}+[3]_qA^*A{A^*}^2-A{A^*}^3 &=&0.
\end{eqnarray}
\end{thm}
In fact, the result of the above theorem is expressed more efficiently as follows
\begin{thm}\cite{Ter03}
\label{abcd}
Let $(A, A^*)$ denote a tridiagonal pair on $V$ of diameter $d$. Then the following are equivalent:\\
i) $A, A^*$ satisfy the $q$-Serre relations.\\
ii) There exist eigenvalue and dual eigenvalue sequences for $A, A^*$ which satisfy 
\[\theta_i = q^{2i}\theta, ~~ \theta_i^*=q^{2d-2i}\theta^*~~~(0 \le i \le d),\]
for some nonzero scalars $\theta, \theta^*$.
\end{thm}
\begin{thm}\cite{IT}
Let $\mathbb{K}$ denote an algebraically closed field with characteristic 0. Let $V$ denote a vector space over $\mathbb{K}$ with finite positive dimension and let $(A, A^*)$ denote a tridiagonal pair on $V$ of diameter $d$. Let $(\rho_0, \rho_1, \dots, \rho_d)$ denote the corresponding shape vector. Assume $A, A^*$ satisfy the $q$-Serre relations then the entries in this shape vector are upper bounded by binomial coefficients as follows
\[\rho_i \le \left(\begin{array}{c}
d\\
i
\end{array}
\right)~~~~(0 \le i \le d).\]
In particular, $\rho_0 = \rho_d = 1$ for such tridiagonal pairs.
\end{thm}
\begin{lem}\cite{AC}
Let $(A, A^*)$ denote a tridiagonal pair on $V$, and let $D$ and $D^*$ denote the subalgebras of $End(V)$ generated by $A$ and $A^*$, respectively. Fix standard orderings of the eigenspaces $V_0, V_1, \dots, V_d$ and of the eigenspaces $V_0^*, V_1^*, \dots, V_d^*$ of $A^*$. Then the following are equivalent:\\
i) $(A, A^*)$ is a Leonard pair.\\
ii) $V=Dv^*$ for some nonzero $v^* \in V_0^*.$\\
iii) $V=D^*v$ for some nonzero $v \in V_d.$
\end{lem}
The preceding result suggests the following generalization of a Leonard pair, from which the case of Leonard pair is excluded to focus on what is new.
\begin{defn}
Let $(A, A^*)$ denote a tridiagonal pair on $V$. Let $D$ and $D^*$ denote the subalgebras of $End(V)$ generated by $A$ and $A^*$, respectively. We say that $(A, A^*)$ is mild whenever $(A, A^*)$ is not a Leonard pair, but $\rho_0 = \rho_d =1$ and $V = Dv^*+D^*v$ for some nonzero $v^* \in V_0^*$ and $v \in V_d$. 
\end{defn}
\begin{thm}\cite{AC}
Suppose that $(A, A^*)$ is mild, and pick nonzero vectors $v^* \in V_0^*$ and $v \in V_d$. Define
\[v_i^*=(A-\theta_{i-1}I)\dots(A-\theta_1I)(A-\theta_0I)v^*~~~(0 \le i \le d-1),\]
\[v_i= (A^*-\theta^*_{i+1}I)\dots(A^*-\theta^*_{d-1}I)(A^*-\theta_d^*I)v~~~(1 \le i \le d).\]
Then $\{v_0^*, v_1, v_1^*, v_2, v_2^*, \dots, v_{d-1}, v_{d-1}^*, v_d\}$ is basis for $V$.
\label{abc}
\end{thm}
\begin{defn}
Let $(A, A^*)$ denote a tridiagonal pair on $V$. Then $(A, A^*)$ is said to be of $q$-Serre type whenever the following hold:
\begin{eqnarray*}
A^3A^*-[3]_qA^2A^*A+[3]_qAA^*A^2-A^*A^3&=&0,\\
{A^*}^3A-[3]_q{A^*}^2AA^*+[3]_qA^*A{A^*}^2-A{A^*}^3&=&0.
\end{eqnarray*}
\end{defn}
The action of a mild tridiagonal pair of $q$-Serre type on the basis of Theorem \ref{abc} was studied more deeply by H. Alnajjar and B. Curtin.
\begin{thm}\cite{AC}
\label{abc2}
Let $(A, A^*)$ denote a mild tridiagonal pair on $V$ of $q$-Serre type with diameter $d \ge 3$. Fix standard orderings of the eigenspaces of $A$ and $A^*$ for which the corresponding eigenvalue and dual eigenvalue sequences satisfy ii) of Theorem \ref{abcd} for some nonzero scalars $\theta, \theta^*$. Define $v_i^*$ and $v_i~ (0 \le i \le d)$ as in Theorem \ref{abc}. Then there exist nonzero scalars $\lambda, \mu, \mu^* \in \mathbb{K}$ such that
\begin{eqnarray*}
Av_i^*&=& q^{2i}\theta v_i^*+v_{i+1}^*~ (0 \le i \le d-2),\\
Av^*_{d-1}&=& q^{2d-2}\theta v^*_{d-1}+\gamma_d\mu^*v_d,\\
Av_i&=&q^{2i}\theta v_i+\lambda_iv_{i+1}+\gamma_{d-i}\mu v^*_{i+1}~ (1 \le i \le d-2),\\
Av_{d-1}&=&q^{2d-2}\theta v_{d-1}+(\lambda_{d-1}+\gamma_{d-1}[2]_q\mu\mu^*)v_d,\\
Av_d&=&q^{2d}\theta v_d,\\
A^*v_i&=& q^{2d-2i}\theta^* v_i+v_{i-1}~ (2 \le i \le d),\\
A^*v_{1}&=& q^{2d-2}\theta^* v_{1}+\gamma_d\mu v_0^*,\\
A^*v_i^*&=&q^{2d-2i}\theta^* v_i^*+\lambda_{d-i}v_{i-1}^*+\gamma_{i}\mu^* v_{i-1}~ (2 \le i \le d-1),\\
A^*v_{1}^*&=&q^{2d-2}\theta^* v_{1}^*+(\lambda_{d-1}+\gamma_{d-1}[2]_q\mu\mu^*)v_0^*,\\
A^*v_0^*&=&q^{2d}\theta^* v_0^*,
\end{eqnarray*}
where
\begin{eqnarray*}
\lambda_i &=&[i]_q[d-i]_q\lambda ~ (1 \le i \le d-1),\\
\gamma_i &=& \frac{[i]_q![d-1]_q!}{[2]_q[d-i+1]_q}\lambda^{i-2}~ (2 \le i \le d).
\end{eqnarray*}
\end{thm}
Recall that $V$ is irreducible as an $(A, A^*)$-module whenever there is no subalgebra $W$ of $V$ such that both $AW \subseteq W$ and $A^*W \subseteq W$, other than 0 and $V$.
\begin{thm}\cite{AC}

Let $\theta, \theta^*, q, \lambda, \mu^*$, and $\mu$ be nonzero scalars in $\mathbb{K}$. Pick any integer $d \ge 3$, and let $V$ be a vector space of dimension $2d$. Let $A: V \to V$ and $A^*: V \to V$ denote linear transformations which act on some basis $v_0^*, v_1, v_1^*, \dots, v_{d-1}, v^*_{d-1}, v_d$ as in Theorem \ref{abc2}. Further suppose that $V$ is irreducible as an $(A, A^*)$-module. Then $(A, A^*)$ is a mild tridiagonal pair on $V$ of $q$-Serre type.
\end{thm}
Let $(A, A^*)$ denote a tridiagonal pair on $V$. Let $\theta_0, \theta_1,\dots, \theta_d$ (resp. $\theta_0^*, \theta_1^*, \dots, \theta_d^*$) denote a standard ordering of the eigenvalues of $A$ (resp. $A^*$) By \cite{TD00} there exists a unique sequence $U_0, U_1, \dots, U_d$ consisting of subspaces of $V$ suth that
\[V = U_0+U_1+\dots+U_d~~~(\text{direct sum}),\] 
\[(A-\theta_iI)U_i \subseteq U_{i+1}~~(0 \le i <d), ~(A-\theta_dI)U_d = 0,\]
\[(A^*-\theta_i^*I)U_i \subseteq U_{i-1}~~ (0 < i \le d), ~ (A^*-\theta^*_0I)U_0 =0.\]
We call the sequence $U_0, U_1, \dots, U_d$ the split decomposition for $(A, A^*)$ with respect to the orderings $\theta_0, \theta_1, \dots, \theta_d$ and $\theta_0^*, \theta_1^*, \dots, \theta_d^*$.\\
For $0 \le i \le d$, let $F_i: V \to V$ denote the linear transformation which satisfies both
\[(F_i-I)U_i = 0,\]
\[F_iU_j = 0~ \text{if} ~ j \ne i, ~~(0 \le j \le d).\]
In other words, $F_i$ is the projection map from $V$ onto $U_i$. We observe
\[F_iF_j = \delta_{ij}F_i ~~ (0 \le i, j \le d),\]
\[F_0 + F_1+\dots+F_d = I,\]
\[F_iV = U_i~~ (0 \le i \le d).\]
Define 
\[R = A-\sum\limits_{h =0}^d{\theta_hF_h},\]
\[L = A^* - \sum\limits_{h =0}^d{\theta_h^*F_h}.\]
We obtain that 
\[RU_i \subseteq U_{i+1}~~ (0 \le i <d), ~ RU_d =0,\]
\[LU_i \subseteq U_{i-1}~~ (0< i \le d), ~ LU_0 = 0,\]
and 
\[R^{d+1} = 0, ~~ L^{d+1} = 0.\]
We call $R$ (resp. $L$) the raising map (resp. lowering map) for $(A, A^*)$ with respect to $U_0, U_1,\dots, U_d$.\\
Now we consider that $(A, A^*)$ is a tridiagonal pair and satisfies the $q$-Serre relations. Hence there exist nonzero scalars $a, a^*$ such that $\theta_i = aq^{2i-d}, ~ \theta_i^* =a^*q^{d-2i} ~~(0 \le i \le d)$. By \cite{TD00} the maps $R, L$ satisfy the $q$-Serre relations
\begin{eqnarray*}
R^3L - [3]_qR^2LR+[3]_qRLR^2-LR^3&=&0,\\
L^3R-[3]_qL^2RL+[3]_qLRL^2-RL^3&=&0.
\end{eqnarray*}
\begin{thm}\cite{IT03}
Let $(A, A^*)$ denote a tridiagonal pair on $V$ satisfying the $q$-Serre relations. Let $U_0, U_1,\dots, U_d$ denote the split decomposition for $(A, A^*)$ and let $R, L$ denote the corresponding raising and lowering maps. Let $v$ denote a nonzero vector in $U_0$. Then $V$ is spanned by the vectors of the forms $L^{i_1}R^{i_2}L^{i_3}R^{i_4}\dots R^{i_n}v$, where $i_1, i_2, \dots, i_n$ ranges over all sequences such that $n$ is a nonnegative even integer, and $i_1, i_2, \dots, i_n$ are integers satisfying $0 \le i_1 < i_2 < \dots < i_n \le d$.
\end{thm}
\begin{defn}
Let $(A, A^*)$ denote a tridiagonal pair on $V$. Let $(\rho_0, \rho_1, \dots, \rho_d)$ denote the shape vector of $(A, A^*)$. The pair $(A, A^*)$ is said to be sharp whenever $\rho_0 = 1$.
\end{defn}
\begin{thm}\cite{NT}
A tridiagonal pair over an algebraically closed field is sharp.
\end{thm}
\begin{thm}\cite{NT}
Let $(A, A^*)$ denote a sharp tridiagonal pair on $V$. Then there exists a nonzero bilinear form $\left\langle , \right\rangle $ on $V$ such that $\left\langle Au , v \right\rangle = \left\langle u , Av \right\rangle$ and $\left\langle A^*u , v \right\rangle = \left\langle u , A ^*v \right\rangle$ for all $u, v \in V$. This form is unique up to multiplication by a nonzero scalar in $\mathbb{K}$. This form is nondegenerate and symmetric.
\end{thm}
Let $\lambda$ denote an indeterminate and let $\mathbb{K}[\lambda]$ denote the $\mathbb{K}$-algebra consisting of the polynomials in $\lambda$ that have all coefficients in $\mathbb{K}$. Let $\{\theta_i\}_{i=0}^d$ and $\{\theta_i^*\}_{i=0}^d$ denote scalars in $\mathbb{K}$. Then for $0 \le i \le d$ we define the following polynomials in $\mathbb{K}[\lambda]$:
\begin{eqnarray*}
\tau_i&=&(\lambda -\theta_0)(\lambda-\theta_1)\dots(\lambda-\theta_{i-1}),\\
\tau^*_i&=&(\lambda-\theta_0^*)(\lambda-\theta^*_1)\dots(\lambda-\theta^*_{i-1}),\\
\eta_i&=&(\lambda-\theta_d)(\lambda-\theta_{d-1})\dots(\lambda-\theta_{d-i+1}),\\
\eta^*_i&=&(\lambda-\theta^*_d)(\lambda-\theta^*_{d-1})\dots(\lambda-\theta^*_{d-i+1}).
\end{eqnarray*}
Note that each of $\tau_i, \tau_i^*, \eta_i, \eta_i^*$ is monic with degree $i$.
\begin{defn}
Let $\Phi = (A;\{E_i\}_{i=0}^d;A^*;\{E_i^*\}_{i=0}^d)$ denote a tridiagonal system on $V$. We say $\Phi$ is sharp whenever the tridiagonal pair $A, A^*$ is sharp.
\end{defn} 
\begin{defn}
Let $\Phi = (A;\{E_i\}_{i=0}^d;A^*;\{E_i^*\}_{i=0}^d)$ denote a tridiagonal system over $\mathbb{K}$, with the standard ordering of the eigenvalues $\{\theta_i\}_{i+0}^d$ (resp. $\{\theta_i^*\}_{i+0}^d$) of $A$ (resp. $A^*$). By \cite{NT}, for $0 \le i \le d$ there exists a unique $\zeta \in \mathbb{K}$ such that
\[E_0^*\tau_i(A)E_0^*=\frac{\zeta_iE_0^*}{(\theta_0^*-\theta_1^*)(\theta^ *_0-\theta_2^*)\dots(\theta_0^*-\theta_i^*)}\ .\] 
Note that $\zeta_0 = 1$. We call $\{\zeta_i\}_{i=0}^d$ the split sequence of the tridiagonal system. 
\end{defn}
\begin{defn}
Let $\Phi=(A;\{E_i\}_{i=0}^d;A^*;\{E_i^*\}_{i=0}^d$ denote a sharp tridiagonal system. By the parameter array of $\Phi$ we mean the sequence $(\{\theta_i\}_{i=0}^d;\{\theta_i^*\}_{i=0}^d;\{\zeta_i\}_{i=0}^d)$ where $\{\theta_i\}_{i=0}^d$ (resp. $\{\theta_i^*\}_{i=0}^d$) is the standard ordering of the eigenvalues of $A$ (resp. $A^*$) and $\{\zeta_i\}_{i=0}^d$ is the split sequence of $\Phi$.
\end{defn}
\begin{defn}
Let $\Phi=(A;\{E_i\}_{i=0}^d;A^*;\{E_i^*\}_{i=0}^d$ denote a tridiagonal system on $V$ and let $\Phi'=(A';\{{E'}_i\}_{i=0}^d;$ ${A^*}';\{{E_i^*}'\}_{i=0}^d)$ denote a tridiagonal system on $V'$. We say $\Phi$ and $\Phi'$ are isomorphic whenever there exists an isomorphism of $\mathbb{K}$-vector spaces $\gamma: V \to V'$ such that $\gamma A=A'\gamma, \gamma A^* = {A^*}'\gamma$ and $\gamma E_i = {E_i}'\gamma, \gamma E_i^* = {E_i^*}'\gamma$ for $0 \le i \le d$.
\end{defn}
The following result shows the significance of the parameter array.

\begin{thm}\cite{NT}
Two sharp tridiagonal systems over $\mathbb{K}$ are isomorphic if and only if they have the same parameter array.
\end{thm}
\begin{defn}
Let $d$ denote a nonnegative integer and let $(\{\theta_i\}_{i=0}^d;\{\theta_i^*\}_{i=0}^d)$ denote a sequence of scalars taken from $\mathbb{K}$. We call this sequence $q$-Racah whenever the following i), ii) hold.\\
i) $\theta_i \ne \theta_j, \theta_i^* \ne \theta_j^*$ if $ i \ne j ~ (0 \le i, j\le d)$.\\
ii) There exist $q, a, b, c, a^*, b^*, c^*$ that satisfy 
\[\theta_i = a +bq^{2i-d}+cq^{d-2i}~~ (0 \le i \le d),\]
\[\theta_i^*= a^*+b^*q^{2i-d}+c^*q^{d-2i}~~ (0 \le i \le d),\]
\[q, a, b, c, a^*, b^*, c^* \in \overline{\mathbb{K}},\]
\[q \ne 0, ~ q^2 \ne 1,~ q^2 \ne -1,~ bb^*cc^* \ne 0. \]
where $\overline{\mathbb{K}}$ is the algebraic closure of $\mathbb{K}$. 
\end{defn}
\begin{thm}\cite{IT2}
Assume the field $\mathbb{K}$ is algebraically closed and let $d$ denote a nonnegative integer. Let $(\{\theta_i\}_{i=0}^d;\{\theta_i^*\}_{i=0}^d)$ denote a $q$-Racah sequence of scalars of $\mathbb{K}$ and let $\{\zeta_i\}_{i=0}^d$ denote any sequence of scalars in $\mathbb{K}$. Then the following are equivalent:\\
i) There exists a tridiagonal system $\Phi$ over $\mathbb{K}$ that has parameter array $(\{\theta_i\}_{i=0}^d;\{\theta_i^*\}_{i=0}^d;\{\zeta_i\}_{i=0}^d)$\\
ii) $\zeta_0 = 1, \zeta_d \ne 0$, and 
\[0 \ne \sum\limits_{i =0}^d{\eta_{d-i}(\theta_0)\eta^*_{d-i}(\theta_0^*)\zeta_i}.\]
Suppose i), ii) hold. Then $\Phi$ is unique up to isomorphism of tridiagonal systems.
\end{thm}
\begin{thm}\cite{INT}
\label{class}
Let $d$ denote a nonnegative integer and let 
\begin{equation}
\label{pa}
(\{\theta_i\}_{i=0}^d;\{\theta_i^*\}_{i=0}^d;\{\zeta_i\}_{i=0}^d)\end{equation}
denote a sequence of scalars taken from $\mathbb{K}$. Then there exists a sharp tridiagonal system $\Phi$ over $\mathbb{K}$ with parameter (\ref{pa}) if and only if i)-iii) hold below.\\
i) $\theta_i \ne \theta_j, \theta_i^* \ne \theta_j^*$ if $i \ne j ~ (0 \le i, j \le d).$\\
ii) The expressions
\[\frac{\theta_{i-2}-\theta_{i+1}}{\theta_{i-1}-\theta_i},~~~ \frac{\theta^*_{i-2}-\theta^*_{i+1}}{\theta^*_{i-1}-\theta^*_i}\]
are equal and independent of $i$ for $2 \le i \le d-1$.\\
iii) $\zeta_0 = 1, \zeta_d \ne 0,$ and 
\[0 \ne \sum\limits_{i=0}^d{\eta_{d-i}(\theta_0)\eta_{d-i}^*(\theta_0^*)\zeta_i}.\]
Suppose i)-iii) hold. Then $\Phi$ is unique up to isomorphism of tridiagonal systems.
\end{thm}
There is no doubt that if $\mathbb{K}$ is an algebraically closed field, then there exists a tridiagonal system over $\mathbb{K}$ with parameter array (\ref{pa}) if and only if this array satisfies the conditions i), ii) and iii) of Theorem \ref{class}.

\subsection{Connection with the theory of orthogonal polynomials}
In the literature, the first connection between the theory of orthogonal polynomials and Leonard pairs arises in the analysis of finite dimensional representation of Zhedanov's algebra \cite{Zhe} (also called the Askey-Wilson algebra). It is shown that the Askey-Wilson polynomials are the overlap coefficients between the `dual' basis in which two generators are diagonalized respectively.\vspace{1mm}

 Furthermore, the theory of Leonard pairs gives a nice algebraic framework for the orthogonal polynomials of the Askey-scheme. Historically, there is a theorem due to Leonard \cite{Leo1}, \cite[page 260]{BI} that gives a characterization of the $q$-Racah polynomials and some related polynomials in the Askey scheme \cite{ARS, AW, AW1, KS, Koo}. In this Section, the connection between the theory of Leonard pairs and the Askey-scheme of orthogonal polynomials is recalled \cite{Ter03, T05, Ter06}.\\

First, we recall hypergeometric orthogonal polynomials which appeared in \cite{KS}. \\
Define $q$-analogue of the Pochhammer-symbol $(a)_k$
\[(a)_0=1~\text{and}~ (a)_k=a(a+1)(a+2)\dots(a+k-1), ~ k = 1, 2, 3, \dots.\]
This $q$-extension is given by
\[(a;q)_0=1~ \text{and}~ (a;q)_k=(1-a)(1-aq)(1-aq^2)\dots(1-aq^{k-1}),~ k = 1, 2, 3\dots.\]
It is clear that 
\[\mathop {\lim }\limits_{q \to 1} \frac{{{{\left( {{q^\alpha };q} \right)}_k}}}{{{{\left( {1 - q} \right)}^k}}} = {\left( \alpha  \right)_k}.\]
The symbols $(a;q)_k$ are called $q$-shifted factorials. They can also be defined for negative values of $k$ as
\[(a;q)_k=\frac{1}{(1-aq^{-1})(1-aq^{-2})\dots(1-aq^k)},~ a \ne q, q^2, q^3, \dots, q^{-k}, k = -1, -2, -3, \dots.\]
We can also define 
\[(a;q)_{\infty}= \prod\limits_{k=0}^{\infty}{(1-aq^k)}.\]
This implies that
\[(a;q)_n=\frac{(a;q)_{\infty}}{(aq^n;q)_{\infty}},\]
and, for any complex number $\lambda$,
\[(a;q)_{\lambda}=\frac{(a;q)_\infty}{(aq^{\lambda};q)_{\infty}}\]
The hypergeometric series $_rF_s$ is defined by
\[_rF_s\left(\left.\begin{array}{c}
a_1, \dots, a_r\\
b_1,\dots, b_s 
\end{array}\right|z\right)=\sum\limits_{k=0}^{\infty}{\frac{(a_1, \dots, a_r)_k}{(b_1, \dots, b_s)_k}\frac{z^k}{k!}},\]
where 
\[(a_1, \dots, a_r)_k=(a_1)_k\dots(a_r)_k.\]
The hypergeometric series $_rF_s$ is called balanced if $r = s+1, z=1$ and $a_1+a_2+\dots+a_{s+1}+1=b_1+b_2+\dots+b_s.$\\
The basic hypergeometric series (or $q$-hypergeometric series) $_r\phi_s$ is defined by
\[_r\phi_s\left(\left.\begin{array}{c}
a_1, \dots, a_r\\
b_1,\dots, b_s 
\end{array}\right|q;z\right)=\sum\limits_{k=0}^{\infty}{\frac{(a_1, \dots, a_r;q)_k}{(b_1, \dots, b_s;q)_k}(-1)^{(1+s-r)k}q^{(1+s-r)\left(\begin{array}{c}
k\\2
\end{array}\right)}\frac{z^k}{(q;q)_k}},\]
where 
\[(a_1, \dots, a_r;q)_k=(a_1;q)_k\dots(a_r;q)_k.\]
The special case $r=s+1$ reads
\[_{s+1}\phi_s\left(\left.\begin{array}{c}
a_1, \dots, a_{s+1}\\
b_1,\dots, b_s 
\end{array}\right|q;z\right)=\sum\limits_{k=0}^{\infty}{\frac{(a_1, \dots, a_r;q)_k}{(b_1, \dots, b_s;q)_k}\frac{z^k}{(q;q)_k}}.\]
A basic hypergeometric series is called balanced if $z=q$ and $a_1a_2\dots a_{s+1}q=b_1b_2\dots b_s.$\\

There is a natural correspondence between Leonard pairs and a family of orthogonal polynomials. The following material is taken from \cite{Ter03, T05}.\\

Let $\lambda$ denote an indeterminant, and let $\mathbb{K}[\lambda]$ denote a $\mathbb{K}$-algebra consisting of all polynomials in $\lambda$ that have coefficients in $\mathbb{K}$. Let $\Phi=(A; E_0, E_1, \dots, E_d; A^*; E_0^*, E_1^*,\dots, E_d^*)$ denote a Leonard system over $\mathbb{K}$. Then there exists a unique sequence of monic polynomials $p_0, p_1, \dots, p_{d+1};~ p_0^*, p_1^*, \dots, p_{d+1}^*$ in $\mathbb{K}[\lambda]$ such that\\
\[deg(p_i) =i,~~~ deg(p_i^*) = i~~~(0 \le i \le d+1),\]
\[p_i(A)E_0^*=E_i^*A^iE_0^*,~~ p_i^*(A^*)E_0= E_i{A^*}^iE_0~~~(0 \le i \le d),\]
\[p_{d+1}(A)=0,~~ p_{d+1}^*(A^*)=0.\]
These polynomials satisfy
\begin{eqnarray}
\label{1}
p_0 =1,~~ p_0^* =1,\\
\label{2}
\lambda p_i = p_{i+1}+a_ip_i+x_ip_{i-1} ~~~(0 \le i \le d),\\
\label{3}
\lambda p_i^* = p_{i+1}^* +a_i^*p_i^* + x^*_ip_{i-1}^* ~~~(0 \le i \le d),
\end{eqnarray}
where $x_0, x_0^*, p_{-1}, p_{-1}^*$ are all 0, and where 
\[a_i = tr(E_i^*A), ~~a_i^* = tr(E_iA^*)~~~(0 \le i \le d),\]
\[x_i = tr(E_i^*AE_{i-1}^*A),~~ x_i^*=tr(E_iA^*E_{i-1}A^*).\]
In fact 
\begin{equation}
x_i \ne 0, ~~ x_i^* \ne 0~~~(1 \le i \le d).
\end{equation}
We call $p_0, p_1, \dots, p_{d+1}$ the monic polynomial sequence of $\Phi$, and $p_0^*, p_1^*, \dots, p_{d+1}^*$ the dual monic polynomial sequence of $\Phi$. \\
Let $\theta_0, \theta_1, \dots, \theta_d$ (resp. $\theta_0^*, \theta_1^*, \dots, \theta_d^*$) denote the eigenvalue sequence (resp. dual eigenvalue sequence) of $\Phi$, so that 
\begin{eqnarray}
\label{4}
\theta_i \ne \theta_j, ~~ \theta_i^* \ne \theta_j^* ~~ \text{if}~ i \ne j~~ (0 \le i, j \le d),\\
\label{5}
p_{d+1}(\theta_i) = 0, ~~ p_{d+1}^*(\theta_i^*) = 0 ~~ (0 \le i \le d).
\end{eqnarray}
Then 
\begin{equation}
\label{6}
p_i(\theta_0) \ne 0, ~~p_i^*(\theta_0^*) \ne 0 ~~~(0 \le i \le d),
\end{equation}
and
\begin{equation}
\label{7}
\frac{p_i(\theta_j)}{p_i(\theta_0)}=\frac{p_j^*(\theta_i^*)}{p_j^*(\theta_0^*)}~~~(0 \le i, j \le d). 
\end{equation}
Conversely, given polynomials 
\begin{eqnarray}
\label{8}
p_0, p_1, \dots, p_{d+1},\\
\label{9}
p_0^*, p_1^*, \dots, p_{d+1}^*
\end{eqnarray}
in $\mathbb{K}[\lambda]$ satisfying (\ref{1})-(\ref{3}), and given scalars 
\begin{eqnarray}
\label{10}
\theta_0, \theta_1, \dots, \theta_d,\\
\label{11}
\theta_0^*, \theta_1^*, \dots, \theta_d^*
\end{eqnarray}  
in $\mathbb{K}$ satisfying (\ref{4})-(\ref{7}), there exists a Leonard system $\Phi$ over $\mathbb{K}$ with monic polynomial sequence (\ref{8}), dual monic polynomial sequence (\ref{9}), eigenvalue sequence (\ref{10}), and dual eigenvalue sequence (\ref{11}). The system $\Phi$ is unique up to isomorphism of Leonard systems.\\

Hence, there is a bijection between the Leonard systems and systems (\ref{8})-(\ref{11}) satisfying (\ref{1})-(\ref{7}). For $\mathbb{K} = \mathbb{R}$, the systems (\ref{8})-(\ref{11}) satisfying (\ref{1})-(\ref{7}) were classified by Leonard, Bannai, and Ito. They found the polynomials involved are $q$-Racah polynomials or related polynomials from the Askey scheme.\\

We define the polynomials
\[u_i= \frac{p_i}{p_i(\theta_0)}, ~~ u_i^*= \frac{p_i^*}{p^*_i(\theta_0^*)}~~~ (0 \le i \le d).\]
By (\ref{7}), 
\[u_i(\theta_j)=u_j^*(\theta_i^*)~~~ (0 \le i, j\le d),\]
then $u_i$ and $u_i^*$ are dual sequences of normalized polynomials \cite{Leo1}. \\

Furthermore, there exists a unique sequence of scalars $c_1, c_2,\dots, c_d;~ b_0, b_1, \dots, b_{d-1}$ in $\mathbb{K}$ such that:
\begin{eqnarray*}
x &=& b_{i-1}c_i ~~~(1 \le i \le d),\\
\theta_0 &=& c_i + a_i +b_i ~~~(0 \le i \le d),
\end{eqnarray*}
where $c_0 =0, b_d =0$. Then
\[\lambda u_i=c_iu_{i-1}+a_iu_i+b_iu_{i+1}~~~ (0 \le i \le d-1),\]
and $\lambda u_d-c_du_{d-1}-a_du_d$ vanishes on each of $\theta_0, \theta_1, \dots, \theta_d$, where $\lambda$ is an indeterminant.
By the main theorem of \cite{Leo1}, we have
\[b_i = \varphi_{i+1}\frac{\tau^*_i(\theta_i^*)}{\tau^*_{i+1}(\theta_{i+1}^*)},~~ b_i^* = \varphi_{i+1}\frac{\tau_i(\theta_i)}{\tau_{i+1}(\theta_{i+1})}~~~ (0 \le i \le d-1),\]
\[c_i= \phi_i\frac{\eta^*_{d-i}(\theta_i^*)}{\eta_{d-i+1}^*(\theta^*_{i-1})}, ~~c_i^*= \phi_i\frac{\eta_{d-i}(\theta_i)}{\eta_{d-i+1}(\theta_{i-1})}~~~ (1 \le i \le d),\]
where 
\begin{eqnarray*}
\tau_i(\lambda)&=& (\lambda - \theta_0)(\lambda-\theta_1)\dots(\lambda - \theta_{i-1}),\\
\tau_i^*(\lambda)&=& (\lambda - \theta^*_0)(\lambda-\theta^*_1)\dots(\lambda - \theta^*_{i-1}),\\
\eta_i(\lambda)&=&(\lambda-\theta_d)(\lambda-\theta_{d-1})\dots(\lambda-\theta_{d-i+1}),\\
\eta_i^*(\lambda)&=&(\lambda-\theta^*_d)(\lambda-\theta^*_{d-1})\dots(\lambda-\theta^*_{d-i+1}).\\
\end{eqnarray*}
Hence, the polynomials $u_i, u_i^* ~ (0 \le i \le d)$ are given by
\[u_i = \sum\limits_{h=0}^i{\frac{\tau_h^*(\theta_i^*)}{\varphi_1\varphi_2\dots\varphi_h}\tau_h},~~ u_i^* = \sum\limits_{h=0}^i{\frac{\tau_h(\theta_i)}{\varphi_1\varphi_2\dots\varphi_h}\tau^*_h}.\]
We find that for $0 \le i, j \le d$, the common value of $u_i(\theta_j), u_j^*(\theta^*_i)$
\begin{equation}
\label{12}
\sum\limits_{n=0}^{d}{\frac{(q^{-i};q)_n(s^*q^{i+1};q)_n(q^{-j};q)_n(sq^{j+1};q)_nq^n}{(r_1q;q)_n(r_2q;q)_n(q^{-d};q)_n(q;q)_n}}
\end{equation}

In fact, (\ref{12}) is the basic hypergeometric series 
\[_4\theta_3\left(\left.\begin{array}{c}
q^{-i}, s^*q^{i+1}, q^{-j}, sq^{j+1}\\
r_1q, r_2q, q^{-d}
\end{array}\right|q, q \right),\]
where $r_1, r_2, s, s^*$ are parameters in (\ref{p1})-(\ref{p4}) and it is balanced because $r_1r_2=ss^*q^{d+1}$. 
Then $u_i, u_i^*$ are $q$-Racah polynomials.\\

Now we consider orthogonality of $p_i, p_i^*$ and $u_i, u_i^*$. \\

Put $m_i = tr(E_iE_0^*), ~~ m_i^* = tr(E_i^*E_0), 0 \le i \le d$.  Then each of $m_i, m^*_i$ is nonzero $(0 \le i \le d)$, and the orthogonality for the $p_i$ is 
\[\sum\limits_{r=0}^d{p_i(\theta_r)p_j(\theta_r)m_r}=\delta_{ij}x_1x_2\dots x_i ~~~(0\le i, j\le d),\]
\[\sum\limits_{i =0}^{d}{\frac{p_i(\theta_r)p_i(\theta_s)}{x_1x_2\dots x_i}}=\delta_{rs}m_r^{-1}~~~(0 \le r, s \le d).\]
Observe that $m_0 = m^*_0$, let $\upsilon$ denote the multiplicative inverse of this common value, and set
\[k_i= m_i^*\upsilon ~~~(0 \le i, j \le d).\]
The orthogonality for the $u_i$ is 
\[\sum\limits_{r=0}^d{u_i(\theta_r)u_j(\theta_r)m_r}=\delta_{ij}k_i^{-1}~~~(0 \le i, j \le d),\]
\[\sum\limits_{i=0}^d{u_i(\theta_r)u_i(\theta_s)k_i}=\delta_{rs}m_r^{-1}~~~(0 \le r,s \le d).\]
Remark that \[k_i = \frac{b_0b_1\dots b_{i-1}}{c_1c_2\dots c_i},\]
and \[\upsilon = k_0+k_1+\dots+k_d.\]

As an example, we now consider an infinite family of Leonard pairs and find the relation with hypergeometric series. For any nonnegative integer $d$, the pair
\[A = \left( {\begin{array}{*{20}{c}}
   0 & d & {} & {} & {} & 0  \\
   1 & 0 & {d - 1} & {} & {} & {}  \\
   {} & 2 &  \cdot  &  \cdot  & {} & {}  \\
   {} & {} &  \cdot  &  \cdot  &  \cdot  & {}  \\
   {} & {} & {} &  \cdot  &  \cdot  & 1  \\
   0 & {} & {} & {} & d & 0  \\
\end{array}} \right),~~ A^* = \text{diag}(d, d-2,d-4, \dots, -d)\]
is a Leonard pair on the vector space $\mathbb{K}^{d+1}$, provided the characteristic of $\mathbb{K}$ is zero or an odd prime greater than $d$. One shows $P^2=2^dI$ and $AP=PA^*$, where $P$ denotes the matrix with $ij$ entry
\[P_{ij}=\left(\begin{array}{c}
d\\
j
\end{array}\right)\sum\limits_{n=0}^d{\frac{(-i)_n(-j)_n2^n}{(-d)_nn!}}~~ (0 \le i, j \le d).\]
Further $P_{ij}$ can be written in form of hypergeometric series
\[P_{ij}=_2F_1\left(\left.\begin{array}{c}
-i, -j\\
-d
\end{array}\right|2\right).\]
There exist Leonard pairs similar to the one above in which the series of type $_2F_1$ is replaced by a series of one of the following types: $_3F_2,~ _4F_3,~ _2\phi_1,~ _3\phi_2,~ _4\phi_3.$\\

Let $\mathbb{K}$ denote a field with characteristic 0. Let $d$ denote a nonnegative integer. Put $\Omega =\{d-2i|i = 0, 1, \dots, d\}$. Let $V$ denote a vector space over $\mathbb{K}$ consisting of all functions from $\Omega$ to $\mathbb{K}$. Since the cardinality of $\Omega$ is $d+1$ the dimension of $V$ is $d+1$.\\

Define two transformations $A, A^*$ from $V$ to $V$ as follows.\\

For all $f \in V, \theta \in \Omega$,
\begin{eqnarray}
(Af)(\theta) &=& \theta f(\theta),\\
(A^*f)(\theta) &=&\frac{d+\theta}{2}f(\theta-2)+\frac{d-\theta}{2}f(\theta+2).
\end{eqnarray}
It is easy seen that $A, A^*$ are linear.\\

For $j = 0, \dots, d$ let $K_j$ denote the element in $V$ satisfying
\[K_j(\theta_i)=\frac{d!}{j!(d-j)!}~{_2F_1\left( {\left. {\begin{array}{*{20}{c}}
   { - i, - j}  \\
   { - d}  \\
\end{array}} \right|2} \right)},\]
where $\theta_i = d-2i, 0 \le i \le d$.\\
Observe that $K_j(\theta)$ is a polynomial of degree $j$ in $\theta$. The polynomials $K_0, K_1, \dots, K_d$ are Krawtchouk polynomials and form a basis for $V$. With respect to this basis the matrices representing $A$ and $A^*$ are
\[A = \left( {\begin{array}{*{20}{c}}
   0 & d & {} & {} & {} & 0  \\
   1 & 0 & {d - 1} & {} & {} & {}  \\
   {} & 2 &  \cdot  &  \cdot  & {} & {}  \\
   {} & {} &  \cdot  &  \cdot  &  \cdot  & {}  \\
   {} & {} & {} &  \cdot  &  \cdot  & 1  \\
   0 & {} & {} & {} & d & 0  \\
\end{array}} \right),~~ A^* = \text{diag}(d, d-2,d-4, \dots, -d).\]
Since the pair $(A, A^*)$ is a Leonard pair on $V$, there exists a basis for $V$ with respect to which the matrix representing $A$ is diagonal and the matrix representing $A^*$ is irreducible tridiagonal. Now we display this basis. For $0 \le j \le d$ let $K^*_j$ denote the element in $V$ which satisfies
\[K_j^*(\theta_i)=\delta_{ij} ~~(0 \le i \le d),\]
where $\delta_{ij}$ denotes the Kronecker delta. The sequence $K_0^*, K_1^*,\dots, K_d^*$ forms a basis for $V$. With respect to this basis the matrices representing $A$ and $A^*$ are
\[A^* = \left( {\begin{array}{*{20}{c}}
   0 & d & {} & {} & {} & 0  \\
   1 & 0 & {d - 1} & {} & {} & {}  \\
   {} & 2 &  \cdot  &  \cdot  & {} & {}  \\
   {} & {} &  \cdot  &  \cdot  &  \cdot  & {}  \\
   {} & {} & {} &  \cdot  &  \cdot  & 1  \\
   0 & {} & {} & {} & d & 0  \\
\end{array}} \right),~~ A^* = \text{diag}(d, d-2,d-4, \dots, -d).\]
We have then shown how the above Krawtchouk polynomials correspond to the Leonard pairs.

The polynomials in the following table are related to Leonard pairs in a similar fashion.\\
\begin{center}
\begin{tabular}{|c|c|}
\hline  Type& Polynomial  \\ 
\hline $_4F_3$ & Racah \\ 
$_3F_2$ & Hahn, dual Hahn\\
$_2F_1$ & Krawtchouk \\
$_4\phi_3$ & $q$-Racah \\
$_3\phi_2$ & $q$-Hahn, dual $q$-Hahn\\
$_2\phi_1$ & $q$-Krawtchouk (classical, affine, quantum, dual)\\
\hline 
\end{tabular} 
\end{center}
The above polynomials are defined in Koekoek and Swarttouw \cite{KS}, and the connection to Leonard pairs is given in \cite{T05, BI}. Indeed, these polynomials exhaust all Leonard pairs for which $q\ne -1$. For $\mathbb{K} = \mathbb{R}$, the classification of Leonard pairs amounts to a ``linear algebraic version" of Leonard's theorem \cite{Leo1, BI}.

\chapter{Mathematical Physics: background}
\label{ChapPhys}

In the first part, we recall the two known presentations of the Onsager algebra.  The second presentation (the original one) is given by an infinite set of elements $\{A_n\}, n = 0, \pm1, \pm2, \dots$, $\{G_m\}, m = 1, 2, \dots$  that satisfy the relations (\ref{relaOn1})-(\ref{relaOn3}) \cite{O}. The first presentation introduced by Dolan-Grady is given by $A_0,A_1$ satisfying the relations (\ref{2111})-(\ref{2112}) \cite{DG}. It is also explained that the elements of an Abelian subalgebra of the Onsager algebra provide examples of mutually commuting quantities (see (\ref{cmm})) that generate integrable systems.  Also, the connection with $\widehat{sl_2}$ (more precisely, with the loop algebra of $sl_2$) is described.  For irreducible finite dimensional representation of the Onsager algebra, it is known that the generators of the second representation satisfy additional relations (\ref{dv})-(\ref{dv1}). These relations are usually called the Davies' relations \cite{B.D, B.D1}.

In the second and third parts, by analogy with the undeformed case discussed in the first part, it is explained that the $q-$Onsager algebra admits two presentations. The second presentation is given by the generators $\mathcal{W}_{-k},\mathcal{W}_{k+1}, \mathcal{ G}_{k+1}, \mathcal{\tilde{G}}_{k+1}$ that satisfy the infinite dimensional algebra ${\cal A}_q$ (see Definition \ref{alA}) \cite{BasS}. The first presentation is given
by the standard generators $\mathcal{W}_0, \mathcal{W}_1$ that satisfy the $q-$Dolan-Grady relations \cite{T05, B1, Bas1}. It is explained how mutually commuting quantities (\ref{tcmm}) that generate an Abelian subalgebra of the $q-$Onsager algebra are derived by using Sklyanin's formalism. Also, the connection between the two presentations and $U_q(\widehat{sl_2})$ (more precisely, the quantum loop algebra of $sl_2$) is described.
For most of the examples considered in the literature \cite{B1, BK2, BK3}, the vector space on which the elements act is finite dimensional. As a consequence, quantum analogs of Davies relations naturally appear, see (\ref{qdv1})-(\ref{qdv4}). \vspace{1mm}

In the last part, we briefly recall how the open XXZ spin chain with generic boundary conditions and generic values of $q$ can be formulated using the $q-$Onsager approach.

\section{Historical background}
The exact solution of the planar Ising model in zero magnetic field which was obtained by Onsager \cite{O} has provided a considerable source of developments in the theory of exactly solvable systems of statistical mechanics, or quantum field theory in two dimensions. Onsager's successful approach was originally based on the so-called Onsager algebra and its representation theory. In progress of solving the two dimensional Ising model, he established the transfer matrix of the model in terms of an infinite set of elements $\{A_n\}, n = 0, \pm1, \pm2, \dots$, $\{G_m\}, m = 1, 2, \dots$  that generate the so-called Onsager algebra.

\begin{defn}\cite{O}
The Onsager algebra is a Lie algebra which has generators $A_n, G_m$, $n = 0, \pm1, \pm2, \dots, m = 1, 2, \dots$ such that they satisfy the following relations
\begin{eqnarray}
\label{relaOn1}
[A_n, A_l] &=& 4G_{n-l},\\
~[G_m, A_n] &=& 2 A_{n+m}-2A_{n-m},\\
\label{relaOn3}
~[G_m, G_l] &=& 0.
\end{eqnarray}
\end{defn} 
Use algebraic methods, Onsager derived the largest eigenvalue and the corresponding eigenvector of the transfer matrix of the Ising model.
Despite the important role of the Onsager algebra, it received less attention in the following years than the star triangular relations which originated in \cite{O}, \cite{Wan} and led to the Yang-Baxter equations, the theory of quantum groups, as well as the quantum inverse scattering method. \\

In the 1980s the Onsager algebra appeared \cite{B.D}, \cite{B.D1}, \cite{J.P} to be closely related with the quantum integrable structure discovered by Dolan and Grady in \cite{DG}. In fact, Dolan and Grady considered a self-dual quantum Hamiltonian of the form 
\begin{equation}
\label{hal}
H = \kappa A + \kappa^*\tilde{A},
\end{equation}
where $\kappa, \kappa^*$ are coupling constants and $\tilde{A}$ is the operator dual to $A$ such that both operators satisfy the  condition
\begin{equation}
\label{215}
[A, [A, [A,\tilde{A}]]]= 16[A,\tilde{A}].
\end{equation}
As a consequence of this relation, there exists an infinite set of commuting conserved self-dual charges
\begin{equation}
\label{cmm}
\mathcal{I}_{2n}=\kappa(W_{2n}-\tilde{W}_{2n-2})+\kappa^*(\tilde{W}_{2n}-W_{2n-2}),~ n = 1, 2, \dots,
\end{equation}
where 
\begin{equation}
W_{2n+2}\equiv-\frac{1}{8}[A,[\tilde{A},W_{2n}]]-\tilde{W}_{2n},~ n = 1, 2, \dots,
\end{equation}
$W_0 \equiv A, \mathcal{I}_0 \equiv H$, and the sequence $\{W_{2n}\}, ~n = 1, 2, \dots$ can be extended to $n <0$ by defining $W_{-2n}\equiv -\tilde{W}_{2n-2}$. \\
It was showed that $[\tilde{W}_{2l},W_{2n-2l-2}]= [\tilde{W}_{2l-2},W_{2n-2l}], ~\text{for all}~ n\ge 0,~ l \ge 0$,
then it followed 
\begin{equation}
\label{fcmm}
[H, \mathcal{I}_{2n}] = 0~~ \text{and}~~ [\mathcal{I}_{2n}, \mathcal{I}_{2m}]=0
\end{equation}
Clearly, Dolan and Grady \cite{DG} showed the the Dolan-Grady relations are sufficient to guarantee that there is an infinite sequence of commuting operators of the Hamiltonian $H$. Based on this result the integrability does not depend on the dimension of the system or the nature of the space-time manifold, i.e, lattice, continuum or loop space. Integrable systems are characterized by the existence of a sufficient number of constants of motion, i.e. equal to the number of degrees of freedom. The rigorous connection between simple self-duality and a set of commuting conserved charges was established. For a finite lattice, the set is finite. For an infinite system, the set is infinite. In the XY and Ising modes, the charges coincide with known results. Moreover, von Gehlen and Rittenberg \cite{GR} considered some $Z_n$ symmetric quantum spin chain Hamiltonians (\ref{hal}), and presented strong numerical evidence that they exhibit Ising-like behavior in their spectra. Namely, the Hamiltonians have an infinite set of commuting conserved charges based on the Dolan-Grady relation
\begin{equation}
\label{219}
[A,[A,[A,\tilde{A}]]] = n^2[A,\tilde{A}].
\end{equation}
Obviously, if we put $B = 4n^{-1}A,~ \tilde{B}= 4n^{-1}\tilde{A}$ the equation (\ref{219}) becomes the equation (\ref{215}) for $B, \tilde{B}$. It means that
\begin{equation}
[B,[B,[B,\tilde{B}]]] = 16[B,\tilde{B}].
\end{equation}

In the early 1990s, Davies \cite{B.D}, \cite{B.D1} obtained the relation between the Onsager algebra and the Dolan Grady relations. Actually, Davies did not require the self-duality of operators $A_0, A_1$ in the Hamiltonian 
\begin{equation}
\label{Hamil}
H = \kappa A_0 + \kappa^* A_1,
\end{equation}
where $\kappa, \kappa^*$ are coupling constants, but he gave a pair of conditions
\begin{eqnarray}
\label{2111}
[A_0, [A_0, [A_0, A_1]]] &=& 16[A_0,A_1],\\
\label{2112}
~[A_1,[A_1,[A_1,A_0]]]&=&16[A_1,A_0].
\end{eqnarray} 
 We first recall the way to identify an Onsager algebra from the Dolan-Grady relations \cite{B.D}. Let $A_0, A_1$ denote generators satisfying the Dolan-Grady relations, define the sequences $A_n, G_m, n = 0, \pm1, \pm2, \dots, m = 1, 2, \dots$ by the recursion relations
\begin{eqnarray}
G_1&=&\frac{1}{4}[A_1,A_0],\\
A_{n+1}-A_{n-1}&=&\frac{1}{2}[G_1, A_n],\\
G_n&=&\frac{1}{4}[A_n,A_0].
\end{eqnarray}
The generators $A_n, G_m$ satisfy the defining relations of the Onsager algebra \cite{B.D}. \\

Inversely, if $A_n, G_m,  n = 0, \pm1, \pm2, \dots, m = 1, 2, \dots$ are generators of an Onsager algebra, then $A_0, A_1$ satisfy the Dolan-Grady relations. In fact, every adjacent pairs $A_k, A_{k+1}$ of the sequence $\{A_n\}$ satisfy the Dolan-Grady relations. As a consequence, the Onsager algebra admits two presentations. One given by (\ref{relaOn1}) - (\ref{relaOn3}) and one given by (\ref{2111}) - (\ref{2112}).
The proof of isomorphism between the two presentations is detailed in \cite{B.D}. More recently, see also \cite{El}\\

Around the same time, a relation between the Onsager algebra and the loop algebra of $sl_2$ was exhibited by Davies \cite{B.D1}.
It is argued that if the Onsager algebra acts on a finite dimensional vector space,  then there exists some value of $h$ such that the sequences $\{A_n\}, \{G_m\}$ satisfy the linear recurrence relations of length $(2h+1)$, namely
\begin{eqnarray}
\label{dv}
\sum\limits_{k=-h}^h{\alpha_kA_{k-l}} &=& 0,\\
\label{dv1}
\sum\limits_{k=-h}^h{\alpha_kG_{k-l}} &=& 0
\end{eqnarray}
where $l$ is arbitrary. Then, Davies showed that the finite-dimensional Onsager algebra is the direct sum of $h$ copies of the algebra $sl_2$,
\begin{eqnarray}
\label{ct}
A_n &=& 2\sum\limits_{j = 1}^h{(z_j^nE_j^++z_j^{-n}E_j^-)},\\
\label{ct2}
G_m &=& \sum\limits_{j=1}^h{(z_j^m-z_j^{-m})H_j}
\end{eqnarray}
where $ [E_j^+,E_k^-] = \delta_{jk}H_k,~~[H_j, E_k^{\pm}]=\pm2\delta_{jk}E_k^{\pm}$ are the generators of the $sl_2$ algebra and $z_j$ are called the evaluation parameters of the representation.

Furthermore, after using the expression of $A_n$ in terms of $E_j^{\pm}$ (\ref{ct}), the eigenvalues of the Hamiltonian (\ref{Hamil}) in the sector that is the direct product of $n$ factors of dimension $d_j, ~ 1\le j \le n$ fit the general form
\begin{equation}
\lambda(\kappa, \kappa^*)=\kappa\alpha+\kappa^*\beta +\sum\limits_{j=1}^n{4m_j\sqrt{\kappa^2+{\kappa^*}^2+2\kappa\kappa^*\cos\theta_j}},~~~m_j = -s_j, -s_j+1,\dots, s_j,
\end{equation}
where $z_j = e^{-i\theta_j}$; $\alpha, \beta $ is a pair of eigenvalues of $A_0$, $A_1$; and $d_j=(2s_j+1)$ is the dimension of an irreducible representation of $sl_2$ associated with the pair $z_j, z_j^{-1}$.

\section{Sklyanin's formalism and the $q-$Onsager algebra}
Among the known examples of quadratic algebraic structures, one finds the Yang-Baxter algebra. For further analysis, let us first recall some known results. This algebra consists of a couple $R(u), L(u)$ where the $R$-matrix solves the Yang-Baxter equation
\begin{equation}
R_{\mathcal{V}_0}(u)R_{\mathcal{V}_0\mathcal{V}_0'}(uv)R_{\mathcal{V}_0'}(v)=R_{\mathcal{V}_0'}(v)R_{\mathcal{V}_0\mathcal{V}_0'}(uv)R_{\mathcal{V}_0}(u)
\end{equation}
and the so-called $L$-operator satisfies the quadratic relation
\begin{equation}
\label{224}
R_{\mathcal{V}_0\mathcal{V}_0'}(u/v)(L_{\mathcal{V}_0}(u)\otimes L_{\mathcal{V}_0'}(v))=(L_{\mathcal{V}_0'}(v)\otimes L_{\mathcal{V}_0}(u))R_{\mathcal{V}_0\mathcal{V}_0'}(u/v)
\end{equation}
where $\mathcal{V}_0, \mathcal{V}_0'$ denote finite dimensional auxiliary space representations. Here, the entries of the $L$-operators act on a quantum space denoted $\mathcal{V}$. If one considers a two-dimensional (spin-$\frac{1}{2}$) representation for $\mathcal{V}_0$ and $\mathcal{V}_0'$, a solution $R(u)$ of the Yang-Baxter equation and the $L$-operator can be written in the form %
\begin{eqnarray}
\label{rmatrix}
R(u)&=&\sum\limits_{i, j \in \{0, z, \pm\}}{\omega_{ij}(u)~\sigma_i\otimes\sigma_j},\\
\label{lmatrix}
L(u)&=&\sum\limits_{i, j \in \{0, z, \pm\}}{\omega_{ij}(u)~\sigma_i\otimes S_j},
\end{eqnarray} 
where $\sigma_z, \sigma_{\pm}$ are Pauli matrices, $\sigma_{0} =\mathbb{I}$ and $\omega_{ij}(u)$ are some combinations of functions
\begin{eqnarray}
\omega_{00}(u)&=&\frac{1}{2}(q+1)(u-q^{-1}u^{-1}),\\
\omega_{zz}(u)&=&\frac{1}{2}(q-1)(u+q^{-1}u^{-1}),\\
\omega_{+-}(u)&=&\omega_{-+}(u)=q-q^{-1}.
\end{eqnarray}
Note that the defining relations of the algebra generated by elements $\{S_j\}$ are determined by the equation (\ref{224}). The elements $\{S_j\}$ act on the quantum space $\mathcal{V}$. The corresponding algebra is known as the Sklyanin algebra \cite{Skly88}. It admits a  trigonometric degeneration such that the elements $\{S_j\}$ are identified with the generators $\{S_{\pm}, s_3\}$ of the quantum enveloping algebra $U_q(sl_2)$
\begin{eqnarray}
S_0&=&\frac{q^{s_3}+q^{-s_3}}{q^{1/2}+q^{-1/2}},\\
S_z &=&\frac{q^{s_3}-q^{-s_3}}{q^{1/2}-q^{-1/2}},
\end{eqnarray}
where $[s_3, S_{\pm}]=\pm S_{\pm}$ and $[S_+,S_-]=\frac{q^{2s_3}-q^{-2s_3}}{q-q^{-1}}$, together with the Casimir operator
\begin{equation}
w=qq^{2s_3}+q^{-1}q^{-2s_3}+(q-q^{-1})^2S_-S_+.
\end{equation}

Following \cite{B1}, \cite{Bas1}, let us consider the reflection equation which was first introduced by Cherednik \cite{Cher84} (see also \cite{Skly88}):
\begin{equation}
\label{refequ}
R(u/v)(K(u)\otimes \mathbb{I })R(uv)(\mathbb{I}\otimes K(v))= (\mathbb{I}\otimes K(v))R(uv)(K(u)\otimes\mathbb{I})R(u/v).
\end{equation}
This equation arises, for instance, in the context of the quantum integrable systems with boundaries \cite{Skly88}.
Similarly to (\ref{lmatrix}), in the spin-$\frac{1}{2}$ we introduce a $K$-matrix of the form:
\begin{equation}
\label{form}
K(u)=\sum\limits_{j \in \{0, z, \pm \}}{\sigma_j \otimes \Omega_j(u)}.
\end{equation}
By \cite{B1}, any solutions of the reflection equation (\ref{refequ}) of degree $-2 \le d \le 2$ in the spectral parameter $u$ - with non-commuting entries - can be written in the form (\ref{form}) where
\begin{eqnarray}
\label{2215}
\Omega_0(u) &=&\frac{(A+A^*)(qu-q^{-1}u^{-1})}{2},\\
\Omega_z(u)&=& \frac{(A-A^*)(qu+q^{-1}u^{-1})}{2},\\
\Omega_+(u)&=&-\frac{qu^2+q^{-1}u^{-2}}{c_0c_1(q^2-q^{-2})}-c_1[A^*,A]_q+c_2,\\
\label{2218}
\Omega_-(u)&=&-\frac{qu^2+q^{-1}u^{-2}}{c_1(q^2-q^{-2})}-c_0c_1[A,A^*]_q+c_0c_2
\end{eqnarray}
where the parameters $c_0, c_1 \ne 0$, $c_2$ are arbitrary, and $A, A^*$ have to satisfy the Askey-Wilson relations (\ref{AW1})-(\ref{AW2}) with particular values of the parameters
\begin{eqnarray}
\rho &=& \rho^* = \frac{1}{c_0c_1^2},\\
\omega &=& -\frac{c_2}{c_1}(q-q^{-1}),\\
\gamma &=& \gamma^* =\eta = \eta^* = 0.
\end{eqnarray} 
This explicit relation between the Askey-Wilson algebra (\ref{AW1})- (\ref{AW2}) and the reflection equation algebra through the analysis of $K-$operators suggested to investigate further this new connection. \vspace{1mm}

The generalization of the above connection goes as follows. From the results of \cite{Skly88}, for any parameter $v$ it is for instance known that
\begin{equation}
\label{geso}
K^{(L )}(u)=L_n(uv)\dots L_1(uv)K^{(0 )}(u)L_1(uv^{-1})\dots L_n(uv^{-1}),
\end{equation}
the so-called Sklyanin's operator, gives a family of solutions to (\ref{refequ}). Here $L_j(u)$ is the Lax operator given by (\ref{lmatrix}), and we choose (for simplicity) the trivial solution of (\ref{refequ}) to be $K^{(0)}(u)=(\sigma_+/c_0+\sigma_-)/(q-q^{-1})$. For these choices, the Sklyanin operator acts on the quantum space $\otimes_{j=1}^L\mathcal{V}_j \otimes \mathcal{V}_0$.\vspace{1mm}

Then, the Sklyanin's operator $K^{(L)}$ can be written as follows
\begin{equation}
\label{general}
K^{(L)}(u)=\sum\limits_{j \in \{0, z, \pm \}}{\sigma_j\otimes \Omega_j^{(L)}}(u),
\end{equation}
where the operators $\Omega_j^{(L)}$ are combinations of Laurent polynomials of degree $-2L \le d \le 2L$ in the spectral parameter $u$ and operators acting solely on $\otimes_{j=1}^LU_q(sl_2)$. Baseilhac and Koizumi \cite{BK2} obtained the following result 
\begin{thm}\cite{BK2}
\label{qOnsageral}
For generic values of $L$, the operators $\Omega_j^{(L)}(u)$ are given by:
\begin{eqnarray}
\Omega_0^{(L)}(u)+\Omega_3^{(L)}(u)&=&uq\sum\limits_{k=0}^{L-1}{P_{-k}^{(L)}(u)\mathcal{W}_{-k}^{(L)}}-u^{-1}q^{-1}\sum\limits_{k=0}^{L-1}{P_{-k}^{(L)}(u)\mathcal{W}_{k+1}^{(L)}},\\
\Omega_0^{(L)}(u)-\Omega_3^{(L)}(u)&=&uq\sum\limits_{k=0}^{L-1}{P_{-k}^{(L)}(u)\mathcal{W}_{k+1}^{(L)}}-u^{-1}q^{-1}\sum\limits_{k=0}^{L-1}{P_{-k}^{(L)}(u)\mathcal{W}_{-k}^{L}},\\
\Omega_+^{(L)}(u)&=&\frac{qu^2+q^{-1}u^{-2}}{c_0(q-q^{-1})}P_0^{(L)}(u)+\frac{1}{q+q^{-1}}\sum\limits_{k=0}^{L-1}{P_{-k}^{(L)}}(u)\mathcal{G}_{k+1}^{(L)}+\omega_0^{(L)},~~~~~~~~~~~~~~~\\
\Omega_-^{(L)}(u)&=&\frac{qu^2+q^{-1}u^{-2}}{q-q^{-1}}P_0^{(L)}(u)+\frac{c_0}{q+q^{-1}}\sum\limits_{k=0}^{L-1}{P_{-k}^{(L)}(u)\tilde{\mathcal{G}}_{k+1}^{(L)}}+c_0\omega_0^{(L)},
\end{eqnarray}
where $P_{-k}^{(L)}(u)$ are Laurent polynomials defined by
\begin{eqnarray}
\label{2229}
P^{(L)}_{-k}(u)&=&-\frac{1}{q+q^{-1}}\sum\limits_{n=k}^{L-1}{(\frac{qu^2+q^{-1}u^{-2}}{q+q^{-1}})^{n-k}C_{-n}^{(L)}},\\
\label{2230}
C_{-n}^{(L)}&=&(q+q^{-1})^{n+1}(-1)^{L-n}(\frac{(v^2+v^{-2})w_0^{(j)}}{q+q^{-1}})^{L-n+1}\frac{L!}{(n+1)!(L-n-1)!},~~~~~~~~~~~~\\
~~~~~~~~\omega_0^{(L+1)}&=&-\frac{v^2+v^{-2}}{q+q^{-1}}w_0^{(j)}\omega_0^{(L)},\\
~~~~~~~~\omega_0^{(1)}&=&-\frac{v^2+v^{-2}}{c_0(q-q^{-1})}w_0^{(j)},\\
w_0^{(j)}&=&q^{2j+1}+q^{-2j-1},
\end{eqnarray}
provided the generators $\mathcal W_{-k}^{(L)}, \mathcal W_{k+1}^{(L)}, \mathcal G_{k+1}^{(L)}, \mathcal{\tilde{G}}_{k+1}^{(L)}$ act on $L$-tensor product evaluation representation of $U_q({sl_2})$:
\begin{eqnarray}
\mathcal{W}_0^{(L)}&=& \frac{1}{c_0}vq^{1/4}S_+ q^{s_3/2}\otimes \mathbb{I} + v^{-1}q^{-1/4}S_- q^{s_3/2}\otimes
\mathbb{I} + q^{s_3}\otimes \mathcal{W}_0^{(L-1)},~~~~~~~~~\\
\mathcal{W}_1^{(L)}&=& \frac{1}{c_0}v^{-1}q^{-1/4}S_+ q^{-s_3/2}\otimes \mathbb{I} + vq^{1/4}S_- q^{-s_3/2}\otimes
\mathbb{I} + q^{-s_3}\otimes \mathcal{ W}_1^{(L-1)}\ ,~~~~~~~~~\label{repN+1N}\\
\mathcal{G}_{1}^{(L)}&=& (q-q^{-1})S_-^2\otimes
\mathbb{I}
-\frac{(q^{1/2}+q^{-1/2})}{c_0(q^{1/2}-q^{-1/2})} (v^{2}q^{s_3}+v^{-2}q^{-s_3}) \otimes \mathbb{I}
+\mathbb{I} \otimes \mathcal{ G}_{1}^{(L-1)}~~~~\nonumber\\
&&+(q-q^{-1})\left(
vq^{-1/4}S_-q^{s_3/2}\otimes \mathcal{W}_0^{(L-1)}
+v^{-1}q^{1/4}S_-q^{-s_3/2}\otimes \mathcal{W}_{1}^{(L-1)}
\right)~~~\nonumber\\
&&+\frac{(v^2+v^{-2})w_0^{(j)}}{c_0(q^{1/2}-q^{-1/2})}\otimes \mathbb{I} ,
\end{eqnarray}
\begin{eqnarray}
{\tilde{\mathcal{G}}}_1^{(L)}&=&
\frac{(q-q^{-1})}{c_0^2}S_+^2\otimes \mathbb{I}
-\frac{(q^{1/2}+q^{-1/2})}{c_0(q^{1/2}-q^{-1/2})}(v^{2}q^{-s_3}+v^{-2}q^{s_3})\otimes \mathbb{I}
+ \mathbb{I} \otimes {\tilde {\mathcal G}}_{1}^{(L-1)}\nonumber\\
&&+\frac{(q-q^{-1})}{c_0}\left(
v^{-1}q^{1/4}S_+q^{s_3/2}\otimes \mathcal{W}_0^{(L-1)}
+vq^{-1/4}S_+q^{-s_3/2}\otimes \mathcal{W}_{1}^{(L-1)}
\right)\nonumber\\
&&+\frac{(v^{2}+v^{-2})w_0^{(j)}}{c_0(q^{1/2}-q^{-1/2})}\otimes \mathbb{I} ,
\end{eqnarray}
\begin{eqnarray}
\mathcal{W}_{-k-1}^{(L)}&=&\frac{(w_0^{(j)}-(q^{1/2}+q^{-1/2})q^{s_3})}{(q^{1/2}+q^{-1/2})}\otimes
\mathcal{W}_{k+1}^{(L-1)}
-\frac{(v^2+v^{-2})}{(q^{1/2}+q^{-1/2})}\mathbb{I}\otimes \mathcal{W}_{-k}^{(L-1)}\nonumber\\
&&+\frac{(q^{1/2}-q^{-1/2})}{(q^{1/2}+q^{-1/2})^2}
\left(vq^{1/4}S_+q^{s_3/2}\otimes
\mathcal{G}_{k+1}^{(L-1)}+c_0v^{-1}q^{-1/4}S_-q^{s_3/2}\otimes {\tilde{ \mathcal{G}}}_{k+1}^{(L-1)}\right)\nonumber\\
&&+\frac{(v^2+v^{-2})w_0^{(j)}}{(q^{1/2}+q^{-1/2})^2}\mathcal{W}_{-k}^{(L)}+q^{s_3}\otimes \mathcal{W}_{-k-1}^{(L-1)}\ ,
\end{eqnarray}
\begin{eqnarray}
\mathcal{W}_{k+2}^{(L)}&=&\frac{(w_0^{(j)}-(q^{1/2}+q^{-1/2})q^{-s_3})}{(q^{1/2}+q^{-1/2})}\otimes
\mathcal{W}_{-k}^{(L-1)}
-\frac{(v^2+v^{-2})}{(q^{1/2}+q^{-1/2})}\mathbb{I}\otimes \mathcal{W}_{k+1}^{(L-1)}\nonumber\\
&&+\frac{(q^{1/2}-q^{-1/2})}{(q^{1/2}+q^{-1/2})^2}
\left(v^{-1}q^{-1/4}S_+q^{-s_3/2}\otimes
{\mathcal G}_{k+1}^{(L-1)}+c_0vq^{1/4}S_-q^{-s_3/2}\otimes {\tilde {\mathcal G}}_{k+1}^{(L-1)}\right)\nonumber\\
&&+\frac{(v^2+v^{-2})w_0^{(j)}}{(q^{1/2}+q^{-1/2})^2}\mathcal{W}_{k+1}^{(L)}+q^{-s_3}\otimes {\mathcal W}_{k+2}^{(L-1)},
\end{eqnarray}
\begin{eqnarray}
{\mathcal G}_{k+2}^{(L)}&=& 
\frac{c_0(q^{1/2}-q^{-1/2})^2}{(q^{1/2}+q^{-1/2})}
S_-^2\otimes {\tilde {\mathcal G}}_{k+1}^{(L-1)}
-\frac{1}{(q^{1/2}+q^{-1/2})}(v^{2}q^{s_3}+v^{-2}q^{-s_3})\otimes {\mathcal G}_{k+1}^{(L-1)}\nonumber \\
&&+\mathbb{I} \otimes {\mathcal G}_{k+2}^{(L-1)}+\frac{(v^2+v^{-2})w_0^{(j)}}{(q^{1/2}+q^{-1/2})^2}{\mathcal G}_{k+1}^{(L)}\nonumber\\
&&+ (q-q^{-1})\left(
vq^{-1/4}S_-q^{s_3/2}\otimes \big({\mathcal W}_{-k-1}^{(L-1)}-{\mathcal W}_{k+1}^{(L-1)}\big)\right.\nonumber\\
&&~~~~~~~~~~~~~~~~~+\left. v^{-1}q^{1/4}S_-q^{-s_3/2}\otimes \big({\mathcal W}_{k+2}^{(L-1)}-{\mathcal W}_{-k}^{(L-1)}\big)
\right),
\end{eqnarray}
\begin{eqnarray}
{\tilde {\mathcal G}}_{k+2}^{(L)}&=& 
\frac{(q^{1/2}-q^{-1/2})^2}{c_0(q^{1/2}+q^{-1/2})}
S_+^2\otimes {{\mathcal G}}_{k+1}^{(L-1)}
-\frac{1}{(q^{1/2}+q^{-1/2})}(v^{2}q^{-s_3}+v^{-2}q^{s_3})\otimes {\tilde {\mathcal G}}_{k+1}^{(L-1)}\nonumber\\ 
&&+\mathbb{I} \otimes {\tilde {\mathcal G}}_{k+2}^{(L-1)}+\frac{(v^2+v^{-2})w_0^{(j)}}{(q^{1/2}+q^{-1/2})^2}{\tilde{\mathcal G}}_{k+1}^{(L)}\ \nonumber\nonumber\\
&&+ \frac{(q-q^{-1})}{c_0}\left(
v^{-1}q^{1/4}S_+q^{s_3/2}\otimes \big({\mathcal W}_{-k-1}^{(L-1)}-{\mathcal W}_{k+1}^{(L-1)}\big)\right.\nonumber\\
&&~~~~~~~~~~~~~~~~+\left. vq^{-1/4}S_+q^{-s_3/2}\otimes \big({\mathcal W}_{k+2}^{(L-1)}-{\mathcal W}_{-k}^{(L-1)}\big)
\right),
\end{eqnarray}
for $k\in\{0,1,...,L-2\}$,

and satisfy the analog of Davies relations
\begin{eqnarray}
\label{qdv1}
c_0(q-q^{-1})\omega_0^{(L)}\mathcal W_{0}^{(L)}-\sum\limits_{k=1}^L{C_{-k+1}^{(L)}\mathcal W_{-k}^{(L)}}&=&0,\\
c_0(q-q^{-1})\omega_0^{(L)}\mathcal W_{1}^{(L)}-\sum\limits_{k=1}^L{C_{-k+1}^{(L)}\mathcal W_{k+1}^{(L)}}&=&0,\\
c_0(q-q^{-1})\omega_0^{(L)}\mathcal G_{1}^{(L)}-\sum\limits_{k=1}^L{C_{-k+1}^{(L)}\mathcal G_{k+1}^{(L)}}&=&0,\\
\label{qdv4}
c_0(q-q^{-1})\omega_0^{(L)}\mathcal{\tilde{G}}_{1}^{(L)}-\sum\limits_{k=1}^L{C_{-k+1}^{(L)}\mathcal{\tilde{G}}_{k+1}^{(L)}}&=&0
\end{eqnarray}
with $C^{(L)}_{-k+1}$ given by (\ref{2230}).\\

\end{thm}

As a consequence of the reflection equation algebra and the fact that the vector space on which the elements act is finite dimensional, the generators $\mathcal W_{-k}^{(L)}, \mathcal W_{k+1}^{(L)}, \mathcal G_{k+1}^{(L)}, \mathcal {\tilde{G}}_{k+1}^{(L)}$ generate a quotient of the infinite dimensional algebra ${\cal A}_q$ (\ref{aq28})-(\ref{aq38}) by the relations (qDavies) (\ref{qdv1})-(\ref{qdv4}).\vspace{1mm}

Importantly, in \cite{BK1} it was shown that the first two elements $\mathcal{W}^{(L)}_0, \mathcal{W}^{(L)}_1$ of the family of generators $\mathcal W_{-k}^{(L)}, \mathcal W_{k+1}^{(L)}, \mathcal G_{k+1}^{(L)}, \mathcal{\tilde{G}}_{k+1}^{(L)}$ of the quotient of ${\cal A}_q$ satisfy the $q$-deformed Dolan-Grady relations
\begin{eqnarray}
[\mathcal{W}_0^{(L)},[\mathcal{W}_0^{(L)},[\mathcal{W}_0^{(L)},\mathcal{W}_1^{(L)}]_q]_{q^{-1}}]&=&\rho~[\mathcal{W}_0^{(L)},\mathcal{W}_1^{(L)}],\label{qDGXXZ}\\
~[\mathcal{W}_1^{(L)},[\mathcal{W}_1^{(L)},[\mathcal{W}_1^{(L)},\mathcal{W}_0^{(L)}]_q]_{q^{-1}}]&=&\rho~[\mathcal{W}_1^{(L)},\mathcal{W}_0^{(L)}],
\end{eqnarray}
with
\begin{eqnarray}
\rho= (q+q^{-1})^2k_+k_-.\nonumber
\end{eqnarray}

For all these reasons, in the literature on the subject the algebra 
${\cal A}_q$ is sometimes called the $q-$deformed analog of the Onsager algebra, and the relations (\ref{qdv1})-(\ref{qdv4}) are called the $q-$deformed analog of Davies relations.

\section{Commuting quantities and the $q-$Dolan-Grady hierarchy}
For the Onsager algebra (\ref{relaOn1})-(\ref{relaOn3}), the explicit construction of mutually commuting quantities that generate an Abelian subalgebra of the Onsager algebra (\ref{cmm}) has been considered in details in \cite{O,DG}. For the $q-$Onsager algebra, an analogous construction is, in general, a rather complicated problem. However, using the connection between the reflection equation algebra and the infinite dimensional algebra ${\cal A}_q$ or alternatively the $q-$Onsager algebra, this problem can be handled. Indeed, it is known that starting from solutions of the reflection equation algebra, a transfer matrix that generates mutually commuting quantities can be constructed.
Namely, following the analysis of \cite{Skly88}, for any values of the spectral parameters $u, v$
\begin{equation}
\label{trace}
[t^{(L)}(u), t^{(L)}(v)] = 0 ~~~\text{where}~~ t^{(L)}(u) = tr_0\{K_+(u)K^{(L)}(u)\}
\end{equation}
where $tr_0$ denotes the trace over the two-dimensional auxiliary space. The idea of \cite{BK2} was then to consider the expansion of (\ref{trace}) in order to extract mutually commuting quantities expressed in terms of the generators of the infinite dimensional algebra ${\cal A}_q$. To this end, let us plug the $c-$number solution of the dual reflection equation\footnote{This equation is obtained from (\ref{refequ}) by changing $u \to u^{-1}, v \to v^{-1}$ and $K(u)$ in its transpose} \cite{IK}, \cite{GZ} given by
\begin{equation}
\label{trivial}
K_+(u) = \left(\begin{array}{cc}
uq\kappa+u^{-1}q^{-1}\kappa^* & \kappa_-(q+q^{-1})(q^2u^2-q^{-2}u^{-2})/c_0\\
\kappa_+(q+q^{-1})(q^2u^2-q^{-2}u^{-2}) & uq\kappa^* +u^{-1}q^{-1}\kappa
\end{array}\right)
\end{equation}
where $\kappa^* = \kappa^{-1}$ for any $\kappa \in \mathbb{C}$, and $\kappa_{\pm} \in \mathbb{C}$.
Substitute (\ref{trivial}), and (\ref{general}) with the generators $\mathcal W_{-k}^{(L)}, \mathcal W_{k+1}^{(L)}, \mathcal G_{k+1}^{(L)}, \mathcal{\tilde{G}}_{k+1}^{(L)}$ of Theorem \ref{qOnsageral} into (\ref{trace}). By straightforward calculations, it follows \cite{BK2} 
\begin{equation}
t^{(L)}(u)=\sum_{k=0}^{L-1}{(q^2u^2-q^{-2}u^{-2})P^{(L)}_{-k}(u)\mathcal{I}^{(L)}_{2k+1}+\mathcal{F}(u)\mathbb{I}}
\end{equation}
with (\ref{2229}) and
\begin{equation}
\mathcal{F}(u)=\frac{(q+q^{-1})(q^2u^2-q^{-2}u^{-2})}{c_0}\left(\frac{(qu^2+q^{-1}u^{-2})P_0^{(L)}(u)}{q-q^{-1}}+c_0\omega_0^{(L)}\right)(\kappa_++\kappa_-).
\end{equation}
Here, we have introduced the generators $\mathcal{I}^{(L)}_{2k+1}$ which can be written in terms of the generators of the quotient of ${\cal A}_q$. Explicitly, they read:
\begin{equation}
\label{tcmm}
\mathcal{I}_{2k+1}^{(L)}=\kappa \mathcal W^{(L)}_{-k}+\kappa^* \mathcal W^{(L)}_{k+1}+\kappa_+\mathcal G^{(L)}_{k+1}+\kappa_-\mathcal{\tilde{G}}^{(L)}_{k+1},
\end{equation}
for $k\in \{0, 1, \dots, L-1\}$. As a consequence of the property (\ref{trace}), it leads to
\begin{equation}
\label{qcmm}
[\mathcal{I}^{(L)}_{2k+1},\mathcal{I}^{(L)}_{2l+1}] =0 ~~~ \text{for all}~~ k, l \in \{0, 1, \dots, L-1\}.
\end{equation}
%
The mutually commuting quantities (\ref{tcmm}) are the $q-$deformed analog of the mutually commuting quantities (\ref{cmm}) of the Onsager algebra.\vspace{1mm}

Clearly, an interesting problem is to write these quantities solely in terms of the fundamental generators of the $q-$Onsager algebra, in view of the connection between the quotient of ${\cal A}_q$  and the $q-$Onsager algebra. Such problem was first considered in details in \cite{BK2}, and generalized in \cite{BS}, where it is shown that the mutually conserved quantities are polynomials of the fundamental generators  $\mathcal{W}_0^{(L)}, \mathcal{W}_1^{(L)}$. For instance, 
\begin{eqnarray}
\mathcal G_1^{(L)}&=& [\mathcal{W}_1^{(L)}, \mathcal{W}_0^{(L)}]_q, ~~ \mathcal{\tilde{G}}^{(L)}_1 = [\mathcal{W}_0^{(L)}, \mathcal{W}_1^{(L)}],~~~\\
\mathcal{W}^{(L)}_{-1}&=&\frac{1}{\rho}\left((q^2+q^{-2})\mathcal{W}_0^{(L)}\mathcal{W}_1^{(L)}\mathcal{W}_0^{(L)}-{\mathcal{W}_0^{(L)}}^2\mathcal{W}_1^{(L)}-{\mathcal{W}_1^{(L)}}{\mathcal{W}_0^{(L)}}^2\right)+\mathcal{W}_1^{(L)},~~~~~~\\
\mathcal{W}^{(L)}_{2}&=&\frac{1}{\rho}\left((q^2+q^{-2})\mathcal{W}_1^{(L)}\mathcal{W}_0^{(L)}\mathcal{W}_1^{(L)}-{\mathcal{W}_1^{(L)}}^2\mathcal{W}_0^{(L)}-{\mathcal{W}_0^{(L)}}{\mathcal{W}_1^{(L)}}^2\right)+\mathcal{W}_0^{(L)}.~~~~~~
\end{eqnarray}

\section{The open $XXZ$ spin chain and the $q$-Onsager algebra}
Since Sklyanin's work \cite{Skly88} on spin chains with integrable boundary conditions, finding exact results such as the energy spectrum of a model, corresponding eigenstates and correlation functions for elementary excitations has remained an interesting problem in connection with condensed matter or high-energy physics. Among the known integrable open spin chains that have been considered in details, one finds the open $XXZ$ spin chain. 

For generic boundary conditions, the Hamiltonian of the $XXZ$ spin chain reads
\begin{eqnarray}
\label{equhamilt1}
H^{(L)}_{XXZ}&=&\sum_{k=1}^{L-1}\Big(\sigma_1^{k+1}\sigma_1^{k}+\sigma_2^{k+1}\sigma_2^{k} + \Delta\sigma_z^{k+1}\sigma_z^{k}\Big)\\&& +\ \frac{(q-q^{-1})}{2}\frac{(\epsilon_+ - \epsilon_-)}{(\epsilon_+ + \epsilon_-)}\sigma^1_z + \frac{2}{(\epsilon_+ + \epsilon_-)}\big(k_+\sigma^1_+ + k_-\sigma^1_-\big)     \nonumber\\
\ && +\ \frac{(q-q^{-1})}{2}\frac{(\bar{\epsilon}_+ - \bar{\epsilon}_-)}{(\bar{\epsilon}_+ +\bar{\epsilon}_-)}\sigma^L_z + \frac{2}{(\bar{\epsilon}_+ + \bar{\epsilon}_-)}\big(\bar{k}_+\sigma^L_+ + \bar{k}_-\sigma^L_-\big)     \nonumber
\end{eqnarray}
where $L$ is the number of sites, $\Delta = \frac{q+q^{-1}}{2}$ denotes the anisotropy parameter, and $\sigma_{\pm}$, $\sigma_1$, $\sigma_2$, $\sigma_z$ are the usual Pauli matrices. 

Besides the anisotropy parameter, $\epsilon_\pm,k_\pm$ (resp. $ \bar{\epsilon}_\pm,\bar{k}_\pm$) denote the right (resp. left) boundary parameters associated with the right (resp. left) boundary. Considering a gauge transformation, note that one parameter might be removed. For symmetry reasons, we however keep the boundary parametrization as defined above.

For generic values of $q$ and generic non-diagonal boundary parameters, the Hamiltonian $H$ can be formulated within the so-called Sklyanin's formalism  (boundary quantum inverse scattering)  \cite{Skly88}. In this standard approach, a generating function for all mutually commuting quantities besides $H$ is introduced. It is built from a $R-$matrix solution of the Yang-Baxter equation and two $K-$matrices solutions of the reflection equation \cite{Skly88}.  However, the algebraic setting based on $U_q(\widehat{sl_2})$ is obscured in this formulation, and the standard algebraic Bethe ansatz fails to apply for generic boundary parameters. Instead, an alternative formulation has been proposed \cite{BK1}, which is based on an analog of Onsager's approach for the two-dimensional Ising model: the transfer matrix - denoted $t^{(L)}(\zeta)$ below - is written in terms of mutually commuting quantities ${{\mathcal I}}_{2k+1}$ that generate a $q-$deformed analog of the Onsager-Dolan-Grady's hierarchy\footnote{The Onsager's (also called Dolan-Grady \cite{DG}) hierarchy is an Abelian algebra with elements of the form ${{\mathcal{I}_{2n+1}}}= \bar{\epsilon}_+ (A_n+A_{-n}) + \bar{\epsilon}_-(A_{n+1}+A_{-n+1}) + \kappa G_{n+1}$ with $n\in{\mathbb Z}_+$, generated from the Onsager algebra with defining relations $\big[A_{n},A_{m}\big]= 4G_{n-m},\ \big[G_m,A_n\big] =2A_{n+m}-2A_{n-m}$ and $\ \big[G_n,G_m\big] =0$ for any $n,m\in{\mathbb Z}$ \cite{O}.}. Namely,
\begin{equation}
t^{(L)}(\zeta)= \sum_{k=0}^{L-1}{\mathcal F}_{2k+1}(\zeta)\ {{\mathcal I}}_{2k+1} + {\mathcal F}_0(\zeta)\mathbb I\quad \mbox{with} \qquad  \big[{\mathcal I}_{2k+1},{\mathcal I}_{2l+1}\big]=0 \qquad \label{tfin} 
\end{equation}
for all $k,l\in 0,...,L-1$ where $\zeta$ is the so-called spectral parameter,
\begin{equation}
{{\mathcal I}}_{2k+1}=\bar{\epsilon}_+\mathcal{W}_{-k} + \bar{\epsilon}_- \mathcal{W}_{k+1} + \frac{1}{q^2-q^{-2}}\Big(\frac{\bar{k}_-}{k_-} \mathcal{G}_{k+1} 
+ \frac{\bar{k}_+}{k_+}\tilde {\mathcal{ G}}_{k+1}\Big)\ \label{Ifin}
\end{equation}
and ${\cal F}_{2k+1}(\zeta)$ are given in \cite{BK1}
\begin{eqnarray*}
\mathcal{F}_{2k+1}(\zeta)&=&-\frac{1}{q+q^{-1}}\frac{(q^2\zeta^4+q^{-2}\zeta^{-4}-q^2-q^{-2})}{(q^2+q^{-2}-\zeta^2-\zeta^{-2})^L}\sum\limits_{n =k}^{L-1}{(\frac{q\zeta^2+q^{-1}\zeta^{-2}}{q+q^{-1}})^{n-k}C_{-n}^{(L)}},\\
\mathcal{F}_0(\zeta)&=&(\frac{(q^2\zeta^4+q^{-2}\zeta^{-4}-q^2-q^{-2})(q+q^{-1})}{(q^2+q^{-2}-\zeta^2-\zeta^{-2})})(\frac{k_+k_-}{q-q^{-1}})(\frac{\bar{k}_+}{k_+}+\frac{\bar{k}_-}{k_-})\\
&&\times\left(\frac{\omega_0^{(L)}(q-q^{-1})}{k_+k_-}-\sum\limits_{n=0}^{L-1}{(\frac{q\zeta^2+q^{-1}\zeta^{-2}}{q+q^{-1}})^{n+1}C_{-n}^{(L)}}\right)\\
&&+\frac{((q\zeta^2+q^{-1}\zeta^{-2})\bar{\epsilon}_++(q+q^{-1})\bar{\epsilon}_-)\epsilon_+^{(L)}}{(q^2+q^{-2}-\zeta^2-\zeta^{-2})^L}\\
&&+\frac{((q\zeta^2+q^{-1}\zeta^{-2})\bar{\epsilon}_-+(q+q^{-1})\bar{\epsilon}_+)\epsilon_-^{(L)}}{(q^2+q^{-2}-\zeta^2-\zeta^{-2})^L}
\end{eqnarray*}
The explicit expressions for the coefficients $C_{-n}^{(L)}$ and $\omega_0^{(L)}, \epsilon_{\pm}^{(L)}$ essentially depend on the choice of quantum space representations at each site (two dimensional for (\ref{equhamilt1})) of the spin chain
\begin{equation}
C_{-n}^{(L)}=(-1)^{L-n}(q+q^{-1})^{n+1}\sum\limits_{k_1<\dots<k_{L-n-1}=1}^L{\alpha_{k_1}\dots\alpha_{k_{L-n-1}}}
\end{equation}
where 
\begin{eqnarray}
\alpha_1 &=&\frac{2(q^2+q^{-2})}{(q+q^{-1})}+\frac{\epsilon_+\epsilon_-(q-q^{-1})^2}{k_+k_-(q+q^{-1})},\\
\alpha_k&=&\frac{2(q^2+q^{-2})}{q+q^{-1}} ~~~ \text{for }~ k = 2, \dots, L,
\end{eqnarray}
and for arbitrary values for $L$
\begin{eqnarray}
\epsilon^{(L)}_{\pm}&=&(q^2+q^{-2})\epsilon^{(L-1)}_{\mp}-2\epsilon_{\pm}^{(L-1)},~~  \epsilon_{\pm}^{(0)}=\epsilon_{\pm},\\
\omega_0^{(L)}&=&(-1)^L\frac{k_+k_-}{q-q^{-1}}\left(\frac{2(q^2+q^{-2})}{q+q^{-1}}\right)^{L-1}\left(\frac{2(q^2+q^{-2})}{q+q^{-1}}+\frac{\epsilon_+\epsilon_-(q-q^{-1})^2}{k_+k_-(q+q^{-1})}\right)
\end{eqnarray}
Note that the parameters $\epsilon_\pm$ of the right boundary - which do not appear explicitly in the formula (\ref{Ifin})- are actually hidden in the definition of the elements  $\mathcal{W}_{-k},\mathcal{W}_{k+1}, \mathcal{ G}_{k+1},\mathcal{\tilde{G}}_{k+1}$  \cite{BK1,BK3}. \\

Since the generators $\mathcal{W}_{-k},\mathcal{W}_{k+1}, \mathcal{ G}_{k+1}, \mathcal{\tilde{G}}_{k+1}$ of the $q$-deformed analogue of the Onsager algebra possess the block diagonal and block tridiagonal structure in the eigenbasis of ${\cal W}_0$ or the eigenbasis of ${\cal W}_1$ \cite{BK3}, the general spectral problem of all nonlocal commuting operators $\mathcal{I}_{2k+1}^{(L)}$ is solved
 \begin{equation}
 \mathcal{I}^{(L)}_{2k+1}\Psi^{(L)}_{\Lambda_1^{(L)}}=\Lambda_{2k+1}^{(L)}\Psi^{(L)}_{\Lambda_1^{(L)}}~~~ \text{for}~ k = 0, 1, \dots, L-1.
 \end{equation}
It follows that the eigenvalues $\Lambda_{XXZ}(\zeta)$ of the transfer matrix (\ref{tfin}) are given by: 
\begin{equation}
\Lambda_{XXZ}(\zeta)= \sum\limits_{k=0}^{L-1}{\mathcal{F}_{2k+1}(\zeta)\Lambda^{(L)}_{2k+1}+\mathcal{F}_0(\zeta)}.
\end{equation}
The spectrum $E$ of the Hamiltonian (\ref{equhamilt1}) immediately follows:
\begin{eqnarray}
E &=& \frac{(q-q^{-1})(q+q^{-1})^{-1}}{2(\bar{\epsilon}_++\bar{\epsilon}_-)(\epsilon_++\epsilon_-)}\left(\sum\limits_{k=0}^{L-1}{\frac{d\mathcal{F}_{2k+1}(\zeta)}{d\zeta}|_{\zeta=1}}~\Lambda^{(L)}_{2k+1}+\frac{d\mathcal{F}_0(\zeta)}{d\zeta}|_{\zeta=1}\right)\\
&&-\left(L\Delta+\frac{(q-q^{-1})^2}{2(q+q^{-1})}\right)\nonumber.
\end{eqnarray}
Indeed, it is known that the Hamiltonian of the XXZ open spin chain with generic integrable boundary conditions (\ref{equhamilt1}) is obtained as follows:
\begin{equation}
\frac{d}{d\zeta} ln(t^{(L)}(\zeta))|_{\zeta=1}=  \frac{2}{(q-q^{-1})}H + \Big(\frac{(q-q^{-1})}{(q+q^{-1})} + \frac{2L}{(q-q^{-1})}\Delta \Big)\mathbb I \label{expH}\ .
\end{equation}
More generally, higher mutually commuting local conserved quantities, say $H_n$ with $H_1\equiv H$, can be derived similarly by taking higher derivatives of the transfer matrix (\ref{tfin}).\vspace{1mm}

To resume, in the $q-$Onsager approach of the open XXZ spin chain with generic integrable boundary conditions \cite{B1, BK3} the Hamiltonian is written in terms of the elements (\ref{Ifin}). These elements generate an Abelian subalgebra of the $q-$deformed analog of the Onsager algebra. For the open XXZ spin chain and 
$q$ generic, the vector space on which the elements act is  irreducible finite dimensional. So, stricly speaking the elements $\mathcal W_{-k}, \mathcal W_{k+1}, \mathcal G_{k+1}, \mathcal{\tilde{G}}_{k+1}$ generate a quotient of the infinite dimensional algebra ${\cal A}_q$ (\ref{aq28})-(\ref{aq38}) by the relations (\ref{qdv1})-(\ref{qdv4}).

In \cite{Bas3}, recall that two (dual) bases of the vector space were constructed, on which the operators  ${\cal W}_0,{\cal W}_1$ act as a tridiagonal pair. Remarkably, in these bases the more general operators (called descendants)  $\mathcal{W}_{-k},\mathcal{W}_{k+1}, \mathcal{ G}_{k+1}, \mathcal{\tilde{G}}_{k+1}$ of the $q$-deformed analogue of the Onsager algebra act as block diagonal or block tridiagonal matrices. In \cite{BK3}, the spectrum and eigenstates of the Hamiltonian (\ref{equhamilt1}) were derived using this remarkable property. It was the first solution proposed in the literature for the open XXZ spin chain with generic integrable boundary conditions.

\chapter{Main Results}
In this Chapter, we review the three main results of the thesis. The first two results have their own interest in the context of tridiagonal algebras, $q-$Onsager algebra, coideal subalgebras of $U_q(\widehat{sl_2})$ and their higher rank generalizations: analogs of Lusztig's higher order relations are conjectured, and supporting evidences for these are described. In particular, the theory of tridiagonal pairs recalled in Chapter 1 plays a central role in the analysis for the $q-$Onsager algebra associated with $\widehat{sl_2}$. The third main result of this thesis overlaps between mathematics and physics. In mathematics, quantum universal enveloping algebras have been studied at roots of unity in the literature \cite{Luszt}. The introduction of the divided powers of the Chevalley generators plays a central role in the construction. Here, for the $q-$Onsager algebra, analogs of these elements are introduced and studied in details at roots of unity at least for a certain class of irreducible finite dimensional representations that finds applications in physics. These elements are divided polynomials of the two fundamental generators of the $q-$Onsager algebra. For a special case, they satisfy a pair of relations that share some similarity with the higher order $q-$Dolan-Grady relations previously conjectured. All together, the $q-$Onsager generators and the divided polynomials generate a new quantum algebra. Based on this construction, the commutation relations between the Hamiltonian of the open $XXZ$ chain at roots of unity
and the four generators are studied. 

\section{Higher order relations for the $q$-Onsager algebra} 
\subsection{Motivation}
Consider the quantum universal enveloping algebras for arbitrary Kac-Moody algebras $\widehat{g}$ introduced by Drinfeld \cite{Dr}, and Jimbo \cite{Jim}.  Let  $\{a_{ij}\}$ be the extended Cartan matrix of $\widehat{g}$. Fix coprime integers $d_i$ such that $d_ia_{ij}$ is symmetric. Define $q_i=q^{d_i}$. The quantum universal enveloping algebra $U_q(\widehat{g})$   is generated by the elements $\{h_j,e_j,f_j\}$, $j=0,1,...,rank(g)$, which satisfy the defining relations:
\begin{eqnarray}
 &&[h_i,h_j]=0\ , \quad [h_i,e_j]=a_{ij}e_j\ , \quad
[h_i,f_j]=-a_{ij}f_j\ ,\\
&&[e_i,f_j]=\delta_{ij}\frac{q_i^{h_i}-q_i^{-h_i}}{q_i-q_i^{-1}},  \quad i,j=0,1,...,rank(g) \ , \ \nonumber \\
&&e_i\,e_j=e_j\,e_i\ , \qquad f_i\,f_j=f_j\,f_i\  \,\qquad \mbox{for} \qquad |i-j|>1 \ ,
\nonumber\end{eqnarray}
together with the so-called {\it quantum Serre relations}
($i\neq j$)
\begin{eqnarray}
\label{relqse}
\sum_{k=0}^{1-a_{ij}}(-1)^k
\left[ \begin{array}{c}
1-a_{ij} \\
k 
\end{array}\right]_{q_i}
e_i^{1-a_{ij}-k}\,e_j\,e_i^k&=&0\ ,\\
\sum_{k=0}^{1-a_{ij}}(-1)^k
\left[ \begin{array}{c}
1-a_{ij} \\
k
\end{array}\right]_{q_i}
f_i^{1-a_{ij}-k}\,f_j\,\,f_i^k&=&0\ ,  \label{Uqserre}
\end{eqnarray}
where $\delta_{ij}$ denotes Kronecker delta.\\

In the mathematical literature \cite{Luszt}, generalizations of the relations (\ref{relqse})-(\ref{Uqserre}) - the so-called {\it higher order quantum ($q-$)Serre relations} - have been proposed. For  $\widehat{g}=\widehat{sl_2}$, they read\footnote{For  $\widehat{g}=\widehat{sl_2}$, recall that $a_{ii}=2$, $a_{ij}=-2$  with $i,j=0,1$.}  \cite{Luszt}:
\begin{eqnarray}
\sum_{k=0}^{2r+1} (-1)^{k} \,  \left[ \begin{array}{c}
2r+1 \\
k
\end{array}\right]_{q} e_i^{2 r+1-k} e_j^{r} e_i^{k}&=&0\ , \label{hqSerre}\\
\sum_{k=0}^{2r+1} (-1)^{k}  \,  \left[ \begin{array}{c}
2r+1 \\
k
\end{array}\right]_{q} f_i^{2r+1-k} f_j^{r} f_i^{k}&=&0 \ \ \quad \mbox{for} \quad i\neq j,\ \  i,j=0,1 \ . \label{hqSerre1}
\end{eqnarray}

Consider the relations (\ref{relqse})-(\ref{Uqserre}) for $g = sl_2$, we observe that the $q$-Onsager algebra is closely related with $U_q(\widehat{sl_2})$. In particular, there exists an homomorphism from the $q$-Onsager algebra to a coideal subalgebras of $U_q(\widehat{sl_2})$ (see Chapter 1).

 Recall that the $q-$Onsager algebra is a special case of the tridiagonal algebra: it corresponds to the {\it reduced} parameter sequence $\gamma=0,\gamma^*=0$, $\beta=q^2+q^{-2}$ and $\rho=\rho_0$, $\rho^*=\rho_1$  which exhibits all interesting properties that can be extended to more general parameter sequences. The defining relations of the $q-$Onsager algebra read
\begin{eqnarray}
[A,[A,[A,A^*]_q]_{q^{-1}}]=\rho_0[A,A^*]\
,\qquad
[A^*,[A^*,[A^*,A]_q]_{q^{-1}}]=\rho_1[A^*,A]\
\label{qDG} , \label{qOA}\end{eqnarray}
which can be seen as $\rho_i-$deformed analogues of the $q-$Serre relations (\ref{relqse})-(\ref{Uqserre}) associated with $\widehat{g}\equiv \widehat{sl_2}$. For $q=1$, $\rho_0=\rho_1=16$, note that they coincide with the Dolan-Grady relations \cite{DG}.\\

In the study of tridiagonal algebras and the representation theory associated with the special case $\rho_0=\rho_1=0$,  higher order $q-$Serre relations (\ref{hqSerre}), (\ref{hqSerre1}) play an important role \cite{IT03} in the construction of a basis of the corresponding vector space. As  suggested in \cite[Problem~3.4]{IT03}, for $\rho_0\neq 0$, $\rho_1\neq 0$ finding analogues of the higher order $q-$Serre relations for the $q-$Onsager algebra is an interesting problem. Another interest for the construction of higher order relations  associated with the $q-$Onsager algebra comes from the theory of quantum integrable systems with boundaries. Indeed, by analogy with the case of periodic boundary conditions \cite{DFM}, such relations or similar relations should play a role in the identification of the symmetry of the Hamiltonian of the $XXZ$ open spin chain at $q$ a root of unity and for special boundary parameters. Motivated by these open problems, below we focus on the explicit construction of higher order tridiagonal relations and the special case associated with the $q$-Onsager algebra.
\subsection{Conjecture about the higher order relations of the $q$-Onsager algebra}
By analogy with Lusztig's higher order relations (\ref{hqSerre}), (\ref{hqSerre1}) for quantum universal enveloping algebras, it is natural to expect higher order relations for the $q-$Onsager algebra.
\begin{conj}\label{conj1} Let $A$, $A^*$ be the fundamental generators of the $q-$Onsager algebra (\ref{qDG}), then $A, A^*$ satisfy the higher order $q$-Dolan-Grady relations as follows:
\begin{eqnarray}
\sum_{p=0}^{r}\sum_{j=0}^{2r+1-2p}  (-1)^{j+p}  \rho_0^{p}\, {c}_{j}^{[r,p]}\,  A^{2 r+1-2p-j} {A^*}^{r} A^{j}&=&0\  \label{qDGfinr} \ ,\\
\sum_{p=0}^{r}\sum_{j=0}^{2r+1-2p} (-1)^{j+p} \rho_1^{p} \, {c}_{j}^{[r,p]}\,  {A^*}^{2 r+1-2p-j} {A}^{r} {A^*}^{j}&=&0\  
\ \label{qDGfinr2}
\end{eqnarray}
where $c_{2(r-p)+1-j}^{[r,p]} =c_{j}^{[r,p]}$ and
\begin{eqnarray}
&&c_{j}^{[r,p]} =  \sum_{k=0}^j   \frac{(r-p)!}{(\{\frac{j-k}{2}\})!(r-p-\{\frac{j-k}{2}\})!}
    \sum_{{{\cal P}_k}}  [s_1]^2_{q^2}...[s_p]^2_{q^2}  \frac{[2s_{p+1}]_{q^2}...[2s_{p+k}]_{q^2}}{[s_{p+1}]_{q^2}...[s_{p+k}]_{q^2}},~~~~  \label{cfinr}
\end{eqnarray}
\begin{eqnarray} 
\mbox{with}\quad \ \left\{\begin{array}{cc}
\!\!\! \!\!\! \!\!\! \!\!\!  \!\!\! \!\!\! \!\!\! \!\!\!   \!\!\! \!\!\! \!\!\! \!\!\! j= \overline{0,r-p}\ , \quad s_i\in\{1,2,...,r\}\ ,\\
{{\cal P}_k}: \begin{array}{cc} \ \ s_1<\dots<s_p\ ;\quad \ s_{p+1}<\dots<s_{p+k}\ ,\\
 \{s_{1},\dots,s_{p}\} \cap \{s_{p+1},\dots,s_{p+k}\}=\emptyset \end{array}
\end{array}\right.\ .\nonumber
\end{eqnarray} 
\end{conj}
Below we give several supporting evidences for this conjecture. First, it is shown that a generalized version of the conjecture holds for every TD pair $(A,A^*)$ of $q-$Racah type. As a special case, the higher order relations for the $q-$Onsager algebra are obtained and have the form of (\ref{cfinr}).
Then, the conjecture is explicitely proven for $r=2,3$. Finally, the conjecture is studied recursively. Using a computer program, up to $r=10$ it is checked that the conjecture holds.

\subsection{Higher order relations and tridiagonal pairs}
Let $A,A^*$ act on an irreducible finite dimensional vector space $V$. In this Section, it is shown that $A,A^*$ satisfy higher order relations which are generalizations (\ref{cfinr}). For $r=1$, these relations are the defining relations of the $q$-Onsager algebra.
This is proven using the properties of tridiagonal pairs described in Chapter 1.\\

Let $(A, A^*)$ denote a tridiagonal pair on $V$, the sequence $\theta_0, \theta_1, \dots, \theta_d$ (resp. $\theta_0^*, \theta_1^*, \dots, \theta_d^*$) denote a standard ordering of the eigenvalues of $A$ (resp. $A^*$). For each positive integer $s$, let $\beta_s, \gamma_s, \gamma_s^*, \delta_s, \delta_s^*$ in $\overline{\mathbb{K}}$ satisfying (\ref{pt1}), (\ref{pt2}).
\begin{defn}
Let $x, y$ denote commuting indeterminates. For each positive integer $r$, we define the polynomials $p_r(x,y), p_r^*(x,y)$ as follows:
\begin{eqnarray}
p_r(x,y) = (x-y)\prod_{s=1}^{r} (x^2-\beta_s xy +y^2 - \gamma_s (x+y) -\delta_s)\ ,\label{defpoly}\\
p^*_r(x,y) = (x-y)\prod_{s=1}^{r} (x^2-\beta_s xy +y^2 - \gamma^*_s (x+y) -\delta^*_s)\ .
\end{eqnarray}
We observe $p_r(x,y)$ and $p^*_r(x,y)$ have a total degree $2r+1$ in $x,y$.
\end{defn}
\begin{lem}
\label{lem:polyroot}
For each positive integer $r$, $p_r(\theta_i, \theta_j) = 0,$ and $p_r^*(\theta_i^*,\theta^*_j)=0 $ if $|i-j|\leq r, (0\leq i,j \leq d)$.
\end{lem}
\begin{proof}
Let two integers $i, j$ such that $|i-j|\le r, ~(0 \le i, j \le d)$.\\

If $i = j$, then $\theta_i = \theta_j, \theta_i^* = \theta_j^*$. Hence, $p_r(\theta_i,\theta_j)=p_r^*(\theta_i^*, \theta_j^*)=0$.\\

If $i \ne j$, since $0<|i-j|\le r$, there exists an integer $1\le s \le r$ such that $|i-j|=s$.\\

By Lemma (\ref{lem:polypart}),
\[\theta_i^2-\beta_s\theta_i\theta_j+\theta_j^2-\gamma_s(\theta_i+\theta_j)-\delta_s = 0,\]
\[{\theta^*_i}^2-\beta^*_s\theta^*_i\theta^*_j+{\theta^*_j}^2-\gamma^*_s(\theta^*_i+\theta^*_j)-\delta^*_s = 0.\]

Hence, $p_r(\theta_i,\theta_j)=p_r^*(\theta_i^*, \theta_j^*)=0$.
\end{proof}

\begin{thm}
For each positive integer $r$, 
\begin{eqnarray}
\sum_{i,j=0}^{i+j\leq 2r+1} a_{ij} A^i {A^*}^rA^j =0 \ , \qquad \sum_{i,j=0}^{i+j\leq 2r+1} a^*_{ij} {A^*}^i {A}^r{A^*}^j =0 \ , \label{eq:higherdolan} 
\end{eqnarray}
where the scalars $a_{ij} ,a^*_{ij}$ are defined by:
\begin{eqnarray}
p_r(x,y)=  \sum_{i,j=0}^{i+j\leq 2r+1} a_{ij} x^i y^j \quad  \mbox{and} \quad  p^*_r(x,y)=\sum_{i,j=0}^{i+j\leq 2r+1} a^*_{ij} x^i y^j \ . \label{polyr}
\end{eqnarray}
\end{thm}
\begin{proof}
Let $\Delta_r$ denote the expression of the left-hand side of the first equation of (\ref{eq:higherdolan}). We show $\Delta_r=0$. 

Let $\{E_i\}_{i=0}^d$ denote the standard ordering of the primitive idempotents of $A$ corresponding to $\{\theta_i\}_{i=0}^d$.

For $0\leq i,j \leq d$, one finds $E_i \Delta_r E_j = p_r(\theta_i,\theta_j) \ E_i {A^*}^r E_j $ with (\ref{polyr}). According to Lemma \ref{lem:polyroot} and Lemma  \ref{lem:triplep}, it follows $\Delta_r=0$. Similar arguments are used to show the second equation of (\ref{eq:higherdolan}).
\end{proof}

As a straightforward application, below we focus on the $r-th$ higher order tridiagonal relations associated with the $q-$Onsager algebra, a special case of the tridiagonal algebra.
\begin{defn}
Consider a TD pair $(A, A^*)$ of $q-$Racah type with eigenvalues such that $\alpha = \alpha^*=0$ (see (\ref{eq:const1})-(\ref{eq:const2})). Assume $r=1$. The corresponding tridiagonal relations (\ref{eq:dolan1}), (\ref{eq:dolan2}) are called the $q-$Dolan-Grady relations.
\end{defn}
\begin{example} For a TD pair $(A,A^*)$  of $q-$Racah type with eigenvalues such that $\alpha=\alpha^*=0$,  
 the parameter sequence is given by $\beta_1= q^{2} + q^{-2}\ , \ \gamma_1=\gamma^*_1=0\ , \   \delta_1= - bc(q^{2} - q^{-2})^2 \ ,\ \delta^*_1= - b^*c^*(q^{2} - q^{-2})^2$. Define $\delta_1=\rho_0$, $\delta_1^*=\rho_1$. The $q-$Dolan-Grady relations
are given by:
\begin{eqnarray}
\quad \sum_{j=0}^{3} (-1)^j  \left[ \begin{array}{c} 3 \\  j \end{array}\right]_q   A^{3-j} {A^*} A^{j} - \rho_0 (AA^*-A^*A) &=& 0\ ,\label{qDG1p} \\
\sum_{j=0}^{3} (-1)^j  \left[ \begin{array}{c} 3 \\ j \end{array}\right]_q   {A^*}^{3-j} {A} {A^*}^{j} - \rho_1 (A^*A-AA^*) &=& 0\ . \label{qDG2p}
\end{eqnarray}
\end{example}
\begin{rem} The relations (\ref{qDG1p}), (\ref{qDG2p}) are the defining relations of the $q-$Onsager algebra.
\end{rem}

\begin{defn}
Consider a TD pair $(A, A^*)$ of $q-$Racah type with eigenvalues such that $\alpha = \alpha^*=0$. For any positive integer $r$, the corresponding higher order tridiagonal relations (\ref{eq:higherdolan}) are called the higher order $q-$Dolan-Grady relations.
\end{defn}
\begin{thm}
For a TD pair $(A, A^*)$ of $q$-Racah type with eigenvalues such that $\alpha=\alpha^*=0$, the higher order $q$-Dolan-Grady relations are given by (\ref{qDGfinr})-(\ref{qDGfinr2}) with the identification $\rho_0 = \delta_1$ and $\rho_1 = \delta_1^*$, 
\end{thm}
\begin{proof}  For $\alpha=\alpha^*=0$ in (\ref{eq:const1}), (\ref{eq:const2}), from Lemma {\ref{lem:polypart}} one finds $\beta_s= q^{2s} + q^{-2s}$\ , \ $\gamma_s=\gamma^*_s=0$\ , \   $\delta_s= - bc(q^{2s} - q^{-2s})^2$ \ ,\ $\delta^*_s= - b^*c^*(q^{2s} - q^{-2s})^2$ . Then, the first polynomial generating function (\ref{defpoly}) reads: 
\begin{eqnarray}
p_r(x,y) = (x-y)\prod_{s=1}^{r} \left(x^2- \frac{[2s]_{q^2}}{[s]_{q^2}}xy +y^2  -   [s]^2_{q^2} \rho_0\right)\ ,\label{polyqons}
\end{eqnarray}
where the notation $\beta_s=[2s]_{q^2}/[s]_{q^2}$ and $\delta_s/\rho_0=[s]^2_{q^2}$ has been introduced. Expanding the polynomial in the variables $x,y$ as (\ref{polyr}), one shows that the coefficients $a_{ij}$ in (\ref{polyr}) take the form:
\begin{eqnarray}
 a_{2r+1-2p - j \ j}= (-1)^{j+p} \rho^{p}_0 c_j^{[r,p]}
\end{eqnarray}
where $c_j^{[r,p]}$  solely depend on $q$, and are vanishing otherwise. By induction, one finds that they are given by (\ref{cfinr}).
Replacing $\rho_0\rightarrow \rho_1$, the second relation (\ref{qDGfinr2})  follows.
\end{proof}

For $r=2,3$, the higher order $q-$Dolan-Grady relations (\ref{qDGfinr}), (\ref{qDGfinr2}) can be constructed in a straightforward manner: 

\begin{example}\label{exr2} The first example of higher order $q-$Dolan-Grady relations is given by (\ref{qDGfinr}), (\ref{qDGfinr2})  for $r=2$ with:
\begin{eqnarray}
c^{[2,0]}_0&=&1,~ c^{[2,0]}_1= 1+ [2]_{q^2} + \frac{[4]_{q^2}}{[2]_{q^2}} \equiv \left[ \begin{array}{c} 5\\ 1 \end{array}\right]_q,~ c^{[2,0]}_2= 2+ [2]_{q^2}+ [4]_{q^2} + \frac{[4]_{q^2}}{[2]_{q^2}} \equiv  \left[ \begin{array}{c} 5\\ 2 \end{array}\right]_q,\nonumber \\
c^{[2,1]}_0&=& 1 + [2]^2_{q^2} \equiv q^4+q^{-4} +3 \ , \quad  c^{[2,1]}_1=  1 + [2]^2_{q^2} +  \frac{[4]_{q^2}}{[2]_{q^2}} +  [2]^3_{q^2} \equiv [5]_q[3]_q\ ,\nonumber\\
c^{[2,2]}_0&=& [2]^2_{q^2} \equiv (q^2+q^{-2})^2 \ . \nonumber 
\end{eqnarray}
\end{example}

\begin{example}\label{exr3} The second example of higher order $q-$Dolan-Grady relations is given by (\ref{qDGfinr}), (\ref{qDGfinr2})  for $r=3$ with:
\begin{eqnarray}
c_{j}^{[3,0]} &=&  \left[ \begin{array}{c} 7 \\ j 
\end{array}\right]_q , \quad j=0,...,7\ ,\nonumber\\     
c^{[3,1]}_0&=& 1 + [2]^2_{q^2} + [3]^2_{q^2} \ ,\nonumber\\
c^{[3,1]}_1&=& 1 + [2]^2_{q^2} + [3]^2_{q^2} + (1+ [3]^2_{q^2})\frac{[4]_{q^2}}{[2]_{q^2}} +  ([2]^2_{q^2} + [3]^2_{q^2})[2]_{q^2}  + (1+ [2]^2_{q^2})\frac{[6]_{q^2}}{[3]_{q^2}} \ ,\nonumber\\
c^{[3,1]}_2&=& 2(1 + [2]^2_{q^2} + [3]^2_{q^2}) + (1+ [3]^2_{q^2})\frac{[4]_{q^2}}{[2]_{q^2}} +  ([2]^2_{q^2} + [3]^2_{q^2})[2]_{q^2}  + (1+ [2]^2_{q^2})\frac{[6]_{q^2}}{[3]_{q^2}} \ \nonumber\\
&&+\ \  \frac{[4]_{q^2}[6]_{q^2}}{[2]_{q^2}[3]_{q^2}} + \frac{[2]^3_{q^2}[6]_{q^2}}{[3]_{q^2}}  + [3]^2_{q^2}[4]_{q^2} \ ,\nonumber\\
c^{[3,2]}_0&=& [2]^2_{q^2} + [3]^2_{q^2} +  [2]^2_{q^2}[3]^2_{q^2}\ ,\nonumber\\
c^{[3,2]}_1&=& [2]^2_{q^2} + [3]^2_{q^2} +  [2]^2_{q^2}[3]^2_{q^2} + [2]^2_{q^2}\frac{[6]_{q^2}}{[3]_{q^2}} + [3]^2_{q^2}\frac{[4]_{q^2}}{[2]_{q^2}} + [2]^3_{q^2}[3]^2_{q^2} ,\nonumber\\
c^{[3,3]}_0&=& [2]^2_{q^2}[3]^2_{q^2} \ .\nonumber
\end{eqnarray}
\end{example}
To end up this Section, let us consider the family of relations satisfied by a TD pair of $q-$Racah type such that\footnote{For instance, choose $b,b^*=0$ and/or $c,c^*=0$ in (\ref{eq:const1}), (\ref{eq:const2}).} $\rho_0=\rho_1=0$.\\

We remind a powerful theorem called the $q$-Binomial Theorem in studying the binomial coefficients.
\begin{thm}
For all $n \ge 1$, $\prod\limits_{j=1}^n{(1+xq^j)}= \sum\limits_{k=0}^n{q^{k(k+1)/2}{\left[{\begin{array}{c}
n\\k\end{array}} \right]_q^*}x^k},$\\
where ${\left[ {\begin{array}{l}
n\\k
\end{array}} \right]}^*_q= \frac{(1-q^n)(1-q^{n-1})\dots(1-q^{n-k+1})}{(1-q^k)(1-q^{k-1})\dots(1-q)}$.
\end{thm}
\begin{proof}
By induction
\end{proof}
\begin{thm}
For $\rho_0 = \rho_1 =0$, the higher order $q-$Dolan-Grady relations satisfied by the corresponding TD pair simplify to the well-known Lusztig's higher order  $q-$Serre relations \cite{Luszt}:
\begin{eqnarray}
\sum_{j=0}^{2r+1}  (-1)^{j}  \left[ \begin{array}{c} 2r+1\\ j \end{array}\right]_q  \,  A^{2 r+1-j} {A^*}^{r} A^{j}&=&0\  \label{qserrefinr} \ ,\\
\sum_{j=0}^{2r+1} (-1)^{j}  \left[ \begin{array}{c} 2r+1\\ j \end{array}\right]_q  \,  {A^*}^{2 r+1-j} {A}^{r} {A^*}^{j}&=&0\  .
\end{eqnarray}

\end{thm}
\begin{proof}
Since $\rho_0 = 0$, (\ref{polyqons}) can be written by:
\begin{eqnarray*}
p_r(x,y)&=&x^{r+1}(1-\frac{y}{x})\prod\limits_{s=1}^r{(1-(q^{2s}+q^{-2s})\frac{y}{x}+(\frac{y}{x})^2)}\\
&=&x^{2r+1}(1-\frac{y}{x})\prod\limits_{s=1}^r{(1-q^{2s}\frac{y}{x})(1-q^{-2s}\frac{y}{x})}\\
&=&x^{2r+1}\prod\limits_{s=-r}^r{(1-q^{2s}\frac{y}{x})}.
\end{eqnarray*}
Put $t=s+r+1$, $p_r(x,y)=x^{2r+1}\prod\limits_{t=1}^{2r+1}{(1+q^{2t}(-q^{-2(r+1)}\frac{y}{x}))}.$\\
Applying the q-Binomial Theorem yields 
\[ p_r(x,y)=x^{2r+1}\sum\limits_{j=0}^{2r+1}{(-1)^jq^{j(j+1)}{\left[ \begin{array}{c}
2r+1\\j
\end{array} \right] }^*_{q^2}(q^{-2(r+1)}\frac{y}{x})^j}. \]
Using the relation between $\left[ \begin{array}{c}n\\j\end{array}\right]_q$ and $\left[ \begin{array}{c}n\\j\end{array}\right]_{q^2}^*$ for all $n \ge j$, $\left[ \begin{array}{c}n\\j\end{array}\right]_q = q^{j(j-n)}\left[ \begin{array}{c}n\\j\end{array}\right]_{q^2}^*$.\\
We have \[ p_r(x,y)=\sum\limits_{j=0}^{2r+1}{(-1)^j{\left[ \begin{array}{c}
2r+1\\j
\end{array} \right] }_qx^{2r+1-j}y^j}  .\]
\end{proof}

\subsection{Recursion for generating the coefficients of the higher order q-Dolan-Grady relations}

In the previous Section, it was shown that every TD pair of $q-$Racah type such that $\alpha=\alpha^*=0$ satisfies the $r-th$ higher order $q-$Dolan-Grady relations (\ref{qDGfinr}), (\ref{qDGfinr2}) with (\ref{cfinr}). For the special case $r=1$, these relations coincide with the defining relations of the $q-$Onsager algebra (\ref{qO1}), (\ref{qO2}). This strongly suggests that the conjecture holds in general. Below, the conjecture for $r = 2, 3$ will be proved. Then, using an inductive argument we will study the general structure and derive recursion formulae - independently of the results of the previous part - for the coefficients $c^{[r,p]}_j$. 

Let $A$, $A^*$ be the fundamental generators of the $q-$Onsager algebra (\ref{qO1}), (\ref{qO2}).
By analogy with Lusztig's higher order $q-$Serre relations, we are interested in more complicated linear combinations of monomials  of the type $A^n{A^*}^rA^m$, $n+m=2r+1,2r-1,...,1$,  that are vanishing. The defining relations (\ref{qO1}), (\ref{qO2}) correspond to the case $r=1$ of (\ref{qDGfinr}), (\ref{qDGfinr2}). Below, successively we derive the relations (\ref{qDGfinr}), (\ref{qDGfinr2}) for $r=2,3$ and study the generic case by induction.
\subsubsection{Proof of the relations for $r=2$} Consider the simplest example beyond (\ref{qO1}): we are looking for a linear relation between monomials of the type $A^n{A^*}^2A^m$, $n+m=5,3,1$. According to the defining relations (\ref{qO1}), note that the monomial $A^3A^*$ can be written as:
\begin{eqnarray}
A^3A^* = \alpha A^2A^*A - \alpha AA^*A^2 + A^*A^3 + \rho_0 (AA^* - A^*A)\ \quad \mbox{with}\quad \alpha=[3]_q \ .\label{mon}
\end{eqnarray}
Multiplying from the left by $A$ or $A^2$, the corresponding expressions can be ordered as follows: each time a monomial of the form $A^n A^*A^m$ with $n\geq 3$ arises, it is reduced using (\ref{mon}). It follows:
\begin{eqnarray}
A^4A^* &=& (\alpha^2-\alpha) A^2A^*A^2 + (1- \alpha^2) AA^*A^3 + \alpha A^*A^4 \nonumber
\\& & +~ \rho_0 \left(A^2A^* - \alpha A^*A^2 + (\alpha-1)AA^*A\right)\ ,\nonumber\\
A^5A^* &=& (\alpha^3-2\alpha^2+1) A^2A^*A^3 + \alpha(-\alpha^2+\alpha+1) AA^*A^4 + \alpha(\alpha-1) A^*A^5 \nonumber\\
&&+  \ \rho_0 \left((2\alpha-1)A^2A^*A + \alpha(\alpha-3)AA^*A^2 - (\alpha^2-\alpha-1) A^*A^3\right) \nonumber \\
&&+ \ \rho_0^2(AA^*-A^*A)\ .\nonumber
\label{mon2}
\end{eqnarray}
For our purpose, four different types of monomials may be now considered: $A^5{A^*}^2$, $A^4{A^*}^2A$, $A^3{A^*}^2A^2$ and $A^3{A^*}^2$. Following the ordering prescription, each of these monomials can be reduced as a combination of monomials of the type ($n,m,p,s,t\geq 0$):
\begin{eqnarray}
&&\quad A^n{A^*}^2A^m \ \qquad \ \mbox{with} \quad  \ n\leq 2\ , \ n+m=5,3,1\ ,\label{mongen1}\\
 &&\quad A^p A^* A^s A^* A^t \quad \  \mbox{with} \quad  \ p\leq 2\ , \ 1\le s\le 2\ ,\ p+s+t=5,3,1 \ .\nonumber
\end{eqnarray}
For instance, the monomial  $A^5{A^*}^2$ is reduced to:
\begin{eqnarray}
A^5{A^*}^2 &=& (\alpha^3-2\alpha^2+1) \left(\alpha A^2A^*A^2A^*A - \alpha A^2A^*AA^*A^2 + A^2{A^*}^2 A^3\right) \nonumber\\
&&-\ (\alpha^3-\alpha^2-\alpha) \left((\alpha^2-\alpha)AA^*A^2A^*A^2 +(1-\alpha^2) AA^*AA^*A^3  + \alpha A{A^*}^2A^4 \right)\nonumber\\
&&+ \ (\alpha^2-\alpha)\left( (\alpha^3-2\alpha^2+1)  {A^*}A^2A^*A^3  - (\alpha^3-\alpha^2-1)  {A^*}AA^*A^4 + \alpha(\alpha-1) {A^*}^2A^5  \right)\nonumber \\
&& + \ \rho_0 (\alpha^3-2\alpha^2+2\alpha)\left(A^2A^*AA^* - AA^*A^2A^* + A^*A^2A^*A\right)\nonumber\\
&& + \ \rho_0 (-\alpha^3+2\alpha^2-1)\left(A^2{A^*}^2A + \alpha AA^*AA^*A\right)\nonumber\\
&& + \ \rho_0 (\alpha^4-\alpha^3-\alpha^2)\left(A{A^*}^2A^2\right)\nonumber\\
&& + \ \rho_0 (\alpha^4-3\alpha^3+2\alpha^2-\alpha)\left(A^*{A}A^*A^2\right)\nonumber\\
&& + \ \rho_0 (-\alpha^4+2\alpha^3-\alpha^2+1)\left({A^*}^2A^3\right) \nonumber\\
&& + \ \rho_0^2 \big[A,{A^*}^2\big] \ . \nonumber
\end{eqnarray}
The two other monomials $A^4{A^*}^2A$, $A^3{A^*}^2A^2$ are also ordered using (\ref{mon}). One obtains:
\begin{eqnarray}
A^4{A^*}^2A &=& (\alpha^2-\alpha) \left( A^2A^*A^2A^*A + \alpha A^*A^2A^*A^3\right) + \ \alpha^2{A^*}^2A^5   \nonumber\\
&&+\ (\alpha^2-1) \left(\alpha AA^*AA^*A^3- \alpha AA^*A^2A^*A^2 - \alpha A^*AA^*A^4 - A{A^*}^2A^4\right)\nonumber\\
&& + \ \rho_0 \left(  A^2{A^*}^2 A  - (1-\alpha^2) A{A^*}^2A^2 - \alpha^2 {A^*}^2 A^3\right)
\nonumber \\
&&+ \ \rho_0 \left(\alpha^2-\alpha \right) \left(A^*AA^*A^2-AA^*AA^*A \right)\ ,\nonumber \\
A^3{A^*}^2A^2 &=& \alpha \left(A^2A^*AA^*A^2 - AA^*A^2A^*A^2 + A^*A^2A^*A^3 - A^*AA^*A^4  \right) +{A^*}^2A^5\nonumber\\
&& + \rho_0(A {A^*}^2A^2 - {A^*}^2A^3) \nonumber \ .
\end{eqnarray}
The ordered expression for the fourth monomial $A^3{A^*}^2$ directly follows from (\ref{mon}). 
Having the explicit ordered expressions of $A^5{A^*}^2$, $A^4{A^*}^2A$, $A^3{A^*}^2A^2$ and $A^3{A^*}^2$ in terms of monomials of the type (\ref{mongen1}), let us consider the combination:
\begin{eqnarray}
f_2(A,A^*)= c_0^{[2,0]} A^5{A^*}^2 - c^{[2,0]}_1 A^4{A^*}^2A +   c^{[2,0]}_2 A^3{A^*}^2A^2 -  \rho_0 c^{[2,1]}_0 A^3{A^*}^2 \ 
\end{eqnarray}
with unknown coefficients $c^{[2,0]}_j$, $j=1,2$, $c^{[2,1]}_0$, and normalization $c_0^{[2,0]}=1$.
After simplifications, the combination takes the ordered form:
\begin{eqnarray}
f_2(A,A^*)=   c^{[2,0]}_3 A^2{A^*}^2A^3 - c^{[2,0]}_4 A{A^*}^2A^4 + c^{[2,0]}_5 {A^*}^2A^5  \ + \ g_2(A,A^*)  \ 
\end{eqnarray}
where 
\begin{eqnarray}
 && c^{[2,0]}_3 = \alpha^3-2\alpha^2+1 \ , \quad c^{[2,0]}_4 = \alpha^2(\alpha^2-\alpha-1)+ c^{[2,0]}_1 (1-\alpha^2) \ , \\ && c^{[2,0]}_5 = (\alpha^2-\alpha)^2 - \alpha^2c^{[2,0]}_1 + c^{[2,0]}_2 \ .\nonumber
\end{eqnarray}
Inspired by the structure of Lusztig's higher order $q-$Serre relations, consider the conditions under which the combination $g_2(A,A^*)$ never contains monomials of the form $A^p A^* A^s A^* A^t $ ($p\leq 2,\ 1\le s\le 2$). At the lowest order in $\rho_0$, given a particular monomial  the condition under which its coefficient in $g_2(A,A^*)$ is vanishing is given by:
\begin{eqnarray*}
&& A^2A^*A^2A^*A: \ \alpha^3-2\alpha^2+1 - c^{[2,0]}_1(\alpha-1)=0 \ ,\nonumber\\
&& A^2A^*AA^*A^2: \ -\alpha^3+2\alpha^2-1 + c^{[2,0]}_2 =0\ , \nonumber \\
&& AA^*A^2A^*A^2:  \ (\alpha-1)(-\alpha^3+\alpha^2+\alpha)  - c^{[2,0]}_1(1-\alpha^2) - c^{[2,0]}_2 =0 \ ,\nonumber\\
&& AA^*AA^*A^3:  \ (1-\alpha^2)(-\alpha^3+\alpha^2+\alpha)+ c^{[2,0]}_1 \alpha(1-\alpha^2) =0\ , \nonumber\\
&& A^*A^2A^*A^3:  \  (\alpha-1)(\alpha^3-2\alpha^2+1)  - c^{[2,0]}_1(\alpha^2-\alpha) + c^{[2,0]}_2 =0 \ ,\nonumber\\
&& A^*AA^*A^4:  \ (\alpha-1)(-\alpha^3+\alpha^2+\alpha)- c^{[2,0]}_1 (1-\alpha^2) - c^{[2,0]}_2 =0\ . \nonumber
\end{eqnarray*}
Recall that $\alpha=[3]_q$. The solution  $\{c^{[2,0]}_j$, $j=1,2$\} to this system of equations exists, and it is unique. In terms of $q-$binomials, it reads:
\begin{eqnarray}
c^{[2,0]}_1= \left[ \begin{array}{c} 5\\ 1 \end{array}\right]_q  \ ,\qquad c^{[2,0]}_2= \left[ \begin{array}{c} 5\\ 2 \end{array}\right]_q \ .
\end{eqnarray}
At the next order $\rho_0$, the conditions such that monomials of the type $A^p A^* A^s A^* A^t $ with $p+s+t\leq 3$ and $1\le s \le 2$ are vanishing yield to:
\begin{eqnarray}
c^{[2,1]}_0=q^4+q^{-4}+3\ .\nonumber
\end{eqnarray}
All the other coefficients of the monomials $A^n{A^*}^2A^m \quad \mbox{for} \quad  n+m=5,3,1$ are explicitly determined in terms of $c^{[2,0]}_j$ ($j=0,1,2$), $c^{[2,1]}_0$, $\rho_0$ and $\rho_0^2$. Based on these results, we conclude that the $q-$Dolan-Grady relation (\ref{qO1}) implies the existence of a unique linear relation between monomials of the type $A^n{A^*}^2A^m \quad \mbox{with} \quad  n+m=5,3,1$. This relation can be seen as a $\rho_0-$deformed analogue of the simplest higher order $q-$Serre relation.  Explicitly, one finds:
\begin{eqnarray}
\sum\limits_{j = 0}^5 {{{\left( { - 1} \right)}^j}\left[ {\begin{array}{*{20}{c}}
   5  \\
   j  \\
\end{array}} \right]} {A^{5 - j}}{A^*}^2{A^j} &=& 
\rho_0 ( (q^4 + q^{- 4} + 3)(A^3 {A^*}^2-  {A^{*2}}{A^3})\nonumber \\ &&~~~~ -[ 5 ]_q\left[ 3 \right]_q(A^2{A^{*2}}A - A{A^{*2}}A^2) ) \label{qDG11} \nonumber  \\ 
& &- \ {\rho_0 ^2}{\left( {{q^2} + {q^{ - 2}}} \right)^2}\left(  A{A^{*2} - A^{*2}}A  \right)\ .
\end{eqnarray}
Using the automorphism $A\leftrightarrow A^*$ and $\rho_0\leftrightarrow \rho_1$  which exchanges (\ref{qO1}) and (\ref{qO2}), the second relation generalizing (\ref{qO2}) is obtained.  The coefficients coincide with the ones given in Example \ref{exr2}.
For the special undeformed case $\rho_0=\rho_1=0$, note that both relations reduce to the simplest examples of higher order $q-$Serre relations.
 
\subsubsection{Proof of the relations for $r=3$} Following a similar analysis, the next  example of higher order $q-$Dolan-Grady relations can be also derived. To this end, one is looking for a linear relation between monomials of the type $A^n{A^*}^3A^m$, $n+m=7,5,3,1$. Assume the $q-$Dolan-Grady relation (\ref{qO1}) and its simplest consequence (\ref{qDG11}). Write the four monomials:
\begin{eqnarray}
&&A^7 {A^*}^3 = (A^7{A^*}^2)A^*\ , \quad  A^6{A^*}^3 A = (A^6{A^*}^2)A^*A \ ,\nonumber\\  &&A^5{A^*}^3A^2 = (A^5 {A^*}^2)A^*A^2\ ,\quad A^5{A^*}^3 = (A^5 {A^*}^2)A^*     \ .\nonumber
\end{eqnarray}
Using (\ref{qDG11}) and then (\ref{mon}), they can be expressed solely in terms of   monomials of the type:
\begin{eqnarray}
&&\qquad A^n{A^*}^3A^m \quad \qquad \mbox{with} \quad n\leq 4\ , \ n+m=7,5,3,1\ , \label{mongen3}\\
&&\qquad A^p {A^*}^2 A^s A^* A^t \quad  \mbox{with} \quad  p\leq 4\ ,\ 1\le s\le 2\ ,\ p+s+t=7,5,3,1 \ .\nonumber
\end{eqnarray}
Then, introduce the combination 
\begin{eqnarray}
f_3(A,A^*)= c_0^{[3,0]} A^7{A^*}^3 - c^{[3,0]}_1 A^6{A^*}^3A +  c^{[3,0]}_2 A^5{A^*}^3A^2 -  \rho_0 c^{[3,1]}_0 A^5{A^*}^3\   \label{comb3}
\end{eqnarray}
with unknown coefficients $c^{[3,0]}_j$, $(j=1,2)$, $c^{[3,1]}_0$  and normalization $c_0^{[3,0]}=1$. By straightforward calculations using the ordered expressions of $A^7{A^*}^3$, $A^6{A^*}^3A$, $A^5{A^*}^3A^2$ and $A^5{A^*}^3$, $f_3(A,A^*)$ is reduced to a combination of monomials of the type (\ref{mongen3}). Note that the coefficients of the monomials $A^n{A^*}^3A^m$ for $n+m=5,3,1$ are of order $\rho_0,\rho_0^2, \rho_0^3$, respectively. Identifying the conditions under which the coefficient of any monomial of the form
\begin{eqnarray}
A^p {A^*}^2 A^s A^* A^t \quad \mbox{with} \quad p\leq 4\ ,\ 1\le s\le 2\ , \quad  \ p+s+t=7,5,3,1\ , \nonumber 
\end{eqnarray}
is vanishing, one obtains a system of equations for the coefficients, which solution is unique. Simplifying (\ref{comb3}) according to the explicit solutions $c^{[3,0]}_j$, $j=1,2$ and $c^{[3,1]}_0$, one ends up with the next example of higher order $q-$Dolan-Grady relations. Using the automorphism $A\leftrightarrow A^*$ and $\rho_0\leftrightarrow \rho_1$,
the second relation follows. One finds:  
\begin{eqnarray}
&&\sum_{p=0}^{3}\, \rho_0^{p} \,\sum_{j=0}^{7-2p} (-1)^{j+p}  \,c_{j}^{[3,p]}\,  A^{7-2p-j} {A^*}^3 A^{j}=0 \ , \ \label{qDG13}\\
&&\sum_{p=0}^{3}\, \rho_1^{p}\, \sum_{j=0}^{7-2p} (-1)^{j+p} \, c_{j}^{[3,p]} \, {A^*}^{7-2p-j} { A}^3 {A^*}^{j}=0\ \nonumber
\end{eqnarray}
where $c_{j}^{[3,p]}=c_{7-2p-j}^{[3,p]}$\ , \ $c_{j}^{[3,0]} =  \left[ \begin{array}{c} 7 \\ j 
\end{array}\right]_q$ and     
\begin{eqnarray}
&& c^{[3,1]}_0= \left( {{q^8} + 3{q^4} + 6 + 3{q^{ - 4}} + {q^{ - 8}}} \right),  c^{[3,1]}_1={\left[ 7 \right]_q}\left( {{q^6} + {q^4} + {q^2} + 4 + {q^{ - 2}} + {q^{ - 4}} + {q^{ - 6}}} \right),\nonumber\\
&&  c^{[3,1]}_2=  \left[ 7 \right]_q (q^2-1+q^{-2})(q^2+q^{-2})(q^4+2q^2+4 + 2q^{-2}+q^{-4})\ ,\nonumber\\
&& c^{[3,2]}_0= ( q^{6} + 2q^2 + 3q^{-2}+ q^{-6}) ( q^{6} + 3q^2 + 2q^{-2}+ q^{-6}) \ ,\nonumber\\
&&  c^{[3,2]}_1=\left[ 7 \right]_q \left( {q^{8}} + {q^{6}} + 4{q^{4}} + {q^2} + 7 + {q^{-2}} + 4{q^{-4}} + q^{-6} + {q^{ - 8}} \right)\ ,\nonumber\\
&& c^{[3,3]}_0=  [2]^2_{q^2}[3]^2_{q^2}\ .\nonumber
\end{eqnarray}
It is straightforward to compare the coefficients above with the ones obtained from the expansion of the polynomial generating function $p_3(x,y)$. Although the coefficients above look different, they coincide exactly with the ones reported in Example \ref{exr3}. 
\subsubsection{Relations for $r$ generic}
To Look for a linear relation between monomials of the type $A^n{A^*}^rA^m$, $n+m=2r+1,2r-1,...,1$, 
for $r\geq 1$, relations of the form 
\begin{eqnarray}
&&\sum_{p=0}^{r}\, \rho_0^{p} \,\sum_{j=0}^{2(r-p)+1} (-1)^{j+p}  \,c_{j}^{[r,p]}\,  A^{2(r-p)+1-j} {A^*}^r A^{j}=0 \ , \ \label{qDGr}\\
&&\sum_{p=0}^{r}\, \rho_1^{p}\, \sum_{j=0}^{2(r-p)+1} (-1)^{j+p} \, c_{j}^{[r,p]} \, {A^*}^{2(r-p)+1-j} { A}^r {A^*}^{j}=0\ \label{qDGr2}
\end{eqnarray}
are expected, where the elements $A,A^*$ satisfy the $q-$Dolan-Grady relations (\ref{qO1}) and (\ref{qO2}). Our aim is now to study these relations in details and obtain recursive formulae for the coefficients $c_{j}^{[r,p]}$. \vspace{1mm} 

In order to study the higher order $q-$Dolan-Grady relations (\ref{qDGr}) for generic values of $r$, we proceed by induction. First, assume  the basic relation (\ref{qO1}) holds and implies all relations (\ref{qDGr}) up to $r$ which explicit coefficients $c_{j}^{[r,p]}$ in terms of $q$ are assumed to be known. It is the case for $r=2,3$ as shown above. Our aim is to construct the higher order relation associated with $r+1$ and express the coefficients $c_{j}^{[r+1,p]}$, ($j=0,1,...2r+3-2p$, $p=0,...,r+1$) in terms of $c_{j'}^{[r,p']}$ ($j'=0,1,...2r+1-2p'$, $p'=0,...,r$) . Following the steps described for $r=2,3$, from the relation (\ref{qDGr}) we first deduce:

\begin{eqnarray}
A^{2r+1}{A^*}^r &=&-\sum\limits_{j=1}^{2r+1}{(-1)^jc^{[r,0]}_jA^{2r+1-j}{A^*}^rA^j}\nonumber\\ &&-\sum\limits_{p=1}^r{\rho_0^p\sum\limits_{j=0}^{2(r-p)+1}{(-1)^{j+p}c^{[r,p]}_jA^{2(r-p)+1-j}{A^*}^rA^j}} \label{eqq1},\\
A^{2r+2}{A^*}^r &=&-\sum\limits_{j=2}^{2r+2}{(-1)^jM^{(r,0)}_jA^{2r+2-j}{A^*}^rA^j}\nonumber\\&& -\sum\limits_{p=1}^r{\rho_0^p\sum\limits_{j=0}^{2(r-p)+2}{(-1)^{j+p}M^{(r,p)}_jA^{2(r-p)+2-j}{A^*}^rA^j}}\label{eq2},\\ 
A^{2r+3}{A^*}^r &=&-\sum\limits_{j=3}^{2r+3}{(-1)^jN^{(r,0)}_jA^{2r+3-j}{A^*}^rA^j}\nonumber\\&&-\sum\limits_{p=1}^r{\rho_0^p\sum\limits_{j=0}^{2(r-p)+3}{(-1)^{j+p}N^{(r,p)}_jA^{2(r-p)+3-j}{A^*}^rA^j}} \label{eq3}
\end{eqnarray}
where the coefficients $M^{(r,p)}_j$, $N^{(r,p)}_j$ are determined recursively in terms of $c_{j}^{[r,p]}$ (see Appendix A). Now, write the four monomials:
\begin{eqnarray}
&& A^{2r+3} {A^*}^{r+1} = (A^{2r+3}{A^*}^r)A^*\ , \quad A^{2r+2} {A^*}^{r+1}A = (A^{2r+2}{A^*}^r)A^*A \ ,\nonumber \\ \quad &&  A^{2r+1} {A^*}^{r+1}A^2 = (A^{2r+1}{A^*}^{r})A^*A^2  \ ,\quad A^{2r+1}{A^*}^{r+1} = (A^{2r+1}{A^*}^r)A^* \ . \nonumber
\end{eqnarray}
Using (\ref{eqq1})-(\ref{eq3}), they can be expressed solely in terms of:
\begin{eqnarray}
&& A^n{A^*}^{r+1}A^m \ \mbox{with} \quad  n\leq 2r \ ,\ n+m=2r+3, 2r+1,...,1\ ,~~~~~~~\label{mongenr}\\
&& A^p {A^*}^r A^s A^* A^t  \ \mbox{with} \quad  p\leq 2r \ ,\ 1\le s\le 2\ , \ p+s+t=2r+3, 2r+1,...,1.~~~~~~\label{mongenUNr}~~~~~
\end{eqnarray}
It is however clear from (\ref{eq1})-(\ref{eq3}) that each monomial  $A^{2r+3} {A^*}^{r+1}$, $A^{2r+2} {A^*}^{r+1}A$, $A^{2r+1} {A^*}^{r+1}A^2$ and $A^{2r+1}{A^*}^{r+1}$   can be further reduced using (\ref{qO1}). For instance, 
\begin{eqnarray}
 (A^{2r+3}{A^*}^r)  A^*&=&-\sum\limits_{j=3}^{2r+3}(-1)^jN^{(r,0)}_jA^{2r+3-j}{A^*}^r\!\!\!\!\!\!\!\underbrace{A^j A^*}_{\mbox{reducible}}\nonumber\\&&- \sum\limits_{p=1}^r \rho_0^p\sum\limits_{j=0}^{2(r-p)+3}(-1)^{j+p}N^{(r,p)}_jA^{2(r-p)+3-j}{A^*}^r \!\!\!\!\!\!\!\!\!\!\!\!\!\!\!\!\!\underbrace{A^j A^*}_{\mbox{reducible if $j\geq 3$}} \!\!\!\!\!\!\!\!\!\!\!\!\!\!\!\!\ .\label{reducible}\nonumber
\end{eqnarray}
According to (\ref{qO1}), observe that the monomials $A^jA^*$ (for $j$ even or odd) can be written as:
\begin{eqnarray}
A^{2n+2}A^*&=&\sum\limits_{k=0}^n{\sum\limits_{i=0}^2{\rho_0^{n-k}\eta^{(2n+2)}_{k,i}A^{2-i}A^*A^{2k+i}}} \ ,\label{eqmu1} \\
A^{2n+3}A^*&=&\sum\limits_{k=1}^{n+1}{\sum\limits_{i=0}^2{\rho_0^{n+1-k}\eta^{(2n+3)}_{k,i}A^{2-i}A^*A^{2k-1+i}}}+\rho_0^{n+1}(AA^*-A^*A) \ ,\label{eqmu2}
\end{eqnarray}
where the coefficients $\eta^{(2n+2)}_{k,i},\eta^{(2n+3)}_{k,i}$ are determined recursively in terms of $q$ (see Appendix A). It follows:
\begin{eqnarray}
&& \!\!\!\! \! \! \! \! \! \!\!\!\!\! A^{2r+3}{A^*}^{r+1}=\sum\limits_{i=1}^{r+1}{N^{(r,0)}_{2i+1}A^{2(r-i)+2}{A^*}^r(\sum\limits_{k=1}^i{\sum\limits_{j=0}^2{\rho_0^{i-k}\eta^{(2i+1)}_{k,j}}A^{2-j}A^*A^{2k-1+j}}+\rho_0^i(AA^*-A^*A))} \label{m1}\\ & &- \sum\limits_{i=1}^r{N^{(r,0)}_{2i+2}A^{2(r-i)+1}{A^*}^r(\sum\limits_{k=0}^i{\sum\limits_{j=0}^2{\rho_0^{i-k}\eta^{(2i+2)}_{k,j}A^{2-j}A^*A^{2k+j}}})} \nonumber\\ & &- \sum\limits_{p=1}^r{{(-\rho_0)}^p(N^{(r,p)}_0A^{2(r-p)+3}{A^*}^{r+1}-N^{(r,p)}_1A^{2(r-p)+2}{A^*}^rAA^*+N^{(r,p)}_2A^{2(r-p)+1}{A^*}^rA^2A^*)} \nonumber\\ & &+ \sum\limits_{p=1}^r{{(-\rho_0)}^p\sum\limits_{i=1}^{r-p+1}{N^{(r,p)}_{2i+1}A^{2(r-p-i)+2}{A^*}^r(\sum\limits_{k=1}^i{\sum\limits_{j=0}^2{\rho^{i-k}_0\eta^{(2i+1)}_{k,j}A^{2-j}A^*A^{2k-1+j}}}+\rho_0^i(AA^*-A^*A))}} \nonumber\\& &- \sum\limits_{p=1}^{r-1}{{(-\rho_0)}^p\sum\limits_{i=1}^{r-p}{N^{(r,p)}_{2i+2}A^{2(r-p-i)+1}{A^*}^r(\sum\limits_{k=0}^i{\sum\limits_{j=0}^2{\rho_0^{i-k}\eta^{(2i+2)}_{k,j}A^{2-j}A^*A^{2k+j}}})}}
\nonumber\\& & -(-\rho)^{r+1}(N_0^{(r,r+1)}A{A^*}^{r+1}-N_1^{(r,r+1)}{A^*}^rAA^*)\ .\nonumber
\end{eqnarray}
The three other monomials $A^{2r+2} {A^*}^{r+1}A$, $A^{2r+1} {A^*}^{r+1}A^2$ and $A^{2r+1}{A^*}^{r+1}$  are also further reduced. For simplicity, corresponding expressions are reported in Appendix B. Now, introduce the combination 
\begin{eqnarray}
f_{r+1}(A,A^*)&=& c^{[r+1,0]}_0 A^{2r+3}{A^*}^{r+1} - c^{[r+1,0]}_1 A^{2r+2}{A^*}^{r+1}A \nonumber \\
& & +~  c^{[r+1,0]}_2 A^{2r+1}{A^*}^{r+1}A^2 -   \rho_0c^{[r+1,1]}_0 A^{2r+1}{A^*}^{r+1}\    \label{combr+1}
\end{eqnarray}
with unknown coefficients $c^{[r+1,0]}_j$, $(j=1,2)$, $c^{[r+1,1]}_0$  and normalization $c_0^{[r+1,0]}=1$. Combining all reduced expressions for $A^{2r+3}{A^*}^{r+1}$, $A^{2r+2}{A^*}^{r+1}A$, $A^{2r+1}{A^*}^{r+1}A^2$  and $A^{2r+1}{A^*}^{r+1}$ reported in Appendix B, one observes that  $f_{r+1}(A,A^*)$ generates monomials either of the type (\ref{mongenr}) or  (\ref{mongenUNr}). First, consider monomials of the type (\ref{mongenUNr}) which occur at the lowest order in $\rho_0$, namely $A^{2r}{A^*}^rA^2A^*A$ and $A^{2r}{A^*}^rAA^*A^2$. The conditions under which their coefficients are vanishing read:
\begin{eqnarray}
A^{2r}{A^*}^rA^2A^*A :&&  \qquad  N^{(r,0)}_3\eta^{(3)}_{1,0}+c^{[r+1,0]}_1M_2^{(r,0)}=0\ , \nonumber\\
A^{2r}{A^*}^rAA^*A^2 :&& \qquad N_3^{(r,0)}\eta_{1,1}^{(3)}+c^{[r+1,0]}_2c^{[r,0]}_1=0\ .\nonumber
\end{eqnarray}
Using the explicit expressions for $N^{(r,0)}_3$, $\eta^{(3)}_{1,0}$ and $\eta^{(3)}_{1,1}$ given in Appendices A,B, it is easy to solve these these equations. It yields to:
\begin{eqnarray}
c^{[r+1,0]}_1= \left[ \begin{array}{c} 2r+3\\ 1 \end{array}\right]_q  \ ,\qquad c^{[r+1,0]}_2= \left[ \begin{array}{c} 2r+3\\ 2 \end{array}\right]_q \ .\label{coeffr+10}
\end{eqnarray}
The conditions under which the coefficients of other unwanted monomials of the type (\ref{mongenUNr}) are vanishing have now to be considered. In particular, similarly to the case $r=2,3$ the coefficients $c^{[r+1,0]}_1$, $c^{[r+1,0]}_2$ arise in  the following set of conditions: 
\begin{eqnarray*}
&&A^{2r-1}{A^*}^rAA^*A^3: \qquad \quad N_4^{(r,0)}\eta^{(4)}_{1,1}+c^{[r+1,0]}_1M^{(r,0)}_3\eta^{(3)}_{1,1}=0, ~~~~~~\\
&&A^{2r-1}{A^*}^rA^2A^*A^2: \ \ \quad \quad N^{(r,0)}_4\eta^{(4)}_{1,0}+c^{[r+1,0]}_1M^{(r,0)}_3\eta^{(3)}_{1,0}+c^{[r+1,0]}_2c^{[r,0]}_2=0,~~~~~~\\
&&A^{2(r-i)}{A^*}^rAA^*A^{2i+2}: \ \ \  \ N^{(r,0)}_{2i+3}\eta^{(2i+3)}_{i+1,1}+c^{[r+1,0]}_1M^{(r,0)}_{2i+2}\eta^{(2i+2)}_{i,1}+c^{[r+1,0]}_2c^{[r,0]}_{2i+1}\eta^{(2i+1)}_{i,1}=0,~ i=\overline{1,r},~~~~~~ \\
&&A^{2(r-i)}{A^*}^rA^2A^*A^{2i+1}: \ \ \ N^{(r,0)}_{2i+3}\eta^{(2i+3)}_{i+1,0}+c^{[r+1,0]}_1M^{(r,0)}_{2i+2}\eta^{(2i+2)}_{i,0}+c^{[r+1,0]}_2c^{[r,0]}_{2i+1}\eta^{(2i+1)}_{i,0}=0,~ i=\overline{1,r},~~~ \\
&&A^{2(r-i)+1}{A^*}^rAA^*A^{2i+1}: \ N^{(r,0)}_{2i+2}\eta^{(2i+2)}_{i,1}+c^{[r+1,0]}_1M^{(r,0)}_{2i+1}\eta^{(2i+1)}_{i,1}+c^{[r+1,0]}_2c^{[r,0]}_{2i}\eta^{(2i)}_{i-1,1}=0,~ i=\overline{2,r},~~~~~~\\
&&A^{2(r-i)+1}{A^*}^rA^2A^*A^{2i}: \quad  N^{(r,0)}_{2i+2}\eta^{(2i+2)}_{i,0}+c^{[r+1,0]}_1M^{(r,0)}_{2i+1}\eta^{(2i+1)}_{i,0}+c^{[r+1,0]}_2c^{[r,0]}_{2i}\eta^{(2i)}_{i-1,0}=0\,~ i=\overline{2,r}.~~~~~~
\end{eqnarray*}
Using the recursion relations in Appendices A,B, we have checked that all above equations are satisfied, as expected.\vspace{1mm}

More generally, one determines all other coefficients $c^{[r+1,0]}_j$ for $j\geq 3$. One finds:
\begin{eqnarray}
c^{[r+1,0]}_3&=& N^{(r,0)}_3\eta^{(3)}_{1,2}={\left[ {\begin{array}{*{20}{c}}
   {2r + 3}  \\
   3  \\
\end{array}} \right]_q},\nonumber\\
c^{[r+1,0]}_4&=& N^{(r,0)}_4\eta^{(4)}_{1,2}+c^{[r+1,0]}_1M^{(r,0)}_3\eta^{(3)}_{1,2}={\left[ {\begin{array}{*{20}{c}}
    {2r + 3}  \\
    4  \\
 \end{array}} \right]_q},\nonumber\\
c^{[r+1,0]}_{2k+1}&=& N^{(r,0)}_{2k+1}\eta^{(2k+1)}_{k,2}+c^{[r+1,0]}_1M^{r,0}_{2k}\eta^{(2k)}_{k-1,2}+c^{[r+1,0]}_2c^{[r,0]}_{2k-1}\eta^{(2k-1)}_{k-1,2},\hspace{.3cm} k=\overline{2,r+1},\nonumber \\
c^{[r+1,0]}_{2k+2}&=& N^{(r,0)}_{2k+2}\eta^{(2k+2)}_{k,2}+c^{[r+1,0]}_1M^{(r,0)}_{2k+1}\eta^{(2k+1)}_{k,2}+c^{[r+1,0]}_2c^{[r,0]}_{2k}\eta^{(2k)}_{k-1,2},\hspace{.3cm} k=\overline{2,r}.\nonumber
\end{eqnarray}
For any $j\geq 0$, one finds that the coefficient  $c^{[r+1,0]}_{j}$ can be simply expressed as a $q-$binomial:
\begin{eqnarray}
c^{[r+1,0]}_j= \left[ \begin{array}{c} 2r+3\\ j \end{array}\right]_q \ .\label{coefbinr}
\end{eqnarray}

All coefficients $c^{[r+1,0]}_j$ being obtained, at the lowest order in $\rho_0$ one has  to check that the coefficients of any unwanted term of the type (\ref{mongenUNr}) with $p+s+t= 2r+1,2r-1,...,1$  are systematically vanishing. Using the recursion relations given in Appendices A,B, this has been checked in details.  Then, following the analysis for $r=3$ it remains to determine the coefficient $c^{[r+1,1]}_0$ which contributes at the order $\rho_0$. The condition such that the coefficient of  the monomial $A^{2r} {A^*}^r A A^*$ is vanishing yields to:
\begin{eqnarray}
c^{[r+1,1]}_0={c^{[r,0]}_1}^2-2c^{[r,0]}_2+\frac{c^{[r,0]}_3}{c^{[r,0]}_1}-\frac{c^{[r,1]}_1}{c^{[r,0]}_1}+2c^{[r,1]}_0\ .\label{cr+11}
\end{eqnarray}
Using the explicit expression for $c^{[r+1,0]}_j, j=0,1,2$ and $c^{[r+1,1]}_0$, we have checked in details that $f_{r+1}(A,A^*)$ reduces to a combination of monomials of the type (\ref{mongenr}) only.  The reduced expression $f_{r+1}(A,A^*)$  determines uniquely all the remaining coefficients $c^{[r+1,p]}_j$ for $p\geq 1$. For $r$ generic, in addition to (\ref{coefbinr}) and (\ref{cr+11}) one finally obtains:
\begin{eqnarray}
c^{[r+1,r+1]}_0&=& c^{[r+1,1]}_0c^{[r,r]}_0+N^{(r,r+1)}_0,\label{coeff} \\
c^{[r+1,p]}_0&=& N_0^{(r,p)}+c^{[r+1,1]}_0c^{[r,p-1]}_0,\hspace{.3cm} p=\overline{2,r},\nonumber
\end{eqnarray}
\vspace{0mm}
\begin{eqnarray}
c_1^{[r+1,1]}&=& N_3^{(r,0)}+c^{[r+1,0]}_1M^{(r,1)}_0,\nonumber \\
c_1^{[r+1,2]}&=& -N_5^{(r,0)}+N_3^{(r,1)}+c^{[r+1,1]}_0c^{[r,0]}_3+c^{[r+1,0]}_1M^{(r,2)}_0,\nonumber \\
c^{[r+1,r+1]}_1&=&\sum_{p=0}^r{(-1)^{r+p}N^{(r,p)}_{2(r-p)+3}}+ c^{[r+1,1]}_0\sum_{p=0}^{r-1}{(-1)^{r+p+1}c^{[r,p]}_{2(r-p)+1}},\nonumber \\
c^{[r+1,p]}_1&=&\sum_{j=0}^{p-1}{(-1)^{j+p+1}N^{(r,j)}_{2(p-j)+1}}+c^{[r+1,0]}_1M^{(r,p)}_0 \nonumber \\& & + c_0^{[r+1,1]}\sum_{j=0}^{p-2}{(-1)^{j+p}c^{[r,j]}_{2(p-j)-1}},\hspace{.3cm} p=\overline{3,r},\nonumber 
\end{eqnarray}
\vspace{0mm}
\begin{eqnarray}
c^{[r+1,1]}_2&=&-N^{(r,0)}_4\eta^{(4)}_{0,2}+c^{[r+1,0]}_1M^{(r,0)}_3+c_2^{[r+1,0]}c_0^{[r,1]},\nonumber \\
c^{[r+1,2]}_2&=&N_6^{(r,0)}\eta^{(6)}_{0,2}-N_4^{(r,1)}\eta^{(4)}_{0,2}+c_1^{[r+1,0]}(M_5^{(r,0)}-M_3^{(r,1)})-c^{[r+1,0]}_2c^{[r,2]}_0+c_0^{[r+1,1]}c^{[r,0]}_4\eta^{(4)}_{0,2},\nonumber \\
c^{[r+1,p]}_2&=&\sum_{j=0}^{p-1}{(-1)^{j+p}N^{(r,j)}_{2(p-j)+2}\eta_{0,2}^{(2(p-j)+2)}}+c^{[r+1,0]}_1\sum_{j=0}^{p-1}{(-1)^{j+p+1}M^{(r,j)}_{2(p-j)+1}}\nonumber \\& &+c^{[r+1,0]}_2c^{[r,p]}_0+c_0^{[r+1,1]}\sum_{j=0}^{p-2}{(-1)^{j+p+1}c^{[r,j]}_{2(p-j)}\eta^{(2(p-j))}_{0,2}},\hspace{.3cm} p=\overline{3,r},\nonumber 
\end{eqnarray}
\vspace{0mm}
\begin{eqnarray}
c^{[r+1,1]}_3&=&-(N_5^{(r,0)}\eta^{(5)}_{1,2}-N_3^{(r,1)}\eta^{(3)}_{1,2})-c_1^{[r+1,0]}M_4^{(r,0)}\eta^{(4)}_{0,2}+c_0^{[r+1,1]}c_3^{[r,0]}\eta^{(3)}_{1,2}+c_2^{[r+1,0]}c_3^{[r,0]},\nonumber \\
c_3^{[r+1,p]}&=&\sum_{j=0}^p{(-1)^{j+p}N^{(r,j)}_{2(p-j)+3}\eta^{(2(p-j)+3)}_{1,2}} +c_1^{[r+1,0]}\sum_{j=0}^{p-1}{(-1)^{j+p}M^{(r,j)}_{2(p-j)+2}\eta^{(2(p-j)+2)}_{0,2}} 
\nonumber \\& &+ c^{[r+1,0]}_2\sum_{j=0}^{p-1}{(-1)^{j+p+1}c^{[r,j]}_{2(p-j)+1}}+c_0^{[r+1,1]}\sum_{j=0}^{p-1}{(-1)^{j+p+1}c^{[r,j]}_{2(p-j)+1}\eta^{(2(p-j)+1)}_{1,2}},\hspace{.3cm} j=\overline{2,r},\nonumber
\end{eqnarray}
\vspace{-0mm}
\begin{eqnarray}
c^{[r+1,1]}_4&=&-(N_6^{(r,0)}\eta^{(6)}_{1,2}-N_4^{(r,1)}\eta^{(4)}_{1,2})-c_1^{[r+1,0]}(M_5^{(r,0)}\eta_{1,2}^{(5)}-M_3^{(r,1)}\eta_{1,2}^{(3)})\nonumber\\ &&-c_2^{[r+1,0]}c_4^{[r,0]}\eta^{(4)}_{0,2}+c_0^{[r+1,1]}c_4^{[r,0]}\eta^{(4)}_{1,2},\nonumber \\
c^{[r+1,p]}_4&=&\sum_{j=0}^{p}{(-1)^{j+p}N^{(r,j)}_{2(p-j)+4}\eta^{(2(p-j)+4)}_{1,2}}+c^{[r+1,0]}_1\sum_{j=0}^p{(-1)^{j+p}M^{(r,j)}_{2(p-j)+3}\eta^{(2(p-j)+3)}_{1,2}}\nonumber \\& &+c^{[r+1,0]}_2\sum_{j=0}^{p-1}{(-1)^{j+p}c^{[r,j]}_{2(p-j)+2}\eta^{(2(p-j)+2)}_{0,2}}\nonumber\\ &&+c_0^{[r+1,1]}\sum_{j=0}^{p-1}{(-1)^{j+p+1}c^{[r,j]}_{2(p-j)+2}\eta^{(2(p-j)+2)}_{1,2}},\hspace{.3cm} p=\overline{2,r-1},\nonumber
\end{eqnarray}
\vspace{0mm}
\begin{eqnarray}
c^{[r+1,j-k]}_{2k+3}&=&\sum_{p=0}^{j-k}{(-1)^{p+j+k}N^{(r,p)}_{2(j-p)+4}\eta^{(2(j-p)+4)}_{k+1,2}}+c_1^{[r+1,0]}\sum_{p=0}^{j-k}{(-1)^{p+j+k}M^{(r,p)}_{2(j-p)+2}\eta^{(2(j-p)+2)}_{k,2}}\nonumber \\& &+c_2^{[r+1,0]}\sum_{p=0}^{j-k}{(-1)^{p+j+k}c^{[r,p]}_{2(j-p)+1}\eta^{(2(j-p)+1)}_{k,2}}\nonumber \\& &+c_0^{[r+1,1]}\sum_{p=0}^{j-k-1}{(-1)^{p+j+k+1}c^{[r,p]}_{(2(j-p)+1)}\eta^{(2(j-p)+1)}_{k+1,2}},\hspace{.3cm} j=\overline{3,r}, \qquad k=\overline{1,j-2},\nonumber
\end{eqnarray}
\begin{eqnarray}
c^{[r+1,1]}_{2j+1}&=&-(N^{(r,0)}_{2j+3}\eta^{(2j+3)}_{j,2}-N^{(r,1)}_{2j+1}\eta^{(2j+1)}_{j,2})-c_1^{[r+1,0]}(\eta^{(2j+2)}_{j-1,2}M^{(r,0)}_{2j+2}-M^{(r,1)}_{2j}\eta^{(2j)}_{j-1,2})\nonumber \\& &-c^{[r+1,0]}_2(c^{[r,0]}_{2j+1}\eta^{(2j+1)}_{j-1,2}-c^{[r,1]}_{2j-1}\eta^{(2j-1)}_{j-1,2})+c_0^{[r+1,1]}c^{[r,0]}_{2j+1}\eta^{(2j+1)}_{j,2},\hspace{.3cm} j=\overline{2,r},\nonumber
\end{eqnarray}
\begin{eqnarray}
c^{[r+1,j-k]}_{2k+2}&=&\sum_{p=0}^{j-k}{(-1)^{p+j+k}N^{(r,p)}_{2(j-p)+2}\eta^{(2(j-p)+2)}_{k,2}}+c_1^{[r+1,0]}\sum_{p=0}^{j-k}{(-1)^{p+j+k}M^{(r,p)}_{2(j-p)+1}\eta^{(2(j-p)+1)}_{k,2}}\nonumber \\& & +c_2^{[r+1,0]}\sum_{p=0}^{j-k}{(-1)^{p+j+k}c^{[r,p]}_{2(j-p)}\eta^{(2(j-p))}_{k-1,2}}\nonumber \\& &+c_0^{[r+1,1]}\sum_{p=0}^{j-k-1}{(-1)^{p+j+k+1}c_{(2(j-p))}^{[r,p]}\eta^{(2(j-p))}_{k,2}}, \hspace{.3cm} j=\overline{4,r}, \qquad k=\overline{2,j-2},\nonumber
\end{eqnarray}
\begin{eqnarray}
c^{[r+1,1]}_{2j}&=&c^{[r+1,1]}_0c^{[r,0]}_{2j}\eta^{(2j)}_{j-1,2}-N^{(r,0)}_{2j+2}\eta^{(2j+2)}_{j-1,2}+N^{(r,1)}_{2j}\eta^{(2j)}_{j-1,2}\nonumber \\ & &-c^{[r+1,0]}_1(M^{(r,0)}_{2j+1}\eta^{(2j+1)}_{j-1,2}-M^{(r,1)}_{2j-1}\eta^{(2j-1)}_{j-1,2})\nonumber\\&&-c^{[r+1,0]}_2(c^{[r,0]}_{2j}\eta^{(2j)}_{j-2,2}-c^{[r,1]}_{2j-2}\eta^{(2j-2)}_{j-2,2}),\hspace{.3cm} j=\overline{3,r}.\nonumber 
\end{eqnarray}
According to above results and using the automorphism $A\leftrightarrow A^*$ and $\rho_0\leftrightarrow \rho_1$, we conclude that if $A,A^*$ satisfy the defining relations (\ref{qO1}), (\ref{qO2}), then the higher order $q-$Dolan-Grady relations (\ref{qDGr}),  (\ref{qDGr2}) are such that the coefficients $c^{[r,p]}_j$ are determined recursively by (\ref{coefbinr}), (\ref{cr+11}) and (\ref{coeff}). For $p\geq1$, they can be computed for practical purpose, for $r=2,3$ the coefficients $c_j^{[r,p]}$ are proportional to $[2r+1]_q$ iff $j\neq 0$ or $2r+1$; for a large number of  values $r\geq 4$, this property holds too. As a consequence, the relations (\ref{qDGfinr}) drastically simplify for $q^{2r+1}=\pm1$. This case is however not considered here.  In particular, one observes that $c_{j}^{[r,p]}=c_{2(r-p)+1-j}^{[r,p]}$\ . For $r=4,5,...\leq 10$, using a computer program we have checked in details that $r-th$ higher order relations of the form (\ref{qDGfinr}) hold, and that the coefficients satisfy above recursive formula.

\subsection{Algorithm}
A Maple software program has been constructed to calculate
the expressions for the coefficients $c^{[r+1,p]}_j$  by induction on $r$. For $r=2,3,...,10$, the expressions have been compared with the exact expressions for the coefficients given in (\ref{cfinr}). Both expressions agree, thus giving a strong support to the conjecture.  The  Maple program is reported in Appendix D. Here, we sketch the algorithm.
\begin{itemize}
\item Input: $r, ~ \eta^{(3)}_{1,j},j = \overline{0,2};~~ \eta^{(4)}_{k,j}, k=\overline{0,1}, j=\overline{0,2}$\\
$c^{[1,p]}_j, p=\overline{0,1}, j=\overline{0,3-2p};~~ c^{[2,p]}_j, p=\overline{0,2}, j=\overline{0,5-2p}$

\item Output: $c^{[r+1,p]}_j, ~ p =0, \dots, r+1;~~ j = 0,\dots,2r+3-2p $
\item Algorithm:
\begin{description}
\item[Step 1.] Compute the coefficients $\eta^{(m)}_{k,j},m=\overline{5,2r+3},j=\overline{0,2},\\ k= \left\{\begin{array}{c}
\overline{0,[\frac{m-1}{2}]} ~~\text{If}~ m~ \text{is even}\\
\overline{1,[\frac{m-1}{2}]} ~~\text{If }~ m~ \text{is odd}
\end{array} \right.$ 
\item[Step 2.] Compute $M^{(r,0)}_j, j=\overline{2,2r+2};~~ M^{(r,p)}_j, p=\overline{1,r}, j=\overline{0,2(r-p)+2}$ 
\item[Step 3.] Compute $N^{(r,0)}_j, j=\overline{3,2r+3};~~N^{(r,p)}_j, p=\overline{1,r}, j=\overline{0,2(r-p)+3};$
\item[Step 4.] Compute $c^{[r+1,0]}_0, c^{[r+1,0]}_1, c^{[r+1,0]}_2, c^{[r+1,1]}_0$
\item[Step 5.] Compute $f_{r+1}(A,A^*)$ in the equation (\ref{combr+1}).
\end{description}
\end{itemize}

\section{Higher order relations for the generalized $q-$Onsager algebra}

\subsection{Introduction}
Introduced in \cite{BB1}, the generalized $q-$Onsager algebra ${\cal O}_q({\widehat{g}})$ associated with the affine Lie algebra $\widehat{g}$ is a higher rank generalization of the so-called $q-$Onsager algebra \cite{Ter03,B1}. The usual $q-$Onsager algebra corresponds to the choice $\widehat{g}=\widehat{sl_2}$. The defining relations are determined by the entries of the Cartan matrix
of the algebra considered.
 For $\widehat{g}=a_n^{(1)}$, it can be understood as a $q-$deformation of the $sl_{n+1}$-Onsager algebra introduced by Uglov and Ivanov \cite{Uglov}. By analogy with the $\widehat{sl_2}$ case \cite{B1,IT}, an algebra homomorphism from ${\cal O}_q({\widehat{g}})$ to a certain coideal subalgebra of the Drinfeld-Jimbo \cite{Dr,Jim} quantum universal enveloping algebra ${\cal U}_q(\widehat{g})$  is known \cite{BB1}. From a general point of view, generalized $q-$Onsager algebras appear in the theory of quantum affine symmetric pairs \cite{Kolb}. Note that realizations in terms of finite dimensional quantum algebras may be also considered: for instance, coideal subalgebras of ${\cal U}_q(g)$ studied by Letzter \cite{Lez} or the non-standard ${\cal U}'_q(so_n)$ introduced by Klimyk, Gavrilik and Iorgov \cite{GI,Klim}.

Besides the definition of the generalized $q-$Onsager algebra in terms of generators and relations \cite[Definition 2.1]{BB1}, most of its properties remain to be studied.
\begin{defn}
Let $\{a_{ij}\}$ be the extended Cartan matrix of the affine Lie algebra $\hat{g}$. Fix coprime integers $d_i$ such that $d_ia_{ij}$ is symmetric. The generalized q-Onsager algebra $O_q(\hat{g})$ is an associative algebra with unit 1, generators $A_i$ and scalars $\rho^k_{ij}, \gamma^{kl}_{ij} \in \mathbb{C}$ with $i, j \in \{0, 1, \dots, n\}, k \in \{0,1,\dots, [-\frac{a_{ij}}{2}]-1\}$ and $l \in \{0,1,\dots, -a_{ij}-1-2k\}$ ($k$ and $l$ are positive integers). The defining relations are:
\begin{equation}
{\sum\limits_{r = 0}^{1 - {a_{ij}}} {{{\left( { - 1} \right)}^r}\left[ {\begin{array}{*{20}{c}}
   {1 - {a_{ij}}}  \\
   r  \\
\end{array}} \right]} _{{q_i}}}A_i^{1 - {a_{ij}} - r}{A_j}A_i^r = \sum\limits_{k = 0}^{\left[ { - \frac{{{a_{ij}}}}{2}} \right] - 1} {\rho _{ij}^k} \sum\limits_{l = 0}^{ - 2k - {a_{ij}} - 1} {{{\left( { - 1} \right)}^l}\gamma _{ij}^{kl}A_i^{ - 2k - {a_{ij}} - 1 - l}{A_j}A_i^l,} 
\end{equation}
\end{defn}

Generalized $q-$Onsager algebras are extensions of the $q-$Onsager algebra to {\it higher rank} affine Lie algebras  \cite{BB1}. Inspired by the analysis of \cite{BV1}, analogues of Lusztig's higher order relations for $\mathcal O_q(\widehat{g})$ can be conjectured. First, recall some basic definitions.
\begin{defn}
Let the simply-laced affine Lie algebra $\hat{g}$, the generalized $q-$Onsager algebra $\mathcal{O}_q(\widehat{g})$ is an associative algebra with unit $1$, generators ${A}_i$ and scalars $\rho_i$. The defining relations are:
\begin{eqnarray}
\label{generqOns}
\quad \sum_{k=0}^{2} (-1)^k \left[ \begin{array}{c} 2 \\  k \end{array}\right]_q   A_i^{2-k} {A_j} A_{i}^{k} - \rho_{i} A_j &=& 0\ \quad \mbox{if} \quad i,j \quad \mbox{are linked}\ ,\label{rel1g}\\
\big[A_i,A_j\big]&=&0 \quad \mbox{otherwize}\ .\nonumber
\end{eqnarray}
\end{defn}
\begin{rem} For $\rho_i=0$ the relations (\ref{rel1g}) reduce to the $q-$Serre relations of $U_{q}(\widehat{g})$. 
\end{rem}
\begin{rem}
For $q=1$, the relations (\ref{rel1g}) coincide with the defining relations of the so-called $sl_{n+1}-$Onsager's algebra for $n> 1$ introduced by Uglov and Ivanov \cite{Uglov}. 
\end{rem}

\subsection{Conjecture about the higher order relations of the generalized $q$-Onsager algebra}
By analogy with the $\widehat{sl_2}$ case discussed in details in the previous parts, we expect the following form for the higher order relations

\begin{conj}
Let $\{A_i\}$ be the fundamental generators of the generalized $q-$Onsager algebra (\ref{generqOns}), then $\{A_i\}$ satisfy the higher order relations as follows:
\begin{eqnarray}
&&\sum_{p=0}^{\{\frac{r+1}{2}\}}\sum_{k=0}^{r-2p+1} (-1)^{k+p} \, \rho_i^{p} \, \,c_{k}^{[r,p]}\,  A_i^{r-2p+1-k} A_j^r A_i^k=0 \  \ \quad \mbox{if} \quad i,j \quad \mbox{are linked}\ \label{qDGADE}
\end{eqnarray}
where the coefficients are given by:
\small
\begin{eqnarray}
\label{coefADE}
c_{k}^{[r,p]} =  \sum_{l=0}^{ \{k/\alpha\}}\frac{(  \{\frac{r+1}{2}\} - k +\alpha l -p)!}{    (\{\frac{\alpha l}{2}\})!( \{\frac{r+1}{2}\} - k + \alpha l - p -\{\frac{\alpha l}{2}\})!}
    \sum_{{{\cal P}_l}}  [s_1]^2_{q}...[s_p]^2_{q}  \frac{[2s_{p+1}]_{q}...[2s_{p+k-\alpha l}]_{q}}{[s_{p+1}]_{q}...[s_{p+k-\alpha l}]_{q}}  \label{coefffin}
\end{eqnarray}
\normalsize
\begin{eqnarray}
\mbox{with}\quad \ \left\{\begin{array}{cc}
\!\!\! \!\!\! \!\!\! \!\!\!  \!\!\! \!\!\! \!\!\! \!\!\!   \!\!\! \!\!\! \!\!\! \!\!\! \quad \quad \quad \quad \quad \quad \quad k= \overline{0,\{\frac{r+1}{2}\}}\ , \quad s_i\in\{   r-2\{\frac{r-1}{2}\},\dots,r-2,r\}\ ,\\
{\cal P}_l: \begin{array}{cc} \ \ s_1<\dots<s_p\ ;\quad \ s_{p+1}<\dots<s_{p+k - \alpha l}\ ,\\
 \{s_{1},\dots,s_{p}\} \cap \{s_{p+1},\dots,s_{p+k-\alpha l}\}=\emptyset \end{array} ,\\
\alpha = \left\{ \begin{array}{cc} 1 \quad \text{if $r$ is even,} \\ \!\!\!\!2 \quad \text{if $r$ is odd}  \end{array}\right .\
\end{array}\right.\ .\nonumber
\end{eqnarray}
\end{conj}
Although it is highly expected that the concept of tridiagonal pair for the $q-$Onsager algebra could be extended to the higher rank generalizations of the $q-$Onsager algebras, in the mathematical literature such object has not been introduced yet. For this reason, the conjecture for the  higher order relations (\ref{qDGADE}) can not be checked for every irreducible finite dimensional vector space on which the generators $A_i$ act. Still, below we provide several supporting evidences for the conjecture. First, the conjecture is proven for $r\leq 5$. Secondly, recursive relations for the coefficients are derived. Using a Maple software, it is found that the coefficients computed from the recursion relations coincide exactly with the ones conjectured. Other checks of the conjecture are considered, thus giving independent supporting evidences.

\subsection{Proof of the higher order relations for $r\le 5$}
For $r=1$, the relations (\ref{qDGADE}) are the defining relations of the generalized $q-$Onsager algebra $\mathcal{O}_{q}(\widehat{g})$.  Assume $A_i$ are the fundamental generators of $\mathcal{O}_{q}(\widehat{g})$. 
To derive the first example of higher order relations, we are looking for a linear relations between monomials of the type $A_i^n A_j^2 A_i^m$ with   $n+m=3,1$. Suppose it is of the form (\ref{qDGADE}) for $r=2$ with yet unknown coefficients $c_{k}^{[r,p]}$. We show $c_{k}^{[r,p]}$ are uniquely determined. First, according to the defining relations (\ref{rel1g}) the monomial $A_i^2A_j$ can be ordered as:

\begin{eqnarray}
A_i^2A_j = [2]_q A_iA_jA_i -  A_jA_i^2 + \rho_i A_j\  .\label{mon1}
\end{eqnarray}
Multiplying from the left and/or right by $A_i,A_j$ ($i\neq j$), the new monomials can be ordered as follows: each time a monomial of the form $A_i^n A_j^2A_i^m$ with $n\geq 2$ arises, it is reduced using (\ref{mon1}). For instance, one has:
\begin{eqnarray}
A_i^3 A_j = ([2]^2_q-1) A_iA_jA_i^2  - [2]_q A_jA^3_i + \rho_i ([2]_q A_jA_i + A_iA_j) \ .
\end{eqnarray}
Now, observe that the first two monomials in  (\ref{qDGADE}) for $r=2$ can be written as $A_i^3 A_j^2\equiv (A_i^3 A_j) A_j$ and $A_i^2 A_j^2A_i \equiv(A_i^2 A_j) A_j A_i$. Following the ordering prescription, each of these monomials can be reduced as a combination of monomials of the type:
\begin{eqnarray}
&&\quad A_i^n{A_j^2}A_i^m \ \qquad \ \mbox{with} \quad  \ n\leq 1\ , \ n+m=3,1\ ,\label{mongen}\\
 &&\quad A_i^p A_j A_i A_j A_i^t \quad \  \mbox{with} \quad  \ p\leq 1\ ,\ p+t=2,0 \ .\label{badterm}
\end{eqnarray}
Plugging the reduced expressions of $A_i^3 A_j^2$ and $A_i^2 A_j^2A_i$  in (\ref{qDGADE}) for $r=2$, one finds that all monomials of the form (\ref{badterm}) cancel provided a simple  system of equations for the coefficients $c_{k}^{[r,p]}$ is satisfied. The solution of this system is unique, given by:
\begin{eqnarray}
&& c^{[2,0]}_k= \left[ \begin{array}{c} 3 \\  k \end{array}\right]_q \ \quad \mbox{for}\quad k=0,1,2,3 \ , \quad \mbox{and}\quad c^{[2,1]}_0=c^{[2,1]}_1=q^2+q^{-2}+2 \nonumber\ .
\end{eqnarray}

For $r=3,4,5$, we proceed similarly: the monomials entering in the relations (\ref{qDGADE}) are ordered according to the prescription described above. Given $r$, the  reduced expression of the corresponding relation (\ref{qDGADE}) holds provided the coefficients $c_{k}^{[r,p]}$  satisfy a system of equation which solution is unique. In each case, one finds:
\begin{eqnarray}
 c^{[r,0]}_k= \left[ \begin{array}{c} r+1 \\  k \end{array}\right]_q \ \quad \mbox{for}\quad k=0,...,r+1 \ ,\  r=3,4,5\ ,\label{qbin}
\end{eqnarray}
whereas for $p\geq 1$, the other coefficients are such that $c_{k}^{[r,p]}=c_{r-2p+1-k}^{[r,p]}$, given by:
\begin{eqnarray}
\mbox{\bf Case $r=3$:}
&& c^{[3,1]}_0= {q^4} + 2q^2+4+2q^{-2} +q^{-4} \ , \qquad c^{[3,1]}_1= [4]_q(q^2 + q^{-2}+3) \ ,\nonumber\\
&& c^{[3,2]}_0= ({q^2} + q^{-2} +1)^2 \ ; \ \nonumber \\
\mbox{\bf Case $r=4$:}&& c_0^{[4,1]}=(q^4+3+q^{-4})[2]_q^2\ ,\ c_1^{[4,1]}=[5]_q[3]_q[2]_q^2,\nonumber \\
&&c_0^{[4,2]}=(q^2+q^{-2})^2[2]_q^4 \ ;\nonumber \\
\mbox{\bf Case $r=5$:}\nonumber \\
c_0^{[5,1]}&=&q^8+2q^6+4q^4+6q^2+9+6q^{-2}+4q^{-4}+2q^{-6}+q^{-8},\nonumber\\
c_1^{[5,1]}&=&[6]_q[3]_q^{-1}(q^8+4q^6+8q^4+14q^2+16+14q^{-2}+8q^{-4}+4q^{-6}+q^{-8}),\nonumber \\
c_2^{[5,1]}&=&[6]_q[2]_q^{-1}[5]_q(q^4+3q^2+6+3q^{-2}+q^{-4}),\nonumber \\
c_0^{[5,2]}&=& q^{12}+4q^{10}+11q^8+20q^6+31q^4+40q^2+45\nonumber\\&&+40q^{-2}+31q^{-4}+20q^{-6}+11q^{-8}+4q^{-10}+q^{-12},\nonumber\\
c_1^{[5,2]}&=&[6]_q[3]_q^{-1}(q^{10}+6q^8+17q^6+32q^4+47q^2+53\nonumber\\&&+47q^{-2}+32q^{-4}+17q^{-6}+6q^{-8}+q^{-10}),\nonumber \\
c_0^{[5,3]}&=&[3]_q^2[5]_q^2\ . \nonumber
\end{eqnarray}
\subsection{Recursion relations of the coefficients of the higher order relations in generic case $r$} 
Above examples suggest that higher order relations of the form (\ref{qDGADE}) exist for generic values of $r$. To derive the coefficients recursively, one first assumes that given $r$, the relation (\ref{qDGADE}) exists and that all coefficients   $c_{k}^{[r,p]}$ are already known in terms of $q$. The relation (\ref{qDGADE}) for $r\rightarrow r+1$ is then considered.  In this case, the combination
\begin{eqnarray}
\label{ADE}
f^{ADE}_r(A_i,A_j)= A_i^{r+2}A_j^{r+1}  - c_{1}^{[r+1,0]} A_i^{r+1}A_j^{r+1}A_i
\end{eqnarray}
is introduced. Following the steps described in details in \cite{BV1}, the monomials $A_i^{r+2}A_j^{r+1}$ and $A_i^{r+1}A_j^{r+1}A_i$  are reduced using (\ref{rel1g}) and  (\ref{qDGADE}). The  ordered expression of the first monomial follows:
\begin{eqnarray}
A_i^{r+2}A_j^{r+1}&=&\sum_{k=2}^{r+2}{(-1)^{k+1} M_k^{(r,0)}A_i^{r+2-k}A_j^r\!\!\!\!\!\!\!~ \underbrace{A_i^k A_j}_{~\mbox{reducible}}} \nonumber \\ &&+\sum_{p=1}^{\{\frac{r+1}{2}\}}\sum_{k=0}^{r+2-2p}{(-1)^{p+k+1}\rho_i^pM_k^{(r,p)}A_i^{r+2-2p-k}A_j^r\!\!\!\!\!\!\!\!\!\!\!\!\!\!\!\!\!~\underbrace{A_i^k A_j}_{\mbox{reducible if $k\geq 2$}} \!\!\!\!\!\!\!\!\!\!\!\!\!\!\!\!}\  ,\nonumber
\end{eqnarray}
where the recursive relations for the coefficients $M_k^{(r,p)}$ are reported in Appendix C. Obviously, an ordered expression for the second monomial immediately follows from (\ref{qDGADE}). Then, the whole combination can be further reduced $f^{ADE}_r(A_i,A_j)$. As an intermediate step, one uses  (\ref{rel1g}) to obtain:
\begin{eqnarray}
A_i^{2n+1}A_j&=&\sum_{p=0}^n{\rho_i^p(\eta_0^{(2n+1,p)}A_iA_jA_i^{2n-2p}+\eta_1^{(2n+1,p)}A_jA_i^{(2n-2p+1)})}\ , \nonumber \\
A_i^{2n+2}A_j&=&\sum_{p=0}^n{\rho_i^p(\eta_0^{(2n+2,p)}A_iA_jA_i^{2n+1-2p}+\eta_1^{(2n+2,p)}A_jA_i^{2n+2-2p})}+\rho_i^{n+1}A_j\ , \nonumber 
\end{eqnarray}
where the coefficients  $\eta_j^{(r,p)}$ are given in Appendix C. The ordered expression of $f^{ADE}_r(A_i,A_j)$ is then studied. A detailed analysis shows that all coefficients of monomials of the type  $A_i^{p} A_j^r A_iA_j A_i^t$ (with $p+t=r-1,...,0$ if $r$ is odd, and $p+t=r-1,...,1$ if $r$ is even) vanish provided the coefficients  $c_{k}^{[r,p]}$ satisfy a system of equations. According to the parity of $r$, one finds:\vspace{2mm}

{\bf \underline{Case $r$ odd:}} For $r=2t+1$ and $p=0$:
\begin{eqnarray}
c^{[2t+1,0]}_2&=& M^{(2t,0)}_2\eta^{(2,0)}_1,\nonumber \\
c^{[2t+1,0]}_{2h}&=& M^{(2t,0)}_{2h}\eta^{(2h,0)}_1+c^{[2t+1,0]}_1c^{[2t,0]}_{2h-1}\eta^{(2h-1,0)}_1,\qquad h=\overline{2,t+1},\nonumber \\
c^{[2t+1,0]}_{2h+1}&=& M^{(2t,0)}_{2h+1}\eta^{(2h+1,0)}_1+c^{[2t+1,0]}_1c^{[2t,0]}_{2h}\eta^{(2h,0)}_1, \qquad h=\overline{1,t}.\nonumber
\end{eqnarray}
Using the recursion relations given in Appendix C, it is possible to show that these coefficients can be simply written in terms of $q-$binomials: 
\begin{eqnarray}
 c_{k}^{[r,0]} =  \left[ \begin{array}{c} r+1 \\ k \end{array}\right]_q\ .    \label{formbinom}
\end{eqnarray}
Other coefficients $c^{[2t+1,p]}_0$ for $p \geq 1$ are determined by the following recursion relations: 
\begin{eqnarray}
c^{[2t+1,t+1]}_0&=&\sum_{p=0}^t{(-1)^{p+t+1}M^{(2t,p)}_{2(t+1-p)}},\nonumber \\
c^{[2t+1,1]}_0&=& -M_2^{(2t,0)}+M_0^{(2t,1)},\nonumber \\
c^{[2t+1,h]}_0&=& \sum_{p=0}^{h}{(-1)^{p+h}M^{(2t,p)}_{2(h-p)}}, \qquad h=\overline{2,t},\nonumber 
\end{eqnarray}
\begin{eqnarray}
c^{[2t+1,1]}_1&=&-(M^{(2t,0)}_3\eta^{(3,1)}_1-c^{[2t+1,0]}_1(-c_2^{[2t,0]}+c_0^{[2t,1]})),\nonumber \\
c^{[2t+1,h]}_1&=&\sum_{p=0}^{h-1}{(-1)^{p+h}M^{(2t,p)}_{2(h-p)+1}\eta^{(2(h-p)+1,h-p)}_1} \nonumber \\& &+c^{[2t+1,0]}_1\sum_{p=0}^{h}{(-1)^{p+h}c^{[2t,p]}_{2(h-p)}}, \qquad h=\overline{2,t},\nonumber
\end{eqnarray}
\begin{eqnarray}
c^{[2t+1,1]}_2&=& -M_4^{(2t,0)}\eta^{(4,1)}_1+M^{(2t,1)}_2\eta^{(2,0)}_1-c^{[2t+1,0]}_1c_3^{[2t,0]}\eta^{(3,1)}_1,\nonumber \\
c^{[2t+1,l]}_{2h-2l+1}&=&\sum_{p=0}^l{(-1)^{p+l}M^{(2t,p)}_{2(h-p)+1}}\eta_1^{(2(h-p)+1,l-p)} \nonumber\\& &+c^{[2t+1,0]}_1\sum_{p=0}^l{(-1)^{p+l}c^{[2t,p]}_{2(h-p)}\eta_1^{(2(h-p),l-p)}}, \qquad h=\overline{2,t}, \quad l=\overline{1,h-1},\nonumber \\
c^{[2t+1,l]}_{2h-2l}&=&\sum_{p=0}^l{(-1)^{p+l}M_{2(h-p)}^{(2t,p)}\eta_1^{(2(h-p),l-p)}}\nonumber\\& &+c^{[2t+1,0]}_1\sum_{p=0}^{min(l,h-2)}{(-1)^{p+l}c^{[2t,p]}_{2(h-p)-1}\eta_1^{(2(h-p)-1,l-p)}}, \qquad h=\overline{3,t+1}, \quad l=\overline{1,h-1}.\nonumber
\end{eqnarray}

{\bf \underline{Case $r$ even:}}  For $r=2t+2$, the coefficients $c^{[2t+2,0]}_j$ are given by:
\begin{eqnarray}
c^{[2t+2,0]}_2&=& M_2^{(2t+1,0)}\eta_1^{(2,0)},\nonumber \\
c^{[2t+2,0]}_{2h+1}&=& M_{2h+1}^{(2t+1,0)}\eta^{(2h+1,0)}_1+c^{[2t+2,0]}_1c_{2h}^{[2t+1,0]}\eta_1^{(2h,0)}, \qquad h=\overline{1,t+1},\nonumber \\
c^{[2t+2,0]}_{2h}&=& M_{2h}^{(2t+1,0)}\eta^{(2h,0)}_1+c^{[2t+2,0]}_1c_{2h-1}^{[2t+1,0]}\eta_1^{(2h-1,0)}, \qquad h=\overline{2,t+1}.\nonumber
\end{eqnarray}

According to the relations in Appendix C,  one shows that $c^{[r,0]}_k$ simplify to $q-$binomials (\ref{formbinom}). For $p \geq 1 $, the recursive formulae for all other coefficients are given by:
\begin{eqnarray}
c_0^{[2t+2,1]}&=&-M^{(2t+1,0)}_2+ M_0^{(2t+1,1)},\nonumber \\
c_0^{[2t+2,h]}&=&\sum_{p=0}^{h}{(-1)^{p+h}M_{2(h-p)}^{(2t+1,p)}}, \qquad h=\overline{2,t+1},\nonumber \\
c_1^{[2t+2,1]}&=&-M_3^{(2t+1,0)}\eta_1^{(3,1)}+c_1^{[2t+2,0]}(-c^{[2t+1,0]}_2+c_0^{[2t+1,1]}),\nonumber 
\end{eqnarray}
\begin{eqnarray}
c_1^{[2t+2,h]}&=&\sum_{p=0}^{h-1}{(-1)^{p+h}M_{2(h-p)+1}^{(2t+1,p)}\eta_1^{(2(h-p)+1,h-p)}} \nonumber\\ & &+c_1^{[2t+2,0]}\sum_{p=0}^{h}{(-1)^{p+h}c_{2(h-p)}^{[2t+1,p]}}, \qquad h=\overline{2,t+1},\nonumber \\
c_2^{[2t+2,1]}&=&-M_4^{(2t+1,0)}\eta_1^{(4,1)}+M_2^{(2t+1,1)}\eta^{(2,0)}_1-c_1^{[2t+2,0]}c_3^{[2t+1,0]}\eta_1^{(3,1)},\nonumber
\end{eqnarray}
\begin{eqnarray}
c_{2h-2l}^{[2t+2,l]}&=&\sum_{p=0}^l{(-1)^{p+l}M_{2(h-p)}^{(2t+1,p)}\eta_1^{(2(h-p),l-p)}}\nonumber \\& & +c_1^{[2t+2,0]}\sum_{p=0}^{min(l,h-2)}{(-1)^{p+l}c_{2(h-p)-1}^{[2t+1,p]}\eta_1^{(2(h-p)-1,l-p)}}, \qquad h=\overline{3,t+1}, l=\overline{1,h-1},\nonumber \\
c_{2h-2l+1}^{[2t+2,l]}&=&\sum_{p=0}^l{(-1)^{p+l}M_{2(h-p)+1}^{(2t+1,p)}\eta_1^{(2(h-p)+1,l-p)}}) \nonumber \\&&+c_1^{[2t+2,0]}\sum_{p=0}^l{(-1)^{p+l}c_{2(h-p)}^{[2t+1,p]}\eta_1^{(2(h-p),l-p)}}, \qquad h=\overline{2,t+1}, l=\overline{1,h-1}.\nonumber
\end{eqnarray}
All coefficients $c^{[r,p]}_k$ entering in the higher order relations (\ref{qDGADE}) can be computed recursively for any positive integer $r$. Note that setting $\rho_i=0$, the relations (\ref{qDGADE}) reproduce the higher order $q-$Serre relations (\ref{hqSerre})  of     ${\mathcal U}_q(\widehat{g})$ \cite{Luszt}. Using a computer program, up to $r=10$ we have checked that the results for the coefficients derived from the recursion relations coincide exactly with the ones conjectured in (\ref{qDGADE}).\vspace{1mm} 

The algorithm for the computation of the coefficients is the following (Appendix E)
\begin{itemize}
\item Input: $r+1,~\eta_{0,0}^{(2)}, ~\eta_{0,1}^{(2)}$, $c^{[1,0]}_0,~ c^{[1,0]}_1, ~ c_{2}^{[1,0]}, ~c_0^{[1,1]}$
\item Output: $c^{[r+1,p]}_k,~ p = 0, \dots, \{\frac{r+2}{2}\}, ~~k =0, \dots, r+2-2p$
\item Algorithm:
\begin{description}
\item[Step 1.] Compute the coefficients $\eta^{(h)}_{p,i}, ~ h =3,\dots, r+2,~ p=0,\dots, \{\frac{h-1}{2}\},~ i =0, 1$ 
\item[Step 2.] Compute $M^{(r,p)}_k$\\
If $p=0$, then $k = 2,\dots, r+2$.

If $p = 1,\dots, \{\frac{r+1}{2} \}$ , then $k = 1,\dots, r+2-2p$ 
\item[Step 3.] Compute $c^{[r+1,0]}_1$
\item[Step 4.] Compute $f_{r+1}(A,A^*)$ in the equation (\ref{ADE}).
\end{description}
\end{itemize}

\subsection{A two-variable polynomial generating function}
For the $q-$Onsager algebra, it was shown that the coefficients entering in the $r-th$ higher order relations can be derived from a two-variable generating function. Here, for any simply-laced affine Lie algebras we propose a two-variable generating function for the coefficients $c^{[r,p]}_k$.
\begin{defn}\label{defpolyADE} Let ${r\in \mathbb Z}^+$. Let $x,y$ be commuting indeterminates and $\rho$ a scalar. To any simply-laced affine Lie algebra  $\widehat{g}$, we associate the polynomial generating function  $p^{ADE}_r(x,y)$ such that:
\begin{eqnarray}
p_{2t+1}^{ADE}(x,y)&=& \prod_{l=1}^{t+1} \left(x^2- \frac{[4l-2]_q}{[2l-1]_q} xy+y^2 - \rho[2l-1]_q^2  \right) \ ,\label{polyprod1}\\
p_{2t+2}^{ADE}(x,y)&=& (x-y)\prod_{l=1}^{t+1} \left(x^2-  \frac{[4l]_q}{[2l]_q} xy+y^2 - \rho[2l]_q^2  \right) \ .\label{polyprod2}
\end{eqnarray}
\end{defn}
\begin{lem} 
The polynomial $p_r^{ADE}(x,y)$ can be expanded as:
\begin{eqnarray}
p_r^{ADE}(x,y) = \sum_{p=0}^{\{\frac{r+1}{2}\}}\ \sum_{k=0}^{r-2p+1} (-1)^{k+p}   \rho^{p}\,c_{k}^{[r,p]}\,  x^{r-2p+1-k} y^k\  \label{polyADE}
\end{eqnarray}
where the coefficients $c_k^{[r,p]}$ are given by (\ref{coefADE}).
\end{lem}
\begin{proof}
By induction.
\end{proof}

We claim that the two-variable polynomial (\ref{polyADE}) is the generating function for the coefficients $c^{[r,p]}_k$  entering in the higher order relations (\ref{qDGADE}) in view of the following observations:

\begin{itemize}
\item For $r\leq 5$, it is an exercise to check that the coefficients $c^{[r,p]}_k$ from (\ref{coefffin}) coincide exactly with the ones derived in the previous Section (see cases $r=2,3,4,5$);

\item  Using the $q-$binomial theorem, for $r$ generic  it is easy to check that the coefficients $c^{[r,0]}_k$ obtained from (\ref{coefffin}) are the $q-$binomials (\ref{formbinom}). Namely,

For $\rho = 0$,
\begin{eqnarray*}
p^{ADE}_{2t+1}(x,y)&=&\prod_{l=1}^{t+1}{(x^2-\frac{[4l-2]_q}{[2l-1]_q}xy+y^2)}\\
&=& x^{2(t+t)}\prod_{l=-(t+1)}^t{(1-q^{2l}(q\frac{y}{x}))}.
\end{eqnarray*}
Put $k=l+t+2$, $p^{ADE}_{2t+1}(x,y)=x^{2(t+1)}\prod\limits_{k=1}^{2(t+1)}{(1-q^{2k}(q^{-2t-3}\frac{y}{x}))}$. Apply the q-binomial theorem
\[ p^{ADE}_{2t+1}(x,y)=\sum\limits_{k=1}^{2(t+1)}{(-1)^k \left [ \begin{array}{c}
2t+2\\
k
\end{array}
\right]_qx^{2t+2-k}y^k. } \]

\item  For $r\geq 6$ and $p\geq 1$, the comparison is more involved. However, using a computer program we have checked that the coefficients derived from the recursive formulae coincide exactly with the ones given by (\ref{coefffin}) for a large number of values $r\geq 6$;

\item  Let $\{c_i,\overline{c}_i,w_i\}\in {\mathbb C}$. Let ${\widehat{g}}=a_n^{(1)} (n>1), d_n^{(1)}, e_6^{(1)},e_7^{(1)},e_8^{(1)}$. There exists an algebra homomorphism: ${\cal O}_q(\widehat{g})\ \rightarrow {\cal U}_q(\widehat{g})$ \cite{BB1} given by 
\begin{eqnarray}
A_i \mapsto {\cal A}_i = c_i\,e_iq_i^{\frac{h_i}{2}} +\overline{c}_i\,f_iq_i^{\frac{h_i}{2}} + w_i q_i^{h_i}\ \label{realg}\
\end{eqnarray}
iff the parameters $w_i$ are subject to the constraints: 
$w_i\,\Big(w_j^2+\frac{c_j\,\overline{c}_j}{q+q^{-1}-2}\Big)=0 \ ,\ w_j\,\Big(w_i^2+\frac{c_i\,\overline{c}_i}{q+q^{-1}-2}\Big)=0$ where \ $i,j$ \ are simply linked and $\rho_i\rightarrow c_i\,\overline{c}_i $\ . Let $V$ be the so-called evaluation representation  of ${\cal U}_q(\widehat{g})$ on which  ${\cal A}_i$ act (see e.g. \cite[Proposition 1]{Jim2} for $\widehat{g}=a_n^{(1)}$). For generic parameters $c_i,\overline{c}_i,q$, $V$ is irreducible and ${\cal A}_i$ is diagonalizable on $V$. Let $\theta_k^{(i)}$, $k=0,1,...$ denote the (possibly degenerate)  corresponding eigenvalues of ${\cal A}_i$. For instance, for the fundamental representation\footnote{For $\widehat{g}=a_n^{(1)},\ n>1$, see \cite[Proposition 1]{Jim2}. For $\widehat{g}=d_n^{(1)}$, see for instance \cite{DelG}.} of ${\cal U}_q(a_n^{(1)})$, the eigenvalues take the simple form:
 \begin{eqnarray}
\theta_k^{(i)} =  C^{(i)}(vq^k +  v^{-1}q^{-k})\ ,\label{eigenval}
\end{eqnarray}
where $v,C^{(i)}$ are scalar and $C^{(i)}$ depend on $c_i,\overline{c}_i,q$.
Let $E_k^{(i)}$ be the projector on the eigenspace associated with the eigenvalue $\theta_k^{(i)}$. Denote $\Delta_1^{(i)}$ as the l.h.s of
the first equation in (\ref{rel1g}). The relation (\ref{rel1g}) implies that it must exist integers  $k,l$ such that:
\begin{eqnarray}
E_k^{(i)} \Delta^{(i)}_1 E_l^{(i)} =0 \ \ \Rightarrow \ \  p^{ADE}_1(\theta_k^{(i)},\theta_l^{(i)})  E_k^{(i)} {\cal A}_j E_l^{(i)}  = 0 \quad \mbox{with} \quad \rho\equiv\rho_i \nonumber
\end{eqnarray}
For generic parameters $c_i,\overline{c}_i,q$,  $E_k^{(i)} {\cal A}_j E_l^{(i)}\neq 0$. It implies $p^{ADE}_1(\theta_k^{(i)},\theta_l^{(i)})=0$ which, using (\ref{eigenval}), is consistent with the structure (\ref{polyprod1}) for $t=0$ provided $l=k\pm 1$. The same observation about the structure of the two-variable polynomial can be generalized as follows. Denote $\Delta_r^{(i)}$ as the l.h.s of (\ref{qDGADE}). If the relation (\ref{qDGADE}) with (\ref{coefffin}) holds, then it must exist integers $k,l$ such that:
\begin{eqnarray}
E_k^{(i)} \Delta^{(i)}_r E_l^{(i)} =0 \ \ \Rightarrow \ \   p^{ADE}_r(\theta_k^{(i)},\theta_l^{(i)})  E_k^{(i)} {\cal A}^r_j E_l^{(i)}  = 0 \quad \mbox{with} \quad \rho\equiv\rho_i  \ . \nonumber
\end{eqnarray}

For generic $c_i,\overline{c}_i,q$, $E_k^{(i)} {\cal A}_j^r E_l^{(i)}\neq 0$. It implies $p^{ADE}_r(\theta_k^{(i)},\theta_l^{(i)})=0$ which leads to  the following constraints of the integers $k,l$:
\begin{eqnarray}
&&k=l\pm 1\ ,\ l\pm 3\ , \ l\pm 5\ ,\cdots \ ,\  l\pm r \qquad \mbox{for} \qquad r \ \ \mbox{odd} \ ,\nonumber\\
&&k=l\ ,\  l\pm 2\ , \ l\pm 4\ ,\cdots \ , \ l\pm r \qquad \quad \ \ \mbox{for} \qquad r \ \  \mbox{even} \ .\nonumber
\end{eqnarray}
Again, this is in perfect agreement with the factorized form (\ref{polyprod1}), (\ref{polyprod2}). Thus,  for $\widehat{g}=a_n^{(1)}$ the structure of the two-variable polynomial (\ref{defpolyADE}) is consistent with the spectral properties of ${\cal A}_i$.  
\end{itemize}

\section{The $XXZ$ open spin chain at roots of unity}
In the literature, the $XXZ$ spin chain with periodic boundary conditions at roots of unity $q=e^{i\pi/N}$, $N\in \mathbb{N}\backslash \{0\}$, is known to enjoy a $sl_2$ loop algebra symmetry in certain sectors of the spectrum \cite{DFM}. A rather natural question is whether such phenomena occurs for the open $XXZ$ spin chain and for which class of boundary conditions. \vspace{1mm}

In this Section, we describe the third main result of the thesis.
We consider the open XXZ spin chain within the framework of the $q-$Onsager algebra and its representation theory. Starting from the basic operators that generate the $q-$Onsager algebra, two new operators are introduced at roots of unity. These operators can be understood as analogues of the divided powers of the Chevalley generators (that occur in Lusztig's analysis of $U_q(\widehat{sl_2})$  at roots of unity) acting on a finite dimensional vector space. Some properties of the operators are studied.
For a special class of parameters, it is shown that the new operators satisfy a pair of relations that can be understood as a higher order generalization of the two basic (Dolan-Grady \cite{DG}) defining relations of the classical Onsager algebra. The `mixed relations' between the basic operators and the divided polynomials are also constructed. They can be seen as a different higher order generalization of the Dolan-Grady relations. All relations together provide, to our knowledge, the first example in the literature of an analog of Lusztig quantum group for the $q-$Onsager algebra. As an application in physics, we study some of the symmetries of the Hamiltonian with respect to the generators of the new algebra.

\subsection{A background: the $XXZ$ periodic spin chain at roots of unity}
Recall that the Hamiltonian of the $XXZ$ spin chain with periodic boundary conditions and $L$ sites reads:
\begin{equation}
\label{Hamilper}
H_{0}=\frac{1}{2}\sum_{j=1}^{L}{(\sigma^j_1\sigma^{j+1}_1+\sigma^j_2\sigma_2^{j+1}+\Delta\sigma^j_z\sigma_z^{j+1})},
\end{equation}
Here $\Delta = \frac{q+q^{-1}}{2}$ denotes the anisotropy parameter.  By construction, the Hamiltonian acts on a $2^L$ finite dimensional vector space:
\begin{equation}
{\mathcal V}^{(L)}=\underbrace{ {\mathbb C}^2 \otimes {\mathbb C}^2 \otimes \cdots \cdots\cdots  \otimes  {\mathbb C}^2}_{L \quad times}.\label{spaceV}
 \end{equation}

For $q$ a root of unity and $L$ finite, it is known \cite{Baxter} that additional degeneracies occur in the spectrum of the Hamiltonian. Such degeneracies are associated with the existence of an additional $sl_2$-loop algebra symmetry of the Hamiltonian, as shown in \cite{DFM}.
Let us now recall the main steps and results of \cite{DFM}.\vspace{1mm}

In the works of Jimbo \cite{Jim2},  five basic operators  satisfying the $U_q(\widehat{sl_2})$ defining relations naturally occur in the study of the $XXZ$ spin chain  for $q$ generic:
\begin{eqnarray}
S^z&=&\frac{1}{2}\sum_{j=1}^L{\sigma_z^j},\nonumber\\
S^{\pm}&=&\sum_{j=1}^{L}q^{\sigma_z/2}\otimes\dots\otimes q^{\sigma_z/2}\otimes\sigma^{\pm}_j\otimes q^{-\sigma_z/2}\otimes\dots\otimes q^{-\sigma_z/2},\nonumber\\
T^{\pm}&=&\sum_{j=1}^L{q^{-\sigma_z/2}\otimes\dots\otimes q^{-\sigma_z/2}\otimes \sigma^{\pm}_j\otimes q^{\sigma_z/2}\otimes \dots \otimes q^{\sigma_z/2}}.\nonumber
\end{eqnarray}
By straightforward calculations, $N-$th powers of the basic operators $S_\pm,T_\pm$ are shown to be proportional to $[N]_q!$. Explicitly, one derives:
\begin{eqnarray}
(S^{\pm})^N&=&[N]_q!\sum_{1\leq j_1<\dots<j_N \leq L}{q^{\frac{N}{2}\sigma_z}\otimes \dots \otimes q^{\frac{N}{2}\sigma_z}\otimes \sigma^{\pm}_{j_1}\otimes q^{\frac{(N-2)}{2}\sigma_z}\otimes \dots \otimes q^{\frac{(N-2)}{2}\sigma_z}} \\  \nonumber && \qquad  \qquad\qquad\qquad\otimes\  \sigma^{\pm}_{j_2}\otimes q^{\frac{(N-4)}{2}\sigma_z}\otimes \dots \otimes q^{\frac{(N-4)}{2}\sigma_z} \otimes \sigma^{\pm}_{j_N}\otimes q^{\frac{-N}{2}\sigma_z}\otimes \dots \otimes q^{\frac{-N}{2}\sigma_z}\ ,\label{defS}
\end{eqnarray}
\begin{eqnarray}
(T^{\pm})^N&=&[N]_q!\sum_{1\leq j_1<\dots<j_N \leq L}{q^{-\frac{N}{2}\sigma_z}\otimes \dots \otimes q^{-\frac{N}{2}\sigma_z}\otimes \sigma^{\pm}_{j_1}\otimes q^{-\frac{(N-2)}{2}\sigma_z}\otimes \dots \otimes q^{-\frac{(N-2)}{2}\sigma_z}} \\ \nonumber && \qquad  \qquad\qquad\qquad \otimes \ \sigma^{\pm}_{j_2}\otimes q^{-\frac{(N-4)}{2}\sigma_z}\otimes \dots \otimes q^{\frac{-(N-4)}{2}\sigma_z} \otimes \sigma^{\pm}_{j_N}\otimes q^{\frac{N}{2}\sigma_z}\otimes \dots \otimes q^{\frac{N}{2}\sigma_z}.\label{defT}
\end{eqnarray}
For $q^{2N}\to 1$, it implies $(S^{\pm})^N=(T^{\pm})^N=0$. 

Following Lusztig's works, the authors \cite{DFM} introduce the non-trivial divided powers:
\begin{equation}
S^{\pm(N)}= \lim_{q^{2N}\to 1}(S^{\pm})^N/[N]_q!, \quad T^{\pm(N)}=\lim_{q^{2N}\to 1}(T^{\pm})^N/[N]_q!\ \nonumber
\end{equation}
and study their commutation relations by using, for instance, the so-called Lusztig's higher order $q-$Serre relations \cite{Luszt} (\ref{hqSerre})-(\ref{hqSerre1}).

Let $\theta_1^{(m)}=(S^{+})^m/[m]_q!$, $\theta_2^{(m)}=(T^{-})^m/[m]_q!$. According to Lusztig's work \cite{Luszt},  the following relations can be derived from the higher order $q-$Serre relations:
\begin{eqnarray}
\theta_1^{(3N)} \theta_2^{(N)} &=& \sum_{s'=N}^{3N}  {\gamma}_{s'}\,  \theta_1^{(3N-s')} \theta_2^{(N)} \theta_1^{(s')} \ \label{gamma_rel}
\end{eqnarray}

where 
\begin{eqnarray}
{\gamma}_{s'}= (-1)^{s'+1}q^{s'(N-1)}\sum_{l=0}^{N-1} (-1)^l q^{l(1-s')} \left[ \begin{array}{c} s' \\ l  \end{array}\right]_q.\label{gamL}
\end{eqnarray}
Define 
\begin{eqnarray}
s'=Ns+p \qquad \mbox{with} \qquad p=0,1,...,N-1.\nonumber
\end{eqnarray}
Observe 
\begin{eqnarray}
\theta^{(Ns)}=\frac{[N]_q!^s}{[Ns]_q!}\theta^{(N)s}.\nonumber
\end{eqnarray}
Then, one shows:
\begin{eqnarray}
\lim_{q^{2N}\rightarrow 1} \gamma_{Ns+p}= \delta_{p,0} (-1)^{s+1} \qquad \mbox{and}\qquad
\lim_{q^{2N}\rightarrow 1} \frac{[N]_q!^s}{[Ns]_q!} =\frac{q^{N^2}}{s!}.\nonumber
\end{eqnarray}
It follows that $S^{+(N)},T^{-(N)}$ satisfy the Serre relations of the $sl_2$-loop algebra namely:
\begin{equation}
\sum_{s=0}^{3} \frac{(-1)^{s+1}}{s!(3-s)!} (S^{+(N)})^{(3-s)} (T^{-(N)})(S^{+(N)})^s =0
\end{equation}
Provided the change $S^+\rightarrow S^-$, $T^-\rightarrow T^+$, the same relation holds for $S^-,T^+$.

Either based on straightforward calculations or using some formula\footnote{\begin{equation}
[(S^{+})^m,(S^{-})^n]=\sum_{j=1}^{\rm{min}(m,n)}{m\atopwithdelims[] j} 
{n\atopwithdelims[] j}[j]!(S^{-})^{n-j}(S^{+})^{m-j}\prod_{k=0}^{j-1} \frac{q^{2S^z+m-n-k}-q^{-(2S^z+m-n-k)}}{q-q^{-1}}.
\end{equation}} in \cite{DeCK}, the authors \cite{DFM} derive successively the following relations for $q^{2N}=1$:
\begin{equation}
[S^{+(N)},T^{+(N)}]=[S^{-(N)},T^{-(N)}]=0,
\label{one}
\end{equation}
\begin{equation}
[S^{\pm(N)},S^z]=\pm NS^{\pm(N)},~~~[T^{\pm(N)},S^z]=\pm N T^{\pm(N)}
\label{three}
\end{equation}
and in the sector $S^z\equiv 0 ({\rm mod}~N)$
\begin{equation}
[S^{+(N)},S^{-(N)}]=[T^{+(N)},T^{-(N)}]=-(-q)^N{2\over N}S^z.
\label{two}
\end{equation}

According to the above analysis, define
\begin{eqnarray*}
E_0=S^{+(N)},~ F_0=S^{-(N)},~ E_1=T^{-(N)},~ F_1=T^{+(N)},~ T_0=-T_1=-(-q)^NS^z/N.
\end{eqnarray*}
Then,  the operators $\{E_i,F_i, T_i\}$ satisfy the defining relations of the loop algebra of $sl_2$.  When $S^{z}\equiv 0 ~({\rm mod}~N)$, the Hamiltonian (\ref{Hamilper}) commutes with $S^{\pm(N)}, T^{\pm (N)}$ at $q^{2N}=1$ \cite{DFM}. One has:
\begin{eqnarray}
\label{comm1}
[E_0,H]=[E_1,H]&=&0,\\
\label{comm2}
[F_0,H]=[F_1,H]&=&0.
\end{eqnarray}
Also, by (\ref{two}) and (\ref{comm1})-(\ref{comm2}):
\begin{equation}
[T_{0},H]=[T_1,H]=0.
\end{equation}

Hence, the Hamiltonian of the $XXZ$ periodic spin chain at $q$ a root of unity enjoys a $sl_2$-loop algebra  invariance in the sector $S^{z}\equiv 0$.

\subsection{The case of the open $XXZ$ spin chain} 
The purpose of this Section is to study the open $XXZ$ spin chain with non-diagonal integrable boundary conditions for an anisotropy parameter $(q+q^{-1})/2$ evaluated at roots of unity $q=e^{i\pi/N}$. Inspired by the analysis done for the $XXZ$ spin chain with periodic boundary conditions using the quantum algebra $U_q(\widehat{sl_2})$ \cite{DFM}, it is thus natural to start from the framework of the $q-$Onsager algebra, a coideal subalgebra of $U_q(\widehat{sl_2})$ (see Chapter 1). For generic boundary conditions, recall that the Hamiltonian of the $XXZ$ open spin chain reads
\begin{eqnarray}
\label{equhamilt}
H^{(L)}_{XXZ}&=&\sum_{k=1}^{L-1}\Big(\sigma_1^{k+1}\sigma_1^{k}+\sigma_2^{k+1}\sigma_2^{k} + \Delta\sigma_z^{k+1}\sigma_z^{k}\Big)\\&& +\ \frac{(q-q^{-1})}{2}\frac{(\epsilon_+ - \epsilon_-)}{(\epsilon_+ + \epsilon_-)}\sigma^1_z + \frac{2}{(\epsilon_+ + \epsilon_-)}\big(k_+\sigma^1_+ + k_-\sigma^1_-\big)     \nonumber\\
\ && +\ \frac{(q-q^{-1})}{2}\frac{(\bar{\epsilon}_+ - \bar{\epsilon}_-)}{(\bar{\epsilon}_+ +\bar{\epsilon}_-)}\sigma^L_z + \frac{2}{(\bar{\epsilon}_+ + \bar{\epsilon}_-)}\big(\bar{k}_+\sigma^L_+ + \bar{k}_-\sigma^L_-\big)     \nonumber
\end{eqnarray}
where $L$ is the number of sites, $\Delta = \frac{q+q^{-1}}{2}$ denotes the anisotropy parameter, and $\sigma_{\pm},\sigma_1, \sigma_2, \sigma_z$ are the usual Pauli matrices.\vspace{1mm} 

\subsubsection{The basic operators and the divided polynomials}
Following  \cite{BK1}, the two operators $\mathcal{W}_0, \mathcal{W}_1$ of the $q$-Onsager algebra that naturally occur in the analysis of the $XXZ$ open spin chain with non-diagonal boundary conditions are  known explicitly (see Chapter 2). Denote:
\begin{equation}
w_0=k_+ \sigma_+ + k_-\sigma_-.
\end{equation}
They are given by:
\begin{eqnarray}
\mathcal W_0&=&\sum_{j=1}^L{q^{\sigma_z}\otimes \dots \otimes q^{\sigma_z}\otimes w_{0_j}\otimes \mathbb{I}\otimes \dots \otimes \mathbb{I}}+\epsilon_+ q^{\sigma_z} \otimes \dots \otimes q^{\sigma_z},\label{TDbase}\\
\mathcal W_1&=&\sum_{j=1}^L{q^{-\sigma_z}\otimes \dots \otimes q^{-\sigma_z}\otimes w_{0_j}\otimes \mathbb{I}\otimes \dots \otimes \mathbb{I}}+\epsilon_- q^{-\sigma_z} \otimes \dots \otimes q^{-\sigma_z}.\label{TDbase1}
\end{eqnarray}
with the parameters $\rho = \rho^* = (q+q^{-1})^2k_+k_-$.

The finite dimensional module ${\cal V}^{(L)}$ on which they act is of dimension $2^L$. Indeed, both operators are diagonalizable on the finite $2^L-$ dimensional vector space ${\cal V}^{(L)}$ \cite{Bas3}. For convenience, denote $\epsilon_0=\epsilon_+$,  $\epsilon_1=\epsilon_-$.
\begin{lem}\cite{Bas3}\label{lemspecfund} For generic values of $q$, the operator $\mathcal W_i$, $i=0,1$, has $L+1$ distinct eigenvalues. They read:
\begin{equation}
\theta^{(i)}_n= a_2^{(i)}(q) q^{L-2n} + a_3^{(i)}(q)q^{-L+2n} \qquad \mbox{for} \quad  n=0,1,..., L \ .\label{specgen}
\end{equation}
 
If $k_{\pm}\ne 0$ and use the convenient parametrization
\begin{equation}
\epsilon_i=\cosh(\alpha_i) \quad \mbox{and} \quad k_+=(k_-)^\dagger=-(q-q^{-1})e^{i\eta}/2 \quad \mbox{with} \quad \alpha_i,\eta\in{\mathbb C} \label{paramet}
\end{equation}
then $~a_2^{(i)}(q) = \frac{e^{\alpha_i}}{2}, ~ a_3^{(i)}(q) = \frac{e^{-\alpha_i}}{2}$.

If $k_-=0$, then $a_2^{(i)}(q) = \epsilon_i, ~ a_3^{(i)}(q) = 0$.
\end{lem}


Let $V_0, V_1,\ldots, V_L$ denote the  
eigenspaces of $\mathcal W_0$ with eigenvalues $\theta^{(0)}_0, \theta^{(0)}_1,\dots, \theta^{(0)}_L$ given in Lemma \ref{lemspecfund}. On $\{V_n\}$, the basic operator $\mathcal W_1$ acts as:
\begin{equation}
\mathcal W_1 V_n \subseteq V_{n-1} + V_n+ V_{n+1} \qquad \qquad (0 \leq n \leq L),
\label{eq:tdrecall1}
\end{equation}
where $V_{-1} = 0$, $V_{L+1}= 0$. On the other hand,  let $V^*_0, V^*_1,\ldots, V^*_L$ denote the  eigenspaces of $\mathcal W_1$ with eigenvalues $\theta^{(1)}_0,\theta^{(1)}_1,\dots,\theta^{(1)}_L$. The basic operator $\mathcal W_0$ acts as:
\begin{equation}
\mathcal W_0 V^*_s \subseteq V^*_{s-1} + V^*_s+ V^*_{s+1} \qquad \qquad (0 \leq s \leq L),
\label{eq:tdrecall2}
\end{equation}
where $V^*_{-1} = 0$, $V^*_{L+1}= 0$. Note that the eigenbasis in which $\mathcal W_0$ (resp. $\mathcal W_1$) is diagonalized has been constructed and described in details in \cite{Bas3}, as well as the entries of the block tridiagonal matrix $\mathcal W_1$ (resp. $\mathcal W_0$) in the same basis. \\

Note that $\text{Trace}(\mathcal{W}_i)=\sum\limits_{n=0}^L{\text{dim}\left(V_n\right)\left(a_2^{(i)}(q)q^{L-2n}+a_3^{(i)}(q)q^{-L+2n}\right)}$, and $\text{Trace}(\mathcal{W}_i)=\epsilon_i(q+q^{-1})^L$, then $\text{dim}(V_n) = \left(\begin{array}{c}
L\\
n
\end{array}\right)$, and $a_2^{(i)}(q)+a_3^{(i)}(q)=\epsilon_i$ if $q^2 \ne -1$. In addition, $a_2^{(i)}(q)a_3^{(i)}(q)=-\frac{k_+k_-}{(q-q^{-1})^2}$. Therefore, if $q^2 \ne -1$, then
\begin{eqnarray*}
a_2^{(i)}(q), a_3^{(i)}(q) \in \left\{\frac{1}{2}\left(\epsilon_i+\sqrt{\epsilon_i^2+\frac{4k_+k_-}{(q-q^{-1})^2}}\right),\frac{1}{2}\left(\epsilon_i-\sqrt{\epsilon_i^2+\frac{4k_+k_-}{(q-q^{-1})^2}}\right)  \right\}.
\end{eqnarray*}
Up to now, the parameter $q$ is generic. Let us now consider the case $q= e^{i\pi/N}$.  Starting from the known results for $q$ generic, for $q=e^{i\pi/N}$ one observes that additional degeneracies in the spectra of   $\mathcal W_0, \mathcal W_1$ occur.
\begin{lem}\label{lemm1} For $q=e^{i\pi/N}, N>2$, the operator $\mathcal W_i$, $i=0,1$, has $N$ distinct eigenvalues. They read:
\begin{equation}
\theta^{(i)}_t= a_2^{(i)}(q) q^{L-2t} + a_3^{(i)}(q)q^{-L+2t} \qquad \mbox{for}  \quad t=0,1,...,N-1.\label{specfund}
\end{equation}

If $k_{\pm}\ne 0$ and use the parametrization (\ref{paramet}), then $~a_2^{(i)}(q) = \frac{e^{\alpha_i}}{2}, ~ a_3^{(i)}(q) = \frac{e^{-\alpha_i}}{2}$.

If $k_-=0$, then $a_2^{(i)}(q) = \epsilon_i, ~ a_3^{(i)}(q) = 0$.

Furthermore, let $V_t^{(N)}$ (resp. $V_t^{*(N)}$) denote the eigenspace associated with $\theta^{(0)}_t$ (resp. $\theta^{(1)}_t$), one has the decomposition:
\begin{equation}
V_t^{(N)}= \bigoplus_{k=0}^{\{\frac{L}{N}\}} V_{t+kN} \qquad \mbox{and} \qquad V_t^{*(N)}= \bigoplus_{k=0}^{\{\frac{L}{N}\}} V^{*}_{t+kN}.\label{esp}
\end{equation}
\end{lem}
\begin{proof} For generic parameters $\alpha_i,\eta$, the eigenvectors with explicit formulae given in \cite{Bas3} remain well-defined and linearly independent for $q=e^{i\pi/N}$. Given $n$, there exist uniquely integers $t,k$ such that $n=t+kN$ with $0\leq n \leq L$, and $0\le t \le N-1$. From (\ref{specgen}), one immediately finds (\ref{specfund}).
We now show (\ref{esp}). Let $t$ be fixed, the eigenspace denoted $V^{(N)}_t$ (resp. $V^{*(N)}_t$ ) is spanned by the eigenvectors of $V_t,V_{t+N},...$ (resp. $V^*_t,V^*_{t+N},...$) evaluated at $q=e^{i\pi/N}$.  
\end{proof}

The dimension of the eigenspaces $V_t^{(N)}$, $V^{*(N)}_t$ is computed based on the above decomposition. It follows:
\begin{equation}
dim(V_t^{(N)})=dim({V_t^*}^{(N)})=\sum\limits_{k=0}^{\{\frac{L}{N}\}} \left(\begin{array}{c} L\\ t+kN \end{array}\right).
\end{equation}

From above results,  using (\ref{eq:tdrecall1}), (\ref{eq:tdrecall2}) it follows:
\begin{cor}\label{cor1} Let $V_t^{(N)},V_t^{*(N)}$ be defined as in Lemma \ref{lemm1}. 
\begin{eqnarray}
&&\mathcal W_0 V_t^{(N)} \subseteq V_{t}^{(N)},\label{tdroot1}\\
&&\mathcal W_1 V_t^{(N)} \subseteq V_{t+1}^{(N)} + V_t^{(N)} + V_{t-1}^{(N)}  \qquad \qquad (0 \leq t \leq N-1),\nonumber
\end{eqnarray}
where $V_{-1}^{(N)} = V_{N-1}^{(N)}$ and $V_{N}^{(N)} = V_{0}^{(N)}$
and
\begin{eqnarray}
&&\mathcal W_1 V_t^{*(N)} \subseteq V_{t}^{*(N)},\label{tdroot2}\\
&&\mathcal W_0 V_t^{*(N)} \subseteq V_{t+1}^{*(N)} + V_t^{*(N)} + V_{t-1}^{*(N)}  \qquad \qquad (0 \leq t \leq N-1),\nonumber
\end{eqnarray}
where $V_{-1}^{*(N)} = V_{N-1}^{*(N)}$ and $V_{N}^{*(N)} = V_{0}^{*(N)}$.
\end{cor}

For $q$ a root of unity, besides the two fundamental operators of the $q-$Onsager algebra, two additional operators that are called the {\it divided polynomials} can be introduced\footnote{These operators can be seen as analogs of the divided powers of the Chevalley elements that generate the so-called {\it full} or {\it Lusztig} quantum group $U_q(\widehat{sl_2})$ with $q=e^{i\pi/N}$, acting on ${\cal V}^{(L)}$. } for the analysis of the spin chain. In particular, the spectrum of these  operators is described in details below. These new operators are introduced as follows. Recall that
 the fundamental generators $S^\pm,T^\pm$ of the quantum loop algebra of $sl_2$ become nilpotent, for $q$ a root of unity. Namely, $(S^\pm)^N=(T^\pm)^N=0$ for $q^{2N}=1$ \cite{DFM}. For non-vanishing values of the non-diagonal boundary parameters $k_\pm$ entering in (\ref{TDbase})-(\ref{TDbase1}), such property doesn't hold for the  fundamental generators $\mathcal W_0,\mathcal W_1$ of the $q-$Onsager algebra. Instead, given $N$, certain polynomials of $\mathcal W_0,\mathcal W_1$ are vanishing for $q^{2N}=1$. These polynomials can be constructed as follows. 

First, for small values of $N$, by straightforward calculations  using the explicit expressions of $\mathcal W_0$ and  $\mathcal W_1$ (see Appendix F)
it is not difficult to construct polynomials such that:
\begin{equation}
P^{(L,i)}_N(\mathcal W_i)=0, \quad i=0,1.\label{relP}
\end{equation}
For instance\footnote{For $N=1$, $q^2=1$ there is no polynomial. In this special case, the operators $\mathcal W_0,\mathcal W_1$ however reduce to elements that generate the undeformed 
Onsager algebra. }:   

\begin{itemize}

\item For $N=2$, $q^4=1$,

\begin{equation}
P^{(L,i)}_2(x) = x^2 + (-1)^{L+1}\left(\epsilon_i^2 - \frac{(1+(-1)^{L+1})}{2}k_+k_-\right) \mathbb{I}.\nonumber
\end{equation}

\vspace{2mm}

\item For $N=3$, $q^6=1$,

\begin{equation}
P^{(L,i)}_3(x) = x^3 -  k_+k_- x  + (-1)^{L+1}\epsilon_i(\epsilon_i^2-k_+k_-)\mathbb{I} \nonumber.
\end{equation}

\vspace{2mm}

\item For $N=4$, $q^8=1$,

\begin{equation}
P^{(L,i)}_4(x) = x^4 - 2k_+k_-x^2 + (-1)^{L+1}\left(\epsilon_i^4-2\epsilon_i^2k_+k_-+ \frac{(1+(-1)^{L+1})}{2}k_+^2k_-^2\right)\mathbb{I} \nonumber.
\end{equation}

\vspace{2mm}

\item For $N=5$, $q^{10}=1$, recall $\rho=k_+k_-(q+q^{-1})^2$,

\begin{eqnarray*}
P^{(L,i)}_5(x) &=& x^5  - \rho (q^4+q^{-4}+3) x^3    +   \rho^2 (q^2+q^{-2})^2 x\\ && +(-1)^{L+1}\frac{(\sqrt{5}+3)}{4}\epsilon_i(\epsilon_i^2-k_+k_-)(3\epsilon_i^2-\sqrt{5}\epsilon_i^2-2k_+k_-)\mathbb{I}.
\nonumber
\end{eqnarray*}

\vspace{2mm}

\item For $N=6$, $q^{12}=1$,

\begin{eqnarray*}
P^{(L,i)}_6(x) &=& x^6 -6k_+k_-x^4 +9k_+^2k_-^2x^2\\&& +(-1)^{L+1}\left(\epsilon_i^6 -6k_+k_-\epsilon_i^4 +9k_+^2k_-^2\epsilon_i^2 -2(1+(-1)^{L+1})k_+^3k_-^3\right)\mathbb{I}.\nonumber 
\end{eqnarray*}
\end{itemize}
Recall that $\mathcal W_i$, $i=0,1$ are both diagonalizable on ${\cal V}^{(L)}$. As above mentioned, for $q=e^{i\pi/N}$, there are exactly $N$ distinct eigenvalues denoted $\theta^{(i)}_n$, $n=0,1,...,N-1$. Then, observe that the above polynomials $P^{(L,i)}_N(x)$ for $N=2,3,...,6$ are nothing but the minimal polynomials associated with $\mathcal W_i$, $i=0,1$. More generally, it follows:

\begin{prop}
Let $e_k(\theta_0,\cdots,\theta_{N-1})$, $k=0,1,...,N-1$ denote the so-called elementary symmetric polynomials in the variables $\theta_n$, given by:
\begin{equation}
e_k(\theta_0,\cdots,\theta_{N-1})=\!\!\!\!\!\sum_{0\leq j_1< j_2< \cdots < j_k\leq N-1}\!\!\!\!\!\!\!\!\!\!\!\!\! \theta_{j_1}\theta_{j_2}\cdots  \theta_{j_{k}} \quad \mbox{with} \quad e_0(\theta_0,\cdots,\theta_{N-1})=1.
\end{equation}

For $N\geq 2$ generic:
\begin{equation}
P^{(L,i)}_N(x) = \sum_{k=0}^N e_k(\theta^{(i)}_0,\cdots,\theta^{(i)}_{N-1})x^{N-k}\ . \label{polymin}
\end{equation}
\end{prop}

\begin{proof}
Recall  that the distinct eigenvalues are given by $\theta^{(i)}_0,...,\theta^{(i)}_{N-1}$ in (\ref{specfund}). Introduce the minimal polynomials of $\{\mathcal W_i\}_{i=0,1}$:
\begin{equation}
P^{(L,i)}_N(x)=\prod_{n=0}^{N-1}(x-\theta^{(i)}_n).
\end{equation}
According to the Cayley-Hamilton theorem, one has $P^{(L,i)}_N(\mathcal W_i)=0$. Expanding the factorized expression above in $x$, one obtains (\ref{polymin}).
\end{proof}

For $q$ a root of unity, recall that the {\it small} or {\it restricted} quantum loop algebra of $sl_2$ can be defined: it is the quotient of the  quantum loop algebra of $sl_2$ by the relations $(S^\pm)^N=(T^\pm)^N=0$ and $S^z=0$ mod $N$ for $q^{2N}=1$. Then, the {\it full} or {\it Lusztig} quantum group can be introduced. It is generated by the fundamental Chevalley elements and additionally by the divided powers  $(S^\pm)^N/[N]_q!$, $(T^\pm)^N/[N]_q!$. The corresponding defining relations can be found in \cite{DFM}. By analogy,  let us introduce the divided polynomials:
%
\begin{equation}
{\stackrel{\footnotesize[N]}{\mathcal W}}_i = \mbox{lim}_{q\rightarrow e^\frac{I\pi}{N}} \ \frac{P_N^{(L,i)}(\mathcal W_i)}{[N]_q!} \qquad \mbox{with} \qquad i=0,1.\label{WNnew}
\end{equation}
Note that using the formulae of Appendix F, it is possible to derive the generic expression of ${\stackrel{\footnotesize[N]}{\mathcal W}}_i$, $i=0,1$ for arbitrary $N$.\vspace{2mm} 

According to the spectral properties of the fundamental operators $\mathcal W_0,\mathcal W_1$, the eigenvalues of the divided polynomials (\ref{WNnew}) can be easily obtained.
\begin{lem} Let $N >2$ denote an integer. Let $\tilde{\theta}^{(i)}_n$, $n=0,1,...,L$ denote the eigenvalues of the  divided polynomials  ${\stackrel{\footnotesize[N]}{\mathcal W}}_i$, $i=0,1$. For all $n=0,1,...,L$ there exist uniquely $k,t$ non-negative integers such that $n=t+kN$, $0\leq t \leq N-1$. For $n=0,1,...,L$, the eigenvalues read:
\begin{equation}
\tilde{\theta}^{(i)}_n = C_0^{(n,i)}\left({\partial_q(a_2^{(i)})}(q^{L-2n+1}-q^{-L+2n+1})+ (L-2n)(a^{(i)}_2(q) q^{L-2n} - a^{(i)}_3(q) q^{-L+2n})\right)   \label{spectnew}
\end{equation}
where $q=e^{I\pi/N}$ and
\begin{equation}
C_0^{(n,i)}=-\frac{(q-q^{-1})}{2N[N-1]_q!}  \prod_{j=0,j\neq t}^{N-1} (\theta^{(i)}_t-\theta^{(i)}_j).\nonumber
\end{equation}
The eigenspace of ${\stackrel{\footnotesize[N]}{\mathcal W}}_0$ (resp. ${\stackrel{\footnotesize[N]}{\mathcal W}}_1$) corresponding to $\tilde{\theta}^{(0)}_n$ (resp. $\tilde{\theta}^{(1)}_n$) is $V_n$ (resp. $V_n^*$) at $q=e^{I\pi/N}$.
\end{lem}

\begin{proof}
By Lemma \ref{lemspecfund}, recall that $\mathcal{W}_i, i = 0,1$ has $L+1$ distinct eigenvalues $\theta_0^{(i)}(q)$,  $\theta_1^{(i)}(q)$,$\dots$, $\theta_L^{(i)}(q)$, and $L+1$ eigenspaces $V_0^{(i)}(q),V_1^{(i)}(q),\dots, V_L^{(i)}(q)$, respectively. Here the notations $\theta_n^{(i)}(q), V_n^{(i)}(q), n=0,\dots,L$  are used in order to emphasize that the eigenvalues and eigenspaces of $\mathcal{W}_i$ depend on $q$. Since $\mathcal{W}_0, \mathcal{W}_1$ is a tridiagonal pair, there exist scalars $a_2^{(i)}\ne 0, a_3^{(i)} \ne 0$ such that 
\begin{equation}
\theta_n^{(i)}(q)=a_2^{(i)}(q)q^{L-2n}+a_3^{(i)}(q)q^{-L+2n}, \qquad n=0,\dots,L.\nonumber
\end{equation}
Clearly, if $q=e^{\frac{I\pi}{N}}$, then $\mathcal{W}_i$ has $N$ distinct eigenvalues 

\begin{equation}
\theta_0^{(i)}=\mathop {\lim }\limits_{q\to e^{\frac{I\pi}{N}} }\theta_0^{(i)}(q),~~\theta_1^{(i)}=\mathop {\lim }\limits_{q\to e^{\frac{I\pi}{N}}}\theta_1^{(i)}(q),\dots,~~ \theta_{N-1}^{(i)}=\mathop {\lim }\limits_{q\to e^{\frac{I\pi}{N}} }\theta_{N-1}^{(i)}(q) .\nonumber
\end{equation}
Now, define (\ref{WNnew}).
%
%
Then, $\stackrel{[N]}{\mathcal{W}_i}$ has $L+1$ eigenvalues denoted
$\widetilde{{\theta }}_n^{(i)}$
%
%
and  $V_n = \mathop{\lim}\limits_{q\to e^{\frac{I\pi}{N}}} V_n^{(i)}(q)$ denote the corresponding eigenspaces. Explicitely, $\widetilde{\theta}_n^{(i)}$ can be written as follows:
\[\widetilde{{\theta}}_{n}^{(i)}=\mathop{\lim}\limits_{q\to e^{\frac{I\pi}{N}}}\frac{(\theta_n^{(i)}(q)-\theta_0^{(i)})(\theta_n^{(i)}(q)-\theta_1^{(i)})\dots(\theta_n^{(i)}(q)-\theta_{N-1}^{(i)})}{[N]_q!},\qquad n =0,\dots, L. \]
For all $n = 0,\dots, L$, there exist uniquely $k, t$ non-negative integers, such that $n = kN +t,~~ 0\le t\le N-1$. It yields to:
\begin{equation*}
\label{eq1}
\widetilde{\theta}_n^{(i)}=\frac{(q-q^{-1})\prod\limits_{j=0,j\ne t}^{N-1}{(\theta_t^{(i)}-\theta_j^{(i)})}}{[N-1]_q!}\mathop{\lim}\limits_{q\to e^{\frac{I\pi}{N}}}\frac{\theta_n^{(i)}(q)-\theta_t^{(i)}}{q^N-q^{-N}}.
\end{equation*}
Using the L'Hopital's rule, we have
\begin{eqnarray*}
&&\mathop{\lim}\limits_{q\to e^{\frac{I\pi}{N}}}\frac{\theta_n^{(i)}(q)-\theta_t^{(i)}}{q^N-q^{-N}}=\mathop{\lim}\limits_{q\to e^{\frac{I\pi}{N}}}\frac{(\theta_n^{(i)}(q))^\prime}{-2Nq^{-1}}\\
&&~~~~= \left.\frac{\partial_q{(a_2^{(i)})} \left(q^{L-2n+1}-q^{-L+2n+1}\right)+(L-2n)\left(a_2^{(i)}(q)q^{L-2n}-a_3^{(i)}(q)q^{-L+2n}\right)}{-2N}\right|_{q=e^{I\pi/N}}.
\end{eqnarray*}

It simplifies to
\begin{equation*}
\tilde{\theta}^{(i)}_n = C_0^{(n,i)}\left(\partial_q{(a_2^{(i)})}(q^{L-2n+1}-q^{-L+2n+1})+ (L-2n)(a^{(i)}_2(q) q^{L-2n} - a^{(i)}_3(q) q^{-L+2n})\right)   \end{equation*}
where $q=e^{I\pi/N}$ and
\begin{equation*}
C_0^{(n,i)}=-\frac{(q-q^{-1})}{2N[N-1]_q!}  \prod_{j=0,j\neq t}^{N-1} (\theta^{(i)}_t-\theta^{(i)}_j).\nonumber
\end{equation*}
\end{proof}
\begin{rem}
If $N=2$, then $~\tilde{\theta}^{(i)}_n=(-1)^L(L-2n)a~$ where $~a=\epsilon_i\sqrt{k_+k_--\epsilon_i^2}$.
\end{rem}
\begin{proof}
For $q=e^{I\pi/2}$, $\mathcal{W}_i$ has the two eigenvalues $\theta_0^{(i)}, \theta_1^{(i)}$ such that
\begin{eqnarray*}
\theta_0^{(i)} =-\theta_1^{(i)}\in \left\{\sqrt{(-1)^L\left(\epsilon_i^2-\frac{1+(-1)^{L+1}}{2}k_+k_-\right)}, -\sqrt{(-1)^L\left(\epsilon_i^2-\frac{1+(-1)^{L+1}}{2}k_+k_-\right)}\right\}.
\end{eqnarray*}
We compute
\begin{eqnarray*}
\tilde{\theta}_n^{(i)}&=&\mathop{\lim}\limits_{q\to e^{\frac{I\pi}{2}}}\frac{\left(a_2^{(i)}(q)q^{L-2n}+a_3^{(i)}(q) q^{-L+2n}\right)^2-{\theta_0^{(i)}}^2}{q+q^{-1}}\\
&=&\left.(L-2n)\left(a_2^{(i)}(q)q^{L-2n}+a_3^{(i)}(q) q^{-L+2n}\right)\left(a_2^{(i)}(q)q^{L-2n-1}-a_3^{(i)}(q) q^{-L+2n-1}\right)\right|_{q=e^{I\pi/2}}\\
&=&(-1)^L(L-2n)\epsilon_i\sqrt{k_+k_--\epsilon_i^2}
\end{eqnarray*}
\end{proof}

\begin{rem} For $k_{\pm}\ne 0$ and the parametrization (\ref{paramet}), one has the identification $a_2^{(i)}(q) = \frac{e^{\alpha_i}}{2},\ a_3^{(i)}(q)=\frac{e^{-\alpha_i}}{2}$.
For $k_-=0$, one has $a_2^{(i)}(q) = \epsilon_i, ~ a_3^{(i)}(q) =0$
\end{rem}

Using the explicit expression (\ref{WNnew}) of the divided polynomials $\cWN_0,\cWN_1$ in terms of the basic operators $\mathcal{W}_0, \mathcal{W}_1$, the action of the basic operators in each eigenbasis  implies:

\begin{prop}\label{actWN} On the eigenspaces $V_n$ (resp. $V^*_s$), the divided polynomials ${\stackrel{\footnotesize[N]}{\mathcal W}}_0,{\stackrel{\footnotesize[N]}{\mathcal W}}_1$ act as
\begin{eqnarray}
&&{\stackrel{\footnotesize[N]}{\mathcal W}}_0 V_n \subseteq V_n\ ,\label{Nstruct}\\ 
&&{\stackrel{\footnotesize[N]}{\mathcal W}}_1 V_n \subseteq V_{n-N} +  \cdots + V_n+  \cdots + V_{n+N} \qquad \qquad (0 \leq n \leq L),\nonumber\\
&&{\stackrel{\footnotesize[N]}{\mathcal W}}_1 V^*_s \subseteq V^*_{s}\ ,\nonumber\\ 
&&{\stackrel{\footnotesize[N]}{\mathcal W}}_0 V^*_s \subseteq V^*_{s-N} + \cdots + V^*_s+  \cdots + V^*_{s+N} \qquad \qquad (0 \leq s \leq L).\nonumber
\end{eqnarray}
where $V_{-1}=\dots=V_{-N}=0, V_{L+1}=\dots=V_{L+N}=0,$ and $V^*_{-1}=\dots=V^*_{-N}=0, V^*_{L+1}=\dots=V^*_{L+N}=0$.
\end{prop}
Explicitly, in the basis which diagonalizes $\mathcal W_0$ (resp. $\mathcal W_1$) , the matrix representing  ${\stackrel{\footnotesize[N]}{\mathcal W}}_1$ (resp. $\mathcal W_0$) is  $(2N+1)$-block diagonal.

\subsubsection{The algebra generated by $\cW_0,\cW_1,\cWN_0,\cWN_1$ for $k_-=0$ or $k_+=0$.}
Let us focus on the special class of parameters $k_+=0$ or $k_-=0$. In this case, the spectrum of the tridiagonal pair ${\cal W}_0,{\cal W}_1$ simplifies: $a^{(i)}_2(q)=\epsilon_i$, $a^{(i)}_3(q)=0$. This has the following consequence on the spectrum of the divided polynomials
$\cWN_0,\cWN_1$:
\begin{lem} Let $\tilde{\theta}^{(i)}_n$, $n=0,1,...,L$ denote the distinct eigenvalues of the  divided polynomials  ${\cWN}_i$, $i=0,1$. 
 Assume $a_2^{(i)}(q)=\epsilon_i$ and $a_3^{(i)}(q)=0$.  Then 
\begin{eqnarray}
\tilde{\theta}^{(i)}_n = \epsilon_i^N\frac{(q-q^{-1})^N}{2N}(-1)^{\frac{N+1}{2}+L}(L-2n).
\quad n=0,1,...,L. \label{spectnewred}
\end{eqnarray}
\end{lem}

For this special choice of parameters, note that the spectrum of the divided polynomials has an arithmetic progression with respect to the integer $n=0,1,2,...,L$. In this special case, it is not difficult to show that the divided polynomials satisfy a pair of polynomial relations. Indeed, define
 \begin{defn}
 Let $x,y$ denote commuting indeterminates. For each positive integer $N$, we define the two two-variable polynomials $p^{(i)}_{N}(x,y)$, $i=0,1$, as follows:
 \begin{eqnarray}
 \qquad p^{(i)}_{N}(x,y) = (x-y)\prod_{s=1}^{N} (x^2- 2xy +y^2 -\rho_is^2)\label{defpoly}
 \end{eqnarray}
 with
 \begin{eqnarray}
 \rho_i=(-1)^{N+1}\epsilon_i^{2N}\frac{(q-q^{-1})^{2N}}{N^2}, \quad q=e^{i\pi/N}.\nonumber
 \end{eqnarray}
 Observe $p^{(i)}_{N}(x,y)$ have a total degree $2N+1$ in $x,y$.
 \end{defn}

 \begin{lem}\label{polynul} Assume $a_2^{(i)}=\epsilon_i$ and $a_3^{(i)}=0$. For any $i=0,1$, one has:
 \begin{eqnarray}
 p^{(i)}_{N}(\tilde{\theta}_r^{(i)},\tilde{\theta}_s^{(i)})=0 \qquad \mbox{for} \qquad |r-s|\leq N\ .\nonumber
 \end{eqnarray}
 \end{lem}
\begin{proof} For $r=s$, the relation obvioulsy holds. Then, suppose $|r-s|=k$ with $k\leq N$. Due to the arithmetic progression of $\tilde{\theta}_s$ for $a_2=\epsilon_i$ and $a_3=0$, one has $s^2 -2r(s\pm k) + (s\pm k)^2 =k^2$. The relation then follows.
\end{proof}

\vspace{5mm}

Let us now introduce a new algebra. The defining relations of the new algebra can be understood as a generalization of the Dolan-Grady relations (\ref{qDG}) which are recovered for $N=1$. For this reason, we decide to call this new algebra the {\it $N-th$ higher order generalization of the Onsager algebra}.

\begin{prop}\label{pro3} Let $k_-=0$ or $k_+=0$. Then, the divided polynomials
${\stackrel{\footnotesize[N]}{\mathcal W}}_i$, $i=0,1$ satisfy the defining relations of the $N-th$ higher order generalization of the Onsager algebra:
\begin{eqnarray}
\sum_{p=0}^{N}\  \sum_{k=0}^{2N+1 -2p}  (-1)^{k+p}  \rho_0^{p}\, {c}_{k}^{[N,p]}\, \big(\cWN_0\big)^{2N+1-2p-k} \cWN_1 \big(\cWN_0\big)^{k}&=&0,\label{N-DG} \\ 
\sum_{p=0}^{N}\  \sum_{k=0}^{2N+1 -2p}  (-1)^{k+p}  \rho_1^{p}\, {c}_{k}^{[N,p]}\, \big(\cWN_1\big)^{2N+1-2p-k} \cWN_0 \big(\cWN_1\big)^{k}&=&0,\nonumber
\end{eqnarray}
where $c_{2(N-p)+1-k}^{[N,p]} =c_{k}^{[N,p]}$ and
\begin{eqnarray}
c^{[N,p]}_k=\left(\begin{array}{cc}
 2N+1-2p\\
 k
 \end{array}\right)\sum\limits_{1\le s_1<\dots <s_p\le N}{s_1^2s_2^2\dots s_p^2}
\end{eqnarray}
with $k=\overline{0,N-p}$. Also, one has:
\begin{eqnarray}
\rho_i=\epsilon_i^{2N}(-1)^{N+1}\frac{(q-q^{-1})^{2N}}{N^2},\quad q=e^{i\pi/N}.\nonumber
\end{eqnarray}
\end{prop}
\begin{proof} Let $\Delta^{[N]}$  denote the expression of the left-hand side of the first equation (\ref{N-DG}). We
show $\Delta^{[N]}\equiv 0$. Let $E_n$ denote the projector on the eigenspace $V_n$ associated with the eigenvalue $\tilde{\theta}^{(0)}_n$. For $0\leq n \leq m \leq N$, observe $E_n\Delta^{[N]} E_m = p^{(0)}_{N}(\tilde{\theta}^{(0)}_n, \tilde{\theta}^{(0)}_m) E_n \cWN_1 E_m$ where the two-variable polynomial (\ref{defpoly}) has been expanded in a straigthforward manner as power series in the variable $x,y$. According to Proposition \ref{actWN}, one has $E_n \cWN_1 E_m=0$ for $|n-m|>N$. Using Lemma \ref{polynul}, it follows  $\Delta^{[N]}=0$. Similar arguments are used to show the second relation in (\ref{N-DG}) 
\end{proof}

Note that the defining relations above do {\it not} coincide with the higher order $q-$Dolan-Grady relations introduced in the previous Section, except for $N=r=1$.\vspace{1mm}

It remains to consider the mixed relations between the fundamental operators and the divided polynomials. The derivation of the relations below follows the same steps as the ones for the derivation of the higher $q-$Dolan-Grady relations. For this reason, we skip the details.
\begin{prop}
Let $k_-=0$ or $k_+=0$. Then, the operators $\mathcal{W}_i$ and the divided polynomials $\cWN_i$, $i =0, 1$ satisfy the mixed relations:
\begin{eqnarray}
\sum_{k=0}^{2N-1}(-1)^{k}\left[\begin{array}{cc}
2N-1\\
k
\end{array}\right]_{q=e^{i\pi/N}}\mathcal{W}_i^{2N-1-k}\stackrel{[N]}{\mathcal{W}}_{i+1}\mathcal{W}_i^k &=&0 ,\label{N-mixedDG}\\
\left(\stackrel{[N]}{\mathcal{W}}_i\right)^3\mathcal{W}_{i+1}-3\left(\stackrel{[N]}{\mathcal{W}}_i\right)^2\mathcal{W}_{i+1}\stackrel{[N]}{\mathcal{W}}_i+3\stackrel{[N]}{\mathcal{W}}_i\mathcal{W}_{i+1}\left(\stackrel{[N]}{\mathcal{W}}_i\right)^2-\mathcal{W}_i\left(\stackrel{[N]}{\mathcal{W}}_i\right)^3&=& \rho_i \big[\stackrel{[N]}{\mathcal{W}}_i,\mathcal{W}_{i+1}\big],\nonumber
\end{eqnarray}
where $\rho_i$ is defined in Proposition \ref{pro3}
\end{prop}

To resume, for $k_-=0$ or $k_+=0$ the algebra generated by the basic operators $\mathcal{W}_i$  of the $q-$Onsager algebra and the divided polynomials $\cWN_i$ has defining relations (\ref{qDGXXZ}), (\ref{N-DG}) and (\ref{N-mixedDG}).

\subsubsection{The algebra generated by $\cW_0,\cW_1,\NWex{2}_0,\NWex{2}_1$ for generic parameters.}
For the special case $q=i$, the algebra generated by the basic operators together with the divided polynomials can be easily identified. The generators  $\W_i$ satisfy obviously the defining relations of the $q-$Onsager algebra at $q=e^{i\pi/2}$ or, because $\rho=0$, the $q$-Serre relations at $q^2=-1$ together with
the ``nilpotency'' relations
\begin{eqnarray}
\label{nilp-rel}
\W_i^2 = (-1)^{L}\bigl(\epsilon_i^2 - \frac{(1-(-1)^{L})}{2}k_+k_-\bigr) \mathbb{I}.
\end{eqnarray}
For the mixed relations between $\W_i$ and $\NWex{2}_{i+1}$, $i=0,1$ we obtain the defining relations of the Onsager algebra for $i=0,1$:
\begin{align}
\W_i^3\NWex{2}_{i+1} -3\W_i^2\NWex{2}_{i+1} \W_i +3 \W_i \NWex{2}_{i+1} \W_i^2 - \NWex{2}_{i+1} \W_i^3 &= \rho_i^{(1)}\Bigl[\W_i,\NWex{2}_{i+1}\Bigr],\label{eq:WN2-mix-rel-1}\\
\Bigl(\NWex{2}_{i+1}\Bigr)^3\W_i -3\Bigl(\NWex{2}_{i+1}\Bigr)^2\W_i \NWex{2}_{i+1} 
+ 3 \NWex{2}_{i+1}\W_i \Bigl(\NWex{2}_{i+1}\Bigr)^2 - \W_i\Bigl(\NWex{2}_{i+1}\Bigr)^3 &= \rho_{i+1}^{(2)}\Bigl[\NWex{2}_{i+1},\W_i\Bigr],\label{eq:WN2-mix-rel-2}
\end{align}
where we set $\epsilon_0=\epsilon_+$, $\epsilon_1=\epsilon_-$ and 
\begin{equation}
\rho_i^{(1)} = 4(-1)^L \bigl(\epsilon_i^2- \frac{1-(-1)^L}{2}k_+k_-\bigr),\qquad
\rho_i^{(2)} = 4 \epsilon_{i}^2(k_+k_- - \epsilon_{i}^2).
\end{equation}
Finally, for the divided polynomials $\NWex{2}_0$ and $\NWex{2}_1$ we have 
\begin{multline}\label{eq:WN2-rel}
\Biggl[\,\NWex{2}_i\,,
\Bigl(\NWex{2}_{i}\Bigr)^4 \NWex{2}_{i+1}
-4\Bigl(\NWex{2}_{i}\Bigr)^3 \NWex{2}_{i+1}\NWex{2}_{i}
+6\Bigl(\NWex{2}_{i}\Bigr)^2 \NWex{2}_{i+1}\Bigl(\NWex{2}_{i}\Bigr)^2
-4 \NWex{2}_{i}\NWex{2}_{i+1}\Bigl(\NWex{2}_{i}\Bigr)^3 
+\NWex{2}_{i+1}\Bigl(\NWex{2}_{i}\Bigr)^4 \\
-5\rho_i^{(2)}\Bigl\{\Bigl(\NWex{2}_{i}\Bigr)^2 \NWex{2}_{i+1} -2\NWex{2}_{i} \NWex{2}_{i+1}\NWex{2}_{i} + \NWex{2}_{i+1} \Bigl(\NWex{2}_{i}\Bigr)^2 \Bigr\}
+4\bigl(\rho_i^{(2)}\bigr)^2  \NWex{2}_{i+1}
\,\Biggr]=0.
\end{multline}
\subsection{Observations about the symmetries of the Hamiltonian}
In this subsection, we study the commutation relations between the basic generators of the $q-$Onsager algebra or the divided polynomials in the basic generators and the Hamiltonian of the open XXZ spin chain for a certain choice of boundary conditions and $q$ a root of unity. Basically, let us consider the Hamiltonian:
\begin{eqnarray}
H^{(L)}_{XXZ}&=&\sum_{k=1}^{L-1}\Big(\sigma_1^{k+1}\sigma_1^{k}+\sigma_2^{k+1}\sigma_2^{k} + \Delta\sigma_z^{k+1}\sigma_z^{k}\Big) 
+\  \frac{(q-q^{-1})}{2} h_3\sigma^1_z + h_+\sigma^1_+ + h_-\sigma^1_-     \label{Hbord-2}\nonumber\\
 &&\qquad\qquad  \qquad\qquad \qquad\qquad \qquad\qquad \ +\  \frac{(q-q^{-1})}{2}\bar{h}_3\sigma^L_z + \bar{h}_+\sigma^L_+ +  \bar{h}_-\sigma^L_-  .   \nonumber
\end{eqnarray}
$\bullet$ The special case $k_-=0$.\\
In this case, the operators satsify the relations considered in the previous Section. We find:
\begin{eqnarray}
[H_{XXZ}^{(L)},\stackrel{[N]}{\mathcal{W}}_i ] = 0\nonumber
\end{eqnarray}
if and only if
\begin{eqnarray}
\left\{\begin{array}{l}
\bar h_+=\bar h_-=h_-=0,\\
\bar h_3=(-1)^i,\\
h_+ : \text{free}\\
 h_3 = \frac{(-1)^i(-k_++\epsilon_ih_+)}{k_+}.
\end{array}
\right.
\end{eqnarray}
Note that the conditions $[H_{XXZ}^{(L)},{\mathcal{W}}_i ] = 0$ are satisfied provided the above relations are satisfied too.
From the results above, we conclude that the Hamiltonian cannot commute simultaneously with both $\stackrel{[N]}{\mathcal{W}}_0$ and $\stackrel{[N]}{\mathcal{W}}_1$.  \vspace{1mm}

However, recall that for the XXZ spin chain with periodic boundary conditions, in some sectors of the spin chain
the Hamiltonian commutes with the generators of the $sl_2$ loop algebra \cite{DFM}  (see the first section of this part). By analogy, we are currently investigating the existence of sectors of the spin chain in which a non-Abelian symmetry associated with a subalgebra of the algebra generated $\stackrel{[N]}{\mathcal{W}}_0, \stackrel{[N]}{\mathcal{W}}_1$ may occur.\vspace{1mm}

Note that the results reported here are part of an ongoing project.

\chapter{Perspectives}

This thesis opens at least three new and promising directions for further research. To our knowledge, none of these have been discussed in the literature. 

\section{A family of new integrable hierarchies}
Let $A,A^*$ be the generators of the $q-$Onsager algebra. As shown by Dolan and Grady \cite{DG}, it is possible to construct a hierarchy of mutually commuting quantities that are polynomials in $A,A^*$, of the form:
\begin{eqnarray}
I_1&=&\kappa A + \kappa^* A^*,\nonumber\\
I_3&=&\kappa\left(\frac{1}{\rho}((q^2+q^{-2})AA^*A-A^2A^*-A^*A^2)+A^*\right)\nonumber\\
& & +\kappa^*\left(\frac{1}{\rho}((q^2+q^{-2})A^*AA^*-{A^*}^2A-A{A^*}^2)+A\right) \label{quant}
\end{eqnarray}
In Chapter 2, we explained that examples of such hierarchy can be generated from the transfer matrix formalism, where the transfer matrix is built from solutions to the reflection equation algebra (the so-called Sklyanin's operator). In this thesis, it is conjectured (with several supporting evidences) that the $q-$Onsager algebra automatically implies the existence of higher order relations satisfied by $A,A^*$, the so-called $r-$th higher order $q-$Dolan-Grady relations. \vspace{1mm}

Having this in mind and taking a more general point of view, let us introduce a new associative algebra with generators $X,Y$ that satisfy {\it only} the $r-$th higher order Dolan-Grady relation (\ref{qDGfinr})-(\ref{qDGfinr2}) (i.e. we do {\it not} assume the $k-$th higher order $q-$Dolan-Grady relations for $1 \leq k \leq r-1$). We call this new algebra $O_q^{(r)}$. The following problems may be considered: \vspace{1mm}

{\bf Problem 1}: For $r=1$, the $q-$Onsager algebra and the reflection equation algebra are closely related (see Chapters 1,2). Is there an $r-th$ analogue of the reflection equation algebra that would correspond to this new algebra? What is the interpretation of this new algebra in the context of scattering theory?\vspace{1mm}

{\bf Problem 2}: For $r=1$, the quantities (\ref{quant}) generate the $q-$Dolan-Grady integrable hierarchy. For $r\neq 1$, what is the structure of the new integrable hierarchy ? Find an algorithm to construct explicitly its first few elements. \vspace{1mm}

{\bf Problem 3}: For $r=1$, there is an homomorphism from the $q-$Onsager algebra to a coideal subalgebra of $U_q(\widehat{sl_2})$ (see Chapter 1). What happens for $r\neq 1$?   \vspace{1mm}

{\bf Problem 4}: For $r=1$, there are numerous examples of quantum integrable systems that are generated by the $q-$Onsager algebra for $q=1$ or $q\neq 1$: Ising, superintegrable chiral Potts model, XY, open XXZ spin chain,... For $r\neq 1$, construct explicit examples of Hamiltonian integrable systems with applications to condensed matter physics.\vspace{1mm}

\section{Cyclic tridiagonal pairs}
It is clear that one of the main ingredients in this thesis is the concept of tridiagonal pairs, provided $q$ is not a root of unity. For details, we refer the reader to Chapter 2. In particular, the definition of tridiagonal pairs is given in \cite{TD00}. Among the known examples of tridiagonal pairs, for generic values of $q$ one finds the basic operators that appear in the open XXZ spin chain, see (\ref{TDbase})-(\ref{TDbase1}) and found application to the solution of this model \cite{BK3}. \vspace{1mm}

For $q$ a root of unity, as shown in Chapter 3, the basic operators of the open XXZ spin chain, see (\ref{TDbase})-(\ref{TDbase1}) still satisfy the $q-$Onsager algebra. However, the decomposition of the vector space on which they act is different for $q$ a root of unity compared to the case $q$ generic. Indeed, for $q$ a root of unity the spectrum of the operators admits additional degeneracies. In this case, to each eigenvalue one associates an eigenspace as defined by (\ref{esp}). Recall that for $q$ generic, the ordered eigenspaces are such that $V_{-1}=V_{d+1}=0;~ V^*_{-1}=V^*_{d+1}=0$ (see Definition \ref{defitri}). However, for $q$ a root of unity, one has: $V_{-1}^{(N)}=V_{N-1}^{(N)},~ V_{0}^{(N)}=V_{N}^{(N)}; ~{V^*}_{-1}^{(N)}={V^*}_{N-1}^{(N)},~ {V^*}_{0}^{(N)}={V^*}_{N}^{(N)}$ (see Corollary \ref{cor1})

In other words, for $q$ a root of unity the action of the two basic operators is `cyclic'.\vspace{1mm}

Based on this example, a new concept may be introduced that we call `cyclic tridiagonal pairs'. For this object, we propose the following definition:
\begin{defn}\label{deficitri}
Let $V$ denote a vector space over a field $\mathbb{K}$ with finite
positive dimension. 
By a {\it cyclic tridiagonal pair} on $V$ we mean an ordered pair of linear transformations $A:V \to V$ and 
$A^*:V \to V$ that satisfy the following four conditions.
\begin{itemize}
\item[(i)] Each of $A,A^*$ is diagonalizable.
\item[(ii)] There exists an ordering $\lbrace V_i\rbrace_{i=0}^d$ of the  
eigenspaces of $A$ such that 
\begin{equation}
A^* V_i \subseteq V_{i-1} + V_i+ V_{i+1} \qquad \qquad 0 \leq i \leq d,
\label{eq:ct1}
\end{equation}
where $V_{-1} = V_{d-1}$ and $V_0=V_{d}$.
\item[(iii)] There exists an ordering $\lbrace V^*_i\rbrace_{i=0}^{\delta}$ of
the  
eigenspaces of $A^*$ such that 
\begin{equation}
A V^*_i \subseteq V^*_{i-1} + V^*_i+ V^*_{i+1} 
\qquad \qquad 0 \leq i \leq \delta,
\label{eq:ct2}
\end{equation}
where $V^*_{-1} = V^*_{\delta-1}$ and $V^*_0=V^*_{\delta}$.
\item[(iv)] There does not exist a subspace $W$ of $V$ such  that $AW\subseteq W$,
$A^*W\subseteq W$, $W\not=0$, $W\not=V$.
\end{itemize}
\end{defn}
According to this definition, the following problems may be considered:\vspace{1mm}

{\bf Problem 1}: Classify irreducible finite dimensional representations (up to isomorphisms) of cyclic tridiagonal pairs.
\vspace{1mm}

{\bf Problem 2}: Construct explicit examples of cyclic tridiagonal pairs
using the connection between the solutions of the reflection equation algebra and the $q-$Onsager algebra at roots of unity. \vspace{1mm} 

{\bf Problem 3}: Solve quantum integrable systems using the representation theory of the $q-$Onsager algebra at roots of unity and the properties of cyclic tridiagonal pairs.  \vspace{1mm}

\section{An analog of Lusztig quantum group for the $q-$Onsager algebra at roots of unity}
Let $p$ be an positive integer. Consider the quantum group $U_q(sl_2)$ at root of unity $q=e^{{i\pi}/{p}}$, denote $\overline{U}_q(sl_2)$. It is sometimes referred as the `small quantum group'. The generators are given by $E,F$ and $K^{\pm 1}$ satisfying the standard defining relations of $U_q(sl_2)$.
\begin{equation}\label{Uq-relations}
 KEK^{-1}=q^2E,\quad
  KFK^{-1}=q^{-2}F,\quad
  [E,F]=\frac{K-K^{-1}}{q-q^{-1}},
\end{equation}
with  additional relations,
\begin{equation}
  E^{p}=F^{p}=0,\quad K^{2p}= 1.
\end{equation}
It can be endowed with a Hopf algebra structure. The comultiplication is given by
\begin{equation}
  \Delta(E)= 1\otimes E+E\otimes K,\quad
  \Delta(F)=K^{-1}\otimes F+F\otimes 1,\quad
  \Delta(K)=K\otimes K.\label{Uq-comult-relations}
\end{equation}
This associative algebra is finite-dimensional, namely $\mathrm{dim}\overline{U}_{q}(sl_2) = 2p^3$.\vspace{1mm}

In the literature, the so-called full or Lusztig quantum group has been introduced.
For $q = e^{i\pi/p}$ and for any integer $p\geq 2$, it  is generated by $E$, $F$,
$K^{\pm 1}$ together with $h$ such that $K = e^{\alpha h}$ ($\alpha$ is a scalar) and the so-called Lusztig's divided powers $f\sim F^p/[p]!$ and $e\sim E^p/[p]!$.
The defining relations are given by (\ref{Uq-relations}),
\begin{equation}
  [h,e]=e,\qquad[h,f]=-f,\qquad[e,f]=2h\nonumber
\end{equation}
and the `mixed' relations
\begin{eqnarray}
  [h,K]=0,\qquad[E,e]=0,\qquad[K,e]=0,\qquad[F,f]=0,\qquad[K,f]=0,\label{zero-rel}\nonumber\\
  \big[ F,e\big] = \frac{1}{[p-1]!} K^p
  \frac{q K-q^{-1}K^{-1}}{q-q^{-1}}E^{p-1},\qquad
  [E,f]=\frac{(-1)^{p+1}}{[p-1]!} F^{p-1}\frac{q K-q^{-1} K^{-1}}{q-q^{-1}},\label{Ef-rel}\nonumber\\
  \big[h,E\big]=\frac{1}{2}E A,\quad[h,F]= - \frac{1}{2}A F,\label{hE-hF-rel}\nonumber
\end{eqnarray}
where
\begin{equation}\label{A-element}
  A=\,\sum_{s=1}^{p-1}\frac{(u_s(q^{-s-1})-u_s(q^{s-1}))K
        +q^{s-1}u_s(q^{s-1})-q^{-s-1}u_s(q^{-s-1})}{(q^{s-1}
         -q^{-s-1})u_s(q^{-s-1})u_s(q^{s-1})}\,
        u_s(K)e_s.\nonumber
\end{equation}
Here, the polynomials $u_s(K)=\prod_{n=1,\;n\neq s}^{p-1}(K-q^{s-1-2n})$, and
$e_s$ are central primitive idempotents (see \cite{FGST}) \vspace{1mm}

The Lusztig quantum group can be endowed with a Hopf algebra structure.  For instance, the comultiplication is given by:
\begin{eqnarray}
  \Delta(E)=1\otimes E + E\otimes K,\quad
  \Delta(F)=K^{-1}\otimes F + F\otimes 1,\quad
  \Delta(K)=K\otimes K,\nonumber\\
  \Delta(e)=e\otimes 1 +K^p\otimes e
  +\frac{1}{[p-1]!} \sum_{r=1}^{p-1}\frac{q^{r(p-r)}}{[r]}K^p E^{p-r}\otimes E^r
  K^{-r},\nonumber\\
 \Delta(f)= f\otimes 1 + K^p\otimes f+\frac{(-1)^p}{[p-1]!}
  \sum_{s=1}^{p-1}\frac{q^{-s(p-s)}}{[s]}K^{p+s}F^s\otimes F^{p-s}.\nonumber
\end{eqnarray}

In Chapter 3, we have seen that the operators ${\cal W}_0,{\cal W}_1,{\stackrel{\footnotesize[N]}{\mathcal W}}_0,{\stackrel{\footnotesize[N]}{\mathcal W}}_1$ satisfy certain relations at least for certain choices of the parameters $k_\pm$. According to these results, it is thus natural to consider the following problems:\vspace{1mm}

{\bf Problem 1}: Define an analog of the small quantum group for the $q-$Onsager algebra. Classify irreducible finite dimensional representations. 
\vspace{1mm}

{\bf Problem 2}: Define an analog of the Lusztig quantum group for the $q-$Onsager algebra. Classify irreducible finite dimensional representations.
\vspace{1mm}

{\bf Problem 3}: What is the analog of the Hopf algebra structure? 
\vspace{1mm}

Note that some results are already obtained in relation to Problem 1 and 2, although not published yet.
\chapter{Appendices}
\section{APPENDIX A: Coefficients $ \eta^{(m)}_{k,j}$, $M^{(r,p)}_j$, $N^{(r,p)}_j$}

\vspace{3mm}

The coefficients that appear in eqs. (\ref{eqmu1}), (\ref{eqmu2}) are such that:
\begin{eqnarray}
\eta^{(3)}_{1,0}&=& [3]_q,\qquad \eta^{(3)}_{1,1}=-[3]_q,\qquad \eta^{(3)}_{1,2}=1,\nonumber \\
\eta^{(4)}_{0,0}&=&1,\qquad \eta^{(4)}_{0,1}=q^2+q^{-2},\qquad \eta^{(4)}_{0,2}=-[3]_q,\nonumber \\
\eta^{(4)}_{1,0}&=&(q^2+q^{-2})[3]_q,\qquad \eta^{(4)}_{1,1}=-(q^2+q^{-2})[2]_q^2,\qquad \eta^{(4)}_{1,2}=[3]_q.\nonumber 
\end{eqnarray}
The recursion relations for $\eta^{(m)}_{k,j}$ are such that:
\begin{eqnarray}
\eta^{(2n+2)}_{0,0} &=& 1,\nonumber \\
\eta^{(2n+2)}_{k,0}&=&[3]_q\eta^{(2n+1)}_{k,0}+\eta^{(2n+1)}_{k,1},\qquad 1 \leq k \leq n, \nonumber \\
\eta^{(2n+2)}_{0,1}&=&\eta^{(2n+1)}_{1,0}-1, \nonumber \\
\eta^{(2n+2)}_{k,1}&=&-[3]_q\eta^{(2n+1)}_{k,0}+\eta^{(2n+1)}_{k+1,0}+\eta^{(2n+1)}_{k,2},\qquad 1 \leq k \leq n-1, \nonumber \\
\eta^{(2n+2)}_{n,1}&=&-[3]_q\eta^{(2n+1)}_{n,0}+\eta^{(2n+1)}_{n,2},\nonumber \\
\eta^{(2n+2)}_{0,2}&=&-\eta^{(2n+1)}_{1,0}, \nonumber \\
\eta^{(2n+2)}_{k,2}&=&\eta^{(2n+1)}_{k,0}-\eta^{(2n+1)}_{k+1,0},\qquad 1 \leq k \leq n-1, \nonumber \\
\eta^{(2n+2)}_{n,2}&=&\eta^{(2n+1)}_{n,0},\nonumber
\end{eqnarray}
and
\begin{eqnarray}
\eta^{(2n+3)}_{k,0}&=&[3]_q\eta^{(2n+2)}_{k-1,0}+\eta^{(2n+2)}_{k-1,1},\qquad 1 \leq k \leq n+1, \nonumber \\
\eta^{(2n+3)}_{k,1}&=&-[3]_q\eta^{(2n+2)}_{k-1,0}+\eta^{(2n+2)}_{k,0}+\eta^{(2n+2)}_{k-1,2},\qquad 1 \leq k \leq n ,\nonumber \\
\eta^{(2n+3)}_{n+1,1}&=&-[3]_q\eta^{(2n+2)}_{n,0}+\eta^{(2n+2)}_{n,2}, \nonumber \\
\eta^{(2n+3)}_{k,2}&=& \eta^{(2n+2)}_{k-1,0}-\eta^{(2n+2)}_{k,0},\qquad 1 \leq k \leq n, \nonumber \\
\eta^{(2n+3)}_{n+1,2}&=& \eta^{(2n+2)}_{n,0}. \nonumber
\end{eqnarray}

\vspace{3mm}

The coefficients that appear in eqs. (\ref{eq2}), (\ref{eq3}) are such that:

\begin{eqnarray}
M^{(r,0)}_j&=& c^{[r,0]}_j-c^{[r,0]}_1c^{[r,0]}_{j-1},\qquad j=\overline{2,2r+1}, \nonumber \\
M^{(r,0)}_{2r+2}&=&-c^{[r,0]}_1c^{[r,0]}_{2r+1}, \nonumber \\
M^{(r,p)}_0&=&c^{[r,p]}_0, p=\overline{1,r}, \nonumber \\
M^{(r,p)}_j&=&c^{[r,p]}_j-c^{[r,0]}_1c^{[r,p]}_{j-1},\qquad p=\overline{1,r},\quad j=\overline{1,2(r-p)+1}, \nonumber \\
M^{(r,p)}_{2(r-p)+2}&=&-c^{[r,0]}_1c^{[r,p]}_{2(r-p)+1},\qquad p=\overline{1,r},\nonumber 
\end{eqnarray} 
and
\begin{eqnarray}
N^{(r,0)}_j&=&c^{[r,0]}_j-c^{[r,0]}_1c^{[r,0]}_{j-1}+({c^{[r,0]}_1}^2-c^{[r,0]}_2)c^{[r,0]}_{j-2},\qquad j=\overline{3,2r+1},\nonumber \\
N^{(r,0)}_{2r+2}&=&-c^{[r,0]}_1c^{[r,0]}_{2r+1}+({c^{[r,0]}_1}^2-c^{[r,0]}_2)c^{[r,0]}_{2r}, \nonumber \\
N^{(r,0)}_{2r+3}&=&({c^{[r,0]}_1}^2-c^{[r,0]}_2)c^{[r,0]}_{2r+1}, \nonumber \\
N^{(r,1)}_0&=& 0, \nonumber \\
N^{(r,1)}_j&=&({c^{[r,0]}_1}^2-c^{[r,0]}_2)c^{[r,1]}_{j-2}-c^{[r,0]}_1c^{[r,1]}_{j-1}+c^{[r,1]}_j-c^{[r,1]}_0c^{[r,0]}_j,\qquad j=\overline{2,2r-1}, \nonumber \\
N^{(r,1)}_1&=&-c^{[r,0]}_1c^{[r,1]}_0+c^{[r,1]}_1-c^{[r,1]}_0c^{[r,0]}_1,\nonumber\end{eqnarray}
\begin{eqnarray}
N^{(r,1)}_{2r}&=&({c^{[r,0]}_1}^2-c^{[r,0]}_2)c^{[r,1]}_{2r-2}-c^{[r,0]}_1c^{[r,1]}_{2r-1}-c^{[r,1]}_0c^{[r,0]}_{2r}, \nonumber \\
N^{(r,1)}_{2r+1}&=&({c^{[r,0]}_1}^2-c^{[r,0]}_2)c^{[r,1]}_{2r-1}-c^{[r,1]}_0c^{[r,0]}_{2r+1}, \nonumber \\
N^{(r,r+1)}_j&=&-c^{[r,1]}_0c^{[r,r]}_j,\qquad j = \overline{0,1}. \nonumber 
\end{eqnarray}
For $2 \leq p \leq r,$ 
\begin{eqnarray}
N^{(r,p)}_j&=&({c^{[r,0]}_1}^2-c^{[r,0]}_2)c^{[r,p]}_{j-2}-c^{[r,0]}_1c^{[r,p]}_{j-1}+c^{[r,p]}_j-c^{[r,1]}_0c^{[r,p-1]}_j, \quad j=\overline{2, 2(r-p)+1}, \nonumber \\
N^{(r,p)}_0&=&c^{[r,p]}_0-c^{[r,1]}_0c^{[r,p-1]}_0,\nonumber \\
N^{(r,p)}_1&=&-c^{[r,0]}_1c^{[r,p]}_0+c^{[r,p]}_1-c^{[r,1]}_0c^{[r,p-1]}_1, \nonumber \\
N^{(r,p)}_{2(r-p)+2}&=&({c^{[r,0]}_1}^2-c^{[r,0]}_2)c^{[r,p]}_{2(r-p)}-c^{[r,0]}_1c^{[r,p]}_{2(r-p)+1}-c^{[r,1]}_0c^{[r,p-1]}_{2(r-p)+2}, \nonumber \\
N^{(r,p)}_{2(r-p)+3}&=&({c^{[r,0]}_1}^2-c^{[r,0]}_2)c^{[r,p]}_{2(r-p)+1}-c^{[r,1]}_0c^{[r,p-1]}_{2(r-p)+3}. \nonumber 
\end{eqnarray}

\vspace{1cm}

\newpage

\section{APPENDIX B: $A^{2r+2}{A^*}^{r+1}A$, $A^{2r+1}{A^*}^{r+1}A^2$, $A^{2r+1}{A^*}^{r+1}$}

\vspace{3mm}

In addition to (\ref{m1}), the other monomials can be written as:
\begin{eqnarray*}
&& A^{2r+2}{A^*}^{r+1}A =-\sum\limits_{i=1}^r{M^{(r,0)}_{2i+2}A^{2(r-i)}{A^*}^r(\sum\limits_{k=0}^i{\sum\limits_{j=0}^2{\rho_0^{i-k}\eta^{(2i+2)}_{k,j}A^{2-j}A^*A^{2k+j+1}}})} \nonumber\\& &~ + \sum\limits_{i=1}^r{M^{(r,0)}_{2i+1}A^{2(r-i)+1}{A^*}^r(\sum\limits_{k=1}^i{\sum\limits_{j=0}^2{\rho^{i-k}_0\eta^{(2i+1)}_{k,j}A^{2-j}A^*A^{2k+j}}}+\rho_0^i(AA^*A-A^*A^2))}\nonumber \\ & &~ - \sum\limits_{p=1}^r{{(-\rho_0)}^p(M^{(r,p)}_0A^{2(r-p)+2}{A^*}^{r+1}A-M^{(r,p)}_1A^{2(r-p)+1}{A^*}^rAA^*A+M^{(r,p)}_2A^{2(r-p)}{A^*}^rA^2A^*A)} \\& &~ - M^{(r,0)}_2A^{2r}{A^*}^rA^2A^*A-\sum\limits_{p=1}^{r-1}{{(-\rho_0)}^p\sum\limits_{i=1}^{r-p}{M^{(r,p)}_{2i+2}A^{2(r-p-i)}{A^*}^r(\sum\limits_{k=0}^i{\sum\limits_{j=0}^2{\rho_0^{i-k}\eta^{(2i+2)}_{k,j}A^{2-j}A^*A^{2k+j+1}}})}} \nonumber\\& &~ + \sum\limits_{p=1}^{r-1}{{(-\rho_0)}^p\sum\limits_{i=1}^{r-p}{M^{(r,p)}_{2i+1}A^{2(r-p-i)+1}{A^*}^r(\sum\limits_{k=1}^i{\sum\limits_{j=0}^2{\rho_0^{i-k}\eta^{(2i+1)}_{k,j}A^{2-j}A^*A^{2k+j}}}+\rho_0^i(AA^*A-A^*A^2))}}\ ,
\end{eqnarray*}
\vspace{-0.3cm}
\begin{eqnarray}
&& A^{2r+1}{A^*}^{r+1}A^2 = c^{[r,0]}_1A^{2r}{A^*}^rAA^*A^2-c^{[r,0]}_2A^{2r-1}{A^*}^rA^2A^*A^2 \nonumber\\& &~ - \sum\limits_{p=1}^{r-1}{{(-\rho_0)}^p(c^{[r,p]}_0A^{2(r-p)+1}{A^*}^{r+1}A^2-c^{[r,p]}_1A^{2(r-p)}{A^*}^rAA^*A^2+c^{[r,p]}_2A^{2(r-p)-1}{A^*}^rA^2A^*A^2)} \nonumber\\& &~ -{(-\rho_0)}^r(c^{[r,r]}_0A{A^*}^{r+1}A^2-c^{[r,r]}_1{A^*}^rAA^*A^2) \nonumber\\& &~ +\sum\limits_{i=1}^r{c^{[r,0]}_{2i+1}A^{2(r-i)}{A^*}^r(\sum\limits_{k=1}^i{\sum\limits_{j=0}^2{\rho_0^{i-k}\eta^{(2i+1)}_{k,j}A^{2-j}A^*A^{2k+1+j}}}+\rho_0^i(AA^*A^2-A^*A^3))} \nonumber\\ & &~ - \sum\limits_{i=1}^{r-1}{c^{[r,0]}_{2i+2}A^{2(r-i)-1}{A^*}^r(\sum\limits_{k=0}^i{\sum\limits_{j=0}^2{\rho_0^{i-k}\eta^{(2i+2)}_{k,j}A^{2-j}A^*A^{2k+j+2}}})}\nonumber \\& &~  +\sum\limits_{p=1}^{r-1}{{(-\rho_0)}^p\sum\limits_{i=1}^{r-p}{c^{[r,p]}_{2i+1}A^{2(r-p-i)}{A^*}^r(\sum\limits_{k=1}^i{\sum\limits_{j=0}^2{\rho_0^{i-k}\eta^{(2i+1)}_{k,j}A^{2-j}A^*A^{2k+1+j}}}+\rho_0^i(AA^*A^2-A^*A^3))}} \nonumber\\& &~ -\sum\limits_{p=1}^{r-2}{{(-\rho_0)}^p\sum\limits_{i=1}^{r-p-1}{c^{[r,p]}_{2i+2}A^{2(r-p-i)-1}{A^*}^r(\sum\limits_{k=0}^i{\sum\limits_{j=0}^2{\rho_0^{i-k}\eta^{(2i+2)}_{k,j}A^{2-j}A^*A^{2k+j+2}}})}}\ ,\nonumber
\end{eqnarray}
\vspace{-0.3cm}
\begin{eqnarray*}
&& A^{2r+1}{A^*}^{r+1}=c^{[r,0]}_1A^{2r}{A^*}^rAA^*-c^{[r,0]}_2A^{2r-1}{A^*}^rA^2A^* \\& &~ +\sum\limits_{k=1}^r{c^{[r,0]}_{2k+1}A^{2(r-k)}{A^*}^r(\sum\limits_{i=1}^k{\sum\limits_{j=0}^2{\mu^{(2k+1)}_{i,j}A^{2-j}A^*A^{2i-1+j}}}+\rho^k(AA^*-A^*A))} \\& &~ -\sum\limits_{k=1}^{r-1}{c^{[r,0]}_{2k+2}A^{2(r-k)-1}{A^*}^r(\sum\limits_{i=0}^k{\sum\limits_{j=0}^2{\mu^{(2k+2)}_{i,j}A^{2-j}A^*A^{2i+j}}})} \\& &~ -\rho^r(c^{[r,r]}_0A{A^*}^{r+1}-c^{[r,r]}_1{A^*}^rAA^*) \\& &~ -\sum\limits_{p=1}^{r-1}{\rho^p(c^{[r,p]}_0A^{2(r-p)+1}{A^*}^{r+1}-c^{[r,p]}_1A^{2(r-p)}{A^*}^rAA^*+c^{[r,p]}_2A^{2(r-p)-1}{A^*}^rA^2A^*)} \\& &~ + \sum\limits_{p=1}^{r-1}{\rho^p\sum\limits_{k=1}^{r-p}{c^{[r,p]}_{2k+1}A^{2(r-p-k)}{A^*}^r(\sum\limits_{i=1}^k{\sum\limits_{j=0}^2{\mu^{(2k+1)}_{i,j}A^{2-j}A^*A^{2i-1+j}}}+\rho^k(AA^*-A^*A))}} \\& &~ - \sum\limits_{p=1}^{r-2}{\rho^p\sum\limits_{k=1}^{r-p-1}{c^{[r,p]}_{2k+2}A^{2(r-p-k)-1}{A^*}^r(\sum\limits_{i=0}^k{\sum\limits_{j=0}^2{\mu^{(2k+2)}_{i,j}A^{2-j}A^*A^{2i+j}}})}}.
\end{eqnarray*}

\newpage

\section{APPENDIX C: Coefficients $ \eta^{(m)}_{k,j}$, $M^{(r,p)}_j$}
\vspace{3mm}
The initial values of $\eta^{(m)}_{k,j}$ are given by:\\
\[\eta^{(2)}_{0,0}=[2]_q, \qquad \eta^{(2)}_{0,1}=-1.\]
The recursion relations for $\eta^{(m)}_{k,j}$  and $M^{(r,p)}_j$ read:\\
\begin{eqnarray}
\eta^{(2n+1)}_{p,0}&=&[2]_q\eta^{(2n)}_{p,0}+\eta^{(2n)}_{p,1}, \qquad p=\overline{0,n-1},\nonumber \\
\eta^{(2n+1)}_{n,0}&=&1,\nonumber\\
\eta^{(2n+1)}_{p,1} &=&-\eta^{(2n)}_{p,0}+\eta^{(2n)}_{p-1,0},\qquad p=\overline{1,n-1},\nonumber \\
\eta^{(2n+1)}_{0,1}&=&-\eta^{(2n)}_{0,0},\nonumber\\
\eta^{(2n+1)}_{n,1}&=&\eta^{(2n)}_{n-1,0},\nonumber \\
\eta^{(2n+2)}_{p,0}&=&[2]_q\eta^{(2n+1)}_{p,0}+\eta^{(2n+1)}_{p,1}, \qquad p=\overline{0,n},\nonumber 
\end{eqnarray}
\begin{eqnarray}
\eta^{(2n+2)}_{p,1}&=&\eta^{(2n+1)}_{p-1,0}-\eta^{(2n+1)}_{p,0}, \qquad p=\overline{1,n},\nonumber \\
\eta^{(2n+2)}_{0,1}&=&-\eta^{(2n+1)}_{0,0}.\nonumber
\end{eqnarray}
\begin{eqnarray}
M^{(2t,p)}_0&=&c^{[2t,p]}_0,\nonumber \\
M^{(2t,p)}_{2t+2-2p}&=&-c^{[2t,0]}_1c^{[2t,p]}_{2t+1-2p},\nonumber \\
M^{(2t,p)}_k&=&c^{[2t,p]}_k-c^{[2t,0]}_1c^{[2t,p]}_{k-1},\qquad k=\overline{1,2t+1-2p},\nonumber \\
M^{(2t+1,t+1)}_0&=&c^{[2t+1,t+1]}_0,\nonumber\\
M^{(2t+1,t+1)}_1&=&-c^{[2t+1,t+1]}_0c^{[2t+1,0]}_1,\nonumber \\
M^{(2t+1,p)}_k&=&-c^{[2t+1,0]}_1c^{[2t+1,p]}_{k-1}+c^{[2t+1,p]}_k,\qquad k=\overline{1,2t+2-2p},\nonumber \\
M^{(2t+1,p)}_0&=&c^{[2t+1,p]}_0,\nonumber\\
M^{(2t+1,p)}_{2t+3-2p}&=&-c^{[2t+1,0]}_1c^{[2t+1,p]}_{2t+2-2p}.\nonumber
\end{eqnarray}
\newpage

\section{APPENDIX D: Algorithms for the $q$-Onsager algebra}
\begin{verbatim}
# Compute the coefficients eta[m,k,j] of A^(2n+2)A* and A^(2n+3)A*
Funct := proc(n)
  local i, k; 
  global eta, alpha, q;
  # Input the initial values
  alpha := q^2+1+1/q^2;
  eta[3, 1, 0] := alpha;
  eta[3, 1, 1] := -alpha; 
  eta[3, 1, 2] := 1; 
  eta[4, 1, 0] := alpha^2-alpha; 
  eta[4, 1, 1] := 1-alpha^2; 
  eta[4, 1, 2] := alpha; 
  eta[4, 0, 1] := alpha-1;
  eta[4, 0, 0] := 1;
  eta[4, 0, 2] := -alpha; 
  eta[5, 2, 0] := alpha^3-2*alpha^2+1;
  eta[5, 2, 1] := -alpha^3+alpha^2+alpha; 
  eta[5, 2, 2] := alpha*(alpha-1); eta[5, 1, 0] := 2*alpha-1; 
  eta[5, 1, 1] := alpha^2-3*alpha;
  eta[5, 1, 2] := -alpha^2+alpha+1;
  for i from 2 to n do
     eta[2*i+2, 0, 0] := 1;
     for k from 1 to i do 
        eta[2*i+2, k, 0] := simplify(alpha*eta[2*i+1, k, 0]+eta[2*i+1, k, 1]);
     end do;
     eta[2*i+2, 0, 1] := simplify(eta[2*i+1, 1, 0]-1);
     for k from 1 to i-1 do
        eta[2*i+2, k, 1] := simplify(-alpha*eta[2*i+1, k, 0]
                                     +eta[2*i+1, k+1, 0]+eta[2*i+1, k, 2]);
     end do;
     eta[2*i+2, i, 1] := simplify(-alpha*eta[2*i+1, i, 0]+eta[2*i+1, i, 2]);
     eta[2*i+2, 0, 2] := simplify(-eta[2*i+1, 1, 0]);
     for k from 1 to i-1 do
        eta[2*i+2, k, 2] := simplify(eta[2*i+1, k, 0]-eta[2*i+1, k+1, 0]);
     end do;
     eta[2*i+2, i, 2] := simplify(eta[2*i+1, i, 0]);
     for k from 1 to i+1 do
        eta[2*i+3, k, 0] := simplify(alpha*eta[2*i+2, k-1, 0]+eta[2*i+2, k-1, 1]);
     end do;
     for k from 1 to i do
        eta[2*i+3, k, 1] := simplify(-alpha*eta[2*i+2, k-1, 0]+eta[2*i+2, k, 0]
                                     +eta[2*i+2, k-1, 2]);
     end do;
     eta[2*i+3, i+1, 1] := simplify(-alpha*eta[2*i+2, i, 0]+eta[2*i+2, i, 2]);
     for k from 1 to i do
        eta[2*i+3, k, 2] := simplify(eta[2*i+2, k-1, 0]-eta[2*i+2, k, 0]);
     end do;
     eta[2*i+3, i+1, 2] := simplify(eta[2*i+2, i, 0])
  end do;
end proc;
\end{verbatim}
\begin{verbatim}
Result1 := proc (n)
  local i, p, d, l, F, M, N, t, tam, ta1, ta2, ta3; 
  global A, c, alpha, rho, q; 
  # Input the initial values
  alpha := q^2+1+1/q^2; 
  c[1, 0, 0] := 1; 
  c[1, 0, 3] := 1; 
  c[1, 0, 1] := alpha; 
  c[1, 0, 2] := alpha; 
  c[1, 1, 0] := -1; 
  c[1, 1, 1] := -1; 
  c[2, 0, 0] := 1; 
  c[2, 0, 5] := 1; 
  c[2, 0, 1] := expand(simplify(alpha^2-alpha-1)); 
  c[2, 0, 4] := expand(simplify(alpha^2-alpha-1)); 
  c[2, 0, 2] := expand(simplify(alpha^3-2*alpha^2+1));
  c[2, 0, 3] := expand(simplify(alpha^3-2*alpha^2+1));
  c[2, 1, 0] := expand(simplify(-alpha^2+2*alpha-2));
  c[2, 1, 3] := expand(simplify(-alpha^2+2*alpha-2));
  c[2, 1, 1] := expand(simplify(-alpha*(alpha^2-alpha-1))); 
  c[2, 1, 2] := expand(simplify(-alpha*(alpha^2-alpha-1))); 
  c[2, 2, 0] := expand(simplify((alpha-1)^2)); 
  c[2, 2, 1] := expand(simplify((alpha-1)^2));
  # Use the convention that the first factor A is denoted A[1], 
  # the second factor A^* is denoted A[2], and the third factor A is denoted A[3]
  # Input the q-Dolan-Grady relations
  if n = 0 then 
     F[1] := A[1]^3*A[2]-alpha*A[1]^2*A[2]*A[3]+alpha*A[1]*A[2]*A[3]^2-A[2]*A[3]^3
            +rho*(-A[1]*A[2]+A[2]*A[3]) = 0 
  end if;
  # Input the higher order q-Dolan-Grady relations for r = 2 
  if n = 1 then 
     F[2] := A[1]^5*A[2]^2-(alpha^2-alpha-1)*A[1]^4*A[2]^2*A[3]
            +(alpha^3-2*alpha^2+1)*A[1]^3*A[2]^2*A[3]^2
            -(alpha^3-2*alpha^2+1)*A[1]^2*A[2]^2*A[3]^3
            +(alpha^2-alpha-1)*A[1]*A[2]^2*A[3]^4-A[2]^2*A[3]^5
            +rho*(-(alpha^2-2*alpha+2)*A[1]^3*A[2]^2
                  +alpha*(alpha^2-alpha-1)*A[1]^2*A[2]^2*A[3]
                  -alpha*(alpha^2-alpha-1)*A[1]*A[2]^2*A[3]^2
                  +(alpha^2-2*alpha+2)*A[2]^2*A[3]^3)
            +rho^2*((alpha-1)^2*A[1]*A[2]^2-(alpha-1)^2*A[2]^2*A[3]) = 0
  end if; 
  if 1 < n then
     # Compute M[n, p, j], p = 0,..., n; j= 2,..., 2n+2 if p=0, 
             #                     j = 0,..., 2(n-p)+2 if p = 1,...,n
     for i from 2 to n do 
        M[i, 0, 2*i+2] := simplify(-c[i, 0, 1]*c[i, 0, 2*i+1]); 
        for l from 2 to 2*i+1 do 
           M[i, 0, l] := simplify(c[i, 0, l]-c[i, 0, 1]*c[i, 0, l-1]);
        end do; 
        for d to i do 
           M[i, d, 0] := c[i, d, 0]; 
           for l to 2*i-2*d+1 do 
              M[i, d, l] := simplify(c[i, d, l]-c[i, 0, 1]*c[i, d, l-1]);
           end do; 
           M[i, d, 2*i-2*d+2] := simplify(-c[i, 0, 1]*c[i, d, 2*i-2*d+1]); 
        end do; 
        Funct(i); 
        # Expand A^{2n+2}{A^*}^{r+1}A
        t[2] := simplify(-(sum(M[i, 0, 2*h+2]*A[1]^(2*i-2*h)*A[2]^i
                           *(sum(sum(eta[2*h+2, k, j]*A[3]^(2-j)*A[4]*A[5]^(2*k+j+1),
                                     j = 0 .. 1)+eta[2*h+2, k, 2]*A[2]*A[3]^(2*k+3), 
                                     k = 0 .. h)), h = 1 .. i))
                      +sum(M[i, 0, 2*h+1]*A[1]^(2*i-2*h+1)*A[2]^i*
                          (sum(eta[2*h+1, k, 2]*A[2]*A[3]^(2*k+2)
                              +sum(eta[2*h+1, k, j]*A[3]^(2-j)*A[4]*A[5]^(2*k+j), 
                                 j = 0 .. 1),k = 1 .. h) 
                                 +rho^h*(A[3]*A[4]*A[5]-A[2]*A[3]^2)), h = 1 .. i)
                     -M[i, 0, 2]*A[1]^(2*i)*A[2]^i*A[3]^2*A[4]*A[5]
                     -(sum(rho^p*(M[i, p, 0]*A[1]^(2*i-2*p+2)*A[2]^(i+1)*A[3]
                          +sum((-1)^j*M[i, p, j]*A[1]^(2*i-2*p+2-j)*A[2]^i*A[3]^j
                          *A[4]*A[5],j = 1 .. 2)), p = 1 .. i))
                     -(sum(rho^p*(sum(M[i, p, 2*h+2]*A[1]^(2*i-2*p-2*h)*A[2]^i
                                 *(sum(sum(eta[2*h+2, k, j]*A[3]^(2-j)*A[4]
                                 *A[5]^(2*k+j+1),j= 0 .. 1)
                                 +eta[2*h+2, k, 2]*A[2]*A[3]^(2*k+3), k = 0 .. h)), 
                                 h = 1 .. i-p)), p = 1 .. i-1))
                     +sum(rho^p*(sum(M[i, p, 2*h+1]*A[1]^(2*i-2*p-2*h+1)*A[2]^i
                                  *(sum(eta[2*h+1, k, 2]*A[2]*A[3]^(2*k+2)
                                      +sum(eta[2*h+1, k, j]*A[3]^(2-j)*A[4]
                                      *A[5]^(2*k+j),j = 0 .. 1), k = 1 .. h)
                                      +rho^h*(A[3]*A[4]*A[5]-A[2]*A[3]^2)), 
                                       h = 1 .. i-p)), p = 1 .. i-1));
        # Compute N[n, p, j], p = 0,...,n;
        #           j = 3,...,2n+2 if p =0
        #           j = 0,...,2(n-p)+3 if p = 1,...,n
        N[i, 0, 2*i+2] := simplify((c[i, 0, 1]^2-c[i, 0, 2])*c[i, 0, 2*i]
                                   -c[i, 0, 1]*c[i, 0, 2*i+1]); 
        N[i, 0, 2*i+3] := simplify((c[i, 0, 1]^2-c[i, 0, 2])*c[i, 0, 2*i+1]); 
        for l from 3 to 2*i+1 do 
           N[i, 0, l] := simplify((c[i, 0, 1]^2-c[i, 0, 2])*c[i, 0, l-2]
                               -c[i, 0, 1]*c[i, 0, l-1]+c[i, 0, l]); 
        end do; 
        N[i, 1, 1] := simplify(-2*c[i, 0, 1]*c[i, 1, 0]+c[i, 1, 1]); 
        N[i, 1, 2*i] := simplify((c[i, 0, 1]^2-c[i, 0, 2])*c[i, 1, 2*i-2]
                            -c[i, 0, 1]*c[i, 1, 2*i-1]-c[i, 1, 0]*c[i, 0, 2*i]); 
        N[i, 1, 2*i+1] := simplify((c[i, 0, 1]^2-c[i, 0, 2])*c[i, 1, 2*i-1]
                                -c[i, 1, 0]*c[i, 0, 2*i+1]); 
        N[i, 1, 0] := 0; 
        for l from 2 to 2*i-1 do 
           N[i, 1, l] := simplify((c[i, 0, 1]^2-c[i, 0, 2])*c[i, 1, l-2]
                               -c[i, 0, 1]*c[i, 1, l-1]+c[i, 1, l]
                               -c[i, 1, 0]*c[i, 0, l]); 
        end do; 
        for d from 2 to i do 
           for l from 2 to 2*i-2*d+1 do 
              N[i, d, l] := simplify((c[i, 0, 1]^2-c[i, 0, 2])*c[i, d, l-2]
                                 -c[i, 0, 1]*c[i, d, l-1]+c[i, d, l]
                                 -c[i, 1, 0]*c[i, d-1, l]);
           end do; 
           N[i, d, 0] := simplify(c[i, d, 0]-c[i, 1, 0]*c[i, d-1, 0]); 
           N[i, d, 1] := simplify(-c[i, 0, 1]*c[i, d, 0]+c[i, d, 1]
                                  -c[i, 1, 0]*c[i, d-1, 1]); 
           N[i, d, 2*i-2*d+2] := simplify((c[i, 0, 1]^2-c[i, 0, 2])
                                    *c[i, d, 2*i-2*d]
                                    -c[i, 0, 1]*c[i, d, 2*i-2*d+1]
                                    -c[i, 1, 0]*c[i, d-1, 2*i-2*d+2]); 
           N[i, d, 2*i-2*d+3] := simplify((c[i, 0, 1]^2-c[i, 0, 2])
                                       *c[i, d, 2*i-2*d+1]
                                       -c[i, 1, 0]*c[i, d-1, 2*i-2*d+3]); 
        end do; 
        for l from 0 to 1 do 
           N[i, i+1, l] := simplify(-c[i, 1, 0]*c[i, i, l]);
        end do;
        # Expand A^{2n+3}{A^*}^{n+1} 
        t[1] := simplify(sum(N[i, 0, 2*h+1]*A[1]^(2*i+2-2*h)*A[2]^i
                          *(sum(eta[2*h+1, k, 2]*A[2]*A[3]^(2*k+1)
                               +sum(eta[2*h+1, k, j]*A[3]^(2-j)
                                    *A[4]*A[5]^(2*k-1+j),
                                    j = 0 .. 1), k = 1 .. h)
                          +rho^h*(A[3]*A[4]-A[2]*A[3])), h = 1 .. i+1)
                     -(sum(N[i, 0, 2*h+2]*A[1]^(2*i-2*h+1)*A[2]^i
                          *(sum(eta[2*h+2, k, 2]*A[2]*A[3]^(2*k+2)
                               +sum(eta[2*h+2, k, j]*A[3]^(2-j)*A[4]*A[5]^(2*k+j), 
                               j = 0 .. 1), k = 0 .. h)), h = 1 .. i))
                     -(sum(rho^p*(N[i, p, 0]*A[1]^(2*i-2*p+3)*A[2]^(i+1)
                          -N[i, p, 1]*A[1]^(2*i-2*p+2)*A[2]^i*A[3]*A[4]
                          +N[i, p, 2]*A[1]^(2*i-2*p+1)*A[2]^i
                                  *A[3]^2*A[4]), p = 1 .. i))
                     +sum(rho^p*(sum(N[i, p, 2*h+1]*A[1]^(2*i-2*p-2*h+2)*A[2]^i
                                   *(sum(eta[2*h+1, k, 2]*A[2]*A[3]^(2*k+1)
                                        +sum(eta[2*h+1, k, j]*A[3]^(2-j)
                                             *A[4]*A[5]^(2*k-1+j), 
                                             j = 0 .. 1), k = 1 .. h)
                              +rho^h*(A[3]*A[4]-A[2]*A[3])), h = 1 .. i-p+1)), 
                                                             p = 1 .. i)
                     -(sum(rho^p*(sum(N[i, p, 2*h+2]*A[1]^(2*i-2*p-2*h+1)*A[2]^i
                                    *(sum(eta[2*h+2, k, 2]*A[2]*A[3]^(2*k+2)
                                         +sum(eta[2*h+2, k, j]*A[3]^(2-j)
                                         *A[4]*A[5]^(2*k+j),j = 0 .. 1), k = 0 .. h)),
                                          h = 1 .. i-p)), p = 1 .. i-1))
                     -rho^(i+1)*(N[i, i+1, 0]*A[1]*A[2]^(i+1)
                                 -N[i, i+1, 1]*A[2]^i*A[3]*A[4]));
        # Expand A^{2n+1}{A^*}^{n+1}A^2                 
        t[3] := simplify(c[i, 0, 1]*A[1]^(2*i)*A[2]^i*A[3]*A[4]*A[5]^2
                      -c[i, 0, 2]*A[1]^(2*i-1)*A[2]^i*A[3]^2*A[4]*A[5]^2
                      -(sum(rho^p*(c[i, p, 0]*A[1]^(2*i-2*p+1)*A[2]^(i+1)*A[3]^2
                           -c[i, p, 1]*A[1]^(2*i-2*p)*A[2]^i*A[3]*A[4]*A[5]^2
                           +c[i, p, 2]*A[1]^(2*i-2*p-1)*A[2]^i*A[3]^2*A[4]*A[5]^2), 
                         p = 1 .. i-1))
                      -rho^i*(c[i, i, 0]*A[1]*A[2]^(i+1)*A[3]^2
                             -c[i, i, 1]*A[2]^i*A[3]*A[4]*A[5]^2)
                      +sum(c[i, 0, 2*h+1]*A[1]^(2*i-2*h)*A[2]^i
                          *(sum(eta[2*h+1, k, 2]*A[2]*A[3]^(2*k+3)
                                +sum(eta[2*h+1, k, j]*A[3]^(2-j)*A[4]*A[5]^(2*k+j+1), 
                            j = 0 .. 1), k = 1 .. h)
                            +rho^h*(A[3]*A[4]*A[5]^2-A[2]*A[3]^3)), h = 1 .. i)
                      -(sum(c[i, 0, 2*h+2]*A[1]^(2*i-2*h-1)*A[2]^i
                           *(sum(eta[2*h+2, k, 2]*A[2]*A[3]^(2*k+4)
                                 +sum(eta[2*h+2, k, j]*A[3]^(2-j)*A[4]*A[5]^(2*k+j+2),
                        j = 0 .. 1), k = 0 .. h)), h = 1 .. i-1))
                      +sum(rho^p*(sum(c[i, p, 2*h+1]*A[1]^(2*i-2*p-2*h)*A[2]^i
                                     *(sum(eta[2*h+1, k, 2]*A[2]*A[3]^(2*k+3)
                                          +sum(eta[2*h+1, k, j]*A[3]^(2-j)*A[4]
                                          *A[5]^(2*k+j+1), 
                                       j = 0 .. 1), k = 1 .. h)
                        +rho^h*(A[3]*A[4]*A[5]^2-A[2]*A[3]^3)), 

                          h = 1 .. i-p)),p = 1 .. i-1));
        # Expand A^{2n+1}{A^*}^{n+1}
        t[4] := c[i, 0, 1]*A[1]^(2*i)*A[2]^i*A[3]*A[4]
                -c[i, 0, 2]*A[1]^(2*i-1)*A[2]^i*A[3]^2*A[4]
            +sum(c[i, 0, 2*k+1]*A[1]^(2*i-2*k)*A[2]^i*(sum(eta[2*k+1, h, 2]
                    *A[2]*A[3]^(2*h+1)
                +sum(eta[2*k+1, h, j]*A[3]^(2-j)*A[4]*A[5]^(2*h-1+j), 
                   j = 0 .. 1), h = 1 .. k)
                +rho^k*(A[3]*A[4]-A[2]*A[3])), k = 1 .. i)
            -(sum(c[i, 0, 2*k+2]*A[1]^(2*i-2*k-1)*A[2]^i
                  *(sum(eta[2*k+2, h, 2]*A[2]*A[3]^(2*h+2)
                        +sum(eta[2*k+2, h, j]*A[3]^(2-j)
                           *A[4]*A[5]^(2*h+j), j = 0 .. 1), 
                    h = 0 .. k)), k = 1 .. i-1))
            -rho^i*(c[i, i, 0]*A[1]*A[2]^(i+1)-c[i, i, 1]*A[2]^i*A[3]*A[4])
            -(sum(rho^p*(c[i, p, 0]*A[1]^(2*i-2*p+1)*A[2]^(i+1)
                 -c[i, p, 1]*A[1]^(2*i-2*p)*A[2]^i*A[3]*A[4]
                 +c[i, p, 2]*A[1]^(2*i-2*p-1)*A[2]^i*A[3]^2*A[4]), 
                       p = 1 .. i-1))
            +sum(rho^p*(sum(c[i, p, 2*k+1]*A[1]^(2*i-2*p-2*k)*A[2]^i
                  *(sum(eta[2*k+1, h, 2]*A[2]*A[3]^(2*h+1)
                  +sum(eta[2*k+1, h, j]*A[3]^(2-j)*A[4]*A[5]^(2*h-1+j), 
                       j = 0 .. 1),h = 1 .. k)
                      +rho^k*(A[3]*A[4]-A[2]*A[3])),
                       k = 1 .. i-p)), p = 1 .. i-1);
        if 2 < i then 
           t[4] := t[4]-(sum(rho^p*(sum(c[i, p, 2*k+2]
                            *A[1]^(2*i-2*p-2*k-1)*A[2]^i
                        *(sum(eta[2*k+2, h, 2]*A[2]*A[3]^(2*h+2)
                            +sum(eta[2*k+2, h, j]*A[3]^(2-j)*A[4]*A[5]^(2*h+j),
                            j = 0 .. 1),h = 0 .. k)),
                            k = 1 .. i-p-1)), p = 1 .. i-2));
           t[3] := t[3]-(sum(rho^p*(sum(c[i, p, 2*h+2]*A[1]^(2*i-2*p-2*h-1)
                          *A[2]^i*(sum(eta[2*h+2, k, 2]*A[2]*A[3]^(2*k+4)
                            +sum(eta[2*h+2, k, j]*A[3]^(2-j)*A[4]*A[5]^(2*k+j+2), 
                            j = 0 .. 1),k = 0 .. h)), 
                            h = 1 .. i-p-1)), p = 1 .. i-2));
        end if; 
        t[3] := simplify(t[3]); 
        t[4] := simplify(t[4]); 
        c[i+1, 0, 0] := 1; 
        c[i+1, 0, 1] := expand(simplify((q^(2*i+3)-q^(-2*i-3))/(q-1/q))); 
        c[i+1, 0, 2] := expand(simplify((q^(2*i+3)-q^(-2*i-3))*
                        (q^(2*i+2)-q^(-2*i-2))/((q-1/q)*(q^2-1/q^2))));
        c[i+1, 1, 0] := expand(simplify(-c[i, 0, 1]^2+2*c[i, 0, 2]
                            -c[i, 0, 3]/c[i, 0, 1]-c[i, 1, 1]/c[i, 0, 1]
                            +2*c[i, 1, 0]));
        # Compute the higher order q-Dolan-Grady relations F[i+1]
        F[i+1] := expand(simplify(c[i+1, 0, 0]*t[1]-c[i+1, 0, 1]*t[2]
                               +c[i+1, 0, 2]*t[3]+rho*c[i+1, 1, 0]*t[4]));
        # Rewrite F[i+1] in the distributed form of A[1], A[2], A[3] and rho
        F[i+1] := collect(F[i+1], [A[1], A[2], A[3], rho], distributed);
        # Extract the coefficient in A[1]and A[3] in the sum of 
        # all elements containing A[2]^{i+1} of F[i+1]
        tam := coeff(F[i+1], A[2], i+1);
        for l from 3 to 2*i+3 do 
           # Extract the coefficient in A[1] in the sum of 
           # all elements containing A[3]^l of tam
           ta1[l] := coeff(tam, A[3], l);
           # Extract the coefficient of A[1]^{2*i+3-l} in ta1 
           ta1[l] := coeff(ta1[l], A[1], 2*i+3-l);
           # Define the coefficient c[i+1, 0, l]
           c[i+1, 0, l] := ta1[l]*(-1)^(l+1); 
        end do; 
        for l to 2*i+1 do 
           # Extract the coefficient in A[3] in the sum 
           # of all elements containing A[1]^{2*i+1-l}
           ta2[l] := coeff(tam, A[3], l);
           ta2[l] := coeff(ta2[l], A[1], 2*i+1-l);
           # Define the coefficient c[i+1, 1, l]
           c[i+1, 1, l] := ta2[l]*(-1)^(l+1)/rho;
        end do; 
        for d from 2 to i+1 do 
           for l from 0 to 2*i-2*d+3 do 
              ta3[l] := coeff(tam, A[3], l);
              ta3[l] := coeff(ta3[l], A[1], 2*i-2*d+3-l);
              c[i+1, d, l] := ta3[l]*(-1)^(l+1)*rho^(-d);
           end do
        end do 
     end do 
  end if; 
  F[n+1] 
end proc;
\end{verbatim}
\begin{verbatim}
# Compute the coefficients c[n+1, p, j]
Result2 := proc (n) 
  local i, p, d, l, M, N;
  global A, c, alpha, q; 
  # Input the innitial values
  alpha := q^2+1+1/q^2;
  c[1, 0, 0] := 1; 
  c[1, 0, 3] := 1; 
  c[1, 0, 1] := alpha; 
  c[1, 0, 2] := alpha; 
  c[1, 1, 0] := 1; 
  c[1, 1, 1] := 1; 
  c[2, 0, 0] := 1; 
  c[2, 0, 5] := 1; 
  c[2, 0, 1] := expand(simplify(alpha^2-alpha-1));
  c[2, 0, 4] := expand(simplify(alpha^2-alpha-1));
  c[2, 0, 2] := expand(simplify(alpha^3-2*alpha^2+1)); 
  c[2, 0, 3] := expand(simplify(alpha^3-2*alpha^2+1)); 
  c[2, 1, 0] := expand(simplify(alpha^2-2*alpha+2)); 
  c[2, 1, 3] := expand(simplify(alpha^2-2*alpha+2)); 
  c[2, 1, 1] := expand(simplify(alpha*(alpha^2-alpha-1))); 
  c[2, 1, 2] := expand(simplify(alpha*(alpha^2-alpha-1))); 
  c[2, 2, 0] := expand(simplify((alpha-1)^2));
  c[2, 2, 1] := expand(simplify((alpha-1)^2)); 
  if 1 < n then 
     for i from 2 to n do
        # Compute M[n, p, j], p = 0,..., n; j= 2,..., 2n+2 if p=0, 
        #                     j = 0,..., 2(n-p)+2 if p = 1,...,n 
        M[i, 0, 2*i+2] := simplify(-c[i, 0, 1]*c[i, 0, 2*i+1]); 
        for l from 2 to 2*i+1 do
           M[i, 0, l] := simplify(c[i, 0, l]-c[i, 0, 1]*c[i, 0, l-1]); 
        end do; 
        for d from 1 to i do 
           M[i, d, 0] := c[i, d, 0];
              for l from 1 to 2*i-2*d+1 do
                 M[i, d, l] := simplify(c[i, d, l]-c[i, 0, 1]*c[i, d, l-1]);
              end do;
           M[i, d, 2*i-2*d+2] := simplify(-c[i, 0, 1]*c[i, d, 2*i-2*d+1]);
        end do;
        Funct(i);
        # Compute N[n, p, j], p = 0,...,n;
        #           j = 3,...,2n+2 if p =0
        #           j = 0,...,2(n-p)+3 if p = 1,...,n
        N[i, 0, 2*i+2] := simplify((c[i, 0, 1]^2-c[i, 0, 2])*c[i, 0, 2*i]
                                   -c[i, 0, 1]*c[i, 0, 2*i+1]); 
        N[i, 0, 2*i+3] := simplify((c[i, 0, 1]^2-c[i, 0, 2])*c[i, 0, 2*i+1]); 
        for l from 3 to 2*i+1 do
           N[i, 0, l] := simplify((c[i, 0, 1]^2-c[i, 0, 2])*c[i, 0, l-2]
                                  -c[i, 0, 1]*c[i, 0, l-1]+c[i, 0, l]);
        end do; 
        N[i, 1, 1] := simplify(-2*c[i, 0, 1]*c[i, 1, 0]+c[i, 1, 1]); 
        N[i, 1, 2*i] := simplify((c[i, 0, 1]^2-c[i, 0, 2])*c[i, 1, 2*i-2]
                                  -c[i, 0, 1]*c[i, 1, 2*i-1]
                                  -c[i, 1, 0]*c[i, 0, 2*i]); 
        N[i, 1, 2*i+1] := simplify((c[i, 0, 1]^2-c[i, 0, 2])*c[i, 1, 2*i-1]
                                   -c[i, 1, 0]*c[i, 0, 2*i+1]); 
        N[i, 1, 0] := 0; 
        for l from 2 to 2*i-1 do
           N[i, 1, l] := simplify((c[i, 0, 1]^2-c[i, 0, 2])*c[i, 1, l-2]
                                  -c[i, 0, 1]*c[i, 1, l-1]+c[i, 1, l]
                                  -c[i, 1, 0]*c[i, 0, l]);
        end do; 
        for d from 2 to i do 
           for l from 2 to 2*i-2*d+1 do
              N[i, d, l] := simplify((c[i, 0, 1]^2-c[i, 0, 2])*c[i, d, l-2]
                                     -c[i, 0, 1]*c[i, d, l-1]+c[i, d, l]
                                     -c[i, 1, 0]*c[i, d-1, l]);
           end do; 
           N[i, d, 0] := simplify(c[i, d, 0]-c[i, 1, 0]*c[i, d-1, 0]); 
           N[i, d, 1] := simplify(-c[i, 0, 1]*c[i, d, 0]+c[i, d, 1]
                                  -c[i, 1, 0]*c[i, d-1, 1]); 
           N[i, d, 2*i-2*d+2] := simplify((c[i, 0, 1]^2-c[i, 0, 2])
                                          *c[i, d, 2*i-2*d]
                                          -c[i, 0, 1]*c[i, d, 2*i-2*d+1]
                                          -c[i, 1, 0]*c[i, d-1, 2*i-2*d+2]); 
           N[i, d, 2*i-2*d+3] := simplify((c[i, 0, 1]^2-c[i, 0, 2])
                                          *c[i, d, 2*i-2*d+1]
                                          -c[i, 1, 0]*c[i, d-1, 2*i-2*d+3]) 
        end do; 
        for l from 0 to 1 do 
           N[i, i+1, l] := simplify(-c[i, 1, 0]*c[i, i, l]);
        end do; 
        # Compute the coefficients c[n+1, p, j], p =0,...,n+1; j = 0,...,2n+3-2p.
        c[i+1, 0, 0] := 1; 
        c[i+1, 0, 1] := expand(simplify((q^(2*i+3)-q^(-2*i-3))/(q-1/q)));
        c[i+1, 0, 2] := expand(simplify((q^(2*i+3)-q^(-2*i-3))
                        *(q^(2*i+2)-q^(-2*i-2))/((q^2-1/q^2)*(q-1/q)))); 
        c[i+1, 0, 3] := expand(simplify(N[i, 0, 3]*eta[3, 1, 2])); 
        c[i+1, 0, 4] := expand(simplify(N[i, 0, 4]*eta[4, 1, 2]
                        +c[i+1, 0, 1]*M[i, 0, 3]*eta[3, 1, 2])); 
        for l from 2 to i+1 do    
           c[i+1, 0, 2*l+1] := expand(simplify(N[i, 0, 2*l+1]*eta[2*l+1, l, 2]
                             +c[i+1, 0, 1]*M[i, 0, 2*l]*eta[2*l, l-1, 2]
                             +c[i+1, 0, 2]*c[i, 0, 2*l-1]*eta[2*l-1, l-1, 2]));
        end do; 
        for l from 2 to i do
           c[i+1, 0, 2*l+2] := expand(simplify(N[i, 0, 2*l+2]*eta[2*l+2, l, 2]
                             +c[i+1, 0, 1]*M[i, 0, 2*l+1]*eta[2*l+1, l, 2]
                             +c[i+1, 0, 2]*c[i, 0, 2*l]*eta[2*l, l-1, 2]));
        end do; 
        c[i+1, 1, 0] := expand(simplify(c[i, 0, 1]^2-2*c[i, 0, 2]
                        +c[i, 0, 3]/c[i, 0, 1]
                        -c[i, 1, 1]/c[i, 0, 1]+2*c[i, 1, 0])); 
        c[i+1, i+1, 0] := expand(simplify(N[i, i+1, 0]+c[i+1, 1, 0]
                                          *c[i, i, 0])); 
        for d from 2 to i do
           c[i+1, d, 0] := expand(simplify(N[i, d, 0]+c[i+1, 1, 0]
                                          *c[i, d-1, 0]));
        end do; 
        c[i+1, 1, 1] := expand(simplify(N[i, 0, 3]+c[i+1, 0, 1]*M[i, 1, 0])); 
        c[i+1, 2, 1] := expand(simplify(-N[i, 0, 5]+N[i, 1, 3]
                        +c[i+1, 1, 0]*c[i, 0, 3]+c[i+1, 0, 1]*M[i, 2, 0])); 
        c[i+1, i+1, 1] := expand(simplify((-1)^i*(N[i, 0, 2*i+3]
                        +sum((-1)^p*N[i, p, 2*i-2*p+3], p = 1 .. i))
                        +(-1)^(i+1)*c[i+1, 1, 0]*(c[i, 0, 2*i+1]
                        +sum((-1)^p*c[i, p, 2*i-2*p+1], p = 1 .. i-1))));
        if 2 < i then 
           for d from 3 to i do 
              c[i+1, d, 1] := expand(simplify((-1)^(d+1)*(N[i, 0, 2*d+1]
                        +sum((-1)^j*N[i, j, 2*d-2*j+1], j = 1 .. d-1))
                        +c[i+1, 0, 1]*M[i, d, 0]+(-1)^d
                        *c[i+1, 1, 0]*(c[i, 0, 2*d-1]
                        +sum((-1)^j*c[i, j, 2*d-2*j-1], j = 1 .. d-2)))); 
              c[i+1, d, 2] := expand(simplify((-1)^d*(eta[2*d+2, 0, 2]
                               *N[i, 0, 2*d+2]
                        +sum((-1)^j*N[i, j, 2*d-2*j+2]*eta[2*d-2*j+2, 0, 2], 
                            j = 1 .. d-1))
                        +(-1)^(d+1)*c[i+1, 0, 1]*(M[i, 0, 2*d+1]
                        +sum((-1)^j*M[i, j, 2*d-2*j+1], j = 1 .. d-1))
                              +c[i+1, 0, 2]*c[i, d, 0]
                        +(-1)^(d+1)*c[i+1, 1, 0]*(c[i, 0, 2*d]*eta[2*d, 0, 2]
                        +sum((-1)^j*c[i, j, 2*d-2*j]*eta[2*d-2*j, 0, 2], 
                              j = 1 .. d-2)))); 
           end do; 
           for d from 2 to i-1 do
              c[i+1, d, 4] := expand(simplify((-1)^d*(N[i, 0, 2*d+4]
                               *eta[2*d+4, 1, 2]
                +sum((-1)^j*N[i, j, 2*d-2*j+4]*eta[2*d-2*j+4, 1, 2], j = 1 .. d))
                +(-1)^d*c[i+1, 0, 2]*(c[i, 0, 2*d+2]*eta[2*d+2, 0, 2]
                +sum((-1)^j*c[i, j, 2*d-2*j+2]*eta[2*d-2*j+2, 0, 2], j = 1 .. d-1))
                +(-1)^d*c[i+1, 0, 1]*(M[i, 0, 2*d+3]*eta[2*d+3, 1, 2]
                +sum((-1)^j*M[i, j, 2*d-2*j+3]*eta[2*d-2*j+3, 1, 2], j = 1 .. d))
                +(-1)^(d+1)*c[i+1, 1, 0]*(c[i, 0, 2*d+2]*eta[2*d+2, 1, 2]
                +sum((-1)^j*c[i, j, 2*d-2*j+2]*eta[2*d-2*j+2, 1, 2], j = 1 .. d-1))));
           end do; 
           for d from 3 to i do 
              for l to d-2 do 
                 c[i+1, d-l, 2*l+3] := expand(simplify((-1)^(d+l)
                                              *(N[i, 0, 2*d+3]*eta[2*d+3, l+1, 2]
                         +sum((-1)^j*N[i, j, 2*d-2*j+3]*eta[2*d-2*j+3, l+1, 2], 
                                j = 1 .. d-l))
                         +(-1)^(d+l)*c[i+1, 0, 2]*(c[i, 0, 2*d+1]*eta[2*d+1, l, 2]
                         +sum((-1)^j*c[i, j, 2*d-2*j+1]*eta[2*d-2*j+1, l, 2], 
                                j = 1 .. d-l))
                         +(-1)^(d+l)*c[i+1, 0, 1]*(eta[2*d+2, l, 2]*M[i, 0, 2*d+2]
                         +sum((-1)^j*M[i, j, 2*d-2*j+2]*eta[2*d-2*j+2, l, 2], 
                                j = 1 .. d-l))
                         +(-1)^(d+l+1)*c[i+1, 1, 0]*(c[i, 0, 2*d+1]*eta[2*d+1, l+1, 2]
                         +sum((-1)^j*c[i, j, 2*d-2*j+1]*eta[2*d-2*j+1, l+1, 2], 
                                 j = 1 .. d-l-1))));
              end do 
           end do; 
           for l from 3 to i do 
              c[i+1, 1, 2*l] := expand(simplify(c[i+1, 1, 0]
                                  *c[i, 0, 2*l]*eta[2*l, l-1, 2]
                        -N[i, 0, 2*l+2]*eta[2*l+2, l-1, 2]
                        +N[i, 1, 2*l]*eta[2*l, l-1, 2]
                        -c[i+1, 0, 1]*(M[i, 0, 2*l+1]*eta[2*l+1, l-1, 2]
                        -M[i, 1, 2*l-1]*eta[2*l-1, l-1, 2])
                        -c[i+1, 0, 2]*(c[i, 0, 2*l]*eta[2*l, l-2, 2]
                        -c[i, 1, 2*l-2]*eta[2*l-2, l-2, 2])));
           end do 
        end if; 
        for d from 2 to i do 
           c[i+1, d, 3] := expand(simplify((-1)^d*c[i+1, 0, 1]
                              *(eta[2*d+2, 0, 2]*M[i, 0, 2*d+2]
                           +sum((-1)^j*M[i, j, 2*d-2*j+2]*eta[2*d-2*j+2, 0, 2], 
                                j = 1 .. d-1))
                           +(-1)^(d+1)*c[i+1, 0, 2]*(c[i, 0, 2*d+1]
                           +sum((-1)^j*c[i, j, 2*d-2*j+1], j = 1 .. d-1))
                           +(-1)^d*(N[i, 0, 2*d+3]*eta[2*d+3, 1, 2]
                           +sum((-1)^j*N[i, j, 2*d-2*j+3]*eta[2*d-2*j+3, 1, 2], 
                                j = 1 .. d))
                           +(-1)^(d+1)*c[i+1, 1, 0]*(c[i, 0, 2*d+1]*eta[2*d+1, 1, 2]
                           +sum((-1)^j*c[i, j, 2*d-2*j+1]*eta[2*d-2*j+1, 1, 2], 
                                j = 1 .. d-1))));
        end do; 
        for l from 2 to i do 
           c[i+1, 1, 2*l+1] := expand(simplify(-N[i, 0, 2*l+3]*eta[2*l+3, l, 2]
                           +N[i, 1, 2*l+1]*eta[2*l+1, l, 2]
                           -c[i+1, 0, 2]*(c[i, 0, 2*l+1]*eta[2*l+1, l-1, 2]
                           -c[i, 1, 2*l-1]*eta[2*l-1, l-1, 2])
                           -c[i+1, 0, 1]*(eta[2*l+2, l-1, 2]*M[i, 0, 2*l+2]
                           -M[i, 1, 2*l]*eta[2*l, l-1, 2])
                           +c[i+1, 1, 0]*c[i, 0, 2*l+1]*eta[2*l+1, l, 2]));
        end do; 
        if 3 < i then 
           for d from 4 to i do 
              for l from 2 to d-2 do
                 c[i+1, d-l, 2*l+2] := expand(simplify((-1)^(d+l)*(N[i, 0, 2*d+2]
                                         *eta[2*d+2, l, 2]
                      +sum((-1)^j*N[i, j, 2*d-2*j+2]*eta[2*d-2*j+2, l, 2], 
                           j = 1 .. d-l))
                      +(-1)^(d+l)*c[i+1, 0, 1]*(M[i, 0, 2*d+1]*eta[2*d+1, l, 2]
                      +sum((-1)^j*M[i, j, 2*d-2*j+1]*eta[2*d-2*j+1, l, 2], 
                           j = 1 .. d-l))
                      +(-1)^(d+l)*c[i+1, 0, 2]*(c[i, 0, 2*d]*eta[2*d, l-1, 2]
                      +sum((-1)^j*c[i, j, 2*d-2*j]*eta[2*d-2*j, l-1, 2], 
                           j = 1 .. d-l))
                      +(-1)^(d+l+1)*c[i+1, 1, 0]*(c[i, 0, 2*d]*eta[2*d, l, 2]
                      +sum((-1)^j*c[i, j, 2*d-2*j]*eta[2*d-2*j, l, 2], 
                           j = 1 .. d-l-1)))); 
              end do 
           end do 
        end if; 
        c[i+1, 1, 4] := expand(simplify(-N[i, 0, 6]*eta[6, 1, 2]
                      +N[i, 1, 4]*eta[4, 1, 2]
                      -eta[4, 0, 2]*c[i, 0, 4]*c[i+1, 0, 2]-c[i+1, 0, 1]
                         *(M[i, 0, 5]*eta[5, 1, 2]
                      -M[i, 1, 3]*eta[3, 1, 2])+c[i+1, 1, 0]*c[i, 0, 4]
                         *eta[4, 1, 2])); 
        c[i+1, 1, 3] := expand(simplify(-c[i+1, 0, 1]*eta[4, 0, 2]*M[i, 0, 4]
                      +c[i+1, 1, 0]*c[i, 0, 3]*eta[3, 1, 2]
                      +c[i+1, 0, 2]*c[i, 0, 3]
                      -N[i, 0, 5]*eta[5, 1, 2]+N[i, 1, 3]*eta[3, 1, 2])); 
        c[i+1, 1, 2] := expand(simplify(-eta[4, 0, 2]*N[i, 0, 4]
                      +c[i+1, 0, 1]*M[i, 0, 3]
                      +c[i+1, 0, 2]*c[i, 1, 0])); 
        c[i+1, 2, 2] := -expand(simplify(-eta[6, 0, 2]*N[i, 0, 6]
                      +N[i, 1, 4]*eta[4, 0, 2]
                      -c[i+1, 0, 1]*(-M[i, 0, 5]+M[i, 1, 3])
                      -c[i+1, 0, 2]*c[i, 2, 0]
                      +c[i+1, 1, 0]*c[i, 0, 4]*eta[4, 0, 2])); 
     end do 
  end if;
  # Display the coefficients c[n+1, p, j] = coef[n+1, p, j] 
  for d from 0 to n+1 do 
     for l from 0 to 2*n+3-2*d do
        print(coef[n+1, d, l], "=", c[n+1, d, l]);
     end do 
  end do 
end proc;
\end{verbatim}
The coefficients c[n+1, p, j] which are obtained from the two programming Result1 and Result2  completely coincide together.
\newpage

\begin{center}
\section{APPENDIX E: Algorithms for the generalized $q$-Onsager algebra}
\end{center}
\begin{verbatim}
# Compute the coefficients eta[m,k,j] of A^{2n+1}A^* and A^{2n+2}A^*
Funct := proc (n) 
  local p, i, j; 
  global eta; 
  # Input the initial values
  eta[2, 0, 0] := q+1/q; 
  eta[2, 0, 1] := -1; 
  for i to n do 
     for p from 0 to i-1 do 
        eta[2*i+1, p, 0] := expand(simplify(eta[2*i, p, 0]*(q+1/q)
                                 +eta[2*i, p, 1])); 
     end do; 
     eta[2*i+1, i, 0] := 1; 
     for p to i-1 do 
        eta[2*i+1, p, 1] := expand(simplify(-eta[2*i, p, 0]
                                 +eta[2*i, p-1, 0])); 
     end do; 
     eta[2*i+1, 0, 1] := expand(simplify(-eta[2*i, 0, 0])); 
     eta[2*i+1, i, 1] := expand(simplify(eta[2*i, i-1, 0])); 
     for p from 0 to i do 
        eta[2*i+2, p, 0] := expand(simplify(eta[2*i+1, p, 0]*(q+1/q)
                                              +eta[2*i+1, p, 1]));
     end do; 
     for p to i do 
       eta[2*i+2, p, 1] := expand(simplify(eta[2*i+1, p-1, 0]
                                             -eta[2*i+1, p, 0]));
     end do; 
     eta[2*i+2, 0, 1] := expand(simplify(-eta[2*i+1, 0, 0])) 
  end do 
end proc;
\end{verbatim}
\begin{verbatim}
# Compute A^{r+2}{A^*}^{r+1}-c[r+1,0,1]A^{r+1}{A^*}^{r+1}A
# and output the higher order relations of the generalized q-Onsager algebra
Result1 := proc (r) 
  local k, i, t, p, F, f, tam, ta2, ta1, B;
  global M, c, A, eta; 
  # Input the initial values
  c[1, 0, 0] := 1; 
  c[1, 0, 1] := q+1/q; 
  c[1, 0, 2] := 1; 
  c[1, 1, 0] := 1; 
  c[2, 0, 0] := 1; 
  c[2, 0, 1] := q^2+1+1/q^2; 
  c[2, 0, 2] := q^2+1+1/q^2; 
  c[2, 0, 3] := 1; 
  c[2, 1, 0] := q^2+2+1/q^2; 
  c[2, 1, 1] := q^2+2+1/q^2; 
  c[3, 0, 0] := 1; 
  c[3, 0, 1] := expand((q^4-1/q^4)/(q-1/q)); 
  c[3, 0, 2] := expand((q^2+1/q^2)*(q^2+1+1/q^2)); 
  c[3, 0, 3] := expand((q^4-1/q^4)/(q-1/q)); 
  c[3, 0, 4] := 1; 
  c[3, 1, 0] := q^4+2*q^2+4+2/q^2+1/q^4; 
  c[3, 1, 1] := expand((q+1/q)*(q^2+1/q^2+3)*(q^2+1/q^2)); 
  c[3, 1, 2] := q^4+2*q^2+4+2/q^2+1/q^4; 
  c[3, 2, 0] := expand(q^2+1+1/q^2); 
  if r = 0 then 
     F[1] := A[1]^2*A[2]-(q+1/q)*A[1]*A[2]*A[3]+A[2]*A[3]^2-rho*A[2] = 0 
  end if; 
  if r = 1 then 
     F[2] := A[1]^3*A[2]^2-(q^2+1+1/q^2)*A[1]^2*A[2]^2*A[3]
            +(q^2+1+1/q^2)*A[1]*A[2]^2*A[3]^2
            -A[2]^2*A[3]^3-rho*(q^2+2+1/q^2)*A[1]*A[2]^2
            +rho*(q^2+2+1/q^2)*A[2]^2*A[3] = 0
  end if; 
  if 1 < r then 

     for i from 2 to r do 
        t := trunc((1/2)*i);
        # Compute the coefficients M[r,p,j] 
        for k from 2 to i+1 do 
           M[i, 0, k] := -c[i, 0, 1]*c[i, 0, k-1]+c[i, 0, k]; 
        end do; 
        M[i, 0, i+2] := -c[i, 0, 1]*c[i, 0, i+1];
        if i = 2*t then 
           for p to t do 
              for k to i-2*p+1 do 
                 M[i, p, k] := -c[i, 0, 1]*c[i, p, k-1]+c[i, p, k]; 
              end do; 
              M[i, p, 0] := c[i, p, 0]; 
              M[i, p, i+2-2*p] := -c[i, 0, 1]*c[i, p, i-2*p+1]; 
           end do 
        else 
           for p to t do 
              M[i, p, 0] := c[i, p, 0]; 
              M[i, p, i+2-2*p] := -c[i, 0, 1]*c[i, p, i-2*p+1]; 
              for k to i-2*p+1 do 
                 M[i, p, k] := -c[i, 0, 1]*c[i, p, k-1]+c[i, p, k]; 
              end do 
           end do; 
           M[i, t+1, 0] := c[i, t+1, 0]; 
           M[i, t+1, 1] := -c[i, t+1, 0]*c[i, 0, 1]; 
        end if; 
        Funct(i); 
        if i = 2*t then
           # Expand A^{2r+2}{A^*}^{r+1}A 
           f[1] := simplify(-(sum(M[2*t, 0, 2*h]*A[1]^(2*t+2-2*h)*A[2]^(2*t)
                      *(rho^h*A[2]+sum(rho^l*(eta[2*h, l, 0]
                        *A[3]*A[4]*A[5]^(2*h-1-2*l)
                         +eta[2*h, l, 1]*A[2]*A[3]^(2*h-2*l)), 
                         l = 0 .. h-1)), h = 1 .. t+1))
                  +sum(M[2*t, 0, 2*h+1]*A[1]^(2*t+1-2*h)*A[2]^(2*t)
                      *(sum(rho^l*(eta[2*h+1, l, 0]*A[3]*A[4]*A[5]^(2*h-2*l)
                         +eta[2*h+1, l, 1]*A[2]*A[3]^(2*h-2*l+1)), 
                         l = 0 .. h)), h = 1 .. t)
                  -(sum(rho^d*((-1)^d*M[2*t, d, 0]*A[1]^(2*t+2-2*d)
                     *A[2]^(2*t+1)+(-1)^(d+1)
                  *M[2*t, d, 1]*A[1]^(2*t+1-2*d)*A[2]^(2*t)*A[3]*A[4]),
                     d = 1 .. t))
                  -(sum(rho^d*(sum((-1)^d*M[2*t, d, 2*h]
                        *A[1]^(2*t+2-2*d-2*h)*A[2]^(2*t)
                            *(rho^h*A[2]+sum(rho^l*(eta[2*h, l, 0]
                                      *A[3]*A[4]*A[5]^(2*h-1-2*l)
                            +eta[2*h, l, 1]*A[2]*A[3]^(2*h-2*l)), 
                      l = 0 .. h-1)), h = 1 .. t+1-d)), d = 1 .. t))
                  -(sum(rho^d*(sum((-1)^(d+1)*M[2*t, d, 2*h+1]
                           *A[1]^(2*t+1-2*d-2*h)
                      *A[2]^(2*t)*(sum(rho^l*(eta[2*h+1, l, 0]
                          *A[3]*A[4]*A[5]^(2*h-2*l)
                          +eta[2*h+1, l, 1]*A[2]*A[3]^(2*h-2*l+1)), 
                   l = 0 .. h)), h = 1 .. t-d)), d = 1 .. t)));
           # Expand A^{2r+1}{A^*}^{r+1}A^2 
           f[2] := simplify(c[2*t, 0, 1]*A[1]^(2*t)*A[2]^(2*t)*A[3]*A[4]*A[5]
                -(sum(c[2*t, 0, 2*h]*A[1]^(2*t+1-2*h)*A[2]^(2*t)
                     *(rho^h*A[2]*A[3]
                    +sum(rho^l*(eta[2*h, l, 0]*A[3]*A[4]*A[5]^(2*h-2*l)
                       +eta[2*h, l, 1]*A[2]*A[3]^(2*h-2*l+1)), 
                         l = 0 .. h-1)), h = 1 .. t))
                +sum(c[2*t, 0, 2*h+1]*A[1]^(2*t-2*h)*A[2]^(2*t)
                    *(sum(rho^l*(eta[2*h+1, l, 0]*A[3]*A[4]*A[5]^(2*h-2*l+1)
                         +eta[2*h+1, l, 1]*A[2]*A[3]^(2*h-2*l+2)), 
                            l = 0 .. h)), h = 1 .. t)
                -(sum(rho^d*((-1)^d*c[2*t, d, 0]*A[1]^(2*t+1-2*d)
                   *A[2]^(2*t+1)*A[3]+(-1)^(d+1)
                     *c[2*t, d, 1]*A[1]^(2*t-2*d)*A[2]^(2*t)
                     *A[3]*A[4]*A[5]), d = 1 .. t))); 
           if 2 <= t then 
              f[2] := simplify(f[2]-(sum(rho^d*(sum((-1)^d*c[2*t, d, 2*h]
                             *A[1]^(2*t+1-2*d-2*h)*A[2]^(2*t)*(rho^h*A[2]*A[3]
                                   +sum(rho^l*(eta[2*h, l, 0]*A[3]
                                       *A[4]*A[5]^(2*h-2*l)
                                     +eta[2*h, l, 1]*A[2]*A[3]^(2*h-2*l+1)), 
                                     l = 0 .. h-1)), 
                              h = 1 .. t-d)), d = 1 .. t-1))
                         -(sum(rho^d*(sum((-1)^(d+1)*c[2*t, d, 2*h+1]
                               *A[1]^(2*t-2*d-2*h)
                            *A[2]^(2*t)*(sum(rho^l*(eta[2*h+1, l, 0]
                                   *A[3]*A[4]*A[5]^(2*h-2*l+1)
                              +eta[2*h+1, l, 1]*A[2]*A[3]^(2*h-2*l+2)), 
                         l = 0 .. h)),h = 1 .. t-d)), d = 1 .. t-1))); 
           end if 
        else 
           f[1] := simplify(-(sum(M[2*t+1, 0, 2*h]*A[1]^(2*t+3-2*h)
                      *A[2]^(2*t+1)*(rho^h*A[2]
                         +sum(rho^l*(eta[2*h, l, 0]*A[3]*A[4]*A[5]^(2*h-1-2*l)
                            +eta[2*h, l, 1]*A[2]*A[3]^(2*h-2*l)), 
                            l = 0 .. h-1)),h = 1 .. t+1))
                      +sum(M[2*t+1, 0, 2*h+1]*A[1]^(2*t+2-2*h)*A[2]^(2*t+1)
                        *(sum(rho^l*(eta[2*h+1, l, 0]*A[3]*A[4]*A[5]^(2*h-2*l)
                         +eta[2*h+1, l, 1]*A[2]*A[3]^(2*h-2*l+1)), 
                         l = 0 .. h)), h = 1 .. t+1)
                      -(sum(rho^d*((-1)^d*M[2*t+1, d, 0]*A[1]^(2*t+3-2*d)
                           *A[2]^(2*t+2)
                           +(-1)^(d+1)*M[2*t+1, d, 1]*A[1]^(2*t+2-2*d)
                           *A[2]^(2*t+1)*A[3]*A[4]), d = 1 .. t+1))
                      -(sum(rho^d*(sum((-1)^d*M[2*t+1, d, 2*h]
                             *A[1]^(2*t+3-2*d-2*h)
                         *A[2]^(2*t+1)*(rho^h*A[2]
                         +sum(rho^l*(eta[2*h, l, 0]*A[3]*A[4]
                            *A[5]^(2*h-1-2*l)+eta[2*h, l, 1]*A[2]*A[3]^(2*h-2*l)), 
                         l = 0 .. h-1)), h = 1 .. t+1-d)), d = 1 .. t))
                         -(sum(rho^d*(sum((-1)^(d+1)*M[2*t+1, d, 2*h+1]
                              *A[1]^(2*t+2-2*d-2*h)
                           *A[2]^(2*t+1)*(sum(rho^l*(eta[2*h+1, l, 0]
                           *A[3]*A[4]*A[5]^(2*h-2*l)
                           +eta[2*h+1, l, 1]*A[2]*A[3]^(2*h-2*l+1)), 
                          l = 0 .. h)), h = 1 .. t+1-d)), d = 1 .. t)));
           f[2] := simplify(c[2*t+1, 0, 1]*A[1]^(2*t+1)*A[2]^(2*t+1)*A[3]*A[4]*A[5]
                 -(sum(c[2*t+1, 0, 2*h]*A[1]^(2*t+2-2*h)
                      *A[2]^(2*t+1)*(rho^h*A[2]*A[3]
                    +sum(rho^l*(eta[2*h, l, 0]*A[3]*A[4]*A[5]^(2*h-2*l)
                    +eta[2*h, l, 1]*A[2]*A[3]^(2*h-2*l+1)),
                     l = 0 .. h-1)), h = 1 .. t+1))
                 +sum(c[2*t+1, 0, 2*h+1]*A[1]^(2*t+1-2*h)*A[2]^(2*t+1)
                    *(sum(rho^l*(eta[2*h+1, l, 0]*A[3]*A[4]
                    *A[5]^(2*h-2*l+1)+eta[2*h+1, l, 1]
                      *A[2]*A[3]^(2*h-2*l+2)), l = 0 .. h)), h = 1 .. t)
                 -rho^(t+1)*(-1)^(t+1)*c[2*t+1, t+1, 0]*A[2]^(2*t+2)*A[3]
                 -(sum(rho^d*((-1)^(d+1)*c[2*t+1, d, 1]*A[1]^(2*t+1-2*d)
                    *A[2]^(2*t+1)*A[3]*A[4]*A[5]+(-1)^d*c[2*t+1, d, 0]
                    *A[1]^(2*t+2-2*d)*A[2]^(2*t+2)*A[3]), d = 1 .. t))
                 -(sum(rho^d*(sum((-1)^d*c[2*t+1, d, 2*h]
                       *A[1]^(2*t+2-2*d-2*h)*A[2]^(2*t+1)
                     *(rho^h*A[2]*A[3]+sum(rho^l*(eta[2*h, l, 0]
                     *A[3]*A[4]*A[5]^(2*h-2*l)
                      +eta[2*h, l, 1]*A[2]*A[3]^(2*h-2*l+1)),
                       l = 0 .. h-1)), h = 1 .. t+1-d)), d = 1 .. t))
                 -(sum(rho^d*(sum((-1)^(d+1)*c[2*t+1, d, 2*h+1]
                 *A[1]^(2*t+1-2*d-2*h)*A[2]^(2*t+1)
                     *(sum(rho^l*(eta[2*h+1, l, 0]*A[3]*A[4]*A[5]^(2*h-2*l+1)
                       +eta[2*h+1, l, 1]*A[2]*A[3]^(2*h-2*l+2)), 
                     l = 0 .. h)), h = 1 .. t-d)), d = 1 .. t)));
        end if; 
        c[i+1, 0, 0] := 1; 
        c[i+1, 0, 1] := expand(simplify((q^(i+2)-q^(-i-2))/(q-1/q))); 
        # Compute A^{i+2}{A^*}^{i+1}-c[i+1,0,1]*A^{i+1}{A^*}^{i+1}A
        B[i+1] := expand(simplify(f[1]-c[i+1, 0, 1]*f[2])); 
        # Rewrite B[i+1] in the distributed form of A[1], A[2], A[3] and rho
        B[i+1] := collect(B[i+1], [A[1], A[2], A[3], rho], distributed);
        # Extract the coefficient in A[1] and A[3] in the sum 
        # of all elements containing A[2]^{i+1} of B[i+1]
        tam := coeff(B[i+1], A[2], i+1); 
        for k from 2 to i+2 do
           # Extract the coefficient in A[1] in the sum 
           # of all elements containing A[3]^k of tam 
           ta1[k] := coeff(tam, A[3], k);
           # Extract the coefficient of A[1]^{i+2-k} in ta1[k] 
           ta1[k] := coeff(ta1[k], A[1], i+2-k);
           # Extract the coefficient c[i+1, 0, k] of
            (-1)^{k+l}A[1]^{i+2-k}A[2]^{i+1}A[3]^k
           c[i+1, 0, k] := ta1[k]*(-1)^(k+1); 
        end do; 
        for p to trunc((1/2)*i+1) do 
           for k from 0 to i+2-2*p do 
              # Extract the coefficient in A[1] in the sum 
              # of all elements containing A[3]^k of tam
              ta2[k] := coeff(tam, A[3], k);
              # Extract the coefficient of A[1]^{i+2-2*p-k} in ta2k
              ta2[k] := coeff(ta2[k], A[1], i+2-2*p-k);
              # Extract the coefficient c[i+1,p,k] 
              # of (-1)^{k+p+1}*rho^{p}A[1]^{i+2-2*p-k}A[2]^{i+1}A[3]^k
              c[i+1, p, k] := ta2[k]*rho^(-p)*(-1)^(k+p+1);
           end do 
        end do;
        # Compute the higher order relations of 
        # the generalized q-Onsager algebra for r = i+1 
        F[i+1] := A[1]^(i+2)*A[2]^(i+1)
                  -c[i+1, 0, 1]*A[1]^(i+1)*A[2]^(i+1)*A[3]-B[i+1] =0
     end do 
  end if;
  # Output the higher order relations of the
  # generalized q-Onsager algebra for r+1  
  F[r+1]; 
end proc; 
\end{verbatim}
\begin{verbatim}
# Print the coefficients c[r+1, p, k] of the higher 
# order relations of the generalized q-Onsager algebra
pri := proc (r)
  local p, k, coef; 
  global c; 
  for p from 0 to trunc((1/2)*r+1) do 
     for k from 0 to r+2-2*p do 
        print(coef[r+1, p, k], "=", c[r+1, p, k])
     end do 
  end do 
end proc;
\end{verbatim}
\begin{verbatim}
# Compute the coefficients c[n+1,p,j] by recursion relations
Result2 := proc (n) 
  local p, i, j, k, t, coef; 
  global eta, c, M; 
  # Input the initial values
  c[1, 0, 0] := 1; 
  c[1, 0, 1] := q+1/q; 
  c[1, 0, 2] := 1; 
  c[1, 1, 0] := 1; 
  c[2, 0, 0] := 1; 
  c[2, 0, 1] := q^2+1+1/q^2; 
  c[2, 0, 2] := q^2+1+1/q^2; 
  c[2, 0, 3] := 1; 
  c[2, 1, 0] := q^2+2+1/q^2; 
  c[2, 1, 1] := q^2+2+1/q^2; 
  c[3, 0, 0] := 1; 
  c[3, 0, 1] := expand((q^4-1/q^4)/(q-1/q)); 
  c[3, 0, 2] := expand((q^2+1/q^2)*(q^2+1+1/q^2)); 
  c[3, 0, 3] := expand((q^4-1/q^4)/(q-1/q)); 
  c[3, 0, 4] := 1; 
  c[3, 1, 0] := q^4+2*q^2+4+2/q^2+1/q^4; 
  c[3, 1, 1] := expand((q+1/q)*(q^2+1/q^2+3)*(q^2+1/q^2)); 
  c[3, 1, 2] := q^4+2*q^2+4+2/q^2+1/q^4; 
  c[3, 2, 0] := expand(q^2+1+1/q^2); 
  Function(n);
  if 1 < n then 
     for i from 2 to n do 
        t := trunc((1/2)*i);
        # Compute the coefficients M[r, p, j] of the expansion 
        #  of A^{r+2}{A^*}^{r+1}
        for k from 2 to i+1 do 
           M[i, 0, k] := -c[i, 0, 1]*c[i, 0, k-1]+c[i, 0, k]; 
        end do; 
        M[i, 0, i+2] := -c[i, 0, 1]*c[i, 0, i+1];
        if i = 2*t then 
           for p to t do 
              for k to i-2*p+1 do 
                 M[i, p, k] := -c[i, 0, 1]*c[i, p, k-1]+c[i, p, k]; 
              end do; 
              M[i, p, 0] := c[i, p, 0]; 
              M[i, p, i+2-2*p] := -c[i, 0, 1]*c[i, p, i-2*p+1]; 
           end do 
        else 
           for p to t do 
              M[i, p, 0] := c[i, p, 0]; 
              M[i, p, i+2-2*p] := -c[i, 0, 1]*c[i, p, i-2*p+1]; 
              for k to i-2*p+1 do 
                 M[i, p, k] := -c[i, 0, 1]*c[i, p, k-1]+c[i, p, k]; 

              end do 
           end do; 
           M[i, t+1, 0] := c[i, t+1, 0]; 
           M[i, t+1, 1] := -c[i, t+1, 0]*c[i, 0, 1]
        end if; 
        # Compute the coefficients c[i+1, p, j]
        c[i+1, 0, 0] := 1; 
        c[i+1, 0, 1] := expand(simplify((q^(i+2)-q^(-i-2))/(q-1/q))); 
        if i = 2*t then 
           c[i+1, 0, 2] := expand(simplify(M[i, 0, 2]*eta[2, 0, 1])); 
           for k from 2 to t+1 do 
              c[i+1, 0, 2*k] := expand(simplify(M[i, 0, 2*k]*eta[2*k, 0, 1]
                                 +c[i+1, 0, 1]*c[i, 0, 2*k-1]*eta[2*k-1, 0, 1])); 
           end do; 
           for k to t do 
              c[i+1, 0, 2*k+1] := expand(simplify(M[i, 0, 2*k+1]*eta[2*k+1, 0, 1]
                                 +c[i+1, 0, 1]*c[i, 0, 2*k]*eta[2*k, 0, 1]));
           end do; 
           c[i+1, t+1, 0] := expand(simplify((-1)^t*(-M[i, 0, i+2]
                                 -(sum((-1)^l*M[i, l, i+2-2*l], l = 1 .. t)))));
           c[i+1, 1, 0] := expand(simplify(-M[i, 0, 2]+M[i, 1, 0])); 
           if 2 <= t then 
              for k from 2 to t do 
                 c[i+1, k, 0] := expand(simplify((-1)^k*(sum((-1)^l
                                   *M[i, l, 2*k-2*l], l = 0 .. k))));
              end do 
           end if; 
           c[i+1, 1, 1] := expand(simplify(-M[i, 0, 3]*eta[3, 1, 1]
                                  +c[i+1, 0, 1]*(-c[i, 0, 2]+c[i, 1, 0]))); 
           if 2 <= t then 
              for k from 2 to t do 
                 c[i+1, k, 1] := expand(simplify((-1)^k*(-(sum((-1)^(l+1)
                 *M[i, l, 2*k-2*l+1]
                                    *eta[2*k-2*l+1, k-l, 1], l = 0 .. k-1))
                                +c[i+1, 0, 1]*(sum((-1)^l*c[i, l, 2*k-2*l], 
                                l = 0 .. k)))));
              end do 
           end if; 
           c[i+1, 1, 2] := expand(simplify(-M[i, 0, 4]*eta[4, 1, 1]
                           +M[i, 1, 2]*eta[2, 0, 1]
                           -c[i+1, 0, 1]*c[i, 0, 3]*eta[3, 1, 1]));
           if 2 <= t then 
              for k from 2 to t do 
                 for p to k-1 do 
                    c[i+1, p, 2*k-2*p+1] := expand(simplify((-1)^p*(sum((-1)^l
                                    *M[i, l, 2*k-2*l+1]*eta[2*k-2*l+1, p-l, 1],
                                                         l = 0 .. p)
                                    +c[i+1, 0, 1]*(sum((-1)^l*c[i, l, 2*k-2*l]
                                       *eta[2*k-2*l, p-l, 1], l = 0 .. p)))));
                 end do 
              end do; 
              for k from 3 to t+1 do 
                 for p to k-1 do 
                    c[i+1, p, 2*k-2*p] := expand(simplify((-1)^p*(sum((-1)^l
                         *M[i, l, 2*k-2*l]
                                          *eta[2*k-2*l, p-l, 1], l = 0 .. p)
                                  +c[i+1, 0, 1]*(sum((-1)^l*c[i, l, 2*k-2*l-1]
                             *eta[2*k-2*l-1, p-l, 1], l = 0 .. min(p, k-2))))));
                 end do 
              end do 
           end if 
        else 
           c[i+1, 0, 2] := expand(simplify(M[i, 0, 2]*eta[2, 0, 1])); 
           for k to t+1 do 
              c[i+1, 0, 2*k+1] := expand(simplify(M[i, 0, 2*k+1]
                               *eta[2*k+1, 0, 1]
                              +c[i+1, 0, 1]*c[i, 0, 2*k]*eta[2*k, 0, 1])); 
           end do; 
           for k from 2 to t+1 do 
              c[i+1, 0, 2*k] := expand(simplify(M[i, 0, 2*k]*eta[2*k, 0, 1]
                              +c[i+1, 0, 1]*c[i, 0, 2*k-1]*eta[2*k-1, 0, 1])); 
           end do; 
           c[i+1, 1, 0] := expand(simplify(-M[i, 0, 2]+M[i, 1, 0])); 
           for k from 2 to t+1 do 
              c[i+1, k, 0] := expand(simplify((-1)^k*(sum((-1)^l
                          *M[i, l, 2*k-2*l],l = 0 .. k))));
           end do; 
           c[i+1, 1, 1] := expand(simplify(-M[i, 0, 3]*eta[3, 1, 1]
                                          +c[i+1, 0, 1]*(-c[i, 0, 2]+c[i, 1, 0])));
           for k from 2 to t+1 do 
              c[i+1, k, 1] := expand(simplify((-1)^k*(sum((-1)^l*M[i, l, 2*k-2*l+1]
                                       *eta[2*k-2*l+1, k-l, 1], l = 0 .. k-1)
                                 +c[i+1, 0, 1]*(sum((-1)^l*c[i, l, 2*k-2*l], 
                                 l = 0 .. k))))); 
           end do; 
           c[i+1, 1, 2] := expand(simplify(-M[i, 0, 4]*eta[4, 1, 1]
                              +M[i, 1, 2]*eta[2, 0, 1]
                              -c[i+1, 0, 1]*c[i, 0, 3]*eta[3, 1, 1]));
           if 2 <= t then 
              for k from 3 to t+1 do 
                 for p to k-1 do 
                    c[i+1, p, 2*k-2*p] := expand(simplify((-1)^p
                       *(sum((-1)^l*M[i, l, 2*k-2*l]
                                        *eta[2*k-2*l, p-l, 1], l = 0 .. p)
                                        +c[i+1, 0, 1]*(sum((-1)^l
                                           *c[i, l, 2*k-2*l-1]
                                           *eta[2*k-2*l-1, p-l, 1], 
                                           l = 0 .. min(p, k-2))))));
                 end do 
              end do 
           end if; 
           for k from 2 to t+1 do 
              for p to k-1 do 
                 c[i+1, p, 2*k-2*p+1] := expand(simplify((-1)^p*(sum((-1)^l
                                              *M[i, l, 2*k-2*l+1]
                                         *eta[2*k-2*l+1, p-l, 1], l = 0 .. p)
                                          +c[i+1, 0, 1]*(sum((-1)^l*c[i, l, 2*k-2*l]
                                              *eta[2*k-2*l, p-l, 1],l = 0 .. p))))); 
             end do 
           end do 
        end if 
     end do 
  end if; 
  for p from 0 to trunc((1/2)*n+1) do 
     for k from 0 to n+2-2*p do 
        print(coef[n+1, p, k], "=", c[n+1, p, k]) 
     end do 
  end do 
end proc;
\end{verbatim}

\newpage
\section{APPENDIX F: Powers of $q-$Onsager generators}

\vspace{3mm}
The explicit expressions of the divided polynomials in terms of tensor products of Pauli matrices can be derived using (\ref{polymin}) and computing powers $\mathcal W_i^n$, $i=0,1$ ($n$ is a positive integer) of the fundamental generators of the $q-$Onsager algebra (\ref{TDbase})-(\ref{TDbase1}).  According to the parity of $n$, the expressions slightly differ. For simplicity, let us first study the case $\epsilon_\pm=0$.  To get the results below, we used:
\[w_0^n = \left\{ \begin{array}{l}
 {\left( {{k_ + }{k_ - }} \right)^{\frac{n}{2}}}\mathbb{I} \qquad \quad \text{if $n$ is even}, \\ 
 {\left( {{k_ + }{k_ - }} \right)^{\frac{{n - 1}}{2}}}{w_0} \ \quad \text{if $n$ is odd.} \\ 
 \end{array} \right.\]
 For $N = 2$ and $\epsilon_\pm=0$:
\begin{eqnarray*}
\mathcal W_0^2 & = & [2]!\sum_{1 \le j_1 < j_2 \le L}{q^{2\sigma_z}\otimes \dots \otimes q^{2\sigma_z}\otimes w_{0_{j_1}}\otimes q^{\sigma_z}\otimes \dots \otimes q^{\sigma_z}\otimes w_{0_{j_2}}\otimes \mathbb{I}\otimes \dots \otimes \mathbb{I}}\\
& &+ \sum_{1 \le j\le L}{q^{2\sigma_z}\otimes\dots \otimes q^{2\sigma_z}\otimes w_{0_{j}}^2\otimes\mathbb{I}\otimes\dots \otimes \mathbb{I}}.
\end{eqnarray*}

For $N =3$ and $\epsilon_\pm=0$:
\begin{eqnarray*}
\mathcal W_0^3&=&[3]!\sum_{1\le j_1 <j_2 <j_3\le L}{\left(q^{3\sigma_z}\otimes \dots \otimes q^{3\sigma_z}\otimes w_{0_{j_1}}\otimes q^{2\sigma_z}\otimes \dots \otimes q^{2\sigma_z}\otimes w_{0_{j_2}}\otimes\right.}\\
&&~~~~~~~~~~~~~~~~~~~~~~~~~\left. \otimes q^{\sigma_z}\otimes\dots\otimes q^{\sigma_z}\otimes w_{0_{j_3}}\otimes \mathbb{I}\otimes\dots \otimes \mathbb{I}\right)\\
& &+\sum_{1\le j\le L}{q^{3\sigma_z}\otimes \dots \otimes q^{3\sigma_z}\otimes w_{0_{j}}^3\otimes\mathbb{I}\otimes\dots \otimes\mathbb{I}}\\
& & +[3]\sum_{1\le j_1 < j_2 \le L}{q^{3\sigma_z}\otimes \dots\otimes q^{3\sigma_z}\otimes w_{0_{j_1}}\otimes q^{2\sigma_z}\otimes \dots \otimes q^{2\sigma_z}\otimes w_{0_{j_2}}^2\otimes \mathbb{I}\otimes\dots \otimes \mathbb{I}}\\
& & +[2]\sum_{1\le j_1 < j_2 \le L}{q^{3\sigma_z}\otimes\dots\otimes q^{3\sigma_z}\otimes w_{0_{j_1}}^2\otimes q^{\sigma_z}\otimes\dots\otimes q^{\sigma_z}\otimes w_{0_{j_2}}\otimes \mathbb{I}\otimes \dots\otimes \mathbb{I}}\\
& & +\sum_{1\le j_1 < j_2 \le L}{q^{3\sigma_z}\otimes\dots\otimes q^{3\sigma_z}\otimes q^{\sigma_z}_{j_1}w_{0_{j_1}}^2\otimes q^{\sigma_z}\otimes\dots\otimes q^{\sigma_z}\otimes w_{0_{j_2}}\otimes \mathbb{I}\otimes \dots\otimes \mathbb{I}}.
\end{eqnarray*}

\vspace{5mm}

More generally, by induction one finds the general structure:

\vspace{5mm}

\begin{itemize}
\item Case $n$ is odd and $\epsilon_\pm=0$:
\begin{eqnarray*}
 \mathcal W_0^n &=& \sum\limits_{h = 1}^n {\sum\limits_{\scriptstyle 0 < {i_1},{i_2},\dots,{i_h} \hfill \atop 
  \scriptstyle {i_1} + {i_2} + \dots + {i_h} = n \hfill} {\sum\limits_{1 \le {j_1}<{j_2}<\dots<{j_h} \le L} {\left(C_{\scriptstyle {i_1},\dots,{i_h} \hfill \atop 
  \scriptstyle {j_1},\dots,{j_h} \hfill}^{\left( h \right)} {q^{n{\sigma _z}}} \otimes \dots \otimes {q^{n{\sigma _z}}} \otimes w_{{0_{{j_1}}}}^{{i_1}} \otimes\right.}}}\\
  &&\otimes ~{q^{\left( {n - {i_1}} \right){\sigma _z}}} \otimes \dots \otimes {q^{\left( {n - {i_1}} \right){\sigma _z}}}\otimes w_{{0_{{j_2}}}}^{{i_2}} \otimes {q^{\left( {n - {i_1} - {i_2}} \right){\sigma _z}}} \otimes \dots \otimes {q^{\left( {n - {i_1} - {i_2}} \right){\sigma _z}}} \otimes \dots \\ 
  &&\otimes\left. {q^{\left( {n - {i_1} - \dots - {i_{h - 1}}} \right){\sigma _z}}} \otimes \dots \otimes {q^{\left( {n - {i_1} - \dots - {i_{h - 1}}} \right){\sigma _z}}} \otimes w_{{0_{{j_h}}}}^{{i_h}} \otimes \mathbb{I} \otimes \dots \otimes \mathbb{I}\right) \\ 
  &+& \sum\limits_{h = 2}^{n - 1} \quad {\sum\limits_{u = 1}^{[\frac{n-1}{2}]-|h-[\frac{n+1}{2}]|} \quad { {\sum\limits_{\scriptstyle 0 < {x_1},{x_2},\dots,{x_u} \hfill \atop 
        \scriptstyle {x_1} + {x_2} + \dots + {x_u} = u \hfill}^{{x_1} + {x_2} + \dots + {x_u} = Min\{\left[ {\frac{{n - 1}}{2}} \right],h\}} {\sum\limits_{\scriptstyle 0 < {i_{u+1}},{i_{u+2}},\dots,{i_h} \hfill \atop 
            \scriptstyle {i_{u+1}} + {i_{u+2}} + \dots + {i_h} = n-2(x_1+\dots+x_u) \hfill} {\sum\limits_{P }  } } } } }  \\ 
 & & \left(c_{\scriptstyle {t_1},{x_1},{j_1};\dots;{t_u},{x_u},{j_u} \hfill \atop 
   \scriptstyle {i_{u + 1}},{i_{u + 2}},\dots,{i_h};{{j'}_{u + 1}},\dots,{{j'}_h} \hfill}^{\left( h \right)}
  {q^{n{\sigma _z}}} \otimes \dots \otimes {q^{n{\sigma _z}}} \otimes w_{{0_{j'_{u + 1}}}}^{{i_{u + 1}}} \otimes {q^{\left( {n - {i_{u + 1}}} \right){\sigma _z}}} \otimes \dots \otimes \right.\\
    & & \otimes ~{q^{\left( {n - {i_{u + 1}}} \right){\sigma _z}}} \otimes q_{{j_1}}^{{t_1}}w_{0_{j_1}}^{2{x_1}} \otimes {q^{\left( {n - {i_{u + 1}} - 2{x_1}} \right){\sigma _z}}} \otimes \dots \otimes  {q^{\left( {n - {i_{u + 1}} - \dots - {i_{h - 1}} - 2\left( {{x_1} + \dots + {x_u}} \right)} \right)\sigma_z}} \otimes\\&&\left.\otimes~ w_{{0_{{{j'}_h}}}}^{{i_h}} \otimes \mathbb{I} \otimes \dots \otimes \mathbb{I}\right)
\end{eqnarray*}

\vspace{5mm}

\newpage

\item Case $n$ is even and $\epsilon_\pm=0$:

\begin{eqnarray*}
 \mathcal W_0^n &=& \sum\limits_{h = 1}^n {\sum\limits_{\scriptstyle 0 < {i_1},{i_2},\dots,{i_h} \hfill \atop 
  \scriptstyle {i_1} + {i_2} + \dots + {i_h} = n \hfill} {\sum\limits_{1 \le {j_1}<{j_2}<\dots,<{j_h} \le L} {\left(C_{\scriptstyle {i_1},...,{i_h} \hfill \atop 
  \scriptstyle {j_1},\dots,{j_h} \hfill}^{\left( h \right)} {q^{n{\sigma _z}}} \otimes \dots \otimes {q^{n{\sigma _z}}} \otimes w_{{0_{{j_1}}}}^{{i_1}} \otimes \right.}}}\\
  &\otimes& {q^{\left( {n - {i_1}} \right){\sigma _z}}} \otimes \dots \otimes {q^{\left( {n - {i_1}} \right){\sigma _z}}}\otimes w_{{0_{{j_2}}}}^{{i_2}} \otimes {q^{\left( {n - {i_1} - {i_2}} \right){\sigma _z}}} \otimes \dots \otimes {q^{\left( {n - {i_1} - {i_2}} \right){\sigma _z}}} \otimes \dots \\ 
  &\otimes& \left.{q^{\left( {n - {i_1} - \dots - {i_{h - 1}}} \right){\sigma _z}}} \otimes \dots \otimes {q^{\left( {n - {i_1} - \dots - {i_{h - 1}}} \right){\sigma _z}}} \otimes w_{{0_{{j_h}}}}^{{i_h}} \otimes \mathbb{I} \otimes \dots \otimes \mathbb{I}\right) \\ 
  &+& \sum\limits_{h = 2}^{n/2} \quad {\sum\limits_{u = 1}^{h-1} \quad {{\sum\limits_{\scriptstyle 0 < {x_1},{x_2},\dots,{x_u} \hfill \atop 
          \scriptstyle {x_1} + {x_2} + \dots + {x_u} = u \hfill}^{{x_1} + {x_2} + \dots + {x_u} = Min\{ \left[ {\frac{{n - 1}}{2}} \right],h\}} {\sum\limits_{\scriptstyle 0 < {i_{u+1}},{i_{u+2}},\dots,{i_h} \hfill \atop 
              \scriptstyle {i_{u+1}} + {i_{u+2}} + \dots + {i_h} = n-2(x_1+\dots+x_u) \hfill} {\sum\limits_{P }  } } } } }  \\
  & & \left(c_{\scriptstyle {t_1},{x_1},{j_1};\dots;{t_u},{x_u},{j_u} \hfill \atop 
     \scriptstyle {i_{u + 1}},{i_{u + 2}},\dots,{i_h};{{j'}_{u + 1}},\dots,{{j'}_h} \hfill}^{\left( h \right)}
    {q^{n{\sigma _z}}} \otimes \dots \otimes {q^{n{\sigma _z}}} \otimes w_{{0_{j'{u + 1}}}}^{{i_{u + 1}}} \otimes {q^{\left( {n - {i_{u + 1}}} \right){\sigma _z}}} \otimes \dots \otimes\right.\\&&~~\otimes ~{q^{\left( {n - {i_{u + 1}}} \right){\sigma _z}}} \otimes q_{{j_1}}^{{t_1}}w_{0_{j_1}}^{{x_1}} \otimes {q^{\left( {n - {i_{u + 1}} - 2{x_1}} \right){\sigma _z}}}\otimes \dots \otimes {q^{\left( {n - {i_{u + 1}} - \dots - {i_{h - 1}} - 2\left( {{x_1} + \dots + {x_u}} \right)} \right)}} \otimes\\&&~~\left.\otimes~w_{{0_{{{j'}_h}}}}^{{i_h}} \otimes \mathbb{I} \otimes \dots \otimes \mathbb{I}\right)\\ 
  &+& \sum\limits_{h = n/2+1}^{n - 1} \quad {\sum\limits_{u = 1}^{n-h} \quad { {\sum\limits_{\scriptstyle 0 < {x_1},{x_2},\dots,{x_u} \hfill \atop 
                \scriptstyle {x_1} + {x_2} + \dots + {x_u} = u \hfill}^{{x_1} + {x_2} + \dots + {x_u} = Min\{\left[ {\frac{{n - 1}}{2}} \right],h\}} \quad {\sum\limits_{\scriptstyle 0 < {i_{u+1}},{i_{u+2}},\dots,{i_h} \hfill \atop 
                              \scriptstyle {i_{u+1}} + {i_{u+2}} + \dots + {i_h} = n-2(x_1+\dots+x_u) \hfill} {\sum\limits_{P }  } } } } }  \\ 
 & & \left(c_{\scriptstyle {t_1},{x_1},{j_1};\dots;{t_u},{x_u},{j_u} \hfill \atop 
   \scriptstyle {i_{u + 1}},\dots,{i_h};{{j'}_{u + 1}},\dots,{{j'}_h} \hfill}^{\left( h \right)}
  {q^{n{\sigma _z}}} \otimes \dots \otimes {q^{n{\sigma _z}}} \otimes w_{{0_{j'{u + 1}}}}^{{i_{u + 1}}} \otimes {q^{\left( {n - {i_{u + 1}}} \right){\sigma _z}}} \otimes \dots \otimes \right.\\&& \otimes~{q^{\left( {n - {i_{u + 1}}} \right){\sigma _z}}} \otimes q_{{j_1}}^{{t_1}}w_{0_{j_1}}^{2{x_1}} \otimes {q^{\left( {n - {i_{u + 1}} - 2{x_1}} \right){\sigma _z}}} \otimes \dots \otimes {q^{\left( {n - {i_{u + 1}} - \dots - {i_{h - 1}} - 2\left( {{x_1} + \dots + {x_u}} \right)} \right)}} \otimes\\&&\left.\otimes~ w_{{0_{{{j'}_h}}}}^{{i_h}} \otimes \mathbb{I} \otimes \dots \otimes \mathbb{I}\right)
\end{eqnarray*}
\end{itemize}
where
\[P=\left \{ \begin{array}{l}
1 \le {j_1} < ... < {j_u} \le L,\\
1 \le {{j'}_{u + 1}} < ... < {{j'}_h} \le L,\\
\left\{ {{j_1},...,{j_u}} \right\} \cap \left\{ {{{j'}_{u + 1}},...,{{j'}_h}} \right\} = \emptyset , \\
t_1 \ge t_2 \ge \dots \ge t_u >0,\\
t_k \leq \sum\limits_{v>k}{2x_v}+\sum\limits_{j'_v>j_k}{i_v}, \quad k = 1,\dots,u.
\end{array}\right.
\]

\vspace{3mm}

For $n=1, 2, 3, 4$,  the explicit expressions for the coefficients are:\\
\begin{itemize}
\item For $n =1$: \qquad  $C^{(1)}_{\scriptstyle 1 \hfill \atop \scriptstyle j_1 \hfill}  =1$;
\item For  $n=2$: \qquad $C^{(2)}_{\scriptstyle 1,1 \hfill \atop \scriptstyle j_1,j_2 \hfill}= q+q^{-1}=[2]!,\quad C^{(1)}_{\scriptstyle 2 \hfill \atop \scriptstyle j_1 \hfill}=1  $, where $j_1 < j_2;$
\item For  $n=3$: \\
$C^{(3)}_{\scriptstyle 1,1,1 \hfill \atop \scriptstyle j_1,j_2,j_3 \hfill}=[3]!, \quad C^{(2)}_{\scriptstyle 1,2 \hfill \atop \scriptstyle j_1,j_2 \hfill}=[3], \quad C^{(2)}_{\scriptstyle 2,1 \hfill \atop \scriptstyle j_1,j_2 \hfill}=[2], \quad C^{(1)}_{\scriptstyle 3 \hfill \atop \scriptstyle j_1 \hfill}=1$,\\

$c^{(2)}_{\scriptstyle 1,2,j_1 \hfill \atop \scriptstyle 1,j_2 \hfill}=1$, where $j_1 < j_2 < j_3$;
\item For  $n=4$: \\
$C^{(4)}_{\scriptstyle 1,1,1,1 \hfill \atop \scriptstyle j_1,j_2,j_3,j_4 \hfill}=[4]!, \quad C^{(3)}_{\scriptstyle 2,1,1 \hfill \atop \scriptstyle j_1,j_2,j_3 \hfill}=[3]!,\quad C^{(3)}_{\scriptstyle 1,2,1 \hfill \atop \scriptstyle j_1,j_2,j_3 \hfill}=[2][4], \quad C^{(3)}_{\scriptstyle 1,1,2 \hfill \atop \scriptstyle j_1,j_2,j_3 \hfill}=[3][4],$\\

$ C^{(2)}_{\scriptstyle 3,1 \hfill \atop \scriptstyle j_1,j_2 \hfill}=2[2],\quad C^{(2)}_{\scriptstyle 1,3 \hfill \atop \scriptstyle j_1,j_2 \hfill}=[4],\quad C^{(2)}_{\scriptstyle 2,2 \hfill \atop \scriptstyle j_1,j_2 \hfill}=[3], \quad C^{(1)}_{\scriptstyle 4 \hfill \atop \scriptstyle j_1 \hfill}=1,$\\

$c^{(3)}_{\scriptstyle 1,2,j_2 \hfill \atop \scriptstyle 1,1;j_1,j_3 \hfill}=[4], \quad c^{(3)}_{\scriptstyle 1,2,j_1 \hfill \atop \scriptstyle 1,1;j_2,j_3 \hfill}=[2]^2, \quad c^{(3)}_{\scriptstyle 2,2,j_1 \hfill \atop \scriptstyle 1,1;j_2,j_3 \hfill}=[2], \quad c^{(2)}_{\scriptstyle 1,2,j_1 \hfill \atop \scriptstyle 2,j_2 \hfill}=[2],\quad c^{(2)}_{\scriptstyle 2,2,j_1 \hfill \atop \scriptstyle 2,j_2 \hfill}=1$ , 

where $j_1 < j_2 < j_3 < j_4$.
\end{itemize}

\vspace{5mm}

\begin{rem}
Powers of $\mathcal W_1$ for $\epsilon_-=0$ are readily obtained through the substitutions $q\rightarrow q^{-1}$ in above expressions.
\end{rem}

The results above correspond to the special choice $\epsilon_\pm=0$. For $\epsilon_\pm\neq0$, a similar analysis yields to:

\vspace{2mm}

\begin{eqnarray*}
\mathcal{W}_0^n &=& \mathcal{W}_0^n|_{\epsilon_+=0}+\epsilon_+^n q^{n\sigma_z}\otimes\dots\otimes q^{n\sigma_z}
\\
&+&\sum_{m=1}^{n-1}{\sum_{h=1}^m~{\sum_{P_1}{{\left(\epsilon_+^{n-m}C_{\scriptstyle {i_1},\dots,{i_h} \hfill \atop 
  \scriptstyle {j_1},\dots,{j_h} \hfill}^{\left( m,h \right)}q^{n\sigma_z}\otimes\dots\otimes q^{n\sigma_z}\otimes w_{0_{j_1}}^{i_1}\otimes q^{(n-i_1)\sigma_z}\otimes \dots\otimes w_{0_{j_h}}^{i_h}\otimes\right.}}}}\\&&~~~~~~~~~~~~~~~~~~~~~\left.\otimes ~q^{(n-m)\sigma_z}\otimes \dots\otimes q^{(n-m)\sigma_z}\right)\\
  &+& \sum\limits_{m = 2}^{n - 1} {\sum_{h=1}^{m-1}~{\sum\limits_{u = 1}^{Min([\frac{m}{2}],h)} \quad { {\sum\limits_{\scriptstyle 0 < {x_1},{x_2},\dots,{x_u} \hfill \atop 
          \scriptstyle {x_1} + {x_2} + \dots + {x_u} = u \hfill}^{{x_1} + {x_2} + \dots + {x_u} = [\frac{m}{2}]} {\sum\limits_{\scriptstyle 0 < {i_{u+1}},{i_{u+2}},\dots,{i_h} \hfill \atop 
              \scriptstyle {i_{u+1}} + {i_{u+2}} + \dots + {i_h} = m-2(x_1+\dots+x_u) \hfill} {\sum\limits_{P_2 }  } } } }} }  \\ 
   & &\left(\epsilon_+^{n-m} c_{\scriptstyle {t_1},{x_1},{j_1};\dots;{t_u},{x_u},{j_u} \hfill \atop 
     \scriptstyle {i_{u + 1}},{i_{u + 2}},\dots,{i_h};{{j'}_{u + 1}},\dots,{{j'}_h} \hfill}^{\left( m,h \right)}
    {q^{n{\sigma _z}}} \otimes \dots \otimes {q^{n{\sigma _z}}} \otimes w_{{0_{j'_{u + 1}}}}^{{i_{u + 1}}} \otimes {q^{\left( {n - {i_{u + 1}}} \right){\sigma _z}}} \otimes \dots \otimes {q^{\left( {n - {i_{u + 1}}} \right){\sigma _z}}} \otimes\right.\\&&~\left.\otimes~ q_{{j_1}}^{{t_1}}w_{0_{j_1}}^{2{x_1}}\otimes q^{\left( {n - {i_{u + 1}} - 2{x_1}} \right){\sigma _z}}\otimes\dots\otimes~ q^{\left( {n - {i_{u + 1}} - 2{x_1}} \right){\sigma _z}}\otimes\dots\otimes q^{(n-m)\sigma_z}\otimes\dots\otimes q^{(n-m)\sigma_z}\right)\\
      &+& \sum\limits_{t = 1}^{[\frac{n-1}{2}] } {\sum_{h=1}^{t}~{ { {\sum\limits_{\scriptstyle 0 < {x_1},{x_2},\dots,{x_h} \hfill \atop 
          \scriptstyle {x_1} + {x_2} + \dots + {x_h} = t \hfill}{ {\sum\limits_{P_3 }  } } } }} }\left(\epsilon_+^{n-2t} c_{\scriptstyle {t_1},{x_1},{j_1};\dots;{t_h},{x_h},{j_h} \hfill}^{\left( 2t,h \right)}
    {q^{n{\sigma _z}}} \otimes \dots \otimes {q^{n{\sigma _z}}} \otimes q^{t_1\sigma_z}_{j_1}w_{{0_{j_1}}}^{{2x_1}} \otimes\right.\\&&~~~~~\otimes~ {q^{\left( {n - 2x_1} \right){\sigma _z}}} \otimes \dots \otimes {q^{\left( {n - 2x_1} \right){\sigma _z}}}\otimes q^{t_2\sigma_z}_{j_2}w_{{0_{j_2}}}^{{2x_2}}\otimes {q^{\left( {n - 2x_1-2x_2} \right){\sigma _z}}} \otimes\\&&~~~~~\left.\otimes \dots \otimes {q^{\left( {n - 2x_1-2x_2} \right){\sigma _z}}}\otimes \dots \otimes q^{t_h\sigma_z}_{j_h}w_{{0_{j_h}}}^{{2x_h}}\otimes q^{(n-2t)\sigma_z}\otimes\dots\otimes q^{(n-2t)\sigma_z}\right)
\end{eqnarray*}
where
\[P_1=\left\{\begin{array}{l}
i_1+\dots+i_h =m, \\
 0<j_1<j_2<\dots<j_h\le L\\
 \end{array} \right.; \]
 \[P_2= \left\{\begin{array}{l}
 1 \le j_1 < \dots < j_u \le L,\\
 1 \le j'_{u+1} <\dots <j'_h \le L,\\
 \{j_1,\dots,j_u\} \cap \{j'_{u+1},\dots,j'_{h}\}=\emptyset,\\
 t_1\ge t_2\ge \dots \ge t_u >0,\\
 t_k \le n-2\sum\limits_{i\le k}{x_i}-\sum\limits_{j'_v<j_k}{i_v},\quad k=1,\dots, u.
 \end{array}
 \right. \]
 \[P_3= \left\{\begin{array}{l}
 1 \le j_1 < \dots < j_h \le L,\\
 t_1\ge t_2\ge \dots \ge t_h >0,\\
 t_k \le n-2\sum\limits_{i\le k}{x_i},\quad k=1,\dots, h.
 \end{array}
 \right. \]
For $n=2, 3, 4$,  the explicit expressions for the coefficients are:\\
\begin{itemize}
\item For $n =2$: \qquad  $C^{(1,1)}_{\scriptstyle 1 \hfill \atop \scriptstyle j_1 \hfill}  =[2]$;
\item For  $n=3$: \\
$C^{(2,2)}_{\scriptstyle 1,1 \hfill \atop \scriptstyle j_1,j_2 \hfill}=[3]!, \quad C^{(2,1)}_{\scriptstyle 2 \hfill \atop \scriptstyle j_1 \hfill}=[2],\quad C^{(1,1)}_{\scriptstyle 1 \hfill \atop \scriptstyle j_1 \hfill}=[3]$,\\

$c^{(2,1)}_{1,1,j_1}=1$, where $j_1 < j_2$;
\item For  $n=4$: \\
$C^{(3,3)}_{\scriptstyle 1,1,1 \hfill \atop \scriptstyle j_1,j_2,j_3 \hfill}=[4]!, \quad C^{(3,2)}_{\scriptstyle 2,1 \hfill \atop \scriptstyle j_1,j_2 \hfill}=[3]!,\quad C^{(3,2)}_{\scriptstyle 1,2 \hfill \atop \scriptstyle j_1,j_2,j_3 \hfill}=[2][4], \quad C^{(3,1)}_{\scriptstyle 3 \hfill \atop \scriptstyle j_1 \hfill}=2[2],$\\

$ C^{(2,2)}_{\scriptstyle 1,1 \hfill \atop \scriptstyle j_1,j_2 \hfill}=[3][4],\quad C^{(2,1)}_{\scriptstyle 2 \hfill \atop \scriptstyle j_1 \hfill}=[3],\quad C^{(1,1)}_{\scriptstyle 1 \hfill \atop \scriptstyle j_1 \hfill}=[4],$\\

$c^{(3,2)}_{\scriptstyle 1,1,j_2 \hfill \atop \scriptstyle 1;j_1 \hfill}=[4], \quad c^{(3,2)}_{\scriptstyle 2,1,j_1 \hfill \atop \scriptstyle 1;j_2 \hfill}=[2], \quad c^{(3,2)}_{\scriptstyle 1,1,j_1 \hfill \atop \scriptstyle 1;j_2 \hfill}=[2]^2,$ \\

$\quad c^{(2,1)}_{1,1,j_1}=[2],\quad c^{(2,1)}_{2,1,j_1}=1$ , 

where $j_1 < j_2 < j_3$.

\end{itemize}

\vspace{5mm}

\begin{rem}
Powers of $\mathcal W_1$ for $\epsilon_-\neq0$ are readily obtained through the substitutions $q\rightarrow q^{-1},\epsilon_+\rightarrow \epsilon_-$  in above expressions.
\end{rem}
\newpage



\begin{thebibliography}{10}
\bibitem[AC04]{AC} H. Alnajjar and B. Curtin, \textit{A family of tridiagonal pairs}, Linear Algebra and its Aplications 390 (2004) 369-384. 
%
\bibitem[ARS95]{ARS} N.M. Atakishiyev, M. Rahman, S.K. Suslov, \textit{On classical orthogonal polynomials}, Constr. Approx. 11 (2) (1995) 181-226.
 %
\bibitem[AYP09]{AYP}
H. Au-Yang and J.H.H. Perk,  2011 {\it Quantum loop subalgebra and eigen
vectors of the superintegrable chiral Potts transfer matrices}, J. Phys. {\bf A 44} 025205, {\tt arXiv:0907.0362};\\
H. Au-Yang and J.H.H. Perk, {\it Serre Relations in the Superintegrable Model},  {\tt arXiv:1210.5803}.
%
\bibitem[AW79]{AW} R. Askey, J. Wilson, \textit{A set of orthogonal polynomials that generalize the Racah coefficients or 6 - j symbols}, SIAM J. Math. Anal. 10 (5) (1979) 1008-1016.
%
\bibitem[AW85]{AW1} R. Askey, J. Wilson, \textit{Some basic hypergeometric orthogonal polynomials that generalize Jacobi polynomials}, Mem. Am. Math. Soc. 54 (319) (1985) iv + 55.
%
\bibitem[Bas0404]{B1}
P. Baseilhac, {\it Deformed Dolan-Grady relations in quantum integrable models}, Nucl.Phys. {\bf B 709} (2005) 491-521, {\tt arXiv:hep-th/0404149}.
%
\bibitem[Bas0408]{Bas1} P. Baseilhac, \textit{An integrable structure related to tridiagonal algebras}, Nuclear Phys. B 705 (2005) 605-619, arXiv:math-ph/0408025.
%
\bibitem[Bas06]{Bas2} P. Baseilhac, \textit{The $q$-deformed analogue of the Onsager algebra: beyond the Bethe ansatz approach}, Nuclear Phys. B 754 (2006) 309-382; arXiv:math-ph/0604036.
%
\bibitem[Bas0604]{Bas3} P. Baseilhac, \textit{A family of tridiagonal pairs and related symmetric functions}, J. Phys. A 39 (2006) 11773-11791; arXiv:math-ph/0604035.
%
\bibitem[BCDRS94]{Cor}
E. Corrigan, P.E. Dorey, R.H. Rietdijk and R. Sasaki, {\it Affine Toda field theory on a half line}, Phys. Lett. {\bf B 333} (1994) 83--91, {\tt arXiv:hep-th/9404108};\\ 
P. Bowcock, E. Corrigan, P.E. Dorey and R.H. Rietdijk, {\it Classically integrable boundary conditions for affine Toda field theories}, Nucl. Phys. {\bf B 445} (1995) 469--500, {\tt hep-th/9501098}.
%
%
%
\bibitem[Bax73]{Baxter}
R.J. Baxter, \textit{Eight-vertex model in lattice statistics and one-dimensional anisotropic Heisenberg chain}, Ann. Phys. {\bf 76} (1973) 1; {\bf 76} (1973) 25; {\bf 76} (1973) 48.
%
\bibitem[Bax82]{Baxter1} R. Baxter, \textit{Exactly solved models in statistical mechanics}, New York, Academic Press (1982).
%
\bibitem[BB10]{BS} P. Baseilhac and S. Belliard, \textit{A note on the $O_q(\widehat{sl_2})$ algebra}, arXiv:1012.5261v1.
%
\bibitem[BB11]{BB0}
 P. Baseilhac and S. Belliard,    {\it Central extension of the reflection equations and an analog of Miki's formula},  J. Phys. A {\bf 44} (2011) 415205, {\tt arXiv:1104.1591}. 
%
\bibitem[BB09]{BB1}
P. Baseilhac and S. Belliard, {\it Generalized q-Onsager algebras and boundary affine Toda field theories}, Lett. Math. Phys. {\bf 93} (2010) 213-228, {\tt arXiv:0906.1215}. 
%
\bibitem[BB12]{BB3} P. Baseilhac and S. Belliard, \textit{The half-infinite $XXZ$ chain in Onsager's approach}, Nuclear Phys. B 873 (2013) 550-584, {\tt arXiv:1211.6304}.

\bibitem[Bet31]{Bet} H. Bethe, \textit{Zur theorie der metalle. I. Eigenwerte und eigenfunktionen der linearen atomkette}, Z. Phys. \textbf{71} (1931) 205.

\bibitem[BF11]{BF}
S. Belliard and V. Fomin, {\it Generalized $q-$Onsager algebras and dynamical $K$ matrices}, J. Phys. A: Math. Theor. {\bf 45} (2012) 025201, {\tt arXiv:1106.1317}.

\bibitem[BI84]{BI} E. Bannai and T. Ito, \textit{Algebraic Combinatorics I: Association Schemes}, Benjamin/Cummings, London, 1984.
\bibitem[BGSV15]{BGSV} P. Baseilhac, A.M. Gainutdinov, H. Saleur and T.T. Vu, \textit{The open $XXZ$ spin chain at roots of unity}, 2015, preprint.

\bibitem[BK13]{BK} P. Baseilhac and T. Kojima, \textit{Correction functions of the half-infinite $XXZ$ spin chain with a triangular boundary}, Nuclear Phys. B 880 (2014) 378-413, {\tt arXiv:1309.7785}.

\bibitem[BK14]{BK14} P. Baseilhac and T. Kojima, {\it Form factors of the half-infinite XXZ spin chain with a triangular boundary}, J. Stat. Mech. (2014) P09004, {\tt arXiv:1404.0491}.
%
\bibitem[BK0503]{BK2}
P. Baseilhac and K. Koizumi, {\it A new (in)finite dimensional algebra for quantum integrable models}, Nucl. Phys. {\bf B 720} (2005) 325-347, {\tt arXiv:math-ph/0503036}.
%
\bibitem[BK0507]{BK1}
P. Baseilhac and K. Koizumi, {\it A deformed analogue of Onsager's symmetry in the $XXZ$ open spin chain}, J.Stat.Mech. {\bf 0510} (2005) P005, {\tt arXiv:hep-th/0507053}.
%
\bibitem[BK07]{BK3}
P. Baseilhac and K. Koizumi, {\it Exact spectrum of the $XXZ$ open spin chain from the q-Onsager algebra representation theory}, J. Stat. Mech.  (2007) P09006, {\tt arXiv:hep-th/0703106}.
%
\bibitem[BK02]{BasK3}
P. Baseilhac and K. Koizumi, {\it Sine-Gordon quantum field theory on the half-line with quantum boundary degrees of freedom}, Nucl. Phys. {\bf B 649} (2003) 491--510, {\tt arXiv:hep-th/0208005}.
%
\bibitem[BM15]{BM} P. Baseilhac and X. Martin, \textit{A bispectral $q$-hypergeometric basis for a class of quantum integrable models}, {\tt arXiv:1506.06902}.
%
\bibitem[BS09]{BasS}
P. Baseilhac and K. Shigechi, {\it A new current algebra and the reflection equation}, Lett. Math. Phys. {\bf  92} (2010) 47-65, {\tt arXiv:0906.1215}.
%
\bibitem[BV13]{BV1}
P. Baseilhac and T.T. Vu, {\it Analogues of Lusztig's higher order relations for the $q-$Onsager algebra}, J. Math. Phys. \textbf{55}, 081707 (2014), {\tt arXiv:1312.3433v1}.
%
\bibitem[BV1312]{BV2} P. Baseilhac and T.T. Vu, {\it Higher order relations for ADE-type generalized $q$-Onsager algebras}, Lett. Math. Phys \textbf{55} (2014) 1275-1288, {\tt arXiv:1312.5897v2}.

\bibitem[Cher84]{Cher84}
I.V. Cherednik, {\it Factorizing particles on the half-line and root systems}, Teor. Mat. Fiz. {\bf 61} (1984) 35-44.
%
\bibitem[CP91]{CP} V. Chari and A. Pressley, \textit{Quantum affine algebras}, Commun. Math. Phys. 142 (1991), 261-283.
%
\bibitem[Dav91]{B.D} B. Davies, \textit{Onsager's algebra and the Dolan-Grady condition in the non-self-dual case}, J. Math. Phys. 32 (11) (1991) 2945-2950.
\bibitem[Dav90]{B.D1} B. Davies, \textit{Onsager's algebra and superintegrability}, J. Phys. A \textbf{23} (1990) 2245-2261
%
\bibitem[DeG02]{DelG}
G.W. Delius and A. George, {\it Quantum affine reflection algebras of type $d_n^{(1)}$ and reflection matrices}, Lett. Math. Phys. {\bf 62} (2002) 211--217, {\tt arXiv:math/0208043}. 

%
\bibitem[DeM01]{DM}
G.W. Delius and N.J. MacKay, {\it Quantum group symmetry in sine-Gordon and affine Toda field theories on the half-line}, Commun. Math. Phys. {\bf 233} (2003) 173--190, {\tt arXiv:hep-th/0112023}.
%
\bibitem[DeCK90]{DeCK} C. de Concini and V. Kac, \textit{Representations of quantum groups at roots of 1, in Operator Algebras, Unitary Representation, Enveloping Algebras, and Invariant Theory}, A. Connes et al., eds. (Progress in Mathematics 92, Birkhauser, Basel 1990), p. 471.
\bibitem[DFM99]{DFM}
 T. Deguchi, K. Fabricius and B. M. McCoy, {\it The $sl_2$ loop algebra symmetry of the six-vertex model at roots of unity}, J. Statist. Phys. {\bf 102} (2001) 701-736,  {\tt arXiv:cond-mat/9912141}.

%
\bibitem[DG82]{DG}
L. Dolan and M. Grady, {\it Conserved charges from self-duality}, Phys. Rev. {\bf D 25} (1982) 1587-1604.
%
\bibitem[DR99]{DR} E. Date, S.S. Roan, \textit{The structure of quotients of the Onsager algebra by closed ideals}, J. Phys. A 33 (2000) 3275-3296, {\tt arXiv:math/9911018}.
%

\bibitem[Dr87]{Dr}
V.G. Drinfeld, {\it Quantum groups}, Proceedings ICM 1986, Amer. Math. Soc., 1987,
pp. 798--820.

\bibitem[dVG94]{dVG}
H. J.de Vega and A. Gonza'lez-Ruiz, \textit{Boundary $K$-matrices for the $XYZ$, $XXZ$ and $XXX$ spin chains}, {\it J. Phys.} {\bf A27} (1994) 6129--6137. 

\bibitem[El10]{El} C. El-Chaar, {\it The Onsager algebra}, M.Sc. Math. Thesis (2010) 98 pages, {\tt arXiv:1205.5989}.

\bibitem[FGST05]{FGST} B.L Feigin, A.M Gainutdinov, A.M Semikhatov and I.Y Tipunin, \textit{Modular group representations and fusion in logarithmic conformal field theories and in the quantum group center}, Commun. Math. Phys. \textbf{265} (2006) 47-93, {\tt arXiv:hep-th/0504093v4}.

\bibitem[FK99]{FK}	
H. Furutsu and T. Kojima, {\it $U_q(\widehat{sl_n})$- analog of the X X Z chain with a boundary}, J. Math. Phys. {\bf 41} (2000) 4413-4436, {\tt arXiv: solv-int/9905009}.

\bibitem[FST80]{FST} L.D. Faddeev, E.K. Sklyanin and L.A. Takhtadjan, \textit{The quantum method of the inverse problem}, Theor. Math. Phys. 40 (1980) 688.

\bibitem[FZ80]{FZ}
V. Fateev and A. Zamolodchikov, \textit{A model factorized S-matrix and an integrable spin-1 Heisenberg chain}, Sov. J. Nucl. Phys. {\bf 32} (1980) 298--303.

\bibitem[GI97]{GI}
A.M. Gavrilik and N.Z. Iorgov, {\it $q-$deformed algebras $U_q(so_n)$ and their representations}, Methods Funct. Anal. Topology {\bf 3} (1997), 51.
%
\bibitem[GH96]{GH1} F.A Grumbaum, L. Haine, \textit{The $q$-version of a theorem of Bochner}, J. Comput. Appl. Math. 68 (1996) 103 -114.
%
\bibitem[GH9699]{GH2} F.A Grumbaum, L. Haine, \textit{On a $q$-analogue of the string equation and a generalization of the classical orthogonal polynomials}, in: Algebraic Methods and $q$-Special Functions, Montreal, QC, 1996, Amer. Math. Soc., Providence, RI, 1999, pp. 171-181.
%
\bibitem[GLZ92]{GLZ} Ya.I. Granovskii, I.M Lutzenko, A.S. Zhedanov, \textit{Mutual integrability, quadratic algebras and dynamical symmetry}, Ann. Phys. 217 (1) (1992) 1-20.
%
\bibitem[GNPR05]{dG}
J. de Gier, A. Nichols, P. Pyatov, and V. Rittenberg, {\it Magic in the spectra of the $XXZ$ quantum chain with
boundaries at $\Delta = 0$ and $\Delta = -1/2$}, Nucl. Phys. {\bf B
729} (2005) 387, {\tt arXiv:hep-th/0505062v2}.
%
\bibitem[GR85]{GR} G. vonGehlen and V. Rittenberg, \textit{$Z_n$- symmetric quantum chains with an infinite set of conserved charges and $Z_n$ zero modes}, Nucl. Phys. B \textbf{257}, 351 (1985)
%
\bibitem[GZ94]{GZ}
S.Ghoshal and Al.Zamolodchikov, {\it Boundary $S$ matrix and boundary state in two-dimensional integrable quantum field theory},
{\it Int. J. Mod. Phys.} (1994) {\bf A9}, 3841--3885.
%
\bibitem[IK94]{IK} T. Inami and H. Konno, \textit{Integrable $XYZ$ spin chain with boundaries}, J.Phys. A:Math. Gen.\textbf{27} (1994) L913-L981, {\tt arXiv:hep-th/9409138v1}.

\bibitem[INT10]{INT} T. Ito, K. Nomura and P. Terwilliger, \textit{A classification of sharp tridiagonal pairs}, Linear Algebra and its Applications 435 (2011) 1857-1884, {\tt arXiv:1001.1812v1}.
%
\bibitem[IT07]{IT} T. Ito, P. Terwilliger, \textit{Tridiagonal pairs and the quantum affine algebra $U_q(\widehat{sl_2})$}, Ramanujan J. (2007), 13:39-62.
%
\bibitem[IT03]{IT03}
T. Ito and P. Terwilliger, {\it The shape of a tridiagonal pair}, {\tt arXiv:math/0304244v1}.
%

\bibitem[IT08]{IT2} T. Ito and P. Terwilliger, \textit{Tridiagonal pairs of $q$-Racal type}, J. Algebra 322 (2009) 68-93; arXiv:0807.0271.
%
\bibitem[IT09]{IT004} T. Ito and P. Terwilliger, \textit{The augmented tridiagonal algebra}, Kyushu J. Math. 64 (2010), 81-144, arXiv:0904.2889.
%
\bibitem[IT0904]{IT005} T. Ito and P. Terwilliger, \textit{The $q$-Onsager algebra}, arXiv:0904.2985.
%
\bibitem[ITT99]{TD00}
T.~Ito, K.~Tanabe, and P.~Terwilliger, {\it Some algebra related to ${P}$- and ${Q}$-polynomial association
  schemes},  in:
Codes and Association Schemes (Piscataway NJ, 1999), Amer. Math. Soc., Providence RI, 2001, pp 167--192; 
{\tt arXiv:math.CO/0406556}.


%
\bibitem[J85]{Jim} 
M. Jimbo, {\it A $q-$difference analogue of $U(g)$ and the Yang-Baxter equation}, Lett. Math. Phys. {\bf 11}
(1985), 63--69.

\bibitem[J86]{Jim2} M. Jimbo, {\it A $q-$analogue of $U(gl(n+1))$, Hecke algebra, and the Yang-Baxter equation}, Lett. Math. Phys. {\bf 11} (1986), 247--252.
%
\bibitem[JKKKMW94]{JKKKMW} M. Jimbo, R. Kedem, T. Kojima, H. Konno and T. Miwa, \textit{$XXZ$ chain with a boundary}, Nucl. Phys. B 441 (1995) 437-470, {\tt arXiv:hep-th/9411112v3};\\ M. Jimbo, R. Kedem, H. Konno, T. Miwa and R. Weston, \textit{Difference Equations in Spin Chains with a Boundary}, Nucl. Phys. B 448 (1995) 429-456, {\tt arXiv:hep-th/9502060v1}.

\bibitem[JM95]{JM} M. Jimbo and T. Miwa, \textit{Algebraic Analysis of Solvable Lattice Models}, Conference Board of the Math. Sci., Regional Conference Series in Mathematics \textbf{85} (1995).
\bibitem[Ka83]{Kac}
V. G. Kac, {\it Infinite dimensional Lie algebras}, (1983), Birkha$\ddot{u}$ser, Boston.
%
\bibitem[Kas95]{CK} C. Kassel, \textit{Quantum Groups}, Springer, New York, 1995. 
%

\bibitem[Klim01]{Klim} 
A.U. Klimyk, {\it The nonstandard $q-$deformation of enveloping algebra $U(so_n)$: results and problems},  Czech. J. Phys {\bf  51}
(2001) 331;\\
A.U. Klimyk, {\it Classification of irreducible representations of the $q-$deformed algebra $U_q'(so_n)$}, {\tt arXiv:math/0110038v1}.


\bibitem[KM01]{KM}
C. Korff and B. M. McCoy, {\it Loop symmetry of integrable vertex models at roots of unity}, Nucl. Phys. {\bf B618} (2001) 551--569 
{\tt arXiv:hep-th/0104120}. 

\bibitem[Kol12]{Kolb}
S. Kolb,  {\it Quantum symmetric Kac-Moody pairs}, Adv. Math. 267 (2014), 395-469, {\tt arXiv:1207.6036v1}.
%
\bibitem[Koo86]{Koo} T.H. Koornwinder, \textit{Group theoretic interpretations of Askey's scheme of hypergeometric orthogonal polynomials}, in: Orthogonal Polynomials and Their Applications (Segovia, 1986), Lecture Notes in Mathematics, vol. 1329, Springer, Berlin, 1988, pp. 46-72.
%
\bibitem[KS96]{KS} R. Koekoek, R.F. Swarttouw, \textit{The Askey scheme of hypergeometric orthogonal polynomials and its $q$-analog}, Report 98-17, Delft University of Technology, The Metherlands, 1998, {\tt arXiv:math/9602214v1}.

\bibitem[Lusz93]{Luszt}
G. Lusztig, {\it Introduction to Quantum Groups}, Birkhauser (1993).
%
\bibitem[Leo82]{Leo1} D.A. Leonard, \textit{Orthogonal polynomials, duality, and association schemes}, SIAM J. Math. Anal. 13 (4) (1982) 656-663.
%
\bibitem[Letz99]{Lez}
G. Letzter, {\it Coideal Subalgebras and Quantum Symmetric Pairs}, MSRI volume 1999, Hopf Algebra Workshop, {\tt arXiv:math/0103228}.
%


%
\bibitem[ND06]{ND}
 A. Nishino and T. Deguchi, {\it   The $L(sl_2)$ symmetry of the Bazhanov-Stroganov model associated with the superintegrable chiral Potts model}, Phys. Lett. {\bf A 356} (2006)  366--70, {\tt arXiv:cond-mat/0605551};\\
A. Nishino and T. Deguchi, , {\it An algebraic derivation of the eigens
paces associated with an Ising-like spectrum of the superintegrable chiral Potts model}, J. Stat. Phys.
{\bf 133} (2008) 587--615, {\tt arXiv:0806.1268}.
%
\bibitem[Nom05]{N} K. Nomura, \textit{Tridiagonal pairs and the Askey-Wilson relations}, Linear Algebra and its Applications 397 (2005) 99-106.
%
\bibitem[NT08]{NT} K. Nomura and P. Terwilliger, \textit{The structure of a tridiagonal pair}, Linear Algebra and its Applications 429 (2008), 1647-1662, {\tt arXiv:0802.1096v1}.
%
\bibitem[NT09]{NT:muqrac}
K.~Nomura and P.~Terwilliger, {\it Tridiagonal pairs of $q$-Racah type and the $\mu$-conjecture}, Linear Algebra Appl. Accepted;
{\tt arXiv:0908.3151}.
%
\bibitem[O44]{O} L. Onsager, \textit{Crystal statistics. I. A two-dimensional model with an order-disorder transition}, Phys. Rev. (2) 65 (1944) 117-149.
%
\bibitem[Perk87]{J.P} J.H.H. Perk, \textit{Star-triangle equations, quantum Lax pairs, and higher genus curves}, in: Theta Function-Bowdoin 1987, Part 1 (Brunswick, ME, 1987), AMS, Provindence, RI, 1989, pp. 341-354.
%
\bibitem[PS90]{PS}
V. Pasquier and H. Saleur, {\it Common structures between finite systems and conformal field theories through quantum groups}, Nucl. Phys. {\bf B 330} (1990) 523-556.
%
\bibitem[RS00]{RS}
Yu. G. Stroganov,  {\it The importance of being odd}, J. Phys. {\bf A 34} (2001), L179--L185, {\tt  arXiv:cond-mat/0012035};\\
A.V. Razumov and Yu. G. Stroganov, {\it Spin chains and combinatorics}, J. Phys. {\bf A 34} (2001), 3185-3190, {\tt  arXiv:cond-mat/0012141}.
%
\bibitem[Sk88]{Skly88}
E.K. Sklyanin, {\it Boundary conditions for integrable quantum systems}, J. Phys. {\bf A 21} (1988) 2375--2389.
%
\bibitem[Sk92]{Sk1} E.K. Sklyanin, \textit{Quantum Inverse Scattering Method. Selected Topics}, {\tt arXiv:hep-th/9211111}.
\bibitem[Ter99]{Ter03}
P. Terwilliger, {\it Two relations that generalize the $q-$Serre
relations and the Dolan-Grady relations}, Proceedings of the Nagoya 1999 International workshop on physics and combinatorics. Editors A. N. Kirillov, A. Tsuchiya, H. Umemura. pp 377--398, {\tt math.QA/0307016}.
%
\bibitem[Ter03]{Ter01} P. Terwilliger, \textit{Introduction to Leonard pairs}, Journal of Computational and Applied mathematics 153 (2003) 463-475.

\bibitem[Ter01]{T05} P. Terwilliger, \textit{Two linear transformations each tridiagonal with respect to an eigenbasis of
the other}, \textit{Linear Algebra Appl.} {330} (2001), 149-203, {\tt arXiv:math/0406555v1}.
%
\bibitem[Ter0306]{Ter06} P. Terwilliger, \textit{Leonard pairs and the $q$-Racah polynomials}, Linear Algebra Appli. 387 (2004) 235-276, {\tt arXiv:math/0306301v2}.
%
\bibitem[Ter92]{Terc1} P. Terwilliger, \textit{The subconstituent algebra of an association scheme. I}, J. Algebraic Combin. ,1(4):363-388,1992.
%
\bibitem[Ter93II]{Terc2} P. Terwilliger, \textit{The subconstituent algebra of an association scheme. II}, J. Algebraic Combin. ,2(1):73-103,1993.  
%
\bibitem[Ter93III]{PTIII} P. Terwilliger, \textit{The subconstituent algebra of an association scheme. III}, J. Algebraic Combin. 2(2) (1993) 177-210.
%
\bibitem[TV03]{TV} P. Terwilliger and R. Vidunas, \textit{Leonard pairs and the Askey-Wilson relations}, Discrete Mathematics 308 (2008), 479-495, {\tt arXiv:math/0305356v1}.
%
\bibitem[UI95]{Uglov}
D. Uglov and L. Ivanov, {\it $sl(N)$ Onsager's algebra and integrability}, J. Stat. Phys. {\bf 82} (1996) 87, {\tt arXiv:hep-th/9502068v1}.
%
\bibitem[Wan45]{Wan} G.H. Wannier, \textit{The statistical problem in cooperative phenomena}, Rev. Mod. Phys. \textbf{17} 50 (1945).
\bibitem[Zhe91]{Zhe} A. S. Zhedanov, \textit{``Hidden symmetry" of Askey-Wilson polynomials}, Teoret. Mat. Fiz. 89 (1991) 190-204.
\bibitem[ZK02]{ZK} A. S. Zhedanov, A. Korovnichenko, \textit{``Leonard pairs" in classical mechanics}, J. Phys. A 5 (2002) 5767-5780.

\end{thebibliography}
\end{document}